\documentclass[12pt]{article}

\usepackage{appendix}
\usepackage{amsmath}
\usepackage{changes}
\usepackage{setspace}
\usepackage{bbold}

\newcommand{\can}{\mathrm{can}}
\newcommand{\Prob}{\mathrm{Prob}}

\newcommand{\indicator}{\mathbb{1}}

\DeclareMathOperator*{\Aut}{Aut}

\DeclareMathOperator*{\argmax}{arg\,max}

\usepackage{amsthm}
\usepackage{amsmath}
\usepackage{amssymb}
\usepackage[shortlabels]{enumitem}
\usepackage{xcolor}
\usepackage{listings}
\usepackage{multirow}
\usepackage{color}
\usepackage[pdfencoding=auto, psdextra]{hyperref}
\usepackage{doi}
\usepackage{booktabs}
\usepackage{rotating}
\usepackage{bm}
\usepackage{tablefootnote}
\usepackage[normalem]{ulem}
\usepackage{array}
\usepackage{longtable}
\usepackage{bbm}
\usepackage{cleveref}
\usepackage{tabularx}
\usepackage{bbm}

\newcommand{\ed}{\mathrm{d}}

\newcommand{\id}{{\bf 1}}

\newcolumntype{L}[1]{>{\raggedright\let\newline\\\arraybackslash\hspace{0pt}}p{#1}}

\newtheorem{theorem}{Theorem}
\newtheorem{lemma}[theorem]{Lemma}
\newtheorem{remark}[theorem]{Remark}

\newtheorem{example}[theorem]{Example}
\newtheorem{corollary}[theorem]{Corollary}
\newtheorem{proposition}[theorem]{Proposition}

\theoremstyle{definition}
\newtheorem{definition}[theorem]{Definition}
\numberwithin{equation}{section}
\numberwithin{theorem}{section}

\title{An ergodic theorem for multi-period mutual insurance}
\author{John Armstrong}

\allowdisplaybreaks
\makeatletter
\renewenvironment{proof}[1][\proofname]{\par
  \pushQED{\qed}%
  \normalfont\topsep6\p@\@plus6\p@\relax
  \trivlist
  \item\relax
  {\itshape #1\@addpunct{.}}\enspace\ignorespaces
}{%
  \popQED\endtrivlist\@endpefalse
}
\makeatother

\begin{document}

\maketitle

\begin{abstract}
Suppose there are $N$ heterogeneous agents in a market with idiosyncratic risks but no uninsurable systematic risk factors. These agents may agree arbitrary financial contracts with one another, subject to the condition that contracts are self-enforcing under coalitions of agents in a common state. We show that, under mild conditions, this uniquely determines the limiting utility of every agent as $N$ tends to infinity. The result is an ergodic theorem: as the population grows, the number of degrees of freedom in the problem collapses, so that agents in the same state are treated identically in the limit. We exhibit an explicit, practically realisable mechanism achieving this limit using only short-dated contracts. The model can be applied either to an economy of heterogeneous agents pooling idiosyncratic risk through self-enforcing contracts or to the design of optimal insurance products such as pensions.
\end{abstract}

\section{Introduction}

Suppose that there are $N$ heterogeneous agents in a market with idiosyncratic risks but no uninsurable systematic risk factors. These agents may agree arbitrary financial contracts with one another, subject to the condition that contracts are self-enforcing under coalitions of agents in a common state. We will show that under mild conditions, this uniquely determines the limiting utility of each agent as $N$ tends to infinity.

Because the utility is uniquely determined it is trivially an optimal outcome from the perspective of every agent.

If the only idiosyncratic risk-factor is idiosyncratic mortality risk, this optimal utility can be achieved by a tontine structure. This implies that there is no potential for additional welfare benefit from more complex contracts. A generalised tontine structure can be used to achieve the optimal utility for arbitrary idiosyncratic risks.

The contracts agreed between agents can include arbitrary illiquid long-term insurance contracts that cannot be hedged in the market. We will see that the same optimal utility can be achieved if one restricts to short-dated ``myopic'' contracts that only look one period ahead. Indeed, the generalised tontine design is of this form.

Our results can equally be applied to obtain structural results about an economy of heterogeneous agents (who might be individuals, firms, insurers, or reinsurers) pooling idiosyncratic risk through self-enforcing contracts, or to understand the possible designs of pension schemes themselves. For convenience we use the word ``fund'' to denote the universe in which our agents operate, but we could equally have used the word ``insurance market''. Nothing in the argument depends on the pension interpretation specifically.

\medskip

Our result is an ``ergodic theorem'' because we show that as $N \to \infty$ the number of degrees of freedom in the problem reduces in a similar fashion to thermodynamic problems.

It is standard practice in insurance problems to model infinitely large systems where idiosyncratic risk has been eliminated. In such infinite models, one assumes that all agents in the same state should be treated equally. This assumption drastically reduces the number of degrees of freedom in the problem. This assumption is called the ergodic hypothesis. A central contribution of this paper is that we do not assume the ergodic hypothesis; we prove that it holds.

In a multi-period setting with long-dated contracts, one runs into the problem that, after the first period, agents will no longer have independent financial positions. 
This creates a tension: in order to establish the independence
of agents' financial positions in the limit, one seems to need a forward induction argument, but the key economic arguments describing the optimal behaviour of agents naturally use dynamic programming and hence backward induction.

There is no such tension when $N=\infty$: backward induction arguments suffice in this limit. We therefore would like some way of establishing the existence of an ergodic limit before we attempt to identify optimal strategies.

The key observation is that compactness of the problem is sufficient to establish the existence of an ergodic limit.  Intuitively, there is an inherent contradiction between a problem being compact and having an infinite number of identical degrees of freedom.  By combining the compactness results of Kramkov and Schachermayer \cite{kramkov1999asymptotic,kramkov2003necessary} with the infinite symmetries in our investment problem obtained by permuting members we can obtain our ergodic theorem. The key technical tool that allows one to formalise this heuristic argument is the classical Hewitt--Savage Theorem (\cite{hewittSavage}, stated as Theorem~\ref{thm:hewittSavage}).

\paragraph{Assumptions}

\begin{enumerate}[(i),series=assumptions]
\item \emph{Self-enforcing contracts under coalitions.} With probability 1, all agents should remain solvent. It should not be possible for a group of agents in the same state to obtain a better outcome by reneging on their contract and working together.

\item \emph{Fully insurable systematic risk.} Stated more technically, we will assume that our problem can be written in terms of a complete investment market and independent identically distributed idiosyncratic risks.

If there is uninsurable systematic risk, there can be welfare benefits in mutually insuring against this risk as is shown in \cite{armstrong_dalby} for the case of systematic longevity risk.

\item \emph{No altruism or envy.} Agents are indifferent to the outcomes experienced by other agents. If they were altruistic they may prefer contracts where they simply give away their wealth over the tontine mechanism. Similarly, if they are motivated by envy as in \cite{espinosa2015optimal}, our results will not apply.

\item \emph{No satiation.} An agent's utility must always lie strictly below its upper bound, so that additional wealth is always strictly valued. Without this assumption, agents who are satiated will be willing to give away their wealth.

\item \emph{Finite time horizon.} This eliminates Ponzi-type schemes. See the discussion in \cite{armstrong_dalby_donnelly} for a concrete example of how a fully funded DB scheme can under certain circumstances exploit future wage growth to provide members with an income whose present value is greater than the present value of their contributions.

\item \emph{Time consistency.} An agent with time-inconsistent preferences (in the sense of Kreps and Porteus \cite{krepsPorteus}) may make decisions as though they will become satiated even if the no-satiation property holds.

Although time consistency is a standard economic assumption, it is violated in Prospect Theory so this is a nontrivial restriction.

\item \emph{Monotonicity.} Everyone agrees increased wealth is good.
\end{enumerate}

We also make two additional assumptions which are primarily to simplify the exposition. Both of these assumptions could be weakened significantly, but not eliminated entirely. See Section~\ref{sec:extensions}
for a more detailed discussion of this point. These assumptions are only in force for Section~\ref{sec:finiteFunds}.
\begin{enumerate}[(i),resume=assumptions]
\item \emph{Markovian preferences.} Agents' decisions can depend on their personal history and the market history, but not on their consumption history.
\item \emph{Finite state space.} The set of possible states of an agent is finite.
\end{enumerate}

The problem must of course be well-posed, in the sense that non-trivial solutions exist, but we assume something stronger than this.

\begin{enumerate}[(i),resume=assumptions]
	\item \emph{Compactness.} Unlike the assumptions above, this is a technical rather than an economic condition. We use the word as a shorthand for the requirement that maximisers exist and that the problem has appropriate continuity properties; it is formalised in Section~\ref{sec:compactness}.
\end{enumerate}

\paragraph{Relationship to existing literature}
Borch's mutuality
principle \cite{borch1962,wilson1968} establishes, in a one-period
model, that it is Pareto optimal for a group of agents to pool their risks
so that each agent's outcome depends on their own state only through
the aggregate shock, with idiosyncratic risk pooled away entirely.
Borch's theorem holds for any number of agents,
but assumes concave utility functions and a shared probability model among agents.
What Borch's Theorem does not address is whether, or how,
such an arrangement could be reached or sustained by self-interested
agents who are free to walk away; it is derived on the
assumption that whatever allocation is agreed will simply be
honoured.

The same is true of an extensive applied literature on collective pension
investment designs in which membership of a scheme is compulsory and the
fund pursues a welfare goal set by a central planner, for example
\cite{gollier,cui,bonenkamp,haan}. There the question of walking away does
not arise. Such schemes lie outside the scope of our theorem and our
results in no way invalidate this body of work.

A separate literature
instead asks whether efficient risk sharing can emerge from
decentralised, self-enforcing behaviour. Classically, as a static exchange
economy is replicated, the core (the set of allocations that no
coalition could improve upon by trading only among its own members)
shrinks to the competitive equilibrium allocations
\cite{debreuScarf1963,aumann1964}: this establishes emergence from
self-interested, coalition-proof behaviour as the population grows,
but only in a single period, and the convergence is known to depend
delicately on the convexity of preferences, weakening substantially
in its absence. Becker and Chakrabarti's recursive core
\cite{beckerChakrabarti1995} extends coalitional stability
to a genuinely dynamic, multi-period setting.

The topic of ergodic limits is an area of active research,
particularly the substantial literature on the convergence of 
large-population dynamic games to a mean-field
limit, known in this context as propagation of chaos.
Different routes are used to
prove such results: alongside the classical coupling method of
\cite{sznitman1991}, a compactness-and-exchangeability argument
closely analogous to our own has already been used by
\cite{fischer2017,lacker2016} to identify limits of $N$-player Nash
equilibria. That literature, however, typically relies on a
convexity condition on the population's cost (the Lasry--Lions
monotonicity condition \cite{cardaliaguet2019}) to guarantee that
even the value of the game is unique, since without it distinct
equilibria can carry genuinely different payoffs, not merely
different strategies achieving a common payoff. We instead obtain
uniqueness of the value directly from the self-enforcement condition
itself, via the no-persistent-gap argument of
Proposition~\ref{prop:noPersistentGap}, without assuming any
convexity.

The ``tontine strategy'' we propose has ancient origins. A review of the long history of tontines and the more recent literature can be found in \cite{milevsky}. The broad tontine design we describe in Section~\ref{sec:tontineStrategies} is based on the discrete-time tontine
to insure against single-life mortality risk described in \cite{armstrong_buescu_dalby}. A continuous time-version of the same mechanism for separating market and idiosyncratic risks was described in \cite{donelly_bernhardt}.
Our design generalises this idea to arbitrary idiosyncratic risks.

Our results are essentially structural, describing the form of any system of self-enforcing contracts but without computing the strategy in detail. We do however observe that this task can be completed using the techniques of optimal stochastic control. The classic reference for the optimal investment problem without idiosyncratic risk is \cite{merton1969lifetime}. Stamos shows in \cite{stamos} how to adapt this to solve the optimal-investment problem for a single-life with mortality and a longevity credit. The discrete-time problem that corresponds exactly to the methodology in this paper and the economic setup of Example~\ref{example:blackScholes} is solved in \cite{armstrong_buescu_dalby} for Epstein--Zin preferences with mortality.

\medskip

Our paper makes the following contributions to the literature.

We show that, as the number of participants grows, the value
of this problem converges to a single limiting value. This provides
a rigorous justification for the standard practice in the actuarial literature
of assuming the ergodic hypothesis even in multi-period problems.

We exhibit an explicit and practically realisable mechanism
achieving the same outcome using only short-dated contracts: the
generalised tontine of Section~\ref{sec:finiteFunds}.

We show that even if the contract space is enlarged
to allow arbitrary long-dated and illiquid contracts
between any number of agents, the requirement of coalition stability 
prevents this bringing any welfare
benefit in the large-population limit. This result appears to be new, even
if one assumes the ergodic hypothesis. The classical core literature has no multi-period
contract structure at all; Becker and Chakrabarti's recursive
core is implemented entirely through liquid, spot transfers of
capital. 

Even restricted to a single period, our results offer a
modest improvement on classical arguments: we require
neither a common probability model of market dynamics across agents,
nor concave utility functions.

\medskip

Of particular practical significance is the special case of idiosyncratic mortality risk,
where these results establish the optimal design for a collective pension arrangement.
This design shares single-life longevity risk through a tontine that shields members
from the market risks taken by other members.

In the UK, pension legislation requires that scheme members should be able to transfer out of a scheme
at any time \cite{pensionschemesact2015}. This naturally leads to a model centred on coalitional stability. Novel CDC pension designs
have been proposed for the UK market which attempt to diversify market risk across cohorts \cite{pensionSchemesBill2015}. Our results show
that this is incompatible with coalitional stability.

\paragraph{Proof outline}

The proof will proceed by considering a number of investment problems with differing constraints which we can classify with the following labels:
\begin{itemize}
	\item $N_0$ labels the infinite-fund case, which is constrained by treating all agents in same state in the same way. $N_1$ labels the finite-fund case.
	\item $H_0$ the homogeneous case where insurance is only shared between agents in the same state, $H_1$ the inhomogeneous case.
	\item $T_0$ where contracts are short-dated, $T_1$ where long-term contracts are allowed.
\end{itemize}
The eventual goal is to prove that the value function for the problem $N_0 H_0 T_0$
is equal to that of $N_1 H_1 T_1$. We will write this goal symbolically as $N_0 H_0 T_0=N_1 H_1 T_1$.

We begin by proving that $N_0 H_0 T_0 = N_0 H_0 T_1$. In other words the flexibility of long-dated contracts adds no benefit for homogeneous infinite funds. This uses time-consistency and relies on establishing a dynamic programming principle.

Using this we then prove $N_0 H_1 T_1 = N_0 H_0 T_0$ directly, without a further appeal to dynamic programming: this simultaneously shows that long-dated contracts add no benefit for inhomogeneous funds either ($N_0 H_1 T_0 = N_0 H_1 T_1$) and that there is no benefit from contracting across states. This relies on the positivity of the measure and can be viewed as a multi-agent no arbitrage argument. So we have that all the $N_0 H_* T_*$ are equal (the symbol $*$ is just a wildcard in this informal discussion).

We then show that $N_1 H_1 T_1 \geq N_0 H_* T_*$ using a forward induction and the law of large numbers.

We then perform a backward induction in time.

As induction hypothesis we have that $N_1 H_1 T_1 \geq N_0 H_* T_*$  at all later times. This allows us to simplify the problem after the first time step, so we only need to consider the problem $N_1 H_1 T_0$ when proving the induction step.

To prove the induction step, we use the compactness and symmetry argument to show that any sequence of solutions to the finite problems for $N_1 H_1 T_1$ has a limit point which is a solution to $N_0 H_* T_*$. From this we find  $N_1 H_1 T_1 \leq N_0 H_* T_*$ so we must have equality throughout, including in the budget constraint for the limiting strategy.

\medskip

This sets the agenda for our proof. We must define any problem that we use in the proof (we don't use all 8 possible combinations fortunately) and establish that it is well-posed. This is inevitably somewhat time consuming, but the same compactness property can be used to establish well-posedness throughout in a routine manner.

The most challenging analytical issue is to establish that compactness holds for a broad range of problems. The techniques of Kramkov and Schachermayer provide exactly the tools we need to do this, but unfortunately we are not aware of a source that provides precisely the results we need in our dynamic setting so we must run through the same arguments by hand. The goal of Section~\ref{sec:compactness} is to package the arguments of Kramkov and Schachermayer by defining a formal notion of a compact problem that provides precisely the details needed for the rest of the paper. This isolates the analytical complexity from the economic argument.

\paragraph{Acknowledgements}

This work was funded in part by Nuffield grant FR-000024058.

\section{Modelling with product spaces}
\label{sec:notation} 

The classification of one-period financial markets in \cite{armstrong_classifying_markets}
is greatly simplified by introducing the notion of a ``casino''. This is a complete, atomless probability space on which arbitrary bets can be made, priced by their expectation. If one enriches an existing financial market model by adding a casino, one eliminates the presence of atomic events. Thus if one only classifies markets up to enrichment by a casino, one can ignore the atoms in any probability space leading to a cleaner classification theory.

The principal goal of this section is
to extend the idea of enrichment to multi-period settings and settings with more complex information structures.

In market models with an unambiguous physical probability measure, one can understand preferences using the push-forward of this measure onto the space of consumption streams. Hence, the models of von Neumann--Morgenstern and Kreps--Porteus can be framed in terms of choices between lotteries. Thus comparing lotteries allows one to compare different probability spaces. With no push-forward, 
we need an alternative mechanism to compare probability spaces and classification results up to enrichment will provide precisely the tool we need.

In mathematical terms, the enrichment occurs using product probability spaces. Product probability spaces are also the natural tool to model idiosyncratic risk and symmetries of probability spaces.

We therefore begin by establishing a clear notation for product spaces and the various $\sigma$-algebras that will occur during our proofs. The notation may appear pedantic at first, but it will ultimately pay dividends by allowing us to state complex measurability and independence relationships succinctly.

We then state the classification results we need on the stable equivalence of probability spaces under products. The measure theory underlying these results is standard, but the notion of stable isomorphism provides a particularly convenient packaging.

We then discuss the topics of group actions on probability spaces
and the Hewitt--Savage theorem. These sections also exploit product structures and play a crucial role in the symmetry arguments at the centre of our proof.

\subsection{Product spaces}

Let ${\cal T} = \{0, 1, 2, \ldots, T \}$ be a finite set of time points.
We write ${\cal T}_{\geq s}$ for those points in ${\cal T}$ which are greater than or
equal to $s$, and define ${\cal T}_{\leq s}$, ${\cal T}_{>s}$ and ${\cal T}_{<s}$ similarly.

Let $V$ denote a finite-dimensional real normed vector space.
Given a ${\cal T}^\prime \subseteq {\cal T}$
and a filtered probability space $(\Omega, {\cal F}, ({\cal F}_t)_{t \in {\cal T^\prime}}, {\mathbb Q})$, we will write 
$L^p( V, {\cal F}_{s}, {\mathbb Q})$
for the space of $p$-integrable, $V$-valued, ${\cal F}_s$-measurable random variables, where $s \in {\cal T}^\prime$. 
We will write
\[
L^p( V, ({\cal F}_{s})_{s \in {\cal T}^\prime}, {\mathbb Q}):=\prod_{s \in {\cal T}^\prime} L^p( V, {\cal F}_{s}, {\mathbb Q})
\]
for the corresponding space of $p$-integrable, $V$-valued adapted processes.
Since the measure-class on our probability spaces will always be unambiguous we write
simply
\[
L^0( V, ({\cal F}_{s})_{s \in {\cal T}^\prime})
\]
for the space of $V$-valued adapted processes up to equivalence.

Given filtered probability spaces
$\underline{\Omega}^\alpha:=(\Omega^\alpha, {\cal F}^\alpha, ({\cal F}^\alpha_t)_{t \in {\cal T^\prime}}, {\mathbb Q}^\alpha)$ indexed by $\alpha \in A$, we
will write 
\[
\underline{\Omega}^A:=
\prod_{\alpha \in A}\underline{\Omega}^\alpha:= \left(
\Omega^A,
{\cal F}^{A}, ({\cal F}^{A}_t)_{t \in {\cal T^\prime}}, {\mathbb Q}^A \right)
\]
where  $\Omega^A := \prod_{\alpha \in A} \Omega^\alpha$, ${\mathbb Q}^A:=\prod_{\alpha \in A}{\mathbb Q}^\alpha$. ${\cal F}_t^{A}$ is equal to the completion of the product $\bigotimes_{\alpha \in A}{\cal F}^\alpha_t$ under ${\mathbb Q}^A$ and
${\cal F}^{A}$
is the completion of the product
$\bigotimes_{\alpha\in A}{\cal F}^{\alpha}$
under ${\mathbb Q}^A$.
Given an element $\omega \in \Omega^A$, we will write $\omega^\alpha$ for the component of $\omega$ corresponding to a given $\alpha \in A$.

Given a subset $A^\prime \subseteq A$,
we will write
\begin{equation}
\iota^{\Omega^A}
\left( \prod_{\alpha \in A^\prime} {\cal F}_t^\alpha \right)
= \overline{ \bigotimes_{\alpha \in A} \begin{cases}
\{ \emptyset, \Omega^\alpha \} & \alpha \notin A^\prime \\
{\cal F}^\alpha_t & \alpha \in A^\prime
\end{cases}
}^{\,{\mathbb Q}^A}
\label{eq:defIota}
\end{equation}
where the overline denotes the completion of the $\sigma$-algebra with respect to the given measure. Note that the term on the left-hand side in equation~\eqref{eq:defIota} should be understood purely symbolically rather than as a function applying to an argument. Our notation is intended simply as a shorthand for the term on the right.

When $A=A^\prime$ we use the abbreviation
\begin{equation}
\iota
\left( \prod_{\alpha \in A} {\cal F}_t^\alpha \right)
:=
\iota^{\Omega^A}
\left( \prod_{\alpha \in A} {\cal F}_t^\alpha \right).
\end{equation}

This notation allows us to define $\sigma$-algebras containing only certain
parts of the information available in a product probability space. It is also useful to define the complementary $\sigma$-algebra
\begin{equation}
\iota^{\perp,\Omega^A}
\left( \prod_{\alpha \in A^\prime} {\cal F}_t^\alpha \right)
:= \iota^{\Omega^A}
\left( \prod_{\alpha \in A \setminus A^\prime} {\cal F}_t^\alpha \right).
\label{eq:defIotaPerp}
\end{equation}
Under the product measure ${\mathbb Q}^A$, this complementary
$\sigma$-algebra is independent of
$
\iota^{\Omega^A}
\left( \prod_{\alpha \in A^\prime} {\cal F}_t^\alpha \right)
$.

\subsection{Stable classification of probability spaces}

Let ${\cal B}([0,1])$ be the Borel $\sigma$-algebra on $[0,1]$ 
and let 
$\lambda$ be the Lebesgue measure. We define 
\[
\underline{\Omega}^\can:=([0,1],\overline{{\cal B}([0,1])}^{\, \lambda},\lambda)
\]
Whenever we write an expression such as this to define a probability space, we are implicitly defining the corresponding components of the probability space. In this case
\[
(\Omega^{\can},{\cal F}^{\can}, {\mathbb Q}^\can)=([0,1],\overline{{\cal B}([0,1])}^{\, \lambda},\lambda).
\]
We define 
\[
\underline{\Omega}^\can_{\cal T}:=\prod_{t \in {\cal T}} \underline{\Omega}^{\can,t}
\]
where each $\underline{\Omega}^{\can,t}$ is a copy of $\underline{\Omega}^{\can}$. We equip $\underline{\Omega}^\can_{\cal T}$ with the filtration $({\cal F}^{\can}_{t})_{t \in {\cal T}}$ defined by
\[
{\cal F}^{\can}_{t}:=\iota^{\Omega^{\can}_{\cal T}}\left(
\prod_{s \in {\cal T}_{\leq t}} {\cal F}^{\can,s}\right).
\]

\begin{definition}
A {\em mod 0 isomorphism} between two probability spaces, $\Omega^1$, $\Omega^2$
is a measure-preserving, measurable bijection $\phi:\Omega^1 \setminus N^1 \to \Omega^2 \setminus N^2$ where $N^1$ and $N^2$ are null subsets of $\Omega^1$ and $\Omega^2$ respectively which also has a measurable, measure-preserving inverse.
\end{definition}

\begin{definition}
A probability space $\underline{\Omega}$ 
is called {\em standard} if it is mod $0$ isomorphic to a Borel
probability measure on a complete separable metric space.
\end{definition}

In practice, all the complete probability spaces arising in finance are standard: for example countable disjoint sums, countable products, and completions of these, of standard probability spaces are standard, as is the probability space generated by Brownian motion. Non-standard complete probability spaces are essentially pathological from an applied viewpoint.

\begin{theorem}[Rokhlin's Theorem \cite{rokhlin}]
Every standard probability space is mod $0$ isomorphic to the disjoint
union of a countable collection of atoms and an atomless component of
mass $a\in[0,1]$. The atomless component is uniquely determined up to
mod $0$ isomorphism and is isomorphic to the interval $[0,a]$
equipped with Lebesgue measure.
\end{theorem}

In particular, the only atomless standard probability space up to mod 0 isomorphism is $\underline{\Omega}^\can$.

We will not be interested in the atoms in probability spaces. This motivates the following definition.

Let $\underline{\Omega}^A$ be a filtered probability space with time points ${\cal T}$. Define ${\underline{\Omega}}^{A+}:=\underline{\Omega}^{\can}_{\cal T} \times \underline{\Omega}^A$. Write $\pi^{\Omega^A}:\underline{\Omega}^{A+} \to \underline{\Omega}^{A}$ for the projection onto the second factor.

\begin{definition}
We say that $\underline{\Omega}^A$ is {\em stably equivalent} to $\underline{\Omega}^B$ if $\underline{\Omega}^{A+} \cong \underline{\Omega}^{B+}$ by a filtration preserving mod 0 isomorphism.
\end{definition}

\begin{definition}
A filtered probability space $(\Omega,{\cal F},({\cal F}_t)_{t\in{\cal T}},{\mathbb Q})$ is a {\em standard filtered probability space} if $(\Omega,{\cal F},{\mathbb Q})$ is standard, each ${\cal F}_t$ is complete, and ${\cal F}_T={\cal F}$ modulo null sets.
\end{definition}
The final condition costs no generality: only random variables adapted to the filtration will play a role in our arguments, so one may always replace ${\cal F}$ by ${\cal F}_T$.

\begin{lemma}
Let $\underline{\Omega}$ be a standard filtered probability space, then $\underline{\Omega}$ is stably equivalent to $\underline{\Omega}^{\can}_{\cal T}$.
\label{lemma:stableEquivalence}
\end{lemma}
We prove this in Appendix~\ref{appendix:finiteFiltrationClassification}.

The last result is useful because one typically is not interested in a probability space directly, but only in the random variables that it supports. Given a filtered probability space $\underline{\Omega}$, all the random variables defined on $\underline{\Omega}$ are naturally included in the set of random variables defined on $\underline{\Omega} \times \underline{\Omega}^{\can}_{\cal T}$: one can simply choose to ignore the  
$\underline{\Omega}^{\can}_{\cal T}$ component of the state. As a result, we will be able to safely assume that the state of our agents is modelled by an element of a standard-filtered space.

\medskip

\begin{definition}
Let $\underline{\Omega}^1$ be a standard filtered probability space.
For each $s\in{\cal T}$, choose a standard Borel realization
\[
\underline{\Omega}^1_s
=
(\Omega^1_s,
{\cal B}(\Omega^1_s),
{\mathbb Q}^1_s)
\]
of the probability space
\[
(\Omega^1,{\cal F}^1_s,{\mathbb Q}^1).
\]
We write
\[
\omega^1_s:\Omega^1\to\Omega^1_s
\]
for the corresponding coordinate map. Thus
\[
{\cal F}^1_s
=
(\omega^1_s)^{-1}
{\cal B}(\Omega^1_s)
\]
modulo null sets.
\end{definition}

The importance of standard probability spaces is that they admit
regular conditional probabilities.

In particular, for each
$s\in{\cal T}$ there exists a measurable family of probability
measures
\[
(\mu_x)_{x\in\Omega^1_s}
\]
on $\Omega^1$ such that
\[
{\mathbb Q}^1(A)
=
\int_{\Omega^1_s}
\mu_x(A)\,
{\mathbb Q}^1_s(dx)
\]
for every measurable set $A$, and
\[
\mu_x\!\left(
(\omega^1_s)^{-1}(\{x\})
\right)
=1
\]
for ${\mathbb Q}^1_s$-almost every $x\in\Omega^1_s$.

The measure $\mu_x$ may be interpreted as the probability measure
obtained by conditioning on the event that the time-$s$ state is
equal to $x$.

As a consequence, any random variable
\[
X_t\in L^0({\mathbb R},{\cal F}^1_t,{\mathbb Q}^1)
\]
admits a conditional realization
\[
X_t(x)
:=
\int_{\Omega^1}
X_t(\omega)\,
\mu_x(d\omega),
\]
defined for ${\mathbb Q}^1_s$-almost every
$x\in\Omega^1_s$ whenever the integral exists.

Similarly, conditional expectations can be written as functions of
the state variable:
\[
{\mathbb E}(X_t\mid {\cal F}^1_s)(\omega)
=
X_t(\omega^1_s(\omega))
\]
for a suitable measurable function
\[
x\mapsto X_t(x).
\]

In this paper we will routinely use the notation $X_t(x)$ to denote
the value of a random quantity conditional on the time-$s$ state
being equal to $x$. The underlying theory is that of measure
disintegration and regular conditional probabilities. We refer the
reader to Chang and Pollard \cite{changPollard} or Fremlin
\cite{fremlin41} for further details.

We will later need to identify the conditional state and conditional
law on one standard filtered probability space with those on another,
transported along a filtration-preserving mod 0 isomorphism between
them. The following lemma records that the disintegration described
above is natural with respect to such isomorphisms.

\begin{lemma}[Naturality of disintegration]
\label{lemma:disintegrationNatural}
Let $\underline{\Omega}^1$, $\underline{\Omega}^2$ be standard filtered
probability spaces with time set ${\cal T}$, and let
$\Theta:\underline{\Omega}^1\to\underline{\Omega}^2$ be a
filtration-preserving mod 0 isomorphism. Fix $s\in{\cal T}$ and choose
standard Borel realizations $\Omega^1_s$, $\Omega^2_s$ of
$(\Omega^1,{\cal F}^1_s,{\mathbb Q}^1)$,
$(\Omega^2,{\cal F}^2_s,{\mathbb Q}^2)$ with coordinate maps
$\omega^1_s,\omega^2_s$ and disintegrations
$(\mu^1_x)_{x\in\Omega^1_s}$, $(\mu^2_y)_{y\in\Omega^2_s}$ as above.

Then there is a mod 0 isomorphism
$\Theta_s:\Omega^1_s\to\Omega^2_s$ of the Borel probability spaces
$(\Omega^1_s,{\mathbb Q}^1_s)$ and $(\Omega^2_s,{\mathbb Q}^2_s)$,
unique up to modification on a ${\mathbb Q}^1_s$-null set, such that
\[
\omega^2_s\circ\Theta=\Theta_s\circ\omega^1_s
\qquad
{\mathbb Q}^1\text{-a.s.},
\]
and the disintegration fibres correspond under $\Theta$: for
${\mathbb Q}^1_s$-almost every $x\in\Omega^1_s$,
\[
\Theta_*\mu^1_x=\mu^2_{\Theta_s(x)},
\]
where $\Theta_*\mu^1_x$ denotes the pushforward of $\mu^1_x$ under
$\Theta$. We call $\Theta_s$ the map {\em induced by $\Theta$ on the
time-$s$ disintegration}, and write $\Theta_*$ for $\Theta_s$ when
$s$ is clear from context.
\end{lemma}
A proof is given in Appendix~\ref{appendix:finiteFiltrationClassification}.

In time-dependent problems it is natural to consider stochastic
automorphisms. This motivates our next definition.

\begin{definition}[Adapted randomized automorphism]
\label{def:adaptedRandomizedAutomorphism}
Let $\underline{\Omega}^1$ be a standard filtered probability
space, let $s\in{\cal T}$, and write
$({\cal F}^{1+}_t)_{t\in{\cal T}}$ for the filtration of
$\underline{\Omega}^{1+}$, where ${\cal F}^{1+}_t$ is the
completion of ${\cal F}^{\can}_t\times{\cal F}^1_t$. A jointly
measurable map
\[
\Phi:
\Omega^{\can}_{\cal T}
\times
\Omega^1
\to
\Omega^{\can}_{\cal T}
\]
is an {\em adapted randomized automorphism at time $s$} if
$\Phi(\,\cdot\,,\omega^1)$ is a filtration-preserving mod $0$
automorphism of $\underline{\Omega}^{\can}_{\cal T}$ for every
$\omega^1\in\Omega^1$, and, for every $t\in{\cal T}$, the time-$t$
coordinates of
$(\omega^{\can},\omega^1)\mapsto\Phi(\omega^{\can},\omega^1)$ and of
the map assembled from the fibrewise inverses are measurable with
respect to ${\cal F}^{1+}_{t\vee s}$. We write
\[
\hat\Phi(\omega^{\can},\omega^1)
:=
(\Phi(\omega^{\can},\omega^1),\omega^1).
\]
\end{definition}

Probability spaces routinely arise carrying more information than the
problem at hand requires. We record the relevant notion.

\begin{definition}[Enrichment]
\label{def:enrichment}
Let $\underline{\Omega}^1$ and $\underline{\Omega}^3$ be standard
filtered probability spaces with the same time set ${\cal T}$. An
{\em enrichment} of $\underline{\Omega}^1$ by $\underline{\Omega}^3$ is
a measure-preserving map
$\pi:\underline{\Omega}^3\to\underline{\Omega}^1$ such that, for every
$t\in{\cal T}$, $\pi^{-1}({\cal F}^1_t)\subseteq{\cal F}^3_t$ modulo null
sets and ${\cal F}^3_t$ is conditionally independent of
$\pi^{-1}({\cal F}^1_T)$ given $\pi^{-1}({\cal F}^1_t)$.
\end{definition}

The first condition says that $\underline{\Omega}^3$ has resolved by time
$t$ everything $\underline{\Omega}^1$ has; the second that it has learned
nothing further about the base, so that the adjoined information does not
anticipate $\underline{\Omega}^1$. The second condition is the
conditional-independence formulation of hypothesis $({\cal H})$ of
Br\'emaud and Yor \cite{bremaudYor}: the filtration
$(\pi^{-1}({\cal F}^1_t))_{t\in{\cal T}}$ is immersed in
$({\cal F}^3_t)_{t\in{\cal T}}$, every martingale of the base remaining
a martingale of the enrichment. Both hold whenever
$\underline{\Omega}^3$ is a product of $\underline{\Omega}^1$ with a
further filtered space. More general enrichments are possible, but up to
stable equivalence over $\underline{\Omega}^1$ every enrichment is of
this form.

\begin{lemma}[Stable equivalence of enrichments]
\label{lemma:enrichmentStable}
Let $\pi:\underline{\Omega}^3\to\underline{\Omega}^1$ be an enrichment
(Definition~\ref{def:enrichment}), and write
$p:=\pi\circ\pi^{\Omega^3}:\underline{\Omega}^{3+}\to
\underline{\Omega}^1$ and ${\cal B}_t:=p^{-1}({\cal F}^1_t)$.
\begin{enumerate}[(i)]
\item There is a filtration-preserving mod $0$ isomorphism
\[
\Xi:\underline{\Omega}^{1+}\to\underline{\Omega}^{3+}
\qquad\text{with}\qquad
p\circ\Xi=\pi^{\Omega^1};
\]
that is, $\underline{\Omega}^{3+}$ and $\underline{\Omega}^{1+}$ are
isomorphic over $\underline{\Omega}^1$.
\item If $\Xi'$ is a second such isomorphism then
$\Xi^{-1}\circ\Xi'=\hat\Phi$, where $\Phi$ is an adapted randomized
automorphism at time $0$
(Definition~\ref{def:adaptedRandomizedAutomorphism}), and hence at
every $s\in{\cal T}$.
\item Let $u\in{\cal T}$ and let ${\cal G}$ be a complete
$\sigma$-algebra with
\[
{\cal B}_u
\subseteq
{\cal G}
\subseteq
{\cal F}^3_u.
\]
Then every ${\cal G}$-state $y$ determines a time-$u$ state
$\pi_u(y)\in\Omega^1_u$, the map $p$ pushes ${\mathbb Q}^{3+}_y$
forward to ${\mathbb Q}^1_{\pi_u(y)}$, and there are isomorphisms
$\Xi_y$ as in~(i) between $\underline{\Omega}^{1+}$ conditioned on the
$(\pi^{\Omega^1})^{-1}({\cal F}^1_u)$-state $\pi_u(y)$ and
$\underline{\Omega}^{3+}$ conditioned on the ${\cal G}$-state $y$,
which may be chosen jointly measurably in $y$.
\end{enumerate}
\end{lemma}
Lemma~\ref{lemma:stableEquivalence} is the case in which
$\underline{\Omega}^1$ is trivial.

The two inclusions in~(iii) say respectively that the conditioning
state determines where one is in $\underline{\Omega}^1$, and that it
leaves the canonical factor of the lift free. A proof is given in
Appendix~\ref{appendix:finiteFiltrationClassification}.

\subsection{Group actions on probability spaces}

Let $G$ be a group.
We write $\id_G$ for the identity element of $G$. More generally, for
a set or space $X$ we write $\id_X$ for the identity map on $X$.

Define \(\Aut_{\mathbb Q}(\underline{\Omega})\) to be the group of
measure-preserving
invertible measurable transformations
\(T:\Omega\to\Omega\),
modulo \({\mathbb Q}\)-almost sure equality.

\begin{definition}
A {\em measure-preserving $G$-action} is a homomorphism
$\rho:G\to \Aut_{\mathbb Q}(\underline{\Omega})$.
\end{definition}

For brevity, we write \(g\omega\) for a representative of the
automorphism \(\rho(g)\).
All equalities involving \(g\omega\) are understood modulo null sets. 

The action of $G$ on $\Omega$ induces a covariant action on mod 0 equivalence classes of subsets $A \subseteq \Omega$ defined by $g[ A]=[\{g \cdot a \mid a \in A\}]$, where $[A]$ denotes the equivalence class of $A$. It induces a covariant action on complete $\sigma$-algebras ${\cal H}$ by
$g {\cal H} = \overline{\sigma(\{ gA \mid A \in {\cal H} \})}^{\, {\mathbb Q}}$.

It induces a contravariant action on random variables $X \in L^0(V,{\cal F},{\mathbb Q})$ by $g X = X \circ g^{-1}$. It also induces a contravariant action on measures equivalent to ${\mathbb Q}$ by
\[
(g {\mathbb P})(A) = {\mathbb P}(g^{-1}A).
\]

These definitions ensure that $g$ preserves the scalar values in the sense that
\[
(g X)(g \omega) = X( \omega) \qquad \text{ a.s.},
\]
and
\begin{equation}
(g \mathbb P)(g A) = {\mathbb P}(A).
\label{eq:transformOfMeasure}
\end{equation}

\medskip

Let $\rho_i$ be measure-preserving actions of a group $G$ on probability spaces $\underline{\Omega}^i$ for indices $i \in I=\{1,\ldots,n\}$. 
Given
$(\alpha_i)_{i=1\ldots n}$ with $\alpha_i \in \{0,1\}$ we define
\[
\rho^{\alpha_1, \alpha_2, \ldots, \alpha_n} : G \to \Aut_{[{\mathbb Q}^I]}\left(\underline{\Omega}^I\right)
\]
by defining
\[
\rho^{1}(g,\omega^i):=\rho_i(g)\omega^i
\]
for $\omega^i \in \Omega^i$ and
\[
\rho^{0}(g,\omega^i):=\omega^i.
\]
We then define
\[
\rho^{\alpha_1,\ldots,\alpha_n}(g)
(\omega^1,\ldots,\omega^n)
:=
(\rho^{\alpha_1}(g,\omega^1),
 \ldots,
 \rho^{\alpha_n}(g,\omega^n)).
 \]
It is immediate that
$\rho^{\alpha_1,\ldots,\alpha_n}$
defines a measure-preserving $G$-action on
$\underline{\Omega}^I$.

In the event that $G$ is a finite group, we write $\underline{G}$ for the discrete probability space on $G$ with the uniform measure.
If a component $\underline{\Omega}^i=\underline{G}$, we take the standard left action $\rho^1(g)h=\rho_1(g)h:=gh$. We define $\rho^{-1}(g)h:=hg^{-1}$, which is also a measure-preserving $G$ action as is verified by the calculation
\[
\rho^{-1}(g_1)\rho^{-1}(g_2)h=\rho^{-1}(g_1)(hg_2^{-1})=hg_2^{-1}g_1^{-1}=h(g_1g_2)^{-1}=\rho^{-1}(g_1g_2)h.
\]
This allows us to extend our definition of 
$
\rho^{\alpha_1, \alpha_2, \ldots, \alpha_n}
$ to elements $\alpha_i \in \{-1,0,1\}$ whenever the component $\underline{\Omega}^i$ equals $\underline{G}$. A component $\underline{\Omega}^i$ need not be equipped with a
$G$-action provided that $\alpha_i=0$.
In this case
$\rho^{\alpha_1,\ldots,\alpha_n}$
acts trivially on the \(i\)-th component.

We now relate such actions to each other.
Let $\rho$ be a measure-preserving $G$-action on $\underline{\Omega}$ for a finite group $G$. Consider the product space $\underline{G} \times \underline{G} \times \underline{\Omega}$.
Given $(\alpha,\beta) \in \{-1,0,1\}$ define a map
\[
\phi^{\alpha,\beta}: G \times G \times \Omega 
\to  G \times G \times \Omega 
\]
by
\[
\phi^{\alpha, \beta}(g_1,g_2,\omega)=(\rho^\alpha(g_2)g_1,\rho^\beta(g_1)g_2,\omega).
\]
These maps can be used to provide isomorphisms between different actions on the product.
We use the term isomorphism rather than the ergodic-theoretic conjugacy, to match the algebraic viewpoint adopted in this paper.

We will only need one such isomorphism, but the more general construction motivates our notation. The result we need is given by the next lemma.
\begin{lemma}
The map $\phi^{-1,0}$ defines an isomorphism between the action $\rho^{0,1,0}$ and the action $\rho^{-1,1,0}$. That is to say
\[
\phi^{-1,0} \circ \rho^{0,1,0}=\rho^{-1,1,0} \circ (\id_G \times \phi^{-1,0}).
\]
\label{lemma:isomorphism}
\end{lemma}
\begin{proof}
By abuse of notation we identify an action
\(\rho:G\to\Aut_{[{\mathbb Q}]}(\underline{\Omega})\)
with the associated map
\((g,\omega)\mapsto\rho(g)\omega\).
For \(h\in G\) and \((g_1,g_2,x)\in G\times G\times\Omega\),
\[
\begin{split}
\phi^{-1,0} \rho^{0,1,0} (h,g_1,g_2, x)
&=\phi^{-1,0}( g_1, hg_2, x) \\
&=( g_1 g_2^{-1} h^{-1}, hg_2, x) \\
&=\rho^{-1,1,0}(h,  g_1 g_2^{-1}, g_2,  x) \\
&=\rho^{-1,1,0}(\id_G \times \phi^{-1,0})( h, g_1, g_2, x )
\end{split}
\]
Finally, $\phi^{-1,0}$ is a mod 0 isomorphism: it acts on the finite uniform factor $G \times G$ by the bijection $(g_1,g_2)\mapsto(g_1g_2^{-1},g_2)$ and on $\Omega$ by the identity, so it is a measure-preserving bijection.
\end{proof}

All free measure-preserving $G$-actions of a finite group $G$ on the standard probability space $\underline{\Omega}^\can$ are isomorphic to a product, as is shown by the next Lemma. The underlying fact is standard (finite Borel equivalence relations admit Borel transversals, by the Luzin--Novikov theorem \cite[Theorem~18.10]{kechrisCDST}) but we record an explicit construction because later arguments use the resulting maps.

\begin{lemma}
\label{lemma:sectionExists}
Let $G$ act freely on $\underline{\Omega}^\can$ by a measure-preserving action $\rho$. 
Let $\underline{G}$ be the uniform discrete probability space
with elements $G$.
Let $\pi:\underline{\Omega}^\can\to \underline{\Omega}^\can/ G$ be the projection.
Then there exists a measurable right inverse $\xi$ for $\pi$.
Define $\zeta:\underline{\Omega}^\can \to \underline{G}$ by
\[
\rho(\zeta(\omega), \xi\pi(\omega))=\omega.
\]
Then the map $\Xi:\underline{G} \times (\underline{\Omega}^\can/G) \to \underline{\Omega}^\can$
defined by
\[
\Xi( g, \delta )=\rho(g,\xi(\delta))
\]
is a mod 0 isomorphism which maps the action $\rho^{1,0}$
to the action $\rho$ and has inverse
\[
\Xi^{-1}( \omega ) = (\zeta(\omega),\pi(\omega)).
\]
\end{lemma}
\begin{proof}
Recall that $\underline{\Omega}^{\can}$ is simply $[0,1]$ equipped with the Lebesgue measure. 
Define
\[
\bar{\xi}(\omega)= \min \{ \rho(g, \omega) \mid g \in G \}. 
\]
This is measurable since $G$ is finite. Since 
$\bar{\xi}(\omega)=\bar{\xi}(\rho(g, \omega))$ for all $g \in G$,
it induces a well-defined measurable map $\xi$ on $\underline{\Omega}^\can/G$ satisfying $\xi (\pi(\omega)) = \bar{\xi}(\omega)$. Since the action $\rho$ is free,
$\zeta$ is well defined. Since $G$ is finite and acts measurably, $\zeta$ is also measurable. Thus
$\Xi$ is a bijection with the given inverse. We compute
\[
\Xi( \rho^{1,0}(h,g,\delta))=
\Xi( (hg,\delta))=
\rho(hg,\xi(\delta))
=\rho(h, \rho(g,\xi(\delta)))
=\rho(h, \Xi(g,\delta)).
\]
Finally, $\Xi$ is measure preserving: for measurable
$A\subseteq\underline{\Omega}^{\can}/G$, freeness and invariance of
the action give
\[
\mathbb{Q}(\Xi(\{g\}\times A))
=
\mathbb{Q}(\rho(g,\xi(A)))
=
\mathbb{Q}(\xi(A))
=
\frac{1}{|G|}\,\mathbb{Q}(\pi^{-1}(A))
=
\frac{1}{|G|}\,\mu(A),
\]
where $\mu(A):={\mathbb Q}(\pi^{-1}(A))$ is the quotient measure;
this is the value of the product measure on a generating family of
sets.
\end{proof}

\subsection{Group actions on infinite product spaces}
\label{sec:infiniteProductActions}

We will follow the convention that natural numbers start at zero, so ${\mathbb N}=\{0,1,\ldots \}$. Correspondingly, for $n \geq 1$, we write $S_n$ for the permutation group of $\{0, \ldots, n-1\}$.
Let $S_\infty=\bigcup_{N=1}^\infty S_N$. $S_\infty$ acts on ${\mathbb N}$. Given $g \in S_\infty$ and $i \in {\mathbb N}$ we write $gi \in {\mathbb N}$ for the action of $g$ on $i$.

Given a probability space $\underline{\Omega}^A$,
define $\underline{\Omega}^{A,\infty}:=\prod_{i \in {\mathbb N}} \underline{\Omega}^{A}$. Given $g \in S_\infty$ and $\omega  \in \Omega^{A, \infty}$,
we define $g \omega$ by
$g \omega = (\omega^{g^{-1}i})_{i \in {\mathbb N}}$. This may be written equivalently as $(g \omega)^{i}=\omega^{g^{-1}i}$.
We use a contravariant action because one can think of an element of the Cartesian product as defining a function from the index set to $\Omega^A$ and the natural action of a group on functions is contravariant. Checking this explicitly, we confirm that this defines an $S_\infty$-action on $\underline{\Omega}^{A,\infty}$:
\[
((g_1 g_2) \omega)^i = \omega^{g_2^{-1}g_1^{-1} i}
= (g_2 \omega)^{g_1^{-1}{i}} = (g_1 (g_2 \omega))^i.
\]
This construction is measure-preserving since each element of $S_\infty$ permutes only a finite number of coordinates.

For random variables defined on a Cartesian product, the natural action is $(gX)(\omega)=X( g^{-1}(\omega))$ so that
\[
(gX)((\omega^{i})_{i \in {\mathbb N}})
=X((\omega^{gi})_{i \in {\mathbb N}}).
\]

Define ${\cal E}^{\underline{\Omega}^A,N}$ to be the completion under
${\mathbb Q}^{A,\infty}$ of the exchangeable $\sigma$-algebra of order $N$.
The exchangeable $\sigma$-algebra of order $N$ is the collection of
measurable sets which are invariant under the action of $S_N$.
If $S^{N,i}$ is the stabiliser of $i$ in $S_N$, write ${\cal E}^{\underline{\Omega}^A,N,i}$ for the completion of the $\sigma$-algebra of sets which are invariant under $S^{N,i}$.
If ${\cal F}^a$ is a sub-$\sigma$-algebra of
$\underline{\Omega}^{A,\infty}$, write
\[
{\cal E}^{N}({\cal F}^a)
:=
{\cal F}^a \cap {\cal E}^{\underline{\Omega}^A,N},
\qquad \text{and} \qquad
{\cal E}^{N,i}({\cal F}^a)
:=
{\cal F}^a \cap {\cal E}^{\underline{\Omega}^A,N,i}.
\]

If $\rho$ is a measurable $S_N$-action on a probability space
$\underline{\Omega}$, write
\[
{\cal E}^N_{\rho}
\]
for the completion of the $\sigma$-algebra of sets invariant under
$\rho$.
More generally, if ${\cal F}$ is a sub-$\sigma$-algebra of
$\underline{\Omega}$, write
\[
{\cal E}^N_{\rho}({\cal F})
:=
{\cal F}\cap {\cal E}^N_{\rho}.
\]
We extend this notation in the obvious way to the variants
${\cal E}^{N,i}$ and to product actions.

We conclude this section with a classical fact that we will need
later. Write
\[
{\cal E}^{\underline\Omega^A,\infty}:=\bigcap_{N\ge1}{\cal
E}^{\underline\Omega^A,N}
\]
for the $\sigma$-algebra of sets invariant under the full action of
$S_\infty$: this is the exchangeable $\sigma$-algebra of the coordinate
sequence $(\omega^i)_{i\in{\mathbb N}}$.

In the results below, the space $\Omega^{A,\infty}$ may carry a
probability measure ${\mathbb M}$ other than ${\mathbb Q}^{A,\infty}$.
In that case all completions, including those implicit in
${\cal F}^{A,\infty}$ and ${\cal E}^{\underline\Omega^A,\infty}$, are
understood to be taken with respect to ${\mathbb M}$.

The Hewitt--Savage theorem relates exchangeable sequences to i.i.d. random variables. For a comprehensive account of the classical theory of exchangeability, see \cite[Chapter~1]{kallenberg2005}.

\begin{theorem}[Hewitt--Savage Theorem \cite{hewittSavage}]
\label{thm:hewittSavage}
Let $\underline\Omega^A$ be a standard probability space, and let
${\mathbb M}$ be a probability measure on
$(\Omega^{A,\infty},{\cal F}^{A,\infty})$ invariant under the action of
$S_\infty$ on $\underline\Omega^{A,\infty}$, i.e.\ the
coordinate sequence $(\omega^i)_{i\in{\mathbb N}}$ is exchangeable
under ${\mathbb M}$. Then disintegrating ${\mathbb M}$ with respect to
${\cal E}^{\underline\Omega^A,\infty}$ yields, for each fibre, a measure
under which $(\omega^i)_{i\in{\mathbb N}}$ are i.i.d.
\end{theorem}

We will now explain how this yields a structural result that describes all ``exchangeable sequences'' of random variables, defined below.

\begin{definition}[Exchangeable sequence]
\label{def:exchangeableSequence}
Let $\underline\Omega^A$ be a standard probability space. A sequence
$(X^i)_{i\in{\mathbb N}}$ of $\Omega^A$-valued random variables,
defined on a probability space $(\Omega,{\cal F},{\mathbb P})$, is
{\em exchangeable} if the law of $(X^i)_{i\in{\mathbb N}}$ (that is,
the pushforward of ${\mathbb P}$ under
$\omega\mapsto(X^i(\omega))_{i\in{\mathbb N}}\in\Omega^{A,\infty}$) is
invariant under the action of $S_\infty$ on $\underline\Omega^{A,\infty}$.
\end{definition}

\begin{definition}[Mixing space and de~Finetti mixing variable]
\label{def:mixingSpace}
Let $(X^i)_{i\in{\mathbb N}}$ be an exchangeable sequence
(Definition~\ref{def:exchangeableSequence}), with law
${\mathbb M}:={\rm Law}\bigl((X^i)_{i\in{\mathbb N}}\bigr)$. By
Theorem~\ref{thm:hewittSavage}, since
${\cal E}^{\underline\Omega^A,\infty}$ is a completed, countably
generated sub-$\sigma$-algebra of the standard space
$(\Omega^{A,\infty},{\cal F}^{A,\infty},{\mathbb M})$, it is itself,
modulo null sets, the pullback of the Borel $\sigma$-algebra of a
standard probability space $\underline\Omega^{\rm mix}$ under some
measurable map
\[
Y:\underline\Omega^{A,\infty}\to\underline\Omega^{\rm mix},
\qquad
{\cal E}^{\underline\Omega^A,\infty}=Y^{-1}\bigl({\cal
B}(\Omega^{\rm mix})\bigr)
\]
modulo null sets. We call $\underline\Omega^{\rm mix}$ the {\em mixing
space} and $Y$ the {\em de~Finetti mixing variable} of
$(X^i)_{i\in{\mathbb N}}$; both are unique up to mod~$0$ isomorphism.
\end{definition}

\begin{definition}[De~Finetti representation]
\label{def:deFinettiRepresentation}
Let $(X^i)_{i\in{\mathbb N}}$ be exchangeable, with law ${\mathbb M}$,
mixing space $\underline\Omega^{\rm mix}$, and mixing variable $Y$ as
in Definition~\ref{def:mixingSpace}. Disintegrating ${\mathbb M}$ with
respect to $Y$ gives, for ${\rm Law}(Y)$-almost every
$y\in\Omega^{\rm mix}$, a conditional measure $\mu_y$ on
$\underline\Omega^{A,\infty}$, under which the coordinate sequence is
i.i.d.\ by Theorem~\ref{thm:hewittSavage}; write
\[
\nu_y:=(\pi^0)_*\mu_y
\]
for the common marginal of this i.i.d.\ sequence, i.e.\ the
pushforward of $\mu_y$ under the $0$th coordinate projection
$\pi^0:\Omega^{A,\infty}\to\Omega^A$. We call the family
$(\nu_y)_{y\in\Omega^{\rm mix}}$, together with the law ${\rm Law}(Y)$
of the mixing variable, the {\em de~Finetti representation} of
$(X^i)_{i\in{\mathbb N}}$.
\end{definition}

The Hewitt--Savage theorem now tells us that all exchangeable sequences
can be realised as families of
i.i.d. random variables parameterised by
the mixing variable $Y$. The next theorem makes this precise.

\begin{theorem}[Canonical representation of an exchangeable sequence]
\label{thm:canonicalDeFinetti}
Let $(X^i)_{i\in{\mathbb N}}$ be exchangeable, defined on
$(\Omega,{\cal F},{\mathbb P})$, with mixing variable $Y$ and
de~Finetti representation $(\nu_y)_{y\in\Omega^{\rm mix}}$ as in
Definitions~\ref{def:mixingSpace} and~\ref{def:deFinettiRepresentation}.
Write $X:=(X^i)_{i\in{\mathbb N}}:\Omega\to\Omega^{A,\infty}$ for the
sequence viewed as a single map. Define a measure on
$\underline\Omega^{\rm mix}\times\underline\Omega^{A,\infty}$ by
\[
{\mathbb M}^\circ(\cdot)
:=
\int_{\Omega^{\rm mix}}
\bigl(\delta_y\times\nu_y^{\otimes{\mathbb N}}\bigr)(\cdot)
\,d\,{\rm Law}(Y)(y),
\]
where $\nu_y^{\otimes{\mathbb N}}$ is the i.i.d.\ product measure on
$\underline\Omega^{A,\infty}$ with marginal $\nu_y$. Then:
\begin{enumerate}[(i)]
\item ${\mathbb M}^\circ$ is the law of $(Y(X(\omega)),X(\omega))$,
i.e.\ the pushforward of ${\mathbb P}$ under
$\omega\mapsto(Y(X(\omega)),X(\omega))$.
\item The map
\[
\Xi:(\Omega^{A,\infty},{\mathbb M})
\to
\bigl(\Omega^{\rm mix}\times\Omega^{A,\infty},{\mathbb M}^\circ\bigr),
\qquad
\Xi(\omega^\prime):=(Y(\omega^\prime),\omega^\prime),
\]
is a mod~$0$ isomorphism respecting $Y$ and the coordinate sequence,
and it is fibre-preserving: it commutes with the projection onto
$\underline\Omega^{\rm mix}$, i.e.\ it identifies same-$y$ fibres with
same-$y$ fibres.
\end{enumerate}
In this sense, every exchangeable sequence, viewed through its law, is
fibre-isomorphic to its own canonical de~Finetti reconstruction:
sample the mixing parameter, then generate an i.i.d.\ sequence with
the corresponding conditional marginal.
\end{theorem}
A proof is given in Appendix~\ref{appendix:finiteFiltrationClassification}.

The de Finetti representation collapses the arbitrary correlation structure among infinitely many exchangeable coordinates into a single mixing variable, together with an i.i.d.\ draw from the distribution it determines. This is exactly the reduction in degrees of freedom that our ergodic result will exploit.

\section{Financial modelling}

\subsection{Invariant preferences}

We wish to model the investment decisions involving a heterogeneous group of agents who may have differing beliefs. We now define a suitable notion of preferences 
that will allow us to model such investment decisions.

\begin{definition}
\label{def:invariantPreferences}
Let
\[
\underline{\Omega}^1
=
(\Omega^1,
{\cal F}^1,
({\cal F}^1_t)_{t\in{\cal T}},
{\mathbb Q}^1)
\]
be a standard filtered probability space.

A {\em family of $\underline{\Omega}^1$ preferences} is a collection
of maps
\[
{\cal J}_s(\,\cdot\,\mid x):
L^1(
{\mathbb R}_{\geq0},
({\cal F}^{1+}_t)_{t\in{\cal T}}
)
\to
{\mathbb R}\cup\{-\infty,+\infty\},
\qquad
s\in{\cal T},
\quad
x\in\Omega^1_s,
\]
satisfying the following properties.

\begin{enumerate}[(i)]

\item {\bf (Measurability)}
For every family
\[
(W_r)_{r\in\Omega^{\can}}
\subseteq
L^1(
{\mathbb R}_{\geq0},
({\cal F}^{1+}_t)_{t\in{\cal T}}
)
\]
with $(r,\omega)\mapsto(W_r)_t(\omega)$ jointly measurable for each
$t\in{\cal T}$, the map
$
(r,x)\mapsto{\cal J}_s(W_r\mid x)
$
is Borel measurable. In particular, taking a family constant in $r$,
the map $x\mapsto{\cal J}_s(X\mid x)$ is Borel measurable for every
stream $X$.

\item {\bf (Monotonicity)}
If
\[
X^1_t\geq X^2_t
\qquad
\forall\, t\in{\cal T}
\]
almost surely, then
\[
{\cal J}_s(X^1\mid x)
\geq
{\cal J}_s(X^2\mid x)
\]
for every
$x\in\Omega^1_s$.

\item {\bf (Upper semicontinuity)}
If
$
X_n\to X
$ in $L^1$, then
\[
\limsup_{n\to\infty}
{\cal J}_s(X_n\mid x)
\leq
{\cal J}_s(X\mid x)
\]
for every
$
x\in\Omega^1_s.
$

\item {\bf (Monotone continuity)}
If
$
X_n\to X
$
almost surely, where either
$X_n\leq X_{n+1}\leq X$ for every $n$, or
$X_n\geq X_{n+1}\geq X$ for every $n$,
and all of $X_n$, $X$ lie in
$L^1({\mathbb R}_{\geq0},({\cal F}^{1+}_t)_{t\in{\cal T}})$,
then
\[
{\cal J}_s(X_n\mid x)
\to
{\cal J}_s(X\mid x)
\]
for every
$
x\in\Omega^1_s.
$

\item {\bf (Invariance)}
For every adapted randomized automorphism $\Phi$ at time $s$
(Definition~\ref{def:adaptedRandomizedAutomorphism}),
every
$
x\in\Omega^1_s,
$
and every
\[
X
\in
L^1(
{\mathbb R}_{\geq0},
({\cal F}^{1+}_t)_{t\in{\cal T}}
),
\]
we have
\[
{\cal J}_s(X \mid x)
=
{\cal J}_s
\bigl(
X\circ\hat\Phi
\,\big|\,
x
\bigr).
\]
In particular, taking $\Phi$ independent of $\omega^1$, the value is
unchanged by precomposition with $\phi\times\id_{\Omega^1}$ for every
filtration-preserving mod $0$ automorphism $\phi$ of
$\underline{\Omega}^{\can}_{\cal T}$.

\item {\bf (Continuation consistency)}
Let
$
s<u.
$

Suppose that
\[
X,Y
\in
L^1(
{\mathbb R}_{\geq0},
({\cal F}^{1+}_t)_{t\in{\cal T}}
)
\]
agree at all times \(t<u\), and that
\[
{\cal J}_u(X \mid x')
\geq
{\cal J}_u(Y \mid x')
\]
for every
$
x'\in\Omega^1_u.
$
Then
\[
{\cal J}_s(X\mid x)
\geq
{\cal J}_s(Y\mid x)
\]
for every
$
x\in\Omega^1_s.
$

\item {\bf (Lottery tolerance)}
For $\omega^{\can}\in\Omega^{\can}_{{\cal T}_{\leq s}}$ and a
stream $X$, write $X^{[\omega^{\can}]}$ for the stream obtained by
composing $X$ with the map that replaces the coordinates of the
canonical factor at times ${\leq}\,s$ by $\omega^{\can}$. Let
$(Z^n)$ be streams in
$L^1({\mathbb R}_{\geq0},({\cal F}^{1+}_t)_{t\in{\cal T}})$, let
$(A_n)$ be measurable subsets of
$\Omega^{\can}_{{\cal T}_{\leq s}}$ with
${\mathbb Q}^{\can}_{{\cal T}_{\leq s}}(A_n)\to0$, and suppose
there are $M,m\in{\mathbb R}$ such that, for every $n$ and every
$x\in\Omega^1_s$,
\[
{\cal J}_s(Z^{n,[\omega^{\can}]}\mid x)\ge M
\quad\text{for a.e.\ } \omega^{\can}\notin A_n,
\]
\[
{\cal J}_s(Z^{n,[\omega^{\can}]}\mid x)\ge m
\quad\text{for a.e.\ } \omega^{\can}\in A_n.
\]
Then
$\liminf_n{\cal J}_s(Z^n\mid x)\ge M$ for every $x\in\Omega^1_s$.
\end{enumerate}
\end{definition}

All the key economic assumptions other than the self-enforcement, market-completeness and no-satiation assumptions listed in the Introduction correspond to the axioms above; the no-satiation assumption involves the participation times as well as the preferences, and is formalised at the level of a problem in Section~\ref{sec:compactness}. The upper semicontinuity, monotone continuity, measurability and lottery tolerance axioms exist to rule out mathematical pathologies rather than economically interesting examples. The value $+\infty$ is permitted so that unbounded utility functions may be paired with equivalent changes of measure; the well-posedness conditions of Section~\ref{sec:compactness} will exclude it on feasible sets. The Markovianity assumption and finite state-space assumptions are not required until Section~\ref{sec:finiteFunds}.

\begin{definition}
Let
\[
{\cal J}_s(\,\cdot\,\mid x),
\qquad
x\in\Omega^1_s,
\]
be a family of $\underline{\Omega}^1$ preferences.
For
$
X\in
L^1(
{\mathbb R}_{\geq0},
({\cal F}^{1+}_t)_{t\in{\cal T}}
),
$
define
$
{\cal J}_s(X)
\in
L^0(
{\mathbb R}\cup\{-\infty,+\infty\},
{\cal F}^1_s
)
$
by
$
{\cal J}_s(X)(\omega)
:=
{\cal J}_s
\bigl(
X
\mid
\omega^1_s(\omega^1)
\bigr),
$
where
\[
\omega=(\omega^{\can},\omega^1)
\in
\Omega^{\can}_{\cal T}
\times
\Omega^1.
\]
\end{definition}

The measurability axiom implies that ${\cal J}_s(X)$
is ${\cal F}^1_s$-measurable.

The interpretation of such a family, ${\cal J}$, is that, at time $s$, an agent prefers the income stream $X^1$ to the income stream $X^2$, (both lying in $L^1( {\mathbb R}_{\geq 0}, ({\cal F}^{1+}_t)_{t \in {\cal T}})$) if ${\cal J}_s(X^1)(\omega)>{\cal J}_s(X^2)(\omega)$. The measurability of ${\cal J}_s(X)$ ensures that this preference is known at time $s$.

\medskip

We formulate preferences on an $L^1$ domain rather than the larger
space $L^0$. This excludes models in which finite income streams may have infinite utility. However, we continue to allow the value
$-\infty$, which may be interpreted as representing an unacceptable or infeasible outcome.

Unlike the von Neumann--Morgenstern preference framework, our preferences take income streams, not lotteries, as the primitive
objects. In frameworks whose
primitive objects are lotteries, the mixture operation is part of
the formulation, and linearity in the mixture bounds the value of a
lottery in terms of the values of its branches. Taking streams as
primitive, a replacement for such bounds is needed.

The Lottery tolerance axiom supplies a minimal one: a $\liminf$ bound by a level
attained on almost all branches, robust to exceptional branches of
vanishing probability whose values are uniformly bounded below. Lottery tolerance
is automatic for von Neumann--Morgenstern preferences.

The upper semicontinuity axiom provides the
complementary $\limsup$ bound and is more demanding: von
Neumann--Morgenstern preferences satisfy it only under an additional
condition such as linear growth-domination of the utility, as
implied by concavity.

The continuation-consistency condition serves a similar role to the
time-consistency axiom of Kreps and Porteus
\cite{krepsPorteus}. However, our formulation is expressed directly
in terms of continuation values and is tailored to the recursive
definition of acceptable strategies introduced below.

However, there is an important difference between our formulation and that of Kreps and Porteus. Kreps and Porteus seek to define preferences over income streams in a probability model where the measure is fixed. Our notion of invariance depends upon the choice of measure on $\underline{\Omega}^{\can}_{\cal T}$ but only on the equivalence class of the measure on $\underline{\Omega}^1$.
Our next example shows the importance of this distinction.

\begin{example}[Robust recursive preferences]
\label{example:robustRecursivePreferences}
Let
\[
{\cal P}:=\{{\mathbb P}^{1,i}\}_{i\in{\mathbb I}}
\]
be a family of probability measures equivalent to
${\mathbb Q}^1$, indexed by a finite set
${\mathbb I}$.
Let
\[
u:{\mathbb R}_{\geq 0}\to{\mathbb R}
\]
be a continuous increasing utility function.

Define a family of $\underline{\Omega}^1$ preferences recursively by
\[
{\cal J}^{\cal P}_s(X\mid x)
=
\begin{cases}
u(X_T(x)),
&
s=T,
\\[1ex]
u\!\left(X_s(x)\right)
+
\displaystyle
\min_{i\in{\mathbb I}}
{\mathbb E}_{{\mathbb P}^{1,i}}
\!\left(
{\cal J}^{\cal P}_{s+1}
\bigl(
X
\mid
\omega^1_{s+1}
\bigr)
\,\middle| \,
\omega^1_s=x
\right),
&
s<T.
\end{cases}
\]
Then
\[
{\cal J}^{\cal P}
=
\bigl(
{\cal J}^{\cal P}_s
\bigr)_{s\in{\cal T}}
\]
defines a family of $\underline{\Omega}^1$ preferences,
provided
that $u$ is chosen so that the Upper semicontinuity axiom is
satisfied. The Monotone continuity axiom holds by the monotone
convergence theorem, applied through the recursion, since a minimum
over the finite family ${\cal P}$ preserves monotone limits, and
the Lottery tolerance axiom holds by dominated convergence within
each conditional expectation.

Upper semicontinuity is not automatic. It may fail when $u$ has
asymptotically linear growth and some density
$d{\mathbb P}^{1,i}/d{\mathbb Q}^1$ is unbounded.
Lemma~\ref{lemma:robustUpperSemicontinuity} below provides
sufficient conditions for it to hold.

When ${\cal P}$ contains a single probability measure, this reduces to
a von Neumann--Morgenstern expected-utility preference.
In the general case it may be interpreted as the preference system of
a robust optimiser who evaluates future outcomes under the most adverse
probability measure in the family ${\cal P}$.
\end{example}

\begin{lemma}[Upper semicontinuity of robust recursive preferences]
\label{lemma:robustUpperSemicontinuity}
In Example~\ref{example:robustRecursivePreferences}, suppose that
$u$ is concave and that either
\begin{enumerate}[(i)]
\item $u$ is bounded above; or
\item there are $C>0$ and $\theta\in(0,1)$ such that
$u(x)\le u(0)+Cx^\theta$ for all $x\ge0$, and the one-step
conditional densities of each ${\mathbb P}^{1,i}$ with respect to
${\mathbb Q}^1$ have uniformly bounded conditional moments of order
$\frac{1}{1-\theta}$.
\end{enumerate}
Then ${\cal J}^{\cal P}$ satisfies the Upper semicontinuity axiom.
\end{lemma}

\begin{proof}
Under (i), let $X_n\to X$ in $L^1$. Applying the reverse Fatou
lemma along almost surely convergent subsequences at each step of
the recursion, together with
$\limsup_n\min_i\le\min_i\limsup_n$, gives
$\limsup_n{\cal J}^{\cal P}_s(X_n\mid x)
\le{\cal J}^{\cal P}_s(X\mid x)$.

Under (ii), since the increments of a concave function decrease,
$u$ satisfies the subadditive bound
$u(x^\prime)\le u(x)+u(|x^\prime-x|)-u(0)$. The conditional
H\"older inequality with exponents $\frac{1}{1-\theta}$ and
$\frac{1}{\theta}$, together with the conditional Jensen
inequality, then propagates the estimate
\[
{\cal J}^{\cal P}_s(X^\prime\mid x)
\le
{\cal J}^{\cal P}_s(X\mid x)
+
\sum_{u\in{\cal T}_{\geq s}}
K_u
\,
{\mathbb E}_{{\mathbb Q}^1}
\bigl(
|X^\prime_u-X_u|
\,\big|\,
\omega^1_s=x
\bigr)^{\theta}
\]
backward through the recursion, for constants $K_u$: the moment
hypothesis bounds the change of measure at each step, and a minimum
over the finite family only improves the estimate. Upper
semicontinuity along $L^1$ follows.
\end{proof}

Kreps and Porteus fix a probability measure and study
preferences over income streams relative to that measure.
By contrast, a family of $\underline{\Omega}^1$ preferences may
encode a more general model of belief formation.

The robust optimisation example above illustrates this point:
the preference is determined by a family of probability measures
rather than a single measure. Consequently, our notion of
$\underline{\Omega}^1$ preference should be viewed as encoding both
attitudes towards risk and a rule for evaluating uncertainty in the
evolution of the state.

While we do not insist that agents believe in any particular measure for $\underline{\Omega}^1$, we do insist that they agree that the $\underline{\Omega}^{\can }_{\cal T}$ component is independent of $\underline{\Omega}^1$ and irrelevant to their process of preference formation.

\medskip

We conclude this section with some basic definitions that
allow us to relate preferences on different probability spaces.

\begin{definition}[Preference isomorphism]
	\label{def:preferenceIsomorphism}
	Let $\underline\Omega^1,\underline\Omega^{1^\prime}$ be standard
	filtered probability spaces, and let
	$
	\Xi:\underline\Omega^{1+}\to\underline\Omega^{1^\prime+}
	$
	be a filtration-preserving mod $0$ isomorphism witnessing that
	$\underline\Omega^1$ and $\underline\Omega^{1^\prime}$ are stably
	equivalent. We say
	\[
	X_1\in L^1({\mathbb R}_{\ge0},({\cal F}^{1+}_t)_{t\in{\cal T}}),
	\qquad
	X_{1^\prime}\in L^1({\mathbb R}_{\ge0},({\cal
		F}^{1^\prime+}_t)_{t\in{\cal T}})
	\]
	are {\em preference isomorphic (via $\Xi$)} if
	$
	X_{1^\prime}=X_1\circ\Xi^{-1}.
	$
\end{definition}

A family of $\underline{\Omega}^1$ preferences induces
a family on any product with $\underline{\Omega}^1$ as follows.
\begin{definition}[Preferences induced on a product]
\label{def:inducedPreferences}
Let
\[
{\cal J}_s(\,\cdot\,\mid x),
\qquad
x\in\Omega^1_s,
\]
be a family of $\underline{\Omega}^1$ preferences and let
$\underline{\Omega}^2$ be a standard filtered probability space.
Let
\[
\underline{\Omega}^{1,2}
:=
\underline{\Omega}^1
\times
\underline{\Omega}^2.
\]
The projection $\pi:\underline{\Omega}^{1,2}\to\underline{\Omega}^1$ is an
enrichment (Definition~\ref{def:enrichment}); choose a
filtration-preserving mod $0$ isomorphism
\[
\Xi:
\underline{\Omega}^{1+}
\to
\underline{\Omega}^{(1,2)+}
\]
over $\underline{\Omega}^1$ as in Lemma~\ref{lemma:enrichmentStable}.

The induced family of
$\underline{\Omega}^{1,2}$ preferences is defined by
\[
{\cal J}_s(X\mid (x,y))
:=
{\cal J}_s(X\circ\Xi\mid x),
\]
for all
\[
(x,y)\in\Omega^1_s\times\Omega^2_s
\]
and all
$X\in L^1({\mathbb R}_{\geq0},({\cal F}^{(1,2)+}_t)_{t\in{\cal T}})$.
\end{definition}
The right-hand side does not depend on the choice of $\Xi$: by
Lemma~\ref{lemma:enrichmentStable}(ii) two choices differ by
$\hat\Phi$ for an adapted randomized automorphism $\Phi$, so the two
values agree by the Invariance axiom. The axioms for the induced family follow from those
for ${\cal J}$; in particular, an adapted randomized automorphism for
the induced family transports through $\Xi$ to an adapted randomized
automorphism for ${\cal J}$, giving the induced Invariance axiom.
In these induced preferences, the agent is indifferent to the outcome in ${\underline{\Omega}}^2$.

\medskip

Definition~\ref{def:inducedPreferences} uses only that
$\underline{\Omega}^{1,2}$ enriches $\underline{\Omega}^1$
(Definition~\ref{def:enrichment}), not that the enrichment is a product.
The general form is as follows.

\begin{definition}[Preferences induced along an enrichment]
\label{def:inducedPreferencesEnrichment}
Let ${\cal J}$ be a family of $\underline{\Omega}^1$ preferences and let
$\pi:\underline{\Omega}^3\to\underline{\Omega}^1$ be an enrichment.
Choose $\Xi$ as in Lemma~\ref{lemma:enrichmentStable}. The induced
family of $\underline{\Omega}^3$ preferences is defined by
\[
{\cal J}^\pi_s(X\mid z)
:=
{\cal J}_s\bigl(X\circ\Xi\,\big|\,\pi_s(z)\bigr),
\qquad
z\in\Omega^3_s,
\]
for all
$X\in L^1({\mathbb R}_{\geq0},({\cal F}^{3+}_t)_{t\in{\cal T}})$, where
$\pi_s:\Omega^3_s\to\Omega^1_s$ is the map induced by $\pi$ on time-$s$
states, which exists because $\pi^{-1}({\cal F}^1_s)\subseteq{\cal F}^3_s$.
\end{definition}
The right-hand side does not depend on the choice of $\Xi$: by
Lemma~\ref{lemma:enrichmentStable}(ii) two choices differ by
$\hat\Phi$ for an adapted randomized automorphism $\Phi$, so the two
values agree by the Invariance axiom. The axioms for the induced family
follow from those for ${\cal J}$ as before.
Definition~\ref{def:inducedPreferences} is the case
$\underline{\Omega}^3=\underline{\Omega}^{1,2}$ with $\pi$ the
projection; the induced value there does not depend on the
$\underline{\Omega}^2$ component of the state because that component
belongs to the enrichment rather than to the base.

If the value is required at states of a $\sigma$-algebra
${\cal G}_s\subseteq{\cal F}^3_s$ containing $\pi^{-1}({\cal F}^1_s)$,
the same formula applies with the conditioned isomorphisms of
Lemma~\ref{lemma:enrichmentStable}(iii). What is essential throughout is that
the conditioning state leave the canonical factor of the lift free.
Where instead it resolves that factor, the value must be obtained by
freezing the resolved coordinates, in the manner of the Lottery
tolerance axiom.
\subsection{The financial model for our problems}

As discussed further in Section~\ref{sec:symmetryArgument} below, the financial market is modelled by a filtered probability space
$\underline{\Omega}^M$.  We
choose units so that our numeraire asset
\[
{\cal N} \in L^0({\mathbb R}_{\geq 0}, ({\cal F}^M_t)_{t \in {\cal T}}, \mathbb Q^M)
\]
satisfies ${\cal N}_t=1$ for all $t \in {\cal T}$.

 Again, as discussed in Section~\ref{sec:symmetryArgument} below, the idiosyncratic experience of an agent is modelled by a filtered probability space $\underline{\Omega}^I$. We assume that $\underline{\Omega}^M$ and $\underline{\Omega}^I$ are standard filtered probability spaces.

We assume we have the following additional structures.
Let $\tau$ and $\tau_d$ be stopping times defined on $\underline{\Omega}^M \times \underline{\Omega}^I$ taking values in ${\cal T}$ with $\tau<\tau_d$. The first stopping time $\tau$ represents the first time at which an
agent starts participating in the fund.
They are not permitted to enter into contracts before this time.
The second stopping time $\tau_d$ represents the time at which the
agent leaves the fund.
They are not permitted to enter into new contracts on or after this
time.

Let
\[
\eta \in L^1({\mathbb R}_{\geq 0}, \iota({\cal F}^M_t \times {\cal F}^I_t)_{t \in {\cal T}}, \mathbb Q^M \times \mathbb Q^I)
\]
be a stochastic process representing the amount an agent pays in at each time. We require $\eta_t = 0$ unless $\tau \leq t < \tau_d$, so that contributions are made only during participation.

We assume that the preferences of an agent over consumption are given by a family of $(\underline{\Omega}^M \times \underline{\Omega}^I)$ preferences.

The total spaces for our problems will be given by product spaces which include a factor $\prod_{i \in \mathbb N} \underline{\Omega}^I$. We will label the pull-backs of $\eta$ and $\tau$ on each component as $\eta^i$ and $\tau^i$. We will write ${\cal J}^i$ for the preferences of agent $i$.

Thus we are assuming that all agents have identical
preferences. Agent $i$ cares only about their own state and
history, the state of the market, and their income stream.

Because we do not assume invariance under transformations of the
market factor, different agents need not agree on the
physical measure governing the market. By contrast, our invariance
assumptions imply indifference to relabellings and other
transformations of the states of other agents. In this sense,
all agents agree on the modelling of idiosyncratic risk.

This setting is very abstract; the following example shows how our results may be applied in practice.

\begin{example}[von Neumann agents in a Black--Scholes model]
\label{example:blackScholes}
Let ${\mathbb I}$ be a finite index set describing the possible
agent types. We assume that all agents pay in for $T_1$
years and then retire.

Let $\underline{\Omega}^M$ be the filtered probability space generated
by a Brownian motion $W$, sampled only at times
$t\in{\cal T}$. We write ${\mathbb Q}^M$ for the corresponding pricing
measure.

For each type $\xi\in{\mathbb I}$, choose:
\begin{enumerate}[(i)]
\item a market price of risk $m^\xi$ and the associated physical
measure ${\mathbb P}^\xi$ for the Black--Scholes model;
\item a risk-aversion parameter $
\alpha^\xi\in(-\infty,1)\setminus\{0\}$;
\item a deterministic entry time $\tau^\xi\in{\cal T}$;
\item a mortality model represented by a stopping time
$\tau^\xi_d>\tau^\xi+T_1$.
\end{enumerate}

We assume that, once the type $\xi$ is known, the only state variable
required to describe future idiosyncratic uncertainty is whether the
agent is alive or dead. Since the time horizon ${\cal T}$ is
finite, the probability space
$\underline{\Omega}^\xi$ generated by $\tau^\xi_d$ is finite.

Choose weights
$
\pi^\xi\in[0,1]$
with $
\sum_{\xi\in{\mathbb I}}\pi^\xi=1.
$
For each time $t\in{\cal T}$ define
$
{\mathbb I}_t
:=
\{\xi\in{\mathbb I}\mid\tau^\xi=t\},
$
and
$
\pi_t
:=
\sum_{\xi\in{\mathbb I}_t}\pi^\xi.
$
Let $t_{\max}:=\max\{t\in{\cal T}\mid \pi_t>0\}$.

Let $\underline{\Omega}^{1}$ and
$\underline{\Omega}^{2}$ be copies of
$\underline{\Omega}^{\can}_{\cal T}$.
Write
$\omega^1_t$,
and
$\omega^2_t$,
for the canonical uniform random variables at time $t$ on these
spaces.
Define an auxiliary space
$
\underline{\Omega}^{I'}
:=
\underline{\Omega}^{1}
\times
\underline{\Omega}^{2}
\times
\prod_{\xi\in{\mathbb I}}
\underline{\Omega}^{\xi}$.
The first canonical factor is used to select a cohort.
If the type of an agent has not yet been assigned by time
$t-1$, then we assign a type at time $t$ whenever either
$t=t_{\max}$ or
\[
\omega^1_t\leq \frac{\pi_t}{1-\sum_{s<t}\pi_s}.
\]
Conditional on the decision to assign a type at time $t$, we use $\omega^2_t$ to
choose a type $\xi\in{\mathbb I}_t$
with probability $\frac{\pi^\xi}{\pi_t}$. With these choices, the
probability of remaining unassigned before time $t$ is
$1-\sum_{s<t}\pi_s$, so an agent's type is $\xi$ with probability
$\pi^\xi$.

Let $\underline{\Omega}^{I}$ be the filtered probability space
generated by the random variables
$\xi$ and $\tau_d^\xi$. Thus the idiosyncratic state records both the type of the agent and
their subsequent mortality experience. Define $\tau:=\tau^\xi$.

We assume that the preferences of a type-$\xi$ agent are described
by a time-consistent family of $\underline{\Omega}^M\times
\underline{\Omega}^I$ preferences induced by the CRRA utility
\[
u^\xi(x)=\frac{x^{\alpha^\xi}}{\alpha^\xi}.
\]
Let
\[
Y^\xi_u
:=
\left.
\frac{d{\mathbb Q}^M}{d{\mathbb P}^\xi}
\right|_{{\cal F}^M_u}
\]
denote the state-price density corresponding to the physical
measure ${\mathbb P}^\xi$.
For $x\in\Omega^{M\times I}_t$, define
\[
{\cal J}^\xi_t(X\mid x)
=
{\mathbb E}_{{\mathbb Q}^M\times{\mathbb Q}^I}
\!\left(
\sum_{\tau^\xi+T_1\leq u<\tau^\xi_d}
\frac{1}{Y^\xi_u}\,
\frac{X_u^{\alpha^\xi}}{\alpha^\xi}
\,\middle| \,
\omega^{M\times I}_t=x
\right),
\]
where
\[
\omega^{M\times I}_t
:
\Omega^M\times\Omega^I
\to
\Omega^{M\times I}_t
\]
is the time-$t$ state variable associated with
$\underline{\Omega}^M\times\underline{\Omega}^I$.
The induced preference function is therefore
\[
{\cal J}^\xi_t(X)(\omega)
=
{\cal J}^\xi_t
\bigl(
X
\mid
\omega^{M\times I}_t(\omega)
\bigr).
\]
Thus beliefs about market dynamics are encoded through the family of
measures ${\mathbb P}^\xi$, while risk preferences are encoded through
the parameter $\alpha^\xi$.

Finally, let $c^\xi\geq0$
be a deterministic payment level and define
\[
\eta^\xi_t
=
\begin{cases}
c^\xi,
&
\tau^\xi\leq t<\tau^\xi+T_1,
\\[1ex]
0,
&
\text{otherwise}.
\end{cases}
\]
This includes the case in which the agent's type has not yet been
revealed and so no payments will be made.

\medskip
The continuity axioms hold as in
Example~\ref{example:robustRecursivePreferences}: each $u^\xi$ is
concave and satisfies condition (i) of
Lemma~\ref{lemma:robustUpperSemicontinuity} when $\alpha^\xi<0$ and
condition (ii) with $\theta=\alpha^\xi$ when $\alpha^\xi\in(0,1)$,
the moment condition holding because the state-price densities
$Y^\xi_u$ are lognormal.
\end{example}

As the number of independent agents, $N$, tends to infinity,
the proportion of agents of type $\xi$ converges almost surely
to $\pi^\xi$. Thus our model allows us to
investigate the behaviour of large funds in which the relative sizes
of different cohorts are approximately known.

An alternative approach would specify cohorts of a deterministic
size and study a sequence of funds by scaling them; our approach of
a single idiosyncratic state variable for both cohort membership and
idiosyncratic risk minimises book-keeping and simplifies the proofs.

Our framework also does not require different agent types to share a common model of market dynamics (for example, some agents may use a jump-diffusion model and others a stochastic-volatility model), and allows preferences and the participation times $\tau,\tau_d$ to depend on the market, accommodating products that insure systematic longevity risk or have benefits depending jointly on market and demographic outcomes. More generally, it extends to joint-life insurance, health-insurance contracts, and other collective contracting problems unrelated to pensions.

\section{The infinite fund problem}
\label{sec:infiniteFund}

\subsection{Admissible strategies}
We will now write down the optimal investment problem when we have a continuum of agents. We will later show how this problem arises as the limit as $N \to \infty$ for an associated finite-fund problem.

The definitions given in this section can be posed
relative to a choice of preferences, participation times, and income
other than the ambient one: we call a tuple $({\cal J}',\tau',\tau'_d,\eta')$
a {\em problem}. We will write $\Prob=({\cal J}, \tau, \tau_d, \eta)$ for the ambient problem.
When we wish to emphasise the dependence of a definition on the problem, we relativise the definition by inserting $\Prob$ to the left of a semicolon. See Definition~\ref{def:admissibleNoIdiosyncratic} below for an example. We will not need this notation in this section, but it will prove useful later.
Since ${\cal J}'$ is itself a family of preferences on some underlying
filtered probability space, posing a definition relative to a
non-ambient $\Prob'$ implicitly reinterprets every
space, $\sigma$-algebra, and measure appearing in it on the space underlying ${\cal J}'$. When no problem is specified, we use the ambient one. As a result for most of this section the ambient problem is assumed.

\medskip

We will call the probability space for the infinite-fund problem $\underline{\Omega}^\infty:=\underline{\Omega}^M \times \underline{\Omega}^I$. Instead of thinking of the $I$ component as representing a probability model for an agent, one should think of it as describing the continuum of different types present in the infinite fund. In the infinite fund we will assume all agents of the same type within a fund are treated the same, so payments can be described by random variables on $\underline{\Omega}^I$.

We will want to consider two different investment problems. The first is the conditional investment problem for a homogeneous collective of agents of a known type at time $s \in {\cal T}$. These problems are naturally indexed by $x \in \Omega^{M \times I}_s$ and are associated with a conditional measure ${\mathbb Q}^\infty_x$. The second problem we will want to consider is for an entire
collective. These problems are naturally indexed by
\[
x\in\Omega^M_s,
\]
and are associated with a different conditional measure.

To avoid repeating each definition twice, we shall write
\[
x\in\Omega^M_s\sqcup\Omega^{M\times I}_s
\]
to indicate that the corresponding statement applies both to
collective states in $\Omega^M_s$ and to agent states in
$\Omega^{M\times I}_s$.
The associated conditional measure, obtained by disintegrating
${\mathbb Q}^\infty$ with respect to the relevant state variable,
will always be denoted by ${\mathbb Q}^\infty_x$.

We will write down the possible investment strategies for a collective of agents at a time $s$ given their budgets, $b$
and their consumption history $\overline{X}$. In the formulation below, $\gamma_{u,t}$
for $t>u$ represents the total cashflow that agents
whose state at time $u$ is represented by $x$
will receive at time $t$ given
the contracts agreed at time $u$.
We define $\gamma_u$ to be a stochastic process, and write $\gamma_{u,t}:=(\gamma_u)_t$.
The cashflow $\gamma_{u,t}$ for $t \leq u$ represents
the consumption history at time $u$ and so must satisfy a consistency
condition.

\begin{definition}
\label{def:admissibleNoIdiosyncratic}
Let $s \in {\cal T}$ and let $x \in \Omega^{M}_s \sqcup \Omega^{M \times I}_s$.
Let $b \in L^1({\mathbb R}, {\cal F}^\infty_s,{\mathbb Q}^\infty_x)$ and
$\overline{X} \in  L^1({\mathbb R}_{\geq 0}, ({\cal F}^\infty_t)_{t \in {\cal T}_{< s}},{\mathbb Q}^\infty_x)$.
Let $\Prob'=({\cal J}',\tau',\tau'_d,\eta')$ be a
problem, understood to be the ambient problem
$\Prob=({\cal J},\tau,\tau_d,\eta)$ when omitted.

For $Z \in L^1({\mathbb R},{\cal F}^\infty,{\mathbb Q}^\infty_x)$ and $w \in {\cal T}$ we write
\[
m_w(Z):={\mathbb E}_{{\mathbb Q}^\infty_x}(Z \mid \iota^{\Omega^\infty}{\cal F}^M_w)
\]
for the per-capita value of the cashflow $Z$ given the market
information at time $w$: conditioning on the market factor alone
mutualises idiosyncratic risk across the collective.

We say that
\[
\gamma \in \prod_{u \in {\cal T}_{\geq s}}L^1({\mathbb R}_{\geq 0}, ({\cal F}^\infty_t)_{t \in {\cal T}},{\mathbb Q}^\infty),
\]
is an admissible
strategy for the investment problem $(\Prob';s,x, b, \overline{X})$
if it satisfies
\begin{enumerate}[(i)]
\item The consistency conditions
\begin{equation}
\gamma_{u,t}=\overline{X}_{t} \quad \forall u \in {\cal T}_{\geq s},t \in  {\cal T}_{<s} \qquad {\mathbb Q}^\infty_x\text{-a.s.}
\label{eq:consistencyCondition}
\end{equation}
\begin{equation}
\gamma_{u,t}=\gamma_{t,t} \quad \forall u \in {\cal T}_{\geq s},t \in  {\cal T}_{<u} \cap {\cal T}_{\geq s} \qquad {\mathbb Q}^\infty_x\text{-a.s.}
\label{eq:consistencyConditionRealised}
\end{equation}
\item The non-participation conditions
\begin{equation}
\gamma_{u,t} = 0
\qquad
\forall\, t \in {\cal T}_{>u}, \text{ whenever } u < \tau' \text{ or } u \ge \tau'_d
\qquad {\mathbb Q}^\infty_x\text{-a.s.}
\label{eq:nonParticipation}
\end{equation}
\begin{equation}
\gamma_{u,t}=0
\quad
\text{on } \{t < \tau'\} \cup \{t \ge \tau'_d\}
\qquad
\forall\, u \in {\cal T}_{\geq s},\ t \in {\cal T}_{\geq u}
\qquad {\mathbb Q}^\infty_x\text{-a.s.}
\label{eq:noConsumptionNonParticipation}
\end{equation}
\item The issuance-funding condition
\begin{equation}
\sum_{s \leq t \leq T} m_s( \gamma_{s,t})
\leq
m_s(b) + m_s(\eta'_s)
\qquad {\mathbb Q}^\infty_x\text{-a.s.}
\label{eq:budgetJ}
\end{equation}
\item The self-financing renegotiation conditions
\begin{equation}
\sum_{u \leq t \leq T} m_u( \gamma_{u,t})
\leq
\sum_{u \leq t \leq T} m_u( \gamma_{u-1,t})
+ m_u(\eta'_u)
\qquad {\mathbb Q}^\infty_x\text{-a.s.}
\label{eq:selfFinancing}
\end{equation}
for all $u \in {\cal T}_{>s}$.
\end{enumerate}

The process $\gamma$ is defined on the full probability space
$\underline{\Omega}^\infty$, while admissibility is evaluated under
the conditional measure ${\mathbb Q}^\infty_x$.

We write ${\cal A}^x_s(\Prob';b, \overline{X})$ for
the space of admissible strategies, omitting $\Prob'$
when it is the ambient problem.
If
\[
{\mathbb Q}^\infty_x(b<0)>0,
\]
then no admissible strategy exists because the investment problem
begins with negative wealth on a set of positive conditional
probability
so we define ${\cal A}^x_s(\Prob';b, \overline{X}) = \emptyset$.
\end{definition}

The issuance-funding condition~\eqref{eq:budgetJ} states that the
value of the contract agreed at time $s$ is covered by the initial
budget together with the contribution received at time $s$. The
self-financing renegotiation condition~\eqref{eq:selfFinancing}
states that the contract agreed at time $u$ is purchased by
surrendering the remaining payments promised under the contract
agreed at time $u-1$, together with the contribution received at time
$u$. Thus each contract is funded purely by income received to date.
Indeed, applying $m_s$ to \eqref{eq:selfFinancing} for
$r=s+1,\ldots,u$ and combining with \eqref{eq:budgetJ}, using
\eqref{eq:consistencyConditionRealised} to identify $\gamma_{r,r}$
with realised consumption, gives
\begin{equation}
\sum_{u \leq t \leq T} m_s( \gamma_{u,t})
\leq
m_s(b) + \sum_{s \leq r \leq u} m_s(\eta'_r)
- \sum_{s \leq t < u} m_s(\gamma_{t,t})
\qquad {\mathbb Q}^\infty_x\text{-a.s.}
\label{eq:receivedToDate}
\end{equation}
for every $u \in {\cal T}_{\geq s}$: valued at time $s$, the contract
in force at time $u$ together with the consumption already taken is
funded by the initial budget and the contributions received by time
$u$. Since every payment $\gamma_{u,t}$ is nonnegative, the fund's
position is at all times a portfolio of nonnegative claims, so its
net wealth can never be negative.

Since the only restrictions on contracts are given by the funding conditions, this definition assumes the market is complete.

If $\gamma$ is an admissible strategy and $u \in {\cal T}_{>s}$, the
members of the fund enter time $u$ holding the outstanding payments
$(\gamma_{u-1,t})_{t \geq u}$ promised under the contract agreed at
time $u-1$. This motivates the following definition of the
continuation budget. For $u \in {\cal T}_{>s}$ and
$x^\prime \in \Omega^M_u \sqcup \Omega^{M \times I}_u$ we define
\begin{equation}
\beta_{x^\prime,u}(\gamma)
:=
\sum_{u\le t\le T}
{\mathbb E}_{{\mathbb Q}^\infty}
\!\left(
\gamma_{u-1,t}
\,\middle| \,
\omega^\infty_u=x^\prime
\right).
\label{eq:defBeta}
\end{equation}

The quantity $\beta_{x^\prime,u}(\gamma)$ represents the budget
available at time $u$, before the time-$u$ contract is agreed, to an
agent whose state is $x^\prime$: it is the value of the payments
still owed to that agent under the contract currently in force. The
self-financing renegotiation condition~\eqref{eq:selfFinancing}
states precisely that, per capita, the contract agreed at time $u$ is
affordable from these holdings together with the time-$u$
contribution.

Given a sequence of contracts $\gamma$, we will define the realised
consumption stream $X(\gamma)$ by
\[
X(\gamma)_t:=\gamma_{t,t} \quad \text{for } t \in {\cal T}_{\geq s},
\qquad
X(\gamma)_t:=\overline{X}_t \quad \text{for } t \in {\cal T}_{< s},
\]
where $\overline{X}$ is the consumption history of the problem for
which $\gamma$ is admissible.

\subsection{Acceptable strategies}

\begin{definition}
\label{def:acceptableNoIdiosyncratic}

Let
$
s \in {\cal T}$,
$
x \in \Omega^{M}_s \sqcup \Omega^{M\times I}_s$,
\[
\beta
\in
L^1(
{\mathbb R},
{\cal F}^\infty_s,
{\mathbb Q}^\infty_x
),
\]
and 
\[
\overline X
\in
L^1(
{\mathbb R}_{\ge0},
({\cal F}^\infty_t)_{t\in{\cal T}_{<s}},
{\mathbb Q}^\infty_x
).
\]

We define the sets of {\em acceptable strategies}
$
{\cal G}^x_s(\beta,\overline X)
$
recursively, backwards in time, together with sets of
{\em screened strategies}
${\cal H}^x_s(\beta,\overline X)$.

Suppose that ${\cal G}^{x'}_u$ has been defined for every
$u\in{\cal T}_{>s}$. We define
${\cal H}^x_s(\beta,\overline X)$
to be the set of
$
\gamma
\in
{\cal A}^x_s(\beta,\overline X)
$
such that, for every
$u\in{\cal T}_{>s}$ and every measurable family
$x'\mapsto\gamma'^{,x'}$, defined for
${\mathbb Q}^\infty_x$-almost every $x'\in\Omega^{M\times I}_u$ and
satisfying
$
\gamma'^{,x'}
\in
{\cal G}^{x'}_u
\bigl(
\beta_{x',u}(\gamma),
X(\gamma)
\bigr),
$
we have
\begin{equation}
{\cal J}_u(X(\gamma)\mid x')
\ge
{\cal J}_u(X(\gamma'^{,x'})\mid x')
\qquad
\text{for ${\mathbb Q}^\infty_x$-almost every }x'.
\label{eq:screeningCriterion}
\end{equation}

We then define
$
{\cal G}^x_s(\beta,\overline X)
$
to be the set of
$
\gamma
\in
{\cal H}^x_s(\beta,\overline X)
$
such that, for every measurable family
$x'\mapsto\gamma'^{,x'}$, defined for
${\mathbb Q}^\infty_x$-almost every $x'\in\Omega^{M\times I}_s$ and
satisfying
$
\gamma'^{,x'}
\in
{\cal H}^{x'}_s
\bigl(
\beta(x'),
\overline X
\bigr),
$
where $\beta(x')$ denotes the conditional realization of $\beta$ at
$x'$, we have
\begin{equation}
{\cal J}_s(X(\gamma)\mid x')
\ge
{\cal J}_s(X(\gamma'^{,x'})\mid x')
\qquad
\text{for ${\mathbb Q}^\infty_x$-almost every }x'.
\label{eq:acceptabilityCriterion}
\end{equation}

We will refer to
\eqref{eq:screeningCriterion}
as the {\em screening criterion} and to
\eqref{eq:acceptabilityCriterion}
as the {\em acceptability criterion}.
\end{definition}

At the terminal date the screening criterion is vacuous, so that
${\cal H}^x_T(\beta,\overline X)={\cal A}^x_T(\beta,\overline X)$ and
${\cal G}^x_T(\beta,\overline X)$ consists of the statewise utility
maximisers. The acceptability criterion compares against
${\cal H}^{x'}_s$ rather than ${\cal G}^{x'}_s$ in order to keep the
recursion well-founded; since the comparison is against every
element, this does not change the implied comparison of values. The
criterion expresses the fact that the states of the agents are
already observable when a contract is formed: a coalition of agents
in a common state at time $s$ will not enter a contract if they
could do better by managing their own resources $\beta(x')$,
together with their own contributions, separately.

For a state
$
x'\in\Omega^{M\times I}_u,
$
the quantity
$
\beta_{x',u}(\gamma)
$
is the value of the payments still owed to an agent in state $x'$
under the contract in force when time \(u\) begins. The continuation
problem faced by a breakaway coalition of agents in state \(x'\) is
therefore naturally parameterised by
\(
(\beta_{x',u}(\gamma),X(\gamma))
\),
representing the resources the coalition can walk away with together
with the realised consumption history.

The interpretation is that agents are free to leave a collective
and form a new collective whenever their realised state becomes known.
A strategy is acceptable only if, at every future time and in every
future state, the resulting continuation is at least as desirable as
any acceptable strategy that could be implemented by such a breakaway
collective.

\begin{remark}
In the one-period case, this immunity-to-breakaway requirement is the
standard notion of the {\em core} in cooperative game theory.
Definition~\ref{def:acceptableNoIdiosyncratic}'s dynamic strengthening,
requiring immunity to breakaway at every future time and state, is
what Becker and Chakrabarti \cite{beckerChakrabarti1995} call the {\em
recursive core}.
\end{remark}

Both criteria quantify over spaces of comparison strategies. We
have phrased them in terms of measurable families: each individual
comparison then involves only measurable functions of the
conditioning state, so the almost-everywhere quantifiers are
meaningful and the recursive definition introduces no measurability
difficulties. Here a measurable family is a map into the separable
strategy space of Lemma~\ref{lemma:separabilityA} (Appendix~\ref{appendix:infiniteFund}) which is measurable in the conditioning
state; the measurable-selection machinery of
Section~\ref{sec:compactness} produces comparison families of
exactly this form.

It is useful to record the value associated with
Definition~\ref{def:acceptableNoIdiosyncratic}.

\begin{definition}[Contract-space value]
\label{def:contractValue}
For $s\in{\cal T}$, $x'\in\Omega^{M\times I}_s$, a budget $\beta'$
and a history $\overline X$, define
\[
\bar v_s(x',\beta',\overline X)
:=
\sup\bigl\{
{\cal J}_s(X(\gamma')\mid x')
\;:\;
\gamma'\in{\cal H}^{x'}_s(\beta',\overline X)
\bigr\},
\]
with the convention $\sup\emptyset:=-\infty$.
\end{definition}

The acceptability criterion \eqref{eq:acceptabilityCriterion}
states that an acceptable strategy delivers at least the utility of
every measurable family of alternatives available at formation.
Under the assumptions introduced in the next subsection, this is the
same as delivering utility $\bar v_s(x',\beta(x'),\overline X)$ at
almost every state $x'$: acceptable strategies are the statewise
maximisers, within
${\cal H}^x_s(\beta,\overline X)$, of the contract-space value.

Definition~\ref{def:acceptableNoIdiosyncratic} is recursive and
therefore presupposes the existence of appropriate maximisers at
each stage; in particular, we must ensure that the sets of
acceptable strategies are non-empty. The next subsection introduces
a classical consumption--investment problem, together with the
well-posedness assumption we shall use, and relates the two
problems: the classical value function equals
$\bar v$, and acceptable strategies correspond to optimisers of the
classical problem.

\subsection{A classical reformulation}
We have formulated the problem in terms of an evolving series of
contracts. This is economically natural and necessary in order to
model the purchase of illiquid claims, such as long-term insurance
contracts.

We will now define two continuation problems, formulated directly in
terms of consumption streams. The self-financing problem mirrors the
acceptability criterion, and so does not allow wealth to be
redistributed between states; the redistributive problem mirrors only
the admissibility criterion, and so does.

In a complete-market setting the self-financing problem reduces the
contracting problem for homogeneous collectives to a classical
consumption-investment problem in which agents need trade only liquid
assets, namely assets whose payoffs are realised in the following
period. It is this problem that carries the economic content, and it is
the focus of our arguments. The redistributive problem is introduced
only because it will supply a technical upper bound on utility, which we
use in the proof of Theorem~\ref{thm:generalInfiniteFundLimit}.

\begin{definition}[Classical homogeneous continuation problem]
\label{def:classicalContinuationProblem}

Let
$
s\in{\cal T}$,
$
x\in
\Omega^M_s
\sqcup
\Omega^{M\times I}_s,
$
$
\beta
\in
L^1
(
{\mathbb R},
{\cal F}_s^\infty,
{\mathbb Q}_x^\infty
),
$
and
$
\overline X
\in
L^1
(
{\mathbb R}_{\ge0},
({\cal F}_t^\infty)_{t\in{\cal T}_{<s}},
{\mathbb Q}_x^\infty
).
$
Let $\Prob'=({\cal J}',\tau',\tau'_d,\eta')$ be a
problem in the sense of the discussion before
Definition~\ref{def:admissibleNoIdiosyncratic}, understood to be the
ambient problem when omitted.
A {\em stream} is an element
\[
X
\in
L^1
\!\Bigl(
{\mathbb R}_{\ge0},
({\cal F}_t^\infty)_{t\in{\cal T}},
{\mathbb Q}_x^\infty
\Bigr),
\]
i.e.\ a consumption process indexed by all of \({\cal T}\), matching the
domain on which \({\cal J}'_s(\,\cdot\mid x)\) is defined. For a consumption stream
\(X\), define
\[
C_u(X)
:=
{\mathbb E}_{{\mathbb Q}^\infty}
\!\left(
\sum_{t=u}^T
X_t
\,\middle|\,
{\cal F}_u^\infty
\right),
\qquad
u\in{\cal T}_{\ge s},
\]
the conditional value of the remaining consumption. Define the
{\em self-financing pot process} of $X$ recursively by
$R^X_{T+1}:=0$ and
\[
R^X_u
:=
\Bigl(
X_u
+
{\mathbb E}_{{\mathbb Q}^\infty_x}\bigl(R^X_{u+1}\,\big|\,{\cal F}^\infty_u\bigr)
-
\eta'_u
\Bigr)^{\!+},
\qquad
u\in{\cal T}_{>s},
\]
and, with $m_w$ as in
Definition~\ref{def:admissibleNoIdiosyncratic}, the
{\em redistributive pot process} of $X$ recursively by
$\hat R^X_{T+1}:=0$ and
\[
\hat R^X_u
:=
\Bigl(
m_u(X_u)
+
m_u\bigl(\hat R^X_{u+1}\bigr)
-
m_u(\eta'_u)
\Bigr)^{\!+},
\qquad
u\in{\cal T}_{>s}.
\]

We consider streams satisfying
\begin{enumerate}[(i)]
\item the consistency constraint
\begin{equation}
X_t=\overline X_t
\qquad
\forall\,t\in{\cal T}_{<s}
\qquad
{\mathbb Q}_x^\infty\text{-a.s.};
\label{eq:classicalConsistency}
\end{equation}
\item the non-participation constraint
\begin{equation}
X_t=0
\qquad
\forall\,t<\tau'\text{ or }t\ge\tau'_d
\qquad
{\mathbb Q}_x^\infty\text{-a.s.}
\label{eq:classicalNonParticipation}
\end{equation}
\end{enumerate}

We define the set of {\em self-financing streams}
$
{\cal X}_s^x(\Prob';\beta,\overline X)
$
to consist of all streams satisfying
\eqref{eq:classicalConsistency},
\eqref{eq:classicalNonParticipation}, the self-financing funding
constraint
\begin{equation}
X_s
+
{\mathbb E}_{{\mathbb Q}^\infty_x}\bigl(R^X_{s+1}\,\big|\,{\cal F}^\infty_s\bigr)
\leq
\beta
+
\eta'_s
\qquad
{\mathbb Q}_x^\infty\text{-a.s.},
\label{eq:homogeneousBudget}
\end{equation}
and the requirement that no pot is carried outside participation:
\begin{equation}
{\mathbb E}_{{\mathbb Q}^\infty_x}\bigl(R^X_{u+1}\,\big|\,{\cal F}^\infty_u\bigr)
=
0
\quad
\text{on }
\{u<\tau'\}\cup\{u\ge\tau'_d\},
\quad
\forall\, u \in {\cal T}_{\geq s}
\qquad
{\mathbb Q}_x^\infty\text{-a.s.}
\label{eq:noPotOutsideParticipation}
\end{equation}

We define the set of {\em redistributively feasible streams}
$
\hat{\cal X}_s^x(\Prob';\beta,\overline X)
$
to consist of all streams satisfying
\eqref{eq:classicalConsistency},
\eqref{eq:classicalNonParticipation} and the redistributive funding constraint
\begin{equation}
m_s(X_s)
+
m_s\bigl(\hat R^X_{s+1}\bigr)
\leq
m_s(\beta)
+
m_s(\eta'_s)
\qquad
{\mathbb Q}_x^\infty\text{-a.s.}
\label{eq:redistributiveBudget}
\end{equation}

The associated value function is
\[
V_s(\Prob';x,\beta,\overline X)
:=
\sup_{X\in{\cal X}_s^x(\Prob';\beta,\overline X)}
{\cal J}'_s(X\mid x).
\]
We write $\hat V_s(\Prob';x,\beta,\overline X)$ for the corresponding
supremum over the redistributively feasible set
$\hat{\cal X}_s^x(\Prob';\beta,\overline X)$.

As with ${\cal A}^x_s(\Prob';b,\overline X)$, the
problem argument is omitted when it is the ambient problem, recovering
$V_s(x,\beta,\overline X)$, ${\cal X}_s^x(\beta,\overline X)$, and
$\tau,\tau_d,\eta$ in place of $\tau',\tau'_d,\eta'$ above.
By~\eqref{eq:classicalConsistency}, every
\(X\in{\cal X}_s^x(\Prob';\beta,\overline X)\) already has
\((X_t)_{t<s}=\overline X\); consequently, given a stream \(X\) first
(for instance the realised consumption stream of an admissible strategy),
we write \({\cal X}_s^x(\Prob';\beta,X)\) as shorthand for
\({\cal X}_s^x\bigl(\Prob';\beta,(X_t)_{t<s}\bigr)\), and similarly
for $\hat{\cal X}_s^x$; membership amounts to
$X$ satisfying~\eqref{eq:classicalNonParticipation} together with
the relevant funding constraints, which involve only
\((X_t)_{t\ge s}\), since
\eqref{eq:classicalConsistency} then holds automatically. Thus only the
history of \(X\) prior to time \(s\) is relevant to the continuation
problem.
\end{definition}

The self-financing constraints still permit insurance: the payment
promised for time $u+1$ may be contingent on the state realised then,
so risks not yet resolved may be pooled.

Applying $m_u$ to the defining
inequalities shows that $m_u(R^X_u)\geq\hat R^X_u$ for every $u$, so
every self-financing stream is redistributively feasible whenever $\beta$ is
$\iota^{\Omega^\infty}{\cal F}^M_s$-measurable modulo
${\mathbb Q}^\infty_x$-null sets; this holds in particular for every
homogeneous state.
Hence at homogeneous states
${\cal X}\subseteq\hat{\cal X}$, so $V_s\le\hat V_s$.

The state variable of the self-financing problem of
Definition~\ref{def:classicalContinuationProblem}
consists only of the current continuation budget together with the
realised consumption history.

We now state the well-posedness assumption that we will use.
\begin{definition}
Let
$s\in{\cal T}$,
$x\in\Omega^M_s\sqcup\Omega^{M\times I}_s$,
\[
\beta\in L^1({\mathbb R},{\cal F}^\infty_s,{\mathbb Q}^\infty_x),
\qquad
\overline X
\in
L^1(
{\mathbb R}_{\ge0},
({\cal F}^\infty_t)_{t\in{\cal T}_{<s}},
{\mathbb Q}^\infty_x
).
\]

An {\em arg max} for the problem
$(s,x,\beta,\overline X)$
is a measurable map
\[
y\mapsto X^y,
\]
defined for every $y\in\Omega^{M\times I}_s$ at which the supremum
defining $V_s(y,\beta(y),\overline X)$ is attained, such that
$
X^y
\in
{\cal X}^y_s(\beta(y),\overline X)
$
and
\[
{\cal J}_s(X^y\mid y)
=
V_s(y,\beta(y),\overline X)
\]
for every such $y$, where $\beta(y)$ denotes the conditional
realization of $\beta$ at $y$. Because the family is defined at
every attainment state, rather than merely almost everywhere with
respect to a fixed measure, it may be used under any conditional
measure.
\end{definition}

\begin{definition}
We say that a problem $\Prob:=({\cal J}, \tau, \tau_d, \eta)$
{\em admits regular maximisers} if, for every
$s\in{\cal T}$,
$x\in\Omega^M_s\sqcup\Omega^{M\times I}_s$,
\[
\beta\in L^1({\mathbb R},{\cal F}^\infty_s,{\mathbb Q}^\infty_x),
\qquad
\overline X
\in
L^1(
{\mathbb R}_{\ge0},
({\cal F}^\infty_t)_{t\in{\cal T}_{<s}},
{\mathbb Q}^\infty_x
),
\]
the following hold.

\begin{enumerate}[(i)]

\item
There exists an arg max for
$(s,x,\beta,\overline X)$. This also holds when $\beta$ is replaced
by any ${\cal F}^\infty_s$-measurable integrable random budget
$y\mapsto\beta(y)$; moreover a single map, jointly measurable in the
state, the budget and the history, may be chosen to serve every
problem at the date $s$ at once.

\item
For ${\mathbb Q}^\infty$-almost every
$x\in\Omega^{M\times I}_s$:
${\cal J}_s(X\mid x)<+\infty$ for every redistributively feasible stream
$X\in\hat{\cal X}^x_s(\beta,\overline X)$, and
$V_s(x,\beta,\overline X)<+\infty$.

\item
For ${\mathbb Q}^\infty$-almost every
$x\in\Omega^{M\times I}_s$
satisfying
\[
{\mathbb Q}^\infty_x(\tau\le s<\tau_d)=1,
\]
the map
\[
c\mapsto V_s(x,\beta+c,\overline X)
\]
is strictly increasing and right-continuous on
$\{c:{\cal X}^x_s(\beta+c,\overline X)\neq\emptyset\}$, and
continuous at every interior point of that set; its supremum equals
\[
\sup\bigl\{{\cal J}_s(Y\mid x)
\;:\;
Y\in L^1({\mathbb R}_{\geq0},({\cal F}^\infty_t)_{t\in{\cal T}},{\mathbb Q}^\infty_x)
\bigr\},
\]
the supremum of the utilities of all integrable streams whatsoever;
and this supremum is not attained by any integrable stream.

\item
For ${\mathbb Q}^\infty$-almost every
$x\in\Omega^{M\times I}_s$
satisfying
${\mathbb Q}^\infty_x(\tau\le s<\tau_d)=1$:
\[
V_s(x,\beta,\overline X)>-\infty
\quad
\text{whenever }
{\cal X}^x_s(\beta,\overline X)\neq\emptyset
\text{ and }
\beta(x)+H^x_s>0,
\]
where
$H^x_s:={\mathbb E}_{{\mathbb Q}^\infty}\bigl(\sum_{t\ge s}\eta_t\,\big|\,\omega^\infty_s=x\bigr)$
denotes the value of the remaining contributions.

\end{enumerate}
\end{definition}

Each part of these definitions plays a clear role in
identifying optimal strategies through a backward induction argument.
We will identify sufficient conditions for
this property to hold in the next section using
function-analytic arguments. The aim of this section
is to focus on the algebraic consequences of our definitions.

The assumption that the map in (iii) is strictly increasing excludes
preferences under which an agent
may be willing to sacrifice part of their available budget in order
to improve the outcomes of other agents. The final condition in
(iii) states that a sufficiently large budget can fund utilities
arbitrarily close to the best achievable by any stream; together
with the non-attainment of the supremum, this reflects the
no-satiation assumption of the Introduction.

\begin{proposition}[Self-financing representation]
\label{prop:selfFinancingRepresentation}
Suppose ${\mathbb Q}^\infty_x(\tau'\leq s)=1$ or, more generally,
that $X$ satisfies \eqref{eq:noPotOutsideParticipation}. Given a
stream $X$, define $\Gamma(X)$ by
\begin{gather*}
\Gamma(X)_{u,t}:=X_t \quad (t \leq u),
\\
\Gamma(X)_{u,u+1}
:=
R^X_{u+1},
\\
\Gamma(X)_{u,t}:=0 \quad (t>u+1),
\end{gather*}
for $u \in {\cal T}_{\geq s}$. If
$X\in{\cal X}^x_s(\beta,\overline X)$ then
$\Gamma(X)\in{\cal A}^x_s(\beta,\overline X)$,
$X(\Gamma(X))=X$, and for every $u\in{\cal T}_{>s}$ and
${\mathbb Q}^\infty_x$-almost every $y\in\Omega^{M\times I}_u$,
\[
\beta_{y,u}\bigl(\Gamma(X)\bigr)
=
R^X_u(y).
\]
Under $\Gamma(X)$ the payment promised for time $u+1$ is
contingent on the full time-$(u+1)$ state: each state receives
exactly the pot required to finance its own continuation. This is
the representative through which acceptable strategies will be
constructed.
\end{proposition}

A matching result holds for the redistributive problem: every
redistributively feasible stream is likewise realised by an admissible
strategy, through a redistributive representative $\hat\Gamma(X)$ built from
$\hat R^X$ in place of $R^X$. Since it plays no part in the economic
argument, we defer it to
Proposition~\ref{prop:redistributiveRepresentation} in
Appendix~\ref{appendix:infiniteFund}.

\begin{remark}[Carried-forward form]
\label{rem:carriedForward}
A stream $X$ satisfies \eqref{eq:homogeneousBudget} if and only if
there exists a process $(P_u)_{u\in{\cal T}_{>s}}$, with each $P_u$
nonnegative, integrable and ${\cal F}^\infty_u$-measurable, such that,
setting $P_s:=\beta$ and $P_{T+1}:=0$,
\[
X_u+{\mathbb E}_{{\mathbb Q}^\infty_x}\bigl(P_{u+1}\,\big|\,{\cal F}^\infty_u\bigr)
\leq P_u+\eta'_u
\qquad
{\mathbb Q}_x^\infty\text{-a.s.}
\]
for every $u\in{\cal T}_{\geq s}$: the pot carried into time $u+1$ is
what remains of the pot entering time $u$, plus the contribution
received, after paying for time-$u$ consumption. Indeed, backward
induction shows that any such process satisfies $P_u\geq R^X_u$, while
$(R^X_u)_{u\in{\cal T}_{>s}}$ itself has the stated properties whenever
\eqref{eq:homogeneousBudget} holds. A stream is thus equivalently
described, one period at a time, by its current consumption together
with the pot it hands on to the next period. In particular, if
$X\in{\cal X}^x_s(\beta,\overline X)$ then, for every
$u\in{\cal T}_{>s}$, the deferred stream $X$ satisfies the
constraints of ${\cal X}^y_u\bigl(R^X_u(y),X\bigr)$ for
${\mathbb Q}^\infty_x$-almost every $y\in\Omega^{M\times I}_u$: a
self-financing stream is self-financing along its own tails, with
budgets given by its own self-financing pots. The analogous statement
holds for the redistributive pot process, with $m_u$ in place of the
conditional expectations
${\mathbb E}_{{\mathbb Q}^\infty_x}(\,\cdot\,|\,{\cal F}^\infty_u)$
and \eqref{eq:redistributiveBudget} in place of
\eqref{eq:homogeneousBudget}.
\end{remark}

For
$
X
\in
{\cal X}_s^x(\beta,\overline X)
$
and
$
u>s,
$
suppose we are given a nonnegative, integrable,
${\cal F}^\infty_u$-measurable random budget $y\mapsto B(y)$ on
$\Omega_u^{M\times I}$ and a measurable family
\[
y
\mapsto
Y^y
\in
{\cal X}_u^y
\bigl(
B(y),
X
\bigr),
\qquad
y\in\Omega_u^{M\times I},
\]
defined for ${\mathbb Q}^\infty_x$-almost every compatible $y$.
Define the {\em continuation splice}
$
X\sqcup_{u}Y
$
by
\[
(X\sqcup_{u}Y)_t(\omega)
=
\begin{cases}
X_t(\omega),
&
t<u,
\\[0.5ex]
Y^{\omega^\infty_u(\omega)}_t(\omega),
&
t\ge u.
\end{cases}
\]

\begin{lemma}[Continuation splicing]
\label{lemma:continuationSplicing}
In the setting above, suppose additionally that
\[
B\leq R^X_u
\qquad
{\mathbb Q}^\infty_x\text{-almost surely}.
\]
Then
\[
X\sqcup_{u}Y
\in
{\cal X}_s^x(\beta,\overline X),
\]
and, for ${\mathbb Q}^\infty_x$-almost every $y$, the tail of
$X\sqcup_u Y$ at $y$ is the tail of $Y^y$, so that
\[
{\cal J}_u\bigl(X\sqcup_u Y\,\big|\,y\bigr)
=
{\cal J}_u\bigl(Y^y\,\big|\,y\bigr).
\]
\end{lemma}
See Appendix~\ref{appendix:infiniteFund} for the proofs of
Proposition~\ref{prop:selfFinancingRepresentation} and
Lemma~\ref{lemma:continuationSplicing}.

The lemma formalises the idea of replacing all the continuations at
a given date, without modifying the past: if each state's new
continuation is self-financing from a pot no larger than the pot the
original stream would have carried into that state, then the spliced
stream is again self-financing.

The next proposition is the point at which the
continuation-consistency axiom enters the theory. The self-financing
representation identifies continuation budgets with self-financing pots,
while the splicing lemma allows
continuations to be replaced without altering the realised past.
Continuation consistency then implies that optimality of every
future continuation is equivalent to optimality of the current
continuation.

Consequently the recursive notion of acceptability is reduced to a
classical dynamic programming problem.

\begin{proposition}[Reduction to the classical continuation problem]
\label{prop:reductionClassical}

Assume that $\Prob$ admits regular maximisers.

Let
$
s\in{\cal T}$,
$x\in\Omega_s^{M\times I},
$
and suppose that
${\mathbb Q}^\infty_x(\tau\le s<\tau_d)=1$.
For
$
\gamma
\in
{\cal A}_s^x(\beta,\overline X),
$
the following are equivalent.
\begin{enumerate}[(i)]

\item
$
\gamma
\in
{\cal G}_s^x(\beta,\overline X).
$
\item
$
\gamma
\in
{\cal H}_s^x(\beta,\overline X)
$
and
$
{\cal J}_s(X(\gamma)\mid x)
=
V_s
\bigl(
x,
\beta,
\overline X
\bigr).
$
\end{enumerate}
Moreover, there exists an optimiser
$
X^\circ
\in
{\cal X}_s^x(\beta,\overline X)
$
which is statewise optimal at every future date, and every such
optimiser induces an acceptable strategy through its self-financing
representative
\(\Gamma(X^\circ)\).

\end{proposition}
See Appendix~\ref{appendix:infiniteFund} for the proof.

\begin{corollary}[Coincidence of the contract-space and classical values]
\label{cor:valueCoincidence}
Under the hypotheses of
Proposition~\ref{prop:reductionClassical},
\[
\bar v_s(x,\beta,\overline X)
=
V_s(x,\beta,\overline X),
\]
and the supremum defining $\bar v_s$ is attained, exactly by the
elements of ${\cal G}^x_s(\beta,\overline X)$.
\end{corollary}
\begin{proof}
Lemma~\ref{lemma:deliveredValueBound}(iii), applied at the
homogeneous state $x$, bounds the utility delivered by any element
of ${\cal H}^x_s(\beta,\overline X)$ by $V_s(x,\beta,\overline X)$,
so $\bar v_s\leq V_s$. By
Proposition~\ref{prop:acceptableFromOptimisers} the bound is
attained by an element of
${\cal G}^x_s(\beta,\overline X)\subseteq{\cal H}^x_s(\beta,\overline X)$,
so $\bar v_s=V_s$; and by
Proposition~\ref{prop:reductionClassical} the elements of
${\cal G}^x_s(\beta,\overline X)$ are exactly the members of
${\cal H}^x_s(\beta,\overline X)$ attaining it.
\end{proof}

The proposition identifies the role of continuation consistency: it
is precisely the property that allows the recursive contracting
problem to collapse to a classical dynamic programming problem. In
the notation of the proof outline in the Introduction, this
establishes $N_0H_0T_0=N_0H_0T_1$.

\subsection{No mutually beneficial contracts between states}

The admissible strategy space permits arbitrary long-dated contingent
commitments and transfers of continuation wealth across future
states. Since such contracts are unavailable in a classical
consumption-investment problem, a priori one would expect them to generate
additional welfare gains. The previous section shows that for the homogeneous problem this additional flexibility brings no welfare gains.

The next theorem shows that the same applies to the inhomogeneous problem and, moreover, there is no welfare gain from allowing agents in different states to form contracts with one another.

\begin{theorem}[No mutually beneficial contracts across states]
\label{thm:noMutuallyBeneficialContracts}

Assume that $\Prob$ admits regular maximisers.

Let
$
s\in{\cal T}$,
$
x\in
\Omega^M_s
\sqcup
\Omega^{M\times I}_s$,
and let
$
\gamma
\in
{\cal G}_s^x(\beta,\overline X)$.
Suppose that
${\mathbb Q}^\infty_x(\tau\le s<\tau_d)=1$.
Then:
\begin{enumerate}[(i)]
\item
the issuance-funding condition~\eqref{eq:budgetJ} and every
self-financing renegotiation condition~\eqref{eq:selfFinancing} hold
with equality;
\item
for
\(
{\mathbb Q}_x^\infty
\)-almost every compatible state
$
y\in\Omega_s^{M\times I},
$
\[
{\cal J}_s(X(\gamma)\mid y)
=
V_s\bigl(y,\beta(y),\overline X\bigr);
\]
\item
for every $u\in{\cal T}_{>s}$ and
\(
{\mathbb Q}_x^\infty
\)-almost every compatible state
$
x'\in\Omega_u^{M\times I},
$
\[
{\cal J}_u(X(\gamma)\mid x')
=
V_u\bigl(x',\beta_{x',u}(\gamma),X(\gamma)\bigr).
\]
\end{enumerate}

Consequently the utility delivered to every state, at formation and
at every subsequent date, equals the value of the corresponding
self-financing homogeneous continuation problem. In particular, no
contractual rearrangement of continuation
wealth across states can generate a strict welfare
improvement over statewise self-financing planning.

\end{theorem}

The proof is given in Appendix~\ref{appendix:noMutuallyBeneficial}. We
outline the argument. The acceptability criterion, together
with the existence of acceptable strategies built from statewise
optimisers (Proposition~\ref{prop:reductionClassical}), shows that
every state receives at least its self-financing value. For the
reverse inequality, we measure the cost of what each state actually
receives: for each future state, the statewise budget needed to
deliver the utility that \(\gamma\) delivers there. The breakaway
constraints bound this cost below by the value of the claims the
contract assigns to that state, while a backward induction through
the funding conditions, using the dynamic-programming structure of
\(V\), bounds the total cost above by the resources actually
available. Since the measure is positive, the two bounds force
statewise equality and binding funding conditions simultaneously.

The theorem identifies the precise sense in which the contractual
structure is redundant in a complete market. Although admissible
strategies permit arbitrary long-dated contingent commitments, every
acceptable strategy already coincides almost surely with the statewise
solution of the corresponding homogeneous continuation problems.
Thus the recursive contracting problem and the classical
consumption--investment problem have the same continuation values.

Combining
Theorem~\ref{thm:noMutuallyBeneficialContracts}
with
Proposition~\ref{prop:reductionClassical},
the collective investment problem reduces to a family of classical
dynamic programming problems indexed by the current state. In the
notation of the proof outline in the Introduction, this establishes
$N_0H_1T_1=N_0H_0T_0$, and hence that all of $N_0H_*T_*$ are equal.
Consequently, existence and characterisation results may be obtained
using standard methods from stochastic control. Since our primary
interest is the collective investment problem and its economic
implications, we do not pursue these calculations here.

For the problem described in
Example~\ref{example:blackScholes},
a recursive analytical formula for the optimal investment strategy is
obtained in \cite{armstrong_buescu_dalby}. The authors further extend
the analysis to preferences of Epstein--Zin type. Because the market
is complete, the resulting backward recursion is particularly simple:
at each time step the optimal investment strategy may be recovered by
a duality argument, after which the value function is obtained by
backward iteration.

Finally, the problem decomposes fibrewise across agent types.
Each homogeneous continuation problem may be solved independently, so
agents need not share a common model of market dynamics.
Differences in beliefs are absorbed into the corresponding preference
functionals, and the problem for a given agent type may be solved
without reference to the beliefs of any other agents.

\section{Well-posedness}
\label{sec:compactness}
\subsection{Compact problems}
\label{sec:calculusVariations}

The previous section reduced the recursive contracting problem to a
family of homogeneous continuation problems. The remaining question
is whether these continuation problems are well posed.

The existence argument follows the classical direct method of the
calculus of variations. One first constructs a maximising sequence,
then extracts a convergent subsequence, and finally shows that the
limit remains feasible and attains the value.

The existence theory concerns only the homogeneous continuation
problems introduced in
Definition~\ref{def:classicalContinuationProblem}.
The compactness condition introduced below captures the preference
properties required for these problems to be well posed.

\begin{definition}[Compact problem]
\label{def:compactProblem}

We say that
$
\Prob
=
({\cal J},\tau,\tau_d,\eta)
$ is compact if the following hold for every homogeneous continuation
problem
$
(s,x,\beta,\overline X).
$
\begin{enumerate}[(i)]

\item
{\bf (Fatou property)} Whenever
$
X^n
\in
\hat{\cal X}_s^x(\Prob;\beta,\overline X)
$
with
$
X^n\to X
$
almost surely and
$
\inf_n
{\cal J}_s(X^n\mid x)
>
-\infty,
$
we have
\[
\limsup_{n\to\infty}
{\cal J}_s(X^n\mid x)
\leq
{\cal J}_s(X\mid x).
\]

\item
{\bf (Attainment and stability)} If
${\cal X}_s^x(\beta,\overline X)\neq\emptyset$ and
$V_s(x,\beta,\overline X)>-\infty$, the supremum defining
$V_s(x,\beta,\overline X)$ is attained. Moreover, whenever
$(\beta_n,\eta'_n)\to(\beta,\eta')$ in $L^1$, where $\eta'_n$ are
admissible payment streams perturbing the problem as in
Definition~\ref{def:compactIndependentEta}, with each constraint set
${\cal X}_s^x(\Prob(\eta'_n);\beta_n,\overline X)\neq\emptyset$ and
each value exceeding $-\infty$, every sequence of optimisers
$X^n$ for the perturbed problems with
$\inf_n{\cal J}_s(X^n\mid x)>-\infty$ admits a subsequence
converging in $L^1$ to an optimiser of the limiting problem;
in particular the limiting value exceeds $-\infty$ in this case.

\item
{\bf (Measurability)}
There exists a family
\[
y\mapsto X^y,
\qquad
y\in\Omega_s^{M\times I},
\]
such that
\[
X^y
\in
\argmax_{Z\in{\cal X}_s^y(\Prob;\beta,\overline X)}
{\cal J}_s(Z\mid y)
\]
for every state
\(y\in\Omega_s^{M\times I}\)
at which the supremum is attained; moreover a single such family may
be chosen for all continuation problems at the date $s$ at once,
jointly measurably in the state $y$, the conditional realization
$\beta(y)$ of the budget, and the history $\overline X$, the
histories ranging over the separable space
$L^1({\mathbb R}_{\ge0},({\cal F}^\infty_t)_{t<s},{\mathbb Q}^\infty)$;

\item
{\bf (Non triviality)}
for every homogeneous continuation problem
$(s,x,\beta,\overline X)$, both the self-financing and the
redistributive problems have finite value:
\[
\sup_{X\in{\cal X}^x_s(\beta,\overline X)}
{\cal J}_s(X\mid x)
<
+\infty
\qquad\text{and}\qquad
\sup_{X\in\hat{\cal X}^x_s(\beta,\overline X)}
{\cal J}_s(X\mid x)
<
+\infty .
\]
At homogeneous states ${\cal X}\subseteq\hat{\cal X}$, so the second
is the stronger requirement and implies the first; we state both
because the two are appealed to separately below.

\end{enumerate}

\end{definition}

\begin{definition}[Compactness independent of $\eta$]
\label{def:compactIndependentEta}
We say that $\Prob=({\cal J},\tau,\tau_d,\eta)$ is \emph{compact
independent of $\eta$} if $({\cal J},\tau,\tau_d,\eta')$ is compact
(Definition~\ref{def:compactProblem}) for every admissible
payment stream $\eta'$, i.e.\ every nonnegative $\eta'\in
L^1({\mathbb R}_{\ge0},\iota({\cal F}_t^M\times{\cal
F}_t^I)_{t\in{\cal T}},{\mathbb Q}^M\times{\mathbb Q}^I)$ with
$\eta'_t=0$ unless $\tau\le t<\tau_d$.
\end{definition}

The compactness conditions are analytic properties of the market
and the preferences. The no-satiation assumption of the Introduction
is formalised as follows; it involves the participation window as
well as the preferences.

\begin{definition}[No satiation]
\label{def:noSatiation}
For $s\in{\cal T}$ write $e^s$ for the participation-window stream
$e^s_t:=\indicator_{\{\tau\le t<\tau_d\}}\indicator_{\{t\ge s\}}$.
We say that $\Prob=({\cal J},\tau,\tau_d,\eta)$ satisfies
{\em no satiation} if, for every $s\in{\cal T}$ and
${\mathbb Q}^\infty$-almost every $x\in\Omega^{M\times I}_s$ with
${\mathbb Q}^\infty_x(\tau\le s<\tau_d)=1$:
\begin{enumerate}[(i)]
\item
${\cal J}_s(X+c\,e^s\mid x)>{\cal J}_s(X\mid x)$
for every integrable stream $X$ and every $c>0$;
\item
the supremum of ${\cal J}_s(Y\mid x)$ over all integrable streams
$Y$ is not attained;
\item
${\cal J}_s(c\,e^s\mid x)>-\infty$ for every $c>0$, and
$V_s(x,\beta,\overline X)>-\infty$ for every homogeneous
continuation problem at $x$ with
${\cal X}^x_s(\beta,\overline X)\neq\emptyset$ and
$\beta(x)+H^x_s>0$.
\end{enumerate}
Condition (i) states that additional wealth consumed during the
participation window is always strictly valued; condition (ii) that
no finite consumption stream is fully satiating; and condition (iii)
that only an agent with neither wealth nor future income may face
unboundedly bad utility.
\end{definition}

Compactness axiomatises the outputs of the classical direct method
of the calculus of variations: attainment, stability of the
optimisers, and measurable optimal selections. In
Appendix~\ref{appendix:calculusVariations} we confirm that these
conditions, together with no satiation, are all that is required to
establish the existence of regular maximisers, the proof combining
value-function regularity in the budget variable with standard
measurable selection machinery.

Verifying that a concrete preference system is compact mirrors the
classical direct method itself. Moment bounds on the state-price
density provide the bound classically supplied by coercivity. In the
verification theorems of
Appendix~\ref{appendix:duality}, Koml\'os' theorem \cite{Komlos1967}
together with convexity
of the constraint sets and these bounds supplies the compactness.
Lower semicontinuity of constraints passes feasibility to the limit
via Fatou's lemma, and the Fatou property supplies the upper
semicontinuity of the objective. The
combination of Koml\'os' theorem with convexity is standard in
mathematical finance in the form of the convex-combination lemma of
Delbaen and Schachermayer \cite{DelbaenSchachermayer1994}.

\begin{theorem}[Compact problems admit regular maximisers]
	\label{thm:regularMaximisers}
	
	Every compact problem
	\[
	\Prob
	=
	({\cal J},\tau,\tau_d,\eta)
	\]
	satisfying the no-satiation condition admits regular maximisers.
\end{theorem}
The proof is given in Appendix~\ref{appendix:calculusVariations}.

Throughout the remainder of the paper, assuming that a
problem is compact should be interpreted as excluding
pathological preference systems for which optimal continuation plans
fail to exist or depend discontinuously on continuation wealth.

This assumption encodes our requirements on the market. The results of Kramkov and Schachermayer allow us to compute when the market is compact in our sense. In particular, if our agents are utility maximisers, the requirement for compactness can be interpreted as a condition on the market's state-price density. We view this requirement as essentially equivalent to well-posedness of our optimisation problems, which justifies our earlier claim that the assumptions needed for our main Theorem to hold are rather mild.

\subsection{Establishing compactness}
We now give examples to show how compactness can be established in practice.

\subsubsection{Examples}
\label{sec:compactnessExamples}

This subsection states compactness for two key test-cases and then describes
the general machinery one can use to obtain compactness results by domination arguments.

\begin{definition}[State-price density]
\label{def:statePriceDensity}

The state-price density on the fibre \(x\) is the Radon--Nikodým
density process
\[
Y_t^x
=
E_x\!\left[
\frac{d{\mathbb Q}}{d{\mathbb P}_x}
\,\middle|\,
{\cal F}_t
\right],
\]
where \({\mathbb P}_x\) is a measure equivalent to ${\mathbb Q}$.
\end{definition}

\begin{definition}[Epstein--Zin preferences]
\label{def:EZPreferences}

Let
\[
\beta_t^x\in(0,1),
\qquad
\alpha_t^x<1,\ \alpha_t^x\ne0,
\qquad
\rho_t^x<1,\ \rho_t^x\ne0,
\]
be fibre- and time-dependent parameters. An Epstein--Zin preference
system \cite{epsteinZin1} is defined recursively by
\[
{\cal J}_T(X\mid x)=X_T,
\]
and, for \(t<T\),
\[
{\cal J}_t(X\mid x)
=
\Bigl(
(1-\beta_t^x)X_t^{\rho_t^x}
+
\beta_t^x
M_t^x(X)^{\rho_t^x}
\Bigr)^{1/\rho_t^x},
\]
\[
M_t^x(X)
=
\Bigl(
E_x
\!\left[
{\cal J}_{t+1}(X\mid x)^{\,\alpha_t^x}
\,\middle|\,
{\cal F}_t
\right]
\Bigr)^{1/\alpha_t^x}.
\]
\end{definition}

In Appendix~\ref{appendix:duality} we establish a moment
bound on the state price density which allows one to compute
whether a market is compact for these preferences. The result is summarised
as Theorem~\ref{thm:EZDuality}. The approach is to compactify the
classical, additively-separable CRRA comparison recursion directly,
by the Koml\'os direct method under a moment bound on the
state-price density (the analogue, in our dynamic setting, of the
dual-side well-posedness conditions of Kramkov and Schachermayer
\cite{kramkov1999asymptotic}) and then to transfer this to
Epstein--Zin preferences by a domination argument, valid whenever
$\rho_t^x>0$ (elasticity of intertemporal substitution greater than
$1$) and $\alpha_t^x\le\rho_t^x$: more risk averse than the CRRA
utility matched to the same elasticity of intertemporal substitution
would dictate, i.e.\ a preference for early resolution of uncertainty
in the terminology of \cite{epsteinZin1}. Both conditions hold in the
empirically standard case in the long-run-risk asset-pricing
literature, where they are required to generate realistic risk premia
\cite{bansalYaron}.

For simplicity, for now we only state
the following Corollary.

\begin{corollary}
	\label{cor:EZBlackScholes}
	Suppose that the preference system is Epstein--Zin
	(Definition~\ref{def:EZPreferences}), that $\rho_t^x>0$ and
	$\alpha_t^x\le\rho_t^x$ at every fibre (the empirically relevant
	case discussed above), and that $Y^x_t$ is a lognormal state-price
	density, then $\Prob$ is compact. This covers the case of the
	Black--Scholes model.
\end{corollary}

\begin{remark}[CRRA von Neumann--Morgenstern utilities]
	\label{rem:CRRAFromEZ}
	Specialising $\rho_t^x=\alpha_t^x=:p$ in
	Definition~\ref{def:EZPreferences}, the recursion collapses to
	\[
	{\cal J}_t^p=(1-\beta_t^x)X_t^p+\beta_t^x E_x[{\cal J}_{t+1}^p\mid{\cal
		F}_t],
	\]
	so $W_t:={\cal J}_t^p/p$ satisfies the classical time-additive
	discounted CRRA recursion, for which compactness holds for every
	$p<1$, $p\ne0$ (Theorem~\ref{thm:CRRACompact}), with no
	restriction on the sign of $p$.
\end{remark}

\begin{definition}[Entropic preferences]
\label{def:entropicPreferences}

Let
\[
\rho_t^x<1,\ \rho_t^x\ne0,
\qquad
\beta_t^x\in(0,1),
\qquad
\theta_t^x<0,
\]
be fibre- and time-dependent parameters. An entropic preference
system is defined recursively by
\[
{\cal J}_T(X\mid x)=\frac{X_T^{\rho_T^x}}{\rho_T^x},
\]
and, for $t<T$,
\[
{\cal J}_t(X\mid x)
=
(1-\beta_t^x)\frac{X_t^{\rho_t^x}}{\rho_t^x}
+
\beta_t^x M_t^x(X),
\]
\[
M_t^x(X)
=
\frac1{\theta_t^x}
\log
E_x
\!\left[
e^{\theta_t^x{\cal J}_{t+1}(X\mid x)}
\,\middle|\,
{\cal F}_t
\right].
\]
\end{definition}

\begin{remark}[Entropic preferences and mortality risk]
Bommier~\cite{bommier} argues, from an axiomatic standpoint, that
preferences of this entropic form are the appropriate choice for
life-cycle problems in which mortality is itself a source of risk.
\end{remark}

We establish a moment bound for entropic preferences in Appendix~\ref{appendix:duality}. This proceeds by a dominance argument which shows
that this case follows from the von Neumann--Morgenstern case for CRRA utility.
The result is summarised as Theorem~\ref{thm:entropicDominationCompact}.
Again, for now we simply state the following Corollary.

\begin{corollary}
	\label{cor:entropicBlackScholes}
	Suppose that the preference system is entropic
	(Definition~\ref{def:entropicPreferences}) and $Y^x_t$ is a lognormal
	state-price density, then $\Prob$ is compact.
\end{corollary}

\subsubsection{Domination}
\label{sec:dominationMainText}

Both compactness results above are proved by comparison to the
classical, additively-separable CRRA comparison recursion of
Remark~\ref{rem:CRRAFromEZ}, which is compactified as Theorem~\ref{thm:CRRACompact} of
Appendix~\ref{appendix:duality}. Two independent forms of domination
effect this comparison: domination of a \emph{preference system}
against the CRRA baseline, and domination of a \emph{market} against
a reference market. Both express the same qualitative principle: greater pessimism only ever makes a problem easier to compactify,
never harder. We give the statements here, but defer the  proofs to Appendix~\ref{appendix:duality}.

\begin{lemma}[Domination]
\label{lemma:domination}
Suppose a preference system ${\cal J}$ is built, in place of the
linear conditional expectation in the CRRA comparison recursion, from
a certainty-equivalent operator $M_t^x$ that is monotone and
dominated by conditional expectation (as holds for the entropic
certainty equivalent of Definition~\ref{def:entropicPreferences}):
that is,
\[
{\cal J}_t(X\mid x)
=
(1-\beta_t^x)\frac{X_t^{\rho_t^x}}{\rho_t^x}+\beta_t^xM_t^x\bigl({\cal
J}_{t+1}(X\mid\cdot)\bigr),
\qquad
{\cal J}_T(X\mid x)=\frac{X_T^{\rho_T^x}}{\rho_T^x}.
\]
Then ${\cal J}_t(X\mid x)\le{\cal J}_t^{\rm CRRA}(X\mid x)$ for every
admissible $X$ and every $t$. Consequently condition~(iv) of
Definition~\ref{def:compactProblem} holds for
$\Prob=({\cal J},\tau,\tau_d,\eta)$ whenever it holds for
$\Prob^{\rm CRRA}$, independently of $\eta$.
\end{lemma}

The remaining
compactness conditions are verified for each preference family by
adapting the direct arguments used for the CRRA recursion, with
monotonicity and domination of the certainty equivalent entering
through the same inequalities.

This is what underlies Theorem~\ref{thm:entropicDominationCompact}.
Epstein--Zin preferences are aggregated through a CES
combination of current consumption and the certainty equivalent,
rather than additively, so this Lemma does not apply to ${\cal
J}$ directly: ${\cal J}_t$ and ${\cal J}_t^{\rm CRRA}$ are homogeneous
of degree $1$ and $\rho_t^x$ respectively, and so scale differently in
$X$. When $\rho_t^x>0$ at every fibre, however, raising ${\cal J}_t$
to the power $\rho_t^x$ restores an additively separable recursion of
exactly this Lemma's form, and Theorem~\ref{thm:EZDuality}
is proved by applying it to this reparametrisation
(Appendix~\ref{appendix:duality}).

There is a second, independent form of domination, comparing
\emph{markets} rather than preferences. For a positive random
variable $Z$ (here, a one-period state-price ratio) write $G_Z$
for its quantile function, the increasing generalised inverse of its
distribution function, so that $Z\stackrel{d}{=}G_Z(U)$ for
$U\sim\mathrm{Unif}[0,1]$ and $E[Z^q]=\int_0^1G_Z(x)^q\,\ed x$ for
every $q$ (the standard probability integral transform). The
one-period state-price ratios underlying the compactness theorems of
Appendix~\ref{appendix:duality} are exactly of this form, so the
moment condition those theorems require is a statement purely about
the growth of the corresponding $G_Z$, and we may compare markets by
comparing this growth.

\begin{definition}[Market domination]
\label{def:marketDomination}
Let $G_1,G_2:[0,1]\to(0,\infty)$ be increasing. We say the market with
quantile function $G_1$ is \emph{dominated} by the market with
quantile function $G_2$, written $G_1\preceq G_2$, if
$G_1(x)=O(G_2(x))$ as $x\to1^-$, i.e.\
$\limsup_{x\to1^-}G_1(x)/G_2(x)<\infty$.
\end{definition}

\begin{proposition}[Market domination transfers the moment condition]
\label{prop:marketDomination}
Suppose $G_1\preceq G_2$ and $q>0$. Then $\int_0^1G_2(x)^q\,\ed
x<\infty$ implies $\int_0^1G_1(x)^q\,\ed x<\infty$.
\end{proposition}

\begin{corollary}[Compactness transfers from a dominating market]
\label{cor:marketDomination}
Suppose that, at every date $r$, the one-period market realising the
relevant state-price ratio is dominated (Definition~\ref{def:marketDomination},
applied to $\zeta_r^x$ if $\rho_r^x<0$ or to $1/\zeta_r^x$ if
$0<\rho_r^x<1$) by a reference one-period market with the required
moment finite (for instance a Black--Scholes market, whose
state-price density $\ed{\mathbb Q}/\ed{\mathbb P}$ is lognormal, and
hence so is its reciprocal, so both have finite moments of every
order). Then the moment condition of
Theorem~\ref{thm:CRRACompact}, Theorem~\ref{thm:EZDuality} (for
which $0<\rho_r^x<1$ is the case needed for an empirically realistic
elasticity of intertemporal substitution) and
Theorem~\ref{thm:entropicDominationCompact} holds, and $\Prob$ is
compact.
\end{corollary}

This makes precise the informal claim
that markets in which adverse events occur more often than in the
Black--Scholes model can still be handled: any market whose
state-price density has an upper tail no heavier (in the big-O sense
above) than a lognormal one inherits compactness from Black--Scholes,
without needing to be lognormal itself. The elementary
comparisons underlying Lemma~\ref{lemma:domination} and
Proposition~\ref{prop:marketDomination} are well known in their own
right; we are not, however, aware of a reference that packages them, as
here, into a transfer principle for the compactness of
Definition~\ref{def:compactProblem}, so we have proved
Lemma~\ref{lemma:domination},
Proposition~\ref{prop:marketDomination} and
Corollary~\ref{cor:marketDomination} directly in
Appendix~\ref{appendix:duality}.

\begin{remark}[Extending to other optimisation paradigms]
\label{rem:dominationBeyond}
The domination principle above is not tied to these three specific
families. The robust recursive preferences of
Example~\ref{example:robustRecursivePreferences} are governed by the
same mechanism, since a minimum of expectations is monotone and
dominated by any single one of them. The same idea should extend to
prospect-theory-style preferences that overweight bad outcomes (a
sufficiently pessimistic rank-dependent or Choquet-style distortion
is again dominated by plain expectation, and S-shaped value
functions themselves by comparing their growth and decay against a
CRRA reference), though we do not develop this here. What the
framework cannot accommodate is the standard, time-inconsistent
formulation of prospect theory itself: this paper rests on
continuation-consistency (Kreps--Porteus time-consistency), which
the usual dynamic extension of probability weighting is well known
to violate.
\end{remark}

\section{Finite funds}
\label{sec:finiteFunds}

\subsection{Admissible strategies}

From this point on we make two assumptions, which hold for
the remainder of the paper. First we assume that $\Omega^I$ is a
finite probability space: this is required to make it meaningful to
talk about agents being in the same state.

Second we assume that
preferences are {\em Markovian}, by which we mean that ${\cal J}_t(X)$
does not depend on the components of $(X_s)_{s<t}$. Since consumption
histories are therefore irrelevant to preferences, we no longer need to
track a history $\overline X$. This is made purely to ease the notational burden; we believe the key results hold without it.

More significantly, we now change our interpretation of the probability space $\Omega^I$. We will view it with the conventional probabilistic interpretation as determining stochastically the possible outcomes for a given agent.
The augmented space ${\underline{\Omega}}^{I+}$ contains additional information about the agent that is not relevant for determining their preferences.

Let us now define the total probability space for our finite-fund problem. It is defined by  $\underline{\Omega}^*:=\underline{\Omega}^M \times \underline{\Omega}^\can_{\cal T} \times \prod_{i \in {\mathbb N}} \underline{\Omega}^{I+}$. The cashflows of financial contracts will be modelled using random variables on this total space. The factor $\underline{\Omega}^\can_{\cal T}$ allows contracts to use information from outside the market but which is not tied to any agent.

Write $(\omega^{\can,i},\omega^{I,i}) \in \Omega^{I+} \cong \Omega^{\can}_{\cal T} \times \Omega^I$ for the coordinates giving the state of agent $i$.

We define a $\sigma$-algebra which isolates the financially irrelevant information: that carried neither by the market nor by the idiosyncratic state of any agent. It is defined by
\begin{equation}
{\cal F}^{\perp,M,I}_t:=\iota^{\perp,{\Omega}^*}( {\cal F}^M_t \times \prod_{i \in {\mathbb N}} \iota^{\Omega^{I+}}({\cal F}^I_t)).
\label{eq:defFPerpMI}
\end{equation}
We also define a $\sigma$-algebra which determines which contracts can be purchased at time $s$ with payoff at a later time $t$.

Define
\begin{equation}
{\cal F}^M_{s,t} = \iota^{\Omega^*}\left(
{\cal F}^M_t \times {\cal F}^\can_s \times \prod_{i \in {\mathbb N}}{\cal F}^{I+}_s\right).
\label{eq:defFMst}
\end{equation}

The acceptability criteria and value functions of this section
evaluate an agent's preferences conditionally on states carrying
more information than the agent's own market-and-path state: full
states of $\Omega^*_u$, or states of coarser $\sigma$-algebras,
recording other agents' coordinates and the auxiliary coordinates.

The map recording the market and agent $i$'s own path is an enrichment
(Definition~\ref{def:enrichment}) of $\underline{\Omega}^{M\times I}$
whose $+$-lift is $\underline{\Omega}^*$, and we write ${\cal J}^i$ for
the family of $\underline{\Omega}^*$ preferences induced along it
(Definition~\ref{def:inducedPreferencesEnrichment}). Explicitly, let
$u\in{\cal T}$, let ${\cal G}_u$ be a complete $\sigma$-algebra
generated by market and idiosyncratic information alone and containing
the market information and agent $i$'s own path to time $u$, and let
$y$ be a ${\cal G}_u$-state, with
$x''(i,y)\in\Omega^{M\times I}_u$ the associated market-and-own-path
state. Then, for a stream $Z$ on $\Omega^*$,
\begin{equation}
{\cal J}^i_u(Z\mid y)
:=
{\cal J}_u\bigl(Z\circ\Xi_y\,\big|\,x''(i,y)\bigr),
\label{eq:enrichedValue}
\end{equation}
where $\Xi_y$ is a transport as in
Lemma~\ref{lemma:enrichmentStable}(iii).

A full state $y'\in\Omega^*_u$ refining such a $y$ resolves, under
$\Xi_y$, the canonical coordinates of the $+$-lift at times
${\leq}\,u$, at some value
$\omega^{\can}(y')\in\Omega^{\can}_{{\cal T}_{\leq u}}$, so that
\eqref{eq:enrichedValue} no longer applies. For such states we freeze
the resolved coordinates and define
\begin{equation}
{\cal J}^i_u(Z\mid y')
:=
{\cal J}_u\bigl((Z\circ\Xi_y)^{[\omega^{\can}(y')]}\,\big|\,x''(i,y)\bigr),
\label{eq:frozenEnrichedValue}
\end{equation}
in the notation of the Lottery tolerance axiom. By the Invariance
axiom this depends on neither the choice of $\Xi_y$ nor the choice of
the state $y$ refined.

All expressions of the form ${\cal J}^i_u(Z\mid y)$ in this section,
including those defining the value function of
Definition~\ref{def:homogeneousFinite} and the acceptability
criteria, are understood in this sense. Information resolved by the
conditioning state is frozen in the evaluation: an agent who has
just lost a lottery values the future as the loser.

\begin{definition}
	\label{def:contractSpace}
	Given $N \in {\mathbb Z}_{\geq 1}$ and $s \in {\cal T}$ we define the {\em contract space} ${\cal C}(N,s)$ to be the subspace of
	\[
	\prod_{u \in {\cal T}_{\geq s}}L^1({\mathbb R}^N_{\geq 0}, ({\cal F}^*_t)_{t \in {\cal T}},{\mathbb Q}^*),
	\]
	consisting of elements $\gamma^i_{u,t}$ satisfying the non-participation conditions
	\begin{equation}
		\gamma^i_{u,t} = 0
		\qquad
		\forall\, i \in {\cal I},\ t \in {\cal T}_{\geq u}, \text{ whenever } u < \tau^i \text{ or } u \ge \tau^i_d
		\qquad {\mathbb Q}^*\text{-a.s.}
		\label{eq:nonParticipationI}
	\end{equation}
	\begin{equation}
		\gamma^i_{u,t}=0
		\quad
		\text{on } \{t < \tau^i\} \cup \{t \ge \tau^i_d\}
		\qquad
		\forall\, i \in {\cal I},\ u \in {\cal T}_{\geq s},\ t \in {\cal T}_{\geq u}
		\qquad {\mathbb Q}^*\text{-a.s.}
		\label{eq:noConsumptionNonParticipationI}
	\end{equation}
	We write ${\cal C}^M(N,s)$ for the space of {\em myopic}
	contracts which additionally satisfy
	\begin{equation}
		\gamma^i_{u,t}=0
		\qquad
		\forall\, i \in {\cal I}, u \in {\cal T}_{\geq s}, t \in {\cal T}_{>u+1}
		\qquad {\mathbb Q}^*\text{-a.s.}
		\label{eq:myopicContract}
	\end{equation}
	We write ${\cal C}^S(N,s)$ for the space of {\em symmetric}
	contracts which additionally satisfy
	\begin{equation}
	(\gamma^i_{u,t}-\gamma^j_{u,t})\indicator_{\{\omega^{I,i}=\omega^{I,j}\}}=0 \qquad
	\forall\, i,j \in {\cal I}, u \in {\cal T}_{\geq s}, t \in {\cal T}_{\geq u}
	\qquad {\mathbb Q}^*\text{-a.s.}
	\label{eq:symmetricContract}
	\end{equation}
\end{definition}

\begin{definition}
	\label{def:admissibleIdiosyncratic}
	Let $N \in {\mathbb Z}_{\geq 1}$ be given. Let ${\cal I}=\{0, \ldots, N-1\}$.
	Let $s \in {\cal T}$ and let $x \in \Omega^*_s$.
	Let $b \in L^1({\mathbb R}^N, {\cal F}^*_s,{\mathbb Q}^*_x)$.

	Given
	\[
	\gamma \in {\cal C}(N,s)
	\]
	define, for $u \in {\cal T}_{\geq s}$ and $w \in {\cal T}_{\geq u}$,
	the net position of the fund created by the contract revision at
	time $u$, viewed up to the horizon $w$:
	\begin{equation}
		c_{u,w}:=
		\begin{cases}
		\displaystyle
		\sum_{i \in {\cal I}}\left(
		b^i + \eta^i_s
		- \sum_{s \leq t \leq w} \gamma^i_{s,t} \right),
		& u = s,
		\\[2.5ex]
		\displaystyle
		\sum_{i \in {\cal I}}\left(
		\sum_{u \leq t \leq w} \gamma^i_{u-1,t} + \eta^i_u
		- \sum_{u \leq t \leq w} \gamma^i_{u,t} \right),
		& u \in {\cal T}_{>s}.
		\end{cases}
		\label{eq:defC}
	\end{equation}
	At time $s$ the fund's new contract is funded by the initial
	budgets and the contributions received at time $s$; at each later
	time $u$ it is funded by surrendering the remaining payments
	promised under the contract agreed at time $u-1$, together with the
	contributions received at time $u$. Thus every position is funded
	purely by income received to date.

	We say that $\gamma$ is an admissible
	strategy for the investment problem $(s,x, b)$
	if it satisfies
	\begin{enumerate}[(i)]
		\item The measurability condition $c_{u,w} \in L^0({\mathbb R},{\cal F}^M_{u,w})$,
		with ${\cal F}^M_{u,w}$ as defined in \eqref{eq:defFMst},
		for all $u \in {\cal T}_{\geq s}$ and $w \in {\cal T}_{\geq u}$.
		\item  The funding conditions
		\begin{equation}
			\mathbb{E}_{{\mathbb Q}^*}(c_{u,T} \mid {\cal F}^*_u)\geq 0
			\label{eq:budgetStar}
		\end{equation}
		for all $u \in {\cal T}_{\geq s}$.
	\end{enumerate}
	Since every contract leg is nonnegative,
	${\mathbb E}_{{\mathbb Q}^*}(c_{u,w}\mid{\cal F}^*_u)$ is decreasing
	in the number of new legs counted, so \eqref{eq:budgetStar} for
	$w=T$ is the binding case; the truncated positions $c_{u,w}$ appear
	only in the measurability condition, which expresses the
	complete-market hypothesis at the level of the fund: the net
	position of the fund, at every horizon, must be an exchange-tradable
	claim.

	The contract set $\gamma$ is defined on the full probability space
	$\underline{\Omega}^*$, while admissibility is evaluated under
	the conditional measure ${\mathbb Q}^*_x$.
	
	We write ${\cal A}^{N,x}_s(b)$ for the space of admissible strategies.
	We define the consumption process by $X^i_t(\gamma):=\gamma^i_{t,t}$.
\end{definition}

The interpretation is that $\gamma$ describes the full set of contracts between $N$ agents and their evolution over time. The most interesting condition is the measurability assumption, which captures the complete market hypothesis at the level of the fund rather than of any single contract: it requires the fund's aggregate net position $c_{u,w}$, not each agent's contract $\gamma^i_{u,t}$, to be exchange tradable. The constituent contracts for each member can be bespoke and idiosyncratic to a single agent.

It is helpful to think of $\gamma$ as being a proposal for managing a mutual insurance fund.

To compare outcomes for different proposals $\gamma$, we wish to define some
form of value function that measures the quality of a proposal $\gamma$. This should
depend upon the state of the agent and the market when they first began to participate, but nothing else. In particular it should not depend on the identity of the agent. The next definition identifies an appropriate space ${\Omega}^J$ to parameterise the state of an agent when they join the fund. To ensure this, we build a larger probability space $\Omega^V$ by adjoining factors which allow us to select a random agent and a random state of all the irrelevant factors. This results in a definition of the value function which is invariant under permutations of identity and of irrelevant factors.

This motivates the formal definition below.

\begin{definition}
\label{def:valueFunctionFinite}
Write $\rho$ for the standard action of $S_\infty$ on $\underline{\Omega}^*$ and $(i_1,i_2) \in S_\infty$ for the permutation of ${\cal I}$ swapping the elements $i_1$ and $i_2$.

Choose a filtration-preserving mod 0 isomorphism
\[
\Psi:\underline{\Omega}^M \times \underline{\Omega}^{I+} \times \underline{\Omega}^\can_{\cal T} \to
\underline{\Omega}^*
\]
satisfying
\[ \omega^M=\omega^M \circ \Psi, \]
\[\omega^{I+}=\omega^{I+,i} \circ \rho((0,i)) \circ \Psi \quad \forall i \in {\mathbb N}. \]

Let $\underline{\Omega}^N$ be a standard probability space supporting a uniformly distributed random variable $i^N$ taking values in ${\cal I}$.

Let
\[
\underline{\Omega}^V:=\underline{\Omega}^N \times \underline{\Omega}^M \times \underline{\Omega}^{I+} \times \underline{\Omega}^\can_{\cal T}
\]
and let $\pi^{\perp, N}$ be the projection from $\Omega^V$ onto all the factors orthogonal to $N$. 

Let 
$\underline{\Omega}^J$ be the Borel probability space with underlying Polish space ${\Omega}^J$
obtained by disintegrating ${\mathbb Q}^V$ with respect to the random variable $J:=(\omega^M_\tau, \omega^I_\tau)$. Write ${\mathbb Q}^V_j$ for the associated measure for each $j \in \Omega^J$.

Suppose
\[
\gamma \in \prod_{u \in {\cal T}_{\geq s}}L^1({\mathbb R}^N_{\geq 0}, ({\cal F}^*_t)_{t \in {\cal T}},{\mathbb Q}^*),
\]
and $j \in \Omega^J$, we define the {\em value function} ${\cal V}^\gamma$ by
\[
{\cal V}^\gamma(j):={\cal J}_\tau( X^{i^N}(\gamma) \circ \rho((0,i^N)) \circ \Psi \circ \pi^{\perp,N} \mid j).
\]
\end{definition}

\begin{lemma}[Independence of the choice of $\Psi$]
\label{lem:valueFunctionWellDefined}
The value function ${\cal V}^\gamma$ does not depend on the choice of $\Psi$ in the preceding definition.
\end{lemma}
This is proved by a naturality argument in Appendix~\ref{appendix:admissibleStrategies}: many of the invariance properties we use follow almost automatically from the naturality of our definitions with respect to relabelling and choice of encoding, in the sense that every construction involved commutes with the relevant group actions. Category theory could formalise and streamline such arguments; we do not pursue this here, but flag proofs of this kind as they arise.

\subsection{The homogeneous problem}

We will need to define an appropriate notion of acceptability for our finite fund problem.

The next definition provides an analogue of Definition~\ref{def:classicalContinuationProblem} but for finite funds.
We cannot reduce the problem to the same extent as we did in the infinite fund case because there is idiosyncratic risk that it is useful to insure against. However, we can assume that such insurance is myopic. We will also insist that the strategy is symmetric.

\medskip
We will want to define appropriate $\sigma$-algebras
to ensure that market-irrelevant information is not used in the strategy. This motivates the following definitions. 
Define ${\cal F}^{I+}_{s,t}:= \iota({\cal F}^{\can}_{s} \times {\cal F}^I_{t})$.
Define
\begin{equation}
{\cal F}^{M,I+}_{s,t} = \iota^{\Omega^*}\left(
{\cal F}^M_t \times {\cal F}^\can_s \times \prod_{i \in {\mathbb N}}{\cal F}^{I+}_{s,t}\right).
\label{eq:defFMIst}
\end{equation}

We may now state our definition for the homogeneous problem for finite funds. It is a classical investment-consumption problem, with investment recorded in $\gamma$ and consumption recorded in $X$. It is homogeneous in the sense that the strategies must be symmetric and all agents start in the same state.

\begin{definition}[The homogeneous problem for finite funds]
	\label{def:homogeneousFinite}
	Let $N \in {\mathbb Z}_{\geq 1}$ be given. Let ${\cal I}=\{0, \ldots, N-1\}$.
	Let $s \in {\cal T}$.
	
	Let $x\in
	\Omega^*_s,
	$
	satisfying
	$
	\omega^{I,i}=\omega^{I,j}
	$
	${\mathbb Q}^*_x$-a.s. for all $i, j \in {\cal I}$.
	
	Let
$
\beta
\in
L^1
(
{\mathbb R},
({\cal F}_t^*)_{t\geq s},
{\mathbb Q}_x^*
)
$. Note that $\beta$ is a process in this definition unlike $b$ in Definition~\ref{def:admissibleIdiosyncratic}.
	
	Given
	\[
	\gamma \in {\cal C}^M(N,s) \cap {\cal C}^S(N,s)
	\]
	define, for $u \in {\cal T}_{\geq s}$ and $w \in {\cal T}_{\geq u}$,
	\begin{equation}
		c^h_{u,w}:=
		\begin{cases}
		\displaystyle
		\beta_s +
		\sum_{i \in {\cal I}}\left(
		\eta^i_s
		- \sum_{s \leq t \leq w} \gamma^i_{s,t} \right),
		& u = s,
		\\[2.5ex]
		\displaystyle
		\beta_u +
		\sum_{i \in {\cal I}}\left(
		\sum_{u \leq t \leq w} \gamma^i_{u-1,t} + \eta^i_u
		- \sum_{u \leq t \leq w} \gamma^i_{u,t} \right),
		& u \in {\cal T}_{>s},
		\end{cases}
		\label{eq:defCHomogeneous}
	\end{equation}
	the net position created by the revision at time $u$, funded by the
	exogenous injection $\beta_u$, the surrendered remaining payments of
	the previous contract, and the contributions received at time $u$.
	We say that $\gamma$ is an admissible
	strategy for the homogeneous problem $(s,x, \beta)^h$
	if it satisfies
	\begin{enumerate}[(i)]
		\item The measurability conditions: $c^h_{u,w} \in L^0({\mathbb R},{\cal F}^M_{u,w})$
		for all $u \in {\cal T}_{\geq s}$ and $w \in {\cal T}_{\geq u}$; and
		\begin{equation}
			\gamma^i_{u,t} \in L^0({\mathbb R},\iota^{\Omega^*}({\cal F}^{M,I+}_{s,t}))
			\label{eq:MImeasurabilityHomogeneous}
		\end{equation}
		for all $i \in {\cal I}$, $u \in {\cal T}_{\geq s}$, and $t \in {\cal T}$,
		with ${\cal F}^{M,I+}_{s,t}$ as defined in \eqref{eq:defFMIst},
		ensuring that no leg of the contract uses canonical, market-irrelevant
		information beyond what is already available at time $s$.
		\item  The funding conditions
		\begin{equation}
			\mathbb{E}_{{\mathbb Q}^*}(c^h_{u,T} \mid {\cal F}^*_u)\geq 0
			\label{eq:budgetStarHomogeneous}
		\end{equation}
		for all $u \in {\cal T}_{\geq s}$.
	\end{enumerate}
	We write ${\cal H}^{N,x}_s(\beta)$ for the space of admissible strategies.
	We define the consumption process by $X^i_t(\gamma):=\gamma^i_{t,t}$. The associated value function is
	\[
	V^N_s(x,\beta)
	:=
	\sup_{\gamma\in{\cal H}^{N,x}_s(\beta)}
	{\cal J}_s(X^0(\gamma)\mid x).
	\]
\end{definition}

We wish to relate our finite-fund problems to infinite-fund problems.

Recall from the preliminaries that $\underline\Omega^{\infty+}:=\underline
\Omega^\can_{\cal T}\times\underline\Omega^\infty$, where
$\underline\Omega^\infty=\underline\Omega^M\times\underline\Omega^I$.

We will see that finite-fund problems associated with ${\Prob}$ do not quite correspond to infinite-fund problems associated with ${\Prob}$. Instead they correspond to infinite-fund problems for ${\Prob}^+$ as defined below.

\begin{definition}[The problem $\Prob^+(\tilde{\eta})$]
	\label{def:probPlus}
	Let $\Prob=({\cal J},\tau,\tau_d,\eta)$ be the ambient
	problem.
	
	Let ${\cal J}^+$ be the induced
	family of $\underline\Omega^{\infty+}$ preferences obtained from ${\cal J}$
	by Definition~\ref{def:inducedPreferences}, taking $\underline\Omega^1:=
	\underline\Omega^{M\times I}$ and $\underline\Omega^2:=\underline\Omega^\can_{\cal
		T}$: in particular, ${\cal J}^+_s(X\mid(x,y))$ does not depend on
	$y\in\Omega^\can_{{\cal T},s}$.
	
	Let
	$\tau^+:=\tau\circ\pi^{\Omega^{M\times I}}$,
	$\tau_d^+:=\tau_d\circ\pi^{\Omega^{M\times I}}$, and
	$\eta^+:=\eta\circ\pi^{\Omega^{M\times I}}$.
	
	Given $\tilde{\eta} \in L^1({\mathbb R}_{\geq 0}, ({\cal F}^{\infty+}_t)_{t \in \cal T}, {\mathbb Q}^{\infty+})$ with $\tilde\eta_t=0$ for $t<\tau^+$, we define
	$\Prob^+(\tilde{\eta}):=({\cal J}^+,\tau^+,\tau_d^+,\tilde{\eta})$.
	
	When we write $\Prob^+$ without specifying the parameter $\tilde{\eta}$ we are referring to $\Prob^+(\eta^+)$.
\end{definition}

When working with the homogeneous problem, we will want to use
exchangeable random variables. The following Lemma will therefore be useful.

\begin{lemma}
Define
$
\underline{\Omega}^{{\cal E}^N}:=(\Omega^{\can,\infty}_{\cal T}, {\cal E}^N( {\cal F}^{\can}_{\cal T}),
({\cal E}^N ({\cal F}^{\can}_t))_{t \in {\cal T}}, {\mathbb Q}^{\can,\infty}_{\cal T}).
$

	There exists a measure-preserving map
	$\theta^N:\Omega^{\can,\infty}_{\cal T}\to\Omega^{\can}_{\cal T}$
	satisfying
	\[
	(\theta^N)^{-1}({\cal F}^{\can}_t)={\cal E}^N({\cal F}^{\can}_t)
	\]
	modulo null sets for every $t\in{\cal T}$. Consequently, composition
	with $\theta^N$ defines a bijection, up to almost sure equality,
	between random variables on $\underline{\Omega}^{\can}_{\cal T}$ and
	random variables on $\underline{\Omega}^{{\cal E}^N}$; this bijection
	is law preserving and matches the filtrations.
	\label{lemma:standardFiltered}
\end{lemma}
The proof is given in Appendix~\ref{appendix:finiteHomogeneous}.

\subsection{Tontine strategies}
\label{sec:tontineStrategies}

A basic connection between the finite-fund problem and the infinite-fund problem is given by using the outcome of the infinite-fund strategy as closely as possible in the finite fund case by providing insurance pools. Because this is closely analogous to the use of a tontine to hedge mortality risk we call these strategies tontine strategies. We define two different strategies which differ only in the size of the insurance pools.

\begin{definition}[Tontine strategies]
	\label{def:tontineStrategy}
	Let $N$, ${\cal I}$, be as in Definition~\ref{def:homogeneousFinite}.

	Choose a filtration preserving mod 0 isomorphism
	\[
	\Psi^N: \underline{\Omega}^* \cong \underline{\Omega}^M \times 
	\underline{\Omega}^\can_{\cal T} \times \prod_{i=0}^\infty {\Omega}^{I+} \to 
	\underline{\Omega}^M \times 
	\underline{\Omega}^\can_{\cal T} \times \underline{\Omega}^I
	\cong \underline{\Omega}^{\infty+}
	\]
	which acts trivially on the $\underline{\Omega}^M \times \underline{\Omega}^{I,0}$ components and exchangeably on the first $N$ components of the product.

	We define the {\em restricted tontine strategy} and the {\em broad tontine strategy} recursively. We will use the same symbol $\gamma \in {\cal C}(N,0)$ for both strategies in the recursion step.
	
	{\em Recursion step:}
	
	Suppose that
	$s \in {\cal T}_{<T}$ and that
	$x \in \underline{\Omega}^*_s$ are given
	and that 	
	\[
	\beta^i_s
	\in
	L^1
	(
	{\mathbb R},
	({\cal F}_t^*)_{t\geq s},
	{\mathbb Q}_x^*
	)
	\]
	is given for each $i \in {\cal I}$.
	Divide ${\cal I}$ into a disjoint union of
	sets $({\cal I}^{k}_x)_{k \in {\cal K}_x}$ consisting of agents who
	are all in the same state conditioned on $x$.
	Throughout, we assume, as is standing elsewhere in this section, that $\eta^i_s$ depends on $i$ only through $\omega^{I,i}$; consequently the quantity defined below depends on $i \in {\cal I}^k_x$ only through the class $k$. For each $k \in {\cal K}_x$ define
	\[
	\tilde{\eta}^{(k)}_s=\left(\frac{1}{|{\cal I}^k_x|} \sum_{j \in {\cal I}^k_x} \beta^j_s + \eta^i_s\right) \circ (\Psi^N)^{-1}, \qquad i \in {\cal I}^k_x,
	\]
	and let ${\Prob}^{+,k}_s = ({\cal J}^+, \tau^{+}, \tau^{+}_d, \tilde{\eta}^{(k)}_s)$: this is the problem $\Prob^+$ (Definition~\ref{def:probPlus}) with payment stream $\tilde\eta^{(k)}_s$ in place of the ambient $\eta^+$, i.e.\ the infinite-fund problem faced by a representative member of pool $k$ once payments are pooled and mutualised within the pool.

	{\bf Case 1:}
	If $\tau_d \leq s$ given $x$, set $\gamma_{u,t}=0$ for all $u, t \geq s$.

	{\bf Case 2:}
	Otherwise, write $y=\Psi^N_*(x)$, where $\Psi^N_*$ is the mod 0 isomorphism induced by $\Psi^N$ on the time-$s$ disintegration, in the sense of Lemma~\ref{lemma:disintegrationNatural}. For each $k \in {\cal K}_x$, let $X^k(\beta_s)$ be an arg max for $(\Prob^{+,k}_s; s, y, 0,0 )$.
	Define, for $i \in {\cal I}^k_x$,
	\[
	\gamma^i_{s,s}(\beta_s \mid x):={\mathbb E}_{{\mathbb Q}^{\infty+}_{y}} ( X^k_s(\beta_s))
	\]
	\begin{gather*}
	\tilde{\gamma}^i_{s,s+1}(\beta_s \mid x):=h^k_{s+1}\bigl(\omega^M,\omega^{I,i}\bigr),
	\\
	h^k_{s+1}
	:=
	\sum_{t\in {\cal T}_{> s}} {\mathbb E}_{{\mathbb Q}^{\infty+}_{y}} \bigl( X^k_t(\beta_s) \,\big|\, \iota^{\Omega^{\infty+}}({\cal F}^{M}_{s+1}\times{\cal F}^{I}_{s+1})\bigr),
	\end{gather*}
	the value, given the market information and the representative
	agent's own state at time $s+1$, of the remaining consumption
	prescribed by the infinite-fund optimiser. Here $h^k_{s+1}$ is
	realised, by the Doob--Dynkin lemma, as a measurable function of
	the market path and the representative agent's own path to time
	$s+1$, and is evaluated at agent $i$'s coordinates; conditioning on
	the market and the agent's own state alone ensures that no
	canonical or third-party information enters any leg.
	Divide ${\cal I}$ into disjoint insurance pools $\{P^k_x\}_{k \in {\cal P}_x}$ according to the design of the tontine: for a restricted tontine take $\{P^k_x\}_{k \in {\cal P}_x}
	=\{{\cal I}^k_x\}_{k \in {\cal K}_x}$; for a broad tontine take
	$\{P^k_x\}_{k \in {\cal P}_x}
	=\{{\cal I}\}$.
		 
	For $i \in {\cal P}^k_x$ define the {\em payment} of agent $i$ by
	\[
	C^i_{s+1}(\beta_s \mid x ):= {\mathbb E}_{{\mathbb Q}^*_{x}}( \tilde{\gamma}^i_{s,s+1}(\beta_s \mid x ) \mid {\cal F}^M_{s+1})
	\]
	and set the income received by each member in proportion to the
	realised value of their own continuation, scaled so that the pool
	distributes exactly its resources:
	\[
	\gamma^i_{s,s+1}(\beta_s \mid x ):=
	\mathbb{E}_{{\mathbb Q}^*_x}\bigl( \tilde{\gamma}^i_{s,s+1}(\beta_s \mid x ) \,\big|\, {\cal F}^*_{s+1} \bigr)
	\cdot
	\frac{\sum_{j \in {\cal P}^k_x} C^j_{s+1}(\beta_s \mid x )}{\sum_{j \in {\cal P}^k_x} \mathbb{E}_{{\mathbb Q}^*_x}( \tilde{\gamma}^j_{s,s+1}(\beta_s \mid x ) \mid {\cal F}^*_{s+1} )}
	\]
	with the convention that this is taken to be zero if the denominator vanishes.
	The numerator of the correction ratio is the market value of the
	payments the pool collects; the denominator is the realised value of
	the continuations it must fund. A member whose own state develops
	adversely (in particular, a member who leaves the fund) receives
	the realised value of their own continuation, which is then zero,
	while the ratio spreads the resulting surplus or deficit across the
	pool in proportion to realised values.
	Define
	\[
	\beta_{s+1}^i(\beta \mid x):=\gamma^i_{s,s+1}(\beta_s \mid x)
	\]
	{\em Initialisation:}
	To start the recursion either take $\beta_0=0$ to obtain the restricted strategy $\gamma^{R}$ and broad strategy $\gamma^{B}$. More generally if an initialisation $\beta_s$ is given we obtain strategies $\gamma^{R}(\beta_s)$ and $\gamma^{B}(\beta_s)$.
\end{definition}

For $i,j\in{\cal I}$, write $i\sim_x j$ if $\omega^{I,i}=\omega^{I,j}$ ${\mathbb Q}^*_x$-a.s., so that, at each step $s$ of the recursion, the sets $({\cal I}^k_x)_{k\in{\cal K}_x}$ of Definition~\ref{def:tontineStrategy} are precisely the equivalence classes of $\sim_x$.

The definition of these strategies relies on the existence (and measurable selection) of an argmax and hence we must assume compactness of the problem independent of $\eta$. The next Lemma establishes that, under these circumstances, these strategies are well defined.

\begin{lemma}[Well-posedness of tontine strategies]
\label{lemma:tontineWellPosed}
Suppose $\Prob^+$ is compact independent of $\eta$ (Definition~\ref{def:compactIndependentEta}). Then, for every $N\in{\mathbb Z}_{\ge1}$ and every $x\in\Omega^*_0$, the recursion of Definition~\ref{def:tontineStrategy} with $\beta_0=0$ is well defined: at every step $s\in{\cal T}_{<T}$ and for every $k\in{\cal K}_x$, an arg max $X^k(\beta_s)$ for $\Prob^{+,k}_s$ exists and may be chosen as a measurable function of the conditioning state. Consequently the restricted and broad tontine strategies $\gamma^R,\gamma^B\in{\cal C}(N,0)$ are well defined, and:
\begin{enumerate}[(i)]
\item $\gamma^R,\gamma^B\in{\cal C}^M(N,0)$, i.e.\ both strategies are myopic;
\item $\gamma^R,\gamma^B\in{\cal C}^S(N,0)$, i.e.\ both strategies are symmetric;
\item $\gamma^R,\gamma^B\in{\cal A}^{N,x}_0(0)$, i.e.\ both strategies are admissible for the investment problem $(0,x,0)$ in the sense of Definition~\ref{def:admissibleIdiosyncratic};
\item if $x$ is homogeneous, i.e.\ $i\sim_x j$ for all $i,j\in{\cal I}$ (so that ${\cal H}^{N,x}_0(0)$, Definition~\ref{def:homogeneousFinite}, is meaningful), then $\gamma^R,\gamma^B\in{\cal H}^{N,x}_0(0)$, i.e.\ both strategies are admissible for the homogeneous problem for finite funds $(0,x,0)^h$ in the sense of Definition~\ref{def:homogeneousFinite}, taking $\beta\equiv0$.
\end{enumerate}
\end{lemma}
The proof is in Appendix~\ref{appendix:tontineStrategies}.

We now show that, as the fund grows, a member of either tontine achieves
essentially the same utility as an agent in the (fully diversified)
infinite-fund problem $\Prob^+$. The proof is a forward induction argument using the Strong Law of Large Numbers.

\begin{theorem}[Tontine convergence]

Recall the standing assumption of this section that $\Omega^I$ is finite and that ${\cal J}$ is Markovian.
	
Fix a homogeneous $x_0\in\Omega^*_0$ with common initial idiosyncratic type
$\iota_0$, i.e.\ $\omega^{I,i}_0=\iota_0$ ${\mathbb Q}^*_{x_0}$-a.s.\ for
every $i\in{\mathbb N}$, and write $x_0^+\in\Omega^{M\times I}_0$ for the
corresponding state $(\omega^M_0,\iota_0)$ of a single representative
agent. For each $N$, let $\gamma^R_N,\gamma^B_N\in{\cal C}(N,0)$ be the
restricted and broad tontine strategies of
Definition~\ref{def:tontineStrategy}, built from agents
$\{0,\ldots,N-1\}\subset{\mathbb N}$ of this fixed underlying i.i.d.\
population. For a state $y$ and budget $b$, write
$V^+_s(y,b):=V_s(\Prob^+;y,b,0)$.

\label{thm:tontineConvergence}
Suppose:
\begin{enumerate}[(a)]
\item $\Prob^+$ is compact independent of $\eta$
(Definition~\ref{def:compactIndependentEta}) and satisfies no
satiation, and $V^+_0(x_0^+,0)>-\infty$;
\item for every $s\in{\cal T}_{<T}$, the sharing ratios $R_{N,s+1}$
of both strategies (defined in the proof) are uniformly bounded:
$R_{N,s+1}\le M$ a.s.\ for some constant $M<\infty$ independent of
$N$.
\end{enumerate}
Then
\[
{\cal J}_0\bigl(X^0(\gamma^R_N)\mid x_0\bigr)
\;\to\;
V^+_0(x_0^+,0)
\qquad\text{and}\qquad
{\cal J}_0\bigl(X^0(\gamma^B_N)\mid x_0\bigr)
\;\to\;
V^+_0(x_0^+,0)
\]
as $N\to\infty$, along the entire sequence. The same conclusions
hold with the start $(0,x_0)$ replaced by any $(u,y_0)$ with $y_0$
homogeneous, for the strategies initialised there with
$\beta_u\equiv0$.
\end{theorem}
The proof is given in Appendix~\ref{appendix:tontineStrategies}.

For the restricted tontine strategy, this result formalises the
common understanding of how one can implement theoretical strategies
derived from infinite-fund models in practice which operates as a standing assumption in the insurance literature. What is interesting to note is that the Strong Law of Large Numbers only supplies a lower bound on the utility that can be achieved in finite funds. Thus this argument is not sufficient in itself to justify restricting attention to the infinite-fund problem. In the notation of the proof outline in the Introduction, this establishes $N_1H_1T_1\geq N_0H_*T_*$.

The broad tontine strategy is less well known, and we believe it has not previously been described in this generality. A critical detail in its definition is that the payment is defined conditioned on knowledge of the market at time $s+1$. This means that agents who make poor investment choices will be penalised for their decision without this unduly affecting other members of the pool. The point of the broad tontine design is that it allows longevity risk to be pooled effectively even in small funds.

The special case of using the broad tontine strategy for insuring against idiosyncratic longevity risk has been proposed previously. A continuous-time version is described in \cite{donelly_bernhardt} and a discrete-time
version in \cite{armstrong_buescu_dalby}. Strictly speaking these are different strategies, since they only share the funds of the deceased among survivors, but they embody the same principle of separating investment risk from idiosyncratic risk, which we generalise here.

One problem with the broad tontine strategy is that it does not always converge: we require additional conditions for the broad strategy in Theorem~\ref{thm:tontineConvergence}. However if one adjusts the strategy slightly this can be remedied. The idea is to have an infinite sequence of pools stratified by wealth. Everyone participates in the first pool, wealth over a certain threshold is then only pooled in the second pool and so forth. This design is formalised in \cite{armstrong_dalby_hobbs} for idiosyncratic longevity risk and generalises straightforwardly.

\subsection{Gauge symmetries}

\label{sec:symmetryArgument}

\begin{definition}
Let $\underline{\Omega}^1$ and $\underline{\Omega}^2$ be filtered
probability spaces equipped with distinguished coordinate random
variables
\[
\omega^M
\qquad\text{and}\qquad
(\omega^{I,i})_{i\in{\mathbb N}}.
\]

An {\em $(M,I)$-map}
\[
\Theta:\underline{\Omega}^1\to\underline{\Omega}^2
\]
is a filtration-preserving mod $0$ isomorphism such that
\[
\omega^M \circ \Theta=\omega^M
\]
and
\[
\omega^{I,i}\circ\Theta=\omega^{I,i}
\]
for every $i\in{\mathbb N}$.

If we also have distinguished coordinates ${\omega^{\can,i}}$ on each of $\underline{\Omega}^1$ and $\underline{\Omega}^2$,
an {\em $(M,I+)$-map} is an $(M,I)$ map which additionally preserves
the coordinate
$\omega^{\can,i}$ for each $i$.
\end{definition}

We now wish to define what it means for two possible distributions of income among 
a set of $N$-agents to be financially equivalent. This definition will take into account the idea that the label assigned to an agent is unimportant and so a random relabelling should be unimportant. The definition will also consider two distributions to be equivalent if they are stably related by a mod 0 isomorphism that preserves financially relevant information.
This gives the interpretation of our next definition.

\begin{definition}[$(M,I)$-equivalence]
\label{def:MIequivalenceGeneral}
Let
\[
\Gamma_1,\Gamma_2 \in L^0({\mathbb R}^N,({\cal F}^*_t)_{t \in {\cal T}})
\]
be ${\mathbb R}^N$-valued stochastic processes. We say that $\Gamma_1$
and $\Gamma_2$ are {\em $(M,I)$-equivalent} if there exists a
uniformly distributed $S_N$-valued random variable $g$ on
$\underline{\Omega}^{*,+}$ and a mod $0$ automorphism $\phi$ for
$\underline{\Omega}^{*,+}$ such that
\begin{enumerate}[(i)]
\item
$
\Gamma_1^{g^{-1}i}\circ\rho(g^{-1})\circ\pi^{\Omega^*}\circ\phi=\Gamma_2^i\circ\pi^{\Omega^*}
$
for every $i\in{\cal I}$;
\item $g$ is $\iota^{\Omega^{*,+}}({\cal F}^{\perp,M,I}_0)$-measurable, with
${\cal F}^{\perp,M,I}_0$ as defined in \eqref{eq:defFPerpMI};
\item $\phi$ is an $(M,I)$-map.
\end{enumerate}
\end{definition}

We extend $(M,I)$-equivalence, using the same $g,\phi$ throughout, to
any finite collection of ${\mathbb R}^N$-valued random variables or
processes on $\underline{\Omega}^*$: to the continuation stream
$\gamma$ (viewing $\gamma_u:=(\gamma^i_{u,t})_{i,t}$ as one member of
the collection for each $u$), the budget $b$, the participation data
$(\tau^i,\tau_d^i,\eta^i)_{i\in{\cal I}}$, and the process
$X\in L^0({\mathbb R}^N,({\cal F}^*_t)_{t\in{\cal
T}})$ appearing in Theorem~\ref{theorem:gaugeSelection} below.

\begin{lemma}[Invariance under $(M,I)$-equivalence]
\label{lem:admissibilityValueInvariance}
Suppose $\gamma$ is admissible for $(s,x,b)$ with
participation data $(\tau^i,\tau_d^i,\eta^i)_{i\in{\cal I}}$, and that
$(\tau',\tau_d',\eta',b',\gamma')$ is $(M,I)$-equivalent
to $(\tau,\tau_d,\eta,b,\gamma)$ via $(g,\phi)$.
\begin{enumerate}[(i)]
\item $\gamma'$ is admissible for $(s,x,b')$ with
participation data $(\tau'^i,\tau_d'^i,\eta'^i)_{i\in{\cal I}}$.
\item If $g$ takes values in the stabiliser
$\mathrm{Stab}(0):=\{h\in S_N:h(0)=0\}$ of agent $0$, then
${\cal V}^{\gamma'}={\cal V}^\gamma$.
\end{enumerate}
\end{lemma}
This follows by a naturality argument. A full proof is given in Appendix~\ref{appendix:symmetryArgument}.

This supplies an infinite dimensional-symmetry group for our problem. Following the terminology of differential geometry and physics we will refer to this as a gauge symmetry. Our next result shows that we can find a choice of gauge such that a vector of stochastic processes is highly exchangeable.

\begin{theorem}[Gauge selection]
Given
\[
X \in L^0( {\mathbb R}^N, ({\cal F}^*_t)_{t \in {\cal T}})
\]
there exists
\[
X^* \in L^0( {\mathbb R}^N, ({\cal F}^*_t)_{t \in {\cal T}})
\]
which is $(M,I)$-equivalent to $X$ such that for each $i=0,\ldots, N-1$
\[
(X^*)^i \in L^0({\mathbb R},
\iota^{\Omega^*}(
{\cal F}^M_t\times{\cal F}^{\can}_t\times{\cal E}^{N,i}({\cal F}^{I+}_t))_{t \in {\cal T}}).
\]
Moreover, $X^*$ can be chosen such that
\[
(X^*)^i(\omega)=(X^*)^{g^{-1}i}(g^{-1}\omega)
\]
for all $g \in S_N$.
\label{theorem:gaugeSelection}
\end{theorem}

The idea behind the proof is first to construct a symmetrised vector
$\bar X$ using a random permutation $\hat g$.
To facilitate this construction we repeatedly enlarge the underlying
probability space.
The resulting process $\bar X$ is highly symmetric.
However, the symmetry is initially expressed in terms of the wrong group
action and therefore does not immediately yield the desired
exchangeability properties.
To correct this, we use Lemma~\ref{lemma:isomorphism}, which allows us
to pass between different group actions.
Furthermore, Lemma~\ref{lemma:sectionExists} allows us to realise these
actions in a canonical form, even when they are not initially presented
in the form required by Lemma~\ref{lemma:isomorphism}.
After enlarging the probability space by one further auxiliary factor,
we introduce an additional ``twist'' which converts the symmetry of
$\bar X$ into the desired exchangeable symmetry.
This produces the gauge-selected representative appearing in the
statement of the theorem.

\begin{proof}[Gauge selection, Theorem~\ref{theorem:gaugeSelection}]
For each $N$, let
\[
\theta^{N}:\Omega^{\can,\infty}_{\cal T} \to \Omega^{\can}_{\cal T}
\]
be a measure-preserving map as in Lemma~\ref{lemma:standardFiltered}.

Let ${g}^N$ be a uniformly distributed, ${\cal F}^{\can}_0$-measurable $S_N$-valued random variable defined on $\underline{\Omega}^{\can}_{\cal T}$; such a variable exists since $\underline{\Omega}^{\can,0}$ is atomless. Define
\[
\bar{g}^N:={g}^{N} \circ \theta^{N}.
\]
This ensures that $\bar{g}^N$
is an $N$-exchangeable, uniformly distributed, $S_{N}$-valued random variable defined on
$\Omega^{\can,\infty}_{\cal T}$.

We have
\[
X^i \circ (\id_M \times \theta^N \times \id_{I+,\infty}) \in L^0( {\mathbb R}, \iota({\cal F}^M_t \times {\cal E}^N({\cal F}_{t}^{\can,\infty}) \times {\cal F}^{I+,\infty}_t)_{t \in {\cal T}}).
\]

Let $\underline{\Omega}^{{\cal E}^N,1}$ and $\underline{\Omega}^{{\cal E}^N,2}$ be two copies of $\underline{\Omega}^{{\cal E}^N}$.
Let
\[
\pi^{{\cal E}^N,1}:\Omega^M \times \Omega^{{\cal E}^N,1} \times \Omega^{{\cal E}^N,2} \times \Omega^{I+,\infty}  \to 
\Omega^{{\cal E}^N,1}
\]
be the projection onto $\Omega^{{\cal E}^N,1}$. Let
\[
\pi^{{\cal E}^N,1}_{\perp}:\Omega^M \times \Omega^{{\cal E}^N,1} \times \Omega^{{\cal E}^N,2} \times \Omega^{I+,\infty} \to 
\Omega^M \times \Omega^{{\cal E}^N,2} \times \Omega^{I+,\infty}
\]
be the projection onto the complementary components.

Define
\[
(\hat{X}^N)^i := X^i \circ (\id_{M} \times \theta^{N} \times \id_{I+,\infty}) \circ \pi^{{\cal E}^N,1}_{\perp}.
\]
Then
\[
(\hat{X}^N)^i \in L^0( {\mathbb R}, \iota({\cal F}^M_t \times {\cal F}_t^{{\cal E}^N,1} \times {\cal F}_t^{{\cal E}^N,2} \times {\cal F}^{I+,\infty}_t)_{t \in {\cal T}}).
\]

Define $\hat{g}^N:=\bar{g}^{N} \circ \pi^{{\cal E}^N,1}$. We have
\[
\hat{g}^N \in L^0( S_N, {\cal E}^{N}_{\rho^{0,1,1,1}}({\cal F}^M_t \times {\cal F}_t^{{\cal E}^N,1} \times {\cal F}_t^{{\cal E}^N,2} \times {\cal F}^{I+,\infty}_t)_{t \in {\cal T}}).
\]
By construction $\hat{X}^N$ and $\hat{g}^N$ are independent,
$\hat{X}^N$ has the same distribution as $X$, $\hat{g}^N$ is uniformly distributed and $\hat{g}^N$ is invariant under the $S_N$-action $\rho^{0,1,1,1}$.

Define $\bar{X}^N$, the symmetrisation of $\hat{X}^N$ under $S_N$,
to be given by the standard action of the random group element
$\hat g^N$. To be precise,
decompose elements of the four-fold product space
as
\[
\omega^{(4)}:=(
\omega^M, (\omega^{1,j_1})_{j_1 \in {\mathbb N}},
(\omega^{2,j_2})_{j_2 \in {\mathbb N}},
(\omega^{3,j_3})_{j_3 \in {\mathbb N}} ),
\]
then
\[
\begin{split}
&(\bar{X}^N)^i (\omega^M, (\omega^{1,j_1})_{j_1 \in {\mathbb N}},
(\omega^{2,j_2})_{j_2 \in {\mathbb N}},
(\omega^{3,j_3})_{j_3 \in {\mathbb N}} ):= \\
&\quad
(\hat{X}^N)^{\hat{g}^{-1}i} (\omega^M, (\omega^{1,\hat{g}j_1})_{j_1 \in {\mathbb N}},
(\omega^{2,\hat{g}j_2})_{j_2 \in {\mathbb N}},
(\omega^{3,\hat{g}j_3})_{j_3 \in {\mathbb N}} )
\end{split}
\]
where we are writing $\hat{g}$
as an abbreviation for
\[
\hat{g}^N(\omega^{(4)}).
\]
Defining
\[
\underline{\Omega}^{(4)}: =
{\underline{\Omega}}^M \times {\underline{\Omega}}^{{\cal E}^N,1} \times {\underline{\Omega}}^{{\cal E}^N,2} \times {\underline{\Omega}}^{I+,\infty}
\]
we have
\[
\bar{X}^N \in \prod_{i=0}^{N-1} L^0( {\mathbb R}, ({\cal F}^{(4)}_t)_{t \in {\cal T}}).
\]

Let $S^{N,i}$ denote the stabiliser of $i$ in $S_N$.
Define
\[
\hat{i}:
\{0,\ldots,N-1\}\setminus\{i\}
\to
\{0,\ldots,N-2\}
\]
by
\[
\hat{i}(m)
=
\begin{cases}
m,   & m<i,\\
m-1, & m>i.
\end{cases}
\]
Given $g\in S_N$, define
\[
{\cal S}^i(g)(m)
:=\begin{cases}
i & m = i \\
({\hat{i}}^{-1} \circ 
\widehat{g i}
\circ g)m & m \not= i
\end{cases}
\]
Then
\[
{\cal S}^i(g)\in S^{N,i}.
\]

Recall that $(\bar X^N)^i$ is invariant under the action of
$S^{N,i}$ by $\rho^{0,1,1,1}$.

Let $\underline{S}^{N,i}$ denote the uniform discrete probability
space on $S^{N,i}$, equipped with the trivial filtration.

Let $\underline{\Delta}^{N,i}$ denote the quotient space of
$\underline{\Omega}^{(4)}$ under the action of $S^{N,i}$ by
$\rho^{0,1,1,1}$, and let
\[
\pi^{N,i}:\Omega^{(4)}\to\Delta^{N,i}
\]
be the quotient map.

Repeating the argument of Lemma~\ref{lemma:sectionExists},
and arguing as in the proof of Lemma~\ref{lemma:standardFiltered},
there exists a measurable filtration-preserving right inverse
\[
\xi^{N,i}:\Delta^{N,i}\to\Omega^{(4)}.
\]

Define
\[
\Xi^{N,i}:
\underline{S}^{N,i}\times\underline{\Delta}^{N,i}
\to
\underline{\Omega}^{(4)}
\]
by
\[
\Xi^{N,i}(h,\delta)
=
\rho^{0,1,1,1}(h,\xi^{N,i}(\delta)).
\]
Then $\Xi^{N,i}$ is a filtration-preserving mod $0$ isomorphism
between the action $\rho^{1,0}$ on
$\underline{S}^{N,i}\times\underline{\Delta}^{N,i}$
and the action $\rho^{0,1,1,1}$ on
$\underline{\Omega}^{(4)}$.
Define
\[
\zeta^{N,i}:\Omega^{(4)}\to S^{N,i}
\]
by requiring
\[
\rho^{0,1,1,1}
(
\zeta^{N,i}(\omega),
\xi^{N,i}(\pi^{N,i}(\omega))
)
=
\omega.
\]
Define
\[
\hat g^{N,i}
:=
{\cal S}^i(\hat g^N).
\]
Now define
\[
(\bar{X}^*)^i
\in
L^0(\mathbb R,({\cal F}^{(4)}_t)_{t\in{\cal T}})
\]
by
\[
\begin{split}
(\bar{X}^*)^i(\omega^{(4)})
:=
&
\ \rho^{0,1,1,1}
\circ
(\id\times (\bar X^N)^i)
\\
&
\circ
(\id\times\Xi^{N,i})
\circ
(\phi^{-1,0})^{-1}
\circ
(\id\times\Xi^{N,i})^{-1}
\bigl(
\hat g^{N,i}(\omega^{(4)}),
\omega^{(4)}
\bigr).
\end{split}
\]

We now claim that given  $i\in\{0,\ldots,N-1\}$,  $h \in S^{N,i}$ and $\omega^{(4)}\in\Omega^{(4)}$ we have
\[
(\bar X^*)^i
(
\rho^{0,1,1,1}(h)\omega^{(4)}
)
=
(\bar X^*)^i(\omega^{(4)}).
\]
This follows by construction and naturality, using Lemma~\ref{lemma:isomorphism}; a detailed argument is given as Lemma~\ref{lemma:gaugeInvariant}
in Appendix~\ref{appendix:symmetryArgument}.
The purpose of the conjugation
by
\[
(\id\times\Xi^{N,i})
\]
in the definition of $\bar X^*$ is precisely to place the
$S^{N,i}$-action into a form to which
Lemma~\ref{lemma:isomorphism} may be applied.

It follows that,
\[
(\bar X^*)^i
\in
L^0
\!\left(
\mathbb R,
{\cal E}^{N,i}(({\cal F}^{(4)}_t)_{t\in{\cal T}})
\right).
\]

Define
\[
\Theta^N
:=
\id_M
\times
\theta^N
\times
\theta^N
\times
\id_{I+,\infty}.
\]
Then
\[
\Theta^N:
\underline{\Omega}^{(4)}
\to
\underline{\Omega}^{*+}
\]
is a measure-preserving map which fixes the coordinates $\omega^M$
and $\omega^{I+,l}$ for every $l$, and which, by
Lemma~\ref{lemma:standardFiltered} applied in each of the two middle
coordinates, satisfies
\[
(\Theta^N)^{-1}({\cal F}^{*+}_t)={\cal F}^{(4)}_t
\]
modulo null sets for every $t\in{\cal T}$. By the Doob--Dynkin lemma,
every $V\in L^0({\mathbb R},({\cal F}^{(4)}_t)_{t\in{\cal T}})$ may
therefore be written as
\[
V=V^\flat\circ\Theta^N
\]
for a $V^\flat\in L^0({\mathbb R},({\cal F}^{*+}_t)_{t\in{\cal T}})$
which is unique up to almost sure equality. Moreover,
$\theta^N\circ\rho(g)=\theta^N$ almost surely for every $g\in S_N$,
since $\theta^N$ is
${\cal E}^N({\cal F}^{\can}_{\cal T})$-measurable, and hence
\[
\Theta^N\circ\rho^{0,1,1,1}(g)=\rho^{0,0,0,1}(g)\circ\Theta^N
\]
almost surely, where $\rho^{0,0,0,1}$ denotes the action of $S_N$ on
$\underline{\Omega}^M\times\underline{\Omega}^{\can}_{\cal T}\times\underline{\Omega}^{\can}_{\cal T}\times\underline{\Omega}^{I+,\infty}\cong\underline{\Omega}^{*+}$
which permutes only the factors of $\underline{\Omega}^{I+,\infty}$.
Consequently, if $V$ is invariant under the action $\rho^{0,1,1,1}$
of $S^{N,i}$, then $V^\flat$ is invariant under the action
$\rho^{0,0,0,1}$ of $S^{N,i}$.

Choose a filtration preserving mod 0 isomorphism
\[
\Lambda:\underline{\Omega}^{\can}_{\cal T} \times 
\underline{\Omega}^{\can}_{\cal T} \to \underline{\Omega}^{\can}_{\cal T}
\]
Extend $\Lambda$ by the identity on components other than $\Omega^{\can}_{\cal T}$ to obtain an $(M,I+)$-map
\[
\Lambda^*:\underline{\Omega}^{*+} \cong \underline{\Omega}^{\can}_{\cal T} \times \underline{\Omega}^* \to \underline{\Omega}^*.
\]
Define
\[
(X^*)^i
:=
((\bar X^*)^i)^\flat
\circ
(\Lambda^*)^{-1}
\]
Then
\[
X^*
=
\bigl(
(X^*)^0,\ldots,(X^*)^{N-1}
\bigr)
\in
\prod_{i=0}^{N-1}
L^0(
\mathbb R,
({\cal F}^{*}_t)_{t\in{\cal T}}
).
\]

Recall that $(\bar X^*)^i$ is invariant under the action
$\rho^{0,1,1,1}$ of $S^{N,i}$, so $((\bar X^*)^i)^\flat$ is adapted
and invariant under the action of $S^{N,i}$ on the factor
$\underline{\Omega}^{I+,\infty}$ alone; averaging over the finite
group $S^{N,i}$ shows that it is measurable with respect to
\[
\iota^{\Omega^{*+}}
(
{\cal F}^{\can}_t
\times
{\cal F}^M_t
\times
{\cal F}^{\can}_t
\times
{\cal E}^{N,i}({\cal F}^{I+}_t)
).
\]
Since $\Lambda$ is filtration preserving and $\Lambda^*$ fixes the
market and agent coordinates, $\Lambda^*$ carries this
$\sigma$-algebra onto
\[
\iota^{\Omega^*}
(
{\cal F}^M_t
\times
{\cal F}^{\can}_t
\times
{\cal E}^{N,i}({\cal F}^{I+}_t)
).
\]
Thus,
\[
(X^*)^i
\in
L^0
\!\left(
\mathbb R,
\iota^{\Omega^*}
(
{\cal F}^M_t
\times
{\cal F}^{\can}_t
\times
{\cal E}^{N,i}({\cal F}^{I+}_t)
)_{t\in{\cal T}}
\right).
\]
Finally, every step in the construction of $X^*$ from $X$ is
implemented by either

\begin{enumerate}[(i)]
\item $(M,I+)$-maps, which fix $\omega^M$ and every $\omega^{I,l}$;
\item the single genuine relabelling, by the uniformly distributed,
${\cal F}^{\perp,M,I}_0$-measurable element $\hat g^N$, that produces
$\bar X^N$ from $\hat X^N$: this transports the real agent-index
bundle $\underline{\Omega}^{I+,\infty}$ by $\rho((\hat g^N)^{-1})$,
together with the corresponding superscript change
$i\mapsto(\hat g^N)^{-1}(i)$;
\item enlargements of the probability space by factors measurable
with respect to ${\cal F}^{\perp,M,I}_T$, together with the
corresponding descents, such as $V\mapsto V^\flat$, along
measure-preserving maps fixing $\omega^M$ and every $\omega^{I,l}$.
\end{enumerate}

Conjugating the maps of (i) and (iii) past
$\rho((\hat g^N)^{-1})$ preserves the property of fixing $\omega^M$
and every $\omega^{I,l}$, so composing all the steps exhibits
$X^{(\hat g^N)^{-1}(i)}\circ\rho((\hat g^N)^{-1})\circ\phi=(X^*)^i$ for
a single combined $(M,I+)$-map $\phi$. Consequently $X^*$ is
$(M,I)$-equivalent to $X$: the construction genuinely relabels
agents via $\hat g^N$, and moves their real data with them,
rather than merely relabelling which formula is called $(X^*)^i$.

Finally, the claimed symmetry follows from the fact that the
construction of $\bar X^*$ from $\bar X^N$ is natural with respect to
relabelling by elements of $S_N$.

Indeed, the assignment
\[
(\bar X^N)^i \mapsto (\bar X^*)^i
\]
is built entirely from constructions that commute with the action of
$S_N$: passage to stabilisers, formation of quotients, choice of
sections, the maps $\Xi^{N,i}$, the induced random permutations
$\hat g^{N,i}$ and the isomorphism
$(\phi^{-1,0})^{-1}$.

Consequently the equivariance relation
\[
(\bar X^N)^i(\omega)
=
(\bar X^N)^{g^{-1}i}(g^{-1}\omega)
\]
is preserved by the construction and hence
\[
(\bar X^*)^i(\omega)
=
(\bar X^*)^{g^{-1}i}(g^{-1}\omega).
\]

Since $\Theta^N$ and $\Lambda^*$ are $(M,I+)$-maps, they also commute
with relabelling, and therefore the same identity holds for $X^*$.
\end{proof}

\begin{remark}
\label{rem:gaugeSelectionUniform}
The pair $(g,\phi)$ produced by the proof of
Theorem~\ref{theorem:gaugeSelection} is built entirely from auxiliary
structure ($\theta^N$, $\hat g^N$, $\Theta^N$, $\Lambda^*$, and the
quotient/section maps $\Xi^{N,i}$) none of which refer to the
values of $X$; $X$ enters only in the final step,
$(X^*)^i:=((\bar X^*)^i)^\flat\circ(\Lambda^*)^{-1}$.
Consequently, applying the identical formulas with any other
$Y\in\prod_{i=0}^{N-1}L^0({\mathbb R},({\cal F}^*_t)_{t\in{\cal T}})$
in place of $X$ produces a $Y^*$ satisfying the same measurability and
exchangeability conclusions, via the \emph{same} $(g,\phi)$: one pair
$(g,\phi)$ transforms every member of a finite collection of such
processes at once.
\end{remark}

Having symmetrised the strategy fully, one agent's strategy
will now determine the strategy of the entire fund.

\begin{corollary}[Reduction to one agent's strategy]
\label{cor:exchangeableFund}
Suppose $\gamma$ is admissible for $(0,x,0)$ with participation data
$(\tau^i,\tau_d^i,\eta^i)_{i\in{\cal I}}$. Then:
\begin{enumerate}[(a)]
\item There exists $\gamma^*$, with participation data
$(\tau^{*i},\tau_d^{*i},\eta^{*i})_{i\in{\cal I}}$, $(M,I)$-equivalent
to $\gamma$ (together with its participation data) via some
$g\in\mathrm{Stab}(0)$, such that:
\begin{enumerate}[(i)]
\item $\gamma^*$ is admissible for $(0,x,0)$ with participation data
$(\tau^{*i},\tau_d^{*i},\eta^{*i})_{i\in{\cal I}}$, and
${\cal V}^{\gamma^*}={\cal V}^\gamma$;
\item the data $(\gamma^{*i},\tau^{*i},\tau_d^{*i},\eta^{*i})$ of
agents $1,\ldots,N-1$ is exchangeable: for example
\[
(\gamma^*)^i(\omega)=(\gamma^*)^{h^{-1}i}(h^{-1}\omega),
\]
and likewise for $\tau^*,\tau_d^*,\eta^*$, for every $h\in S_{N-1}$
acting on $\{1,\ldots,N-1\}$ and fixing $0$;
\item for each $i=1,\ldots,N-1$, each of $\gamma^{*i}$, $\tau^{*i}$,
$\tau_d^{*i}$, $\eta^{*i}$ is measurable with respect to
$\iota^{\Omega^*}({\cal F}^M_t\times{\cal F}^{\can}_t\times{\cal E}^{N-1,i}({\cal
F}^{I+}_t))_{t\in{\cal T}}$.
\end{enumerate}
\item There exists $\gamma^\circ$, with participation data
$(\tau^{\circ i},\tau_d^{\circ i},\eta^{\circ i})_{i\in{\cal I}}$,
$(M,I)$-equivalent to $\gamma$ (together with its participation data),
such that for \emph{every} $i=0,\ldots,N-1$, each of $\gamma^{\circ
i}$, $\tau^{\circ i}$, $\tau_d^{\circ i}$, $\eta^{\circ i}$ is
measurable with respect to $\iota^{\Omega^*}({\cal F}^M_t\times{\cal F}^{\can}_t\times{\cal E}^{N,i}({\cal F}^{I+}_t))_{t\in{\cal T}}$. Consequently agents
$0,1,\ldots,N-1$ all have pairwise preference isomorphic contracts
(Definition~\ref{def:preferenceIsomorphism}).
\end{enumerate}
\end{corollary}
The details are given in Appendix~\ref{appendix:symmetryArgument}.

\begin{remark}
While the contracts $(\gamma^{\circ i})_{i=0}^{N-1}$ of part (b) are
preference isomorphic for every $i$, it is not possible to construct
these isomorphisms in a way that also preserves exchangeability for
every $i\ge0$.
\end{remark}

Note that this gauge selection depends on arbitrary choices
and cannot be performed in a uniform way for all $N$. In other words, it cannot be extended to an action of $S_\infty$.
There is a tension between these infinite degrees of freedom and the convergence of a subsequence of strategies implied by compactness. This can only be resolved if an ergodic theorem holds and these degrees of freedom vanish in the infinite limit.

\subsection{Acceptable strategies}

\begin{definition}
	Let $N \in {\mathbb Z}_{\geq 1}$ be given, and let ${\cal I}=\{0,\ldots,N-1\}$.

	Let $s \in {\cal T}$, $x \in \Omega^*_s$ and $b \in L^1({\mathbb R}^N, {\cal F}^*_s,{\mathbb Q}^*_x)$, and let $\gamma \in {\cal A}^{N,x}_s(b)$.

	Given $A \subseteq {\cal I}$, $u \in {\cal T}_{\geq s}$ and $y \in \Omega^*_u$, define the
	resources the coalition $A$ holds when the time-$u$ contract
	revision begins: for $t \in {\cal T}_{\geq u}$,
	\[
	\beta^{A,y,u}_{t} :=
	\begin{cases}
	\displaystyle
	\indicator_{\{t=s\}}\sum_{j \in A} b^j, & u = s,
	\\[2ex]
	\displaystyle
	\sum_{j \in A} \gamma^j_{u-1,t}, & u \in {\cal T}_{>s}.
	\end{cases}
	\]
	These are the payments still owed to the members of $A$ under the
	contract in force when time $u$ begins (at $u=s$, their initial
	budgets): the outside option of a breakaway coalition is to walk
	away with these claims and manage them separately.

	We say that $\gamma$ is either {\em acceptable} or {\em unacceptable} for $A$ at $(u,y)$ according to the following two cases.
	\begin{enumerate}[(a)]
	\item Suppose there is no $B$ with $A \subseteq B \subseteq {\cal I}$ such that $i \sim_y j$ for all $i,j \in B$ and, for every $g' \in S_N$ mapping $B$ to the first $|B|$ integers,
	\[
	{\cal J}^j_u(X^j(\gamma) \mid y) = V^{|B|}_u\bigl(\rho(g')_*(y), \beta^{B,y,u}\bigr)
	\]
	for all $j \in B$. Then $\gamma$ is unacceptable for $A$ at $(u,y)$ if
	\begin{enumerate}[(i)]
	\item $i \sim_y j$ for all $i, j \in A$;
	\item for every $g \in S_N$ mapping $A$ to the first $|A|$ integers,
	\[
	{\cal J}^j_u(X^j(\gamma) \mid y) \leq V^{|A|}_u\bigl(\rho(g)_*(y), \beta^{A,y,u}\bigr)
	\]
	for all $j \in A$, with a strict inequality for at
	least one $j$.
	\end{enumerate}
	\item Suppose instead such a $B$ exists. Then $\gamma$ is acceptable for $A$ at $(u,y)$.
	\end{enumerate}

	We write ${\cal G}^{N,x}_s(b)$ for the set of $\gamma \in {\cal A}^{N,x}_s(b)$
	that are ${\mathbb Q}^*_x$-almost surely acceptable for $A$ at $(u,y)$,
	for every $u \in {\cal T}_{\geq s}$, $y \in \Omega^*_u$ and
	$A \subseteq {\cal I}$.
\end{definition}

Case (b) ensures that once a group is already achieving the optimal value for its own size, we do not additionally require checking whether some subset of that group could improve by splitting off into a smaller fund.

This is precisely the self-enforcement condition described in the Introduction: contracts must be robust to profitable deviation by any coalition of agents in a common state. We do not assume that agents are risk averse; the finite-fund problem here is non-convex even if agents have risk-averse preferences, because this self-enforcement condition is not itself convex.

\begin{theorem}[Existence of acceptable strategies]
\label{thm:existenceOptimalContinuationsFinite}
Suppose $\Prob$ is compact (Definition~\ref{def:compactProblem}) and
that each ${\cal J}_s(\,\cdot\mid x)$ is concave. Let $N \in {\mathbb Z}_{\geq 1}$, $s \in {\cal T}$, and let $x \in \Omega^*_s$ be homogeneous, i.e.\ $\omega^{I,i}=\omega^{I,j}$ ${\mathbb Q}^*_x$-a.s.\ for all $i,j \in {\cal I}$. Let $\beta \in L^1({\mathbb R},({\cal F}^*_t)_{t \geq s},{\mathbb Q}^*_x)$ with $V^N_s(x,\beta)>-\infty$. Then the homogeneous problem for finite funds $(s,x,\beta)^h$ (Definition~\ref{def:homogeneousFinite}) admits an optimiser, i.e.\ the supremum defining $V^N_s(x,\beta)$ is attained. Moreover, if $\beta\equiv0$ and the values arising in the recursion below exceed $-\infty$, then ${\cal G}^{N,x}_s(0)\neq\emptyset$.

Concavity is used only for this existence statement; the remaining
results of this section take acceptable strategies as given.
\end{theorem}
A proof is given in Appendix~\ref{appendix:acceptableStrategiesFinite}.

\begin{definition}
\label{def:valueFunctionInfinite}
Suppose $\Prob^+$ is compact independent of $\eta$ (Definition~\ref{def:compactIndependentEta}).

Let $\underline\Omega^C$ be an auxiliary copy of $\underline\Omega^\can_{\cal T}$, and choose $c \in \Omega^C$. Write $\omega^C_\tau(c) \in \Omega^\can_{{\cal T},\tau}$ for the time-$\tau$ state of $c$, and let $y(j,c) \in \Omega^{\infty+}_\tau$ be the state built from $j \in \Omega^J$ together with $\omega^C_\tau(c)$.

For $j \in \Omega^J$ we define the {\em value function} ${\cal V}^\infty$ by
\[
{\cal V}^\infty(j) := V^+_\tau\bigl(y(j,c),0\bigr).
\]
\end{definition}

\begin{lemma}[Independence of the choice of $c$]
\label{lem:valueFunctionInfiniteWellDefined}
The value function ${\cal V}^\infty$ does not depend on the choice of $c$ in the preceding definition.
\end{lemma}
The proof that this is well-defined follows in exactly the same way as the proof that the value function for finite funds is well-defined (Lemma~\ref{lem:valueFunctionWellDefined}).

\begin{theorem}[Infinite fund limit]
\label{thm:infiniteFundLimit}
Let $(N_a)_{a \in {\mathbb N}}$ be a strictly increasing sequence in ${\mathbb Z}_{\geq 1}$ with $N_a \to \infty$. Let  $(\gamma^{N_a})_{a \in {\mathbb N}}$ be a corresponding sequence
of strategies satisfying $\gamma^{N_a} \in {\cal G}^{N_a,x}_0(0)$ for all $x \in \Omega^*_0$.

Then, for ${\mathbb Q}^J$-almost every $j \in \Omega^J$,
\[
\lim_{a \to \infty} {\cal V}^{\gamma^{N_a}}(j)={\cal V}^\infty(j).
\]
\end{theorem}

A complete proof of this theorem is given in Appendix~\ref{appendix:acceptableStrategiesFinite}. We summarise the argument here.

The proof of Theorem~\ref{thm:infiniteFundLimit}
proceeds as follows.

\begin{itemize}
\item While Theorem~\ref{thm:infiniteFundLimit} is packaged as a
statement about the value function at time $0$, it is proved via the
more general Theorem~\ref{thm:generalInfiniteFundLimit}, stated at an
arbitrary time $u$ so that it can be applied to each cohort of agents
at their entry date.
\item We first symmetrise: Corollary~\ref{cor:exchangeableFund}(a)
replaces $\gamma^a$ by an $(M,I)$-equivalent $\gamma^{a*}$,
exchangeable in $i=1,\ldots,N_a-1$;
Lemma~\ref{lemma:acceptabilityInvariance} confirms that acceptability
transfers under symmetrisation. Exchangeability makes the delivered
utility $v^a$ common to every agent of the block.
\item The lower bound (Lemma~\ref{lemma:coalitionCounting} and
Proposition~\ref{prop:noPersistentGap}): at every date and every
realised state, at most a bounded number of agents can be delivered
less than the infinite-fund value of their own state, for otherwise
the short-changed agents would break away: their outside option
exceeds $V^+-\epsilon$ beyond a threshold coalition size, by the
tontine bound (Corollary~\ref{cor:tontineValueLowerBound}), and the
exemption clause cannot rescue the strategy, because any exempting
set would itself certify a value above theirs. Hence
$\liminf_a v^a\ge V^+_u$. This is where the self-enforcement
condition replaces the convexity conditions of the mean-field
literature.
\item Extraction: the funding conditions bound the per-capita price
of consumption, so the block laws are tight
(Lemma~\ref{lem:exchangeableLimits}); by the Hewitt--Savage theorem in
the form of Theorem~\ref{thm:canonicalDeFinetti}, the limit is
conditionally i.i.d.\ given its de~Finetti mixing variable, which is
absorbed into $\underline\Omega^{\infty+}$ via the stable-equivalence
isomorphism of Lemma~\ref{lemma:stableEquivalence}, giving a candidate
infinite-fund stream $X^{\infty,1}$.
\item The upper bound runs the cost-of-delivered-utility chain of
Lemma~\ref{lemma:deliveredValueBound} on the limit objects: the
statewise cost of what the limit stream delivers at each state is
nonnegative, by the coalition bounds; the verification inequality
propagates these costs backward through the dates; and per-capita
conservation of resources caps their total. The three force the
costs at the starting date to vanish, whence
$\limsup_a v^a\le{\cal J}_u(X^{\infty,1}\mid y_0)\le V^+_u$ by the
Fatou property of compactness.
\item Identification: attainment (condition (ii) of compactness)
together with Proposition~\ref{prop:acceptableFromOptimisers}
produces an acceptable strategy for $\Prob^+$ attaining $V^+_u$, with
binding funding conditions by
Theorem~\ref{thm:noMutuallyBeneficialContracts}(i); a
subsequence-of-every-subsequence argument upgrades convergence along
a subsequence to convergence along the whole sequence.
\item Finally, Theorem~\ref{thm:infiniteFundLimit} follows by
applying Theorem~\ref{thm:generalInfiniteFundLimit} to each pool of
agents sharing a common entry time and type (via
Lemma~\ref{lemma:restrictionAcceptable}, since entering the fund
fresh is exactly a start with budget $0$), and noting that the
randomly selected agent defining ${\cal V}^{\gamma^{N_a}}$ lands in
such a pool with probability $\to1$.
\end{itemize}

Together with the lower bound of Theorem~\ref{thm:tontineConvergence},
this shows $N_1H_1T_1\leq N_0H_*T_*$ and hence equality throughout, in
the notation of the proof outline in the Introduction.

\medskip

It is easy to find contrived examples where the value functions for the infinite problem are different between $\Prob^+$ and $\Prob$. However, for risk-averse von Neumann--Morgenstern preferences and similar concave preferences based around the expectation operator, one can use Jensen's inequality to show that the value functions for the infinite problems for $\Prob^+$ and $\Prob$ are equal, since the value function of the classical continuation problem is then concave in the budget, and pooling replaces several agents' budgets by their mean.

\section{Potential generalisations}
\label{sec:extensions}

We have stated our results in considerable generality, but let us briefly mention three areas where we expect things could be pushed a little further.

One could define a notion of $\epsilon$-acceptable strategies which relaxes the acceptability constraints by a small amount. One could then attempt to prove a similar result for the value functions as both $N\to \infty$ and $\epsilon \to 0$. This would show that our result was robust to such perturbations. It should also allow the concavity hypothesis of Theorem~\ref{thm:existenceOptimalContinuationsFinite} to be dropped: concavity is used only to obtain exact optimisers, whereas $\epsilon$-optimal strategies exist, and admit measurable selections, whenever the value is finite.

One could allow a continuum of types. For this to work, one should choose a topology on the space of types where the preferences are uniformly continuous. One would also need to change the acceptability condition to allow agents of similar types to form coalitions. We expect that some form of uniform continuity would be required as a replacement for our current assumption that the space of types is finite.

One could relax the Markovian assumption in a similar way, insisting that preferences are continuous in consumption history. Again, one would need to adjust the acceptability condition.

Because the last two proposed generalisations would require allowing more coalitions, extending acceptability to $\epsilon$-acceptability should be done first. This explains why we have imposed these simplifying conditions.

\bibliographystyle{plain}
\bibliography{cdc}

\appendix
\appendixpage

The appendices are arranged in two groups.
Appendices~\ref{appendix:infiniteFund}
and~\ref{appendix:finiteFundProofs} contain the proofs of the main
results, for the infinite- and finite-fund problems respectively.
Appendices~\ref{appendix:finiteFiltrationClassification}
and~\ref{appendix:wellPosedProofs} contain the measure-theoretic
toolkit and the compactness verifications: the statements proved
there appear in the body of the paper, and the arguments follow
classical patterns. A reader willing to take those statements on
trust may confine attention to the first two appendices; the second
pair is suitable for publication as online supplementary material.

\section{Infinite funds - proofs}
\label{appendix:infiniteFund}

\subsection{Acceptable strategies - proofs}

\begin{lemma}
\label{lemma:separabilityA}
Let $s\in{\cal T}$,
$x\in\Omega^M_s\sqcup\Omega^{M\times I}_s$,
\[
\beta\in
L^1(
{\mathbb R},
{\cal F}^\infty_s,
{\mathbb Q}^\infty_x
),
\]
and
\[
\overline X
\in
L^1(
{\mathbb R}_{\ge0},
({\cal F}^\infty_t)_{t\in{\cal T}_{<s}},
{\mathbb Q}^\infty_x
).
\]
Then
$
{\cal A}^x_s(\beta,\overline X)
$
is a separable metric space when equipped with the topology
induced from
\[
\prod_{u\in{\cal T}_{\ge s}}
L^1(
{\mathbb R}_{\ge0},
({\cal F}^\infty_t)_{t\in{\cal T}},
{\mathbb Q}^\infty
).
\]
\end{lemma}

\begin{proof}
Since ${\cal T}$ is finite,
\[
\prod_{u\in{\cal T}_{\ge s}}
L^1(
{\mathbb R}_{\ge0},
({\cal F}^\infty_t)_{t\in{\cal T}},
{\mathbb Q}^\infty
)
\]
is a finite product of separable $L^1$ spaces and is therefore
separable.

The consistency conditions~\eqref{eq:consistencyCondition}
and \eqref{eq:consistencyConditionRealised}, and the
non-participation conditions~\eqref{eq:nonParticipation}
and \eqref{eq:noConsumptionNonParticipation},
are linear constraints and define closed subsets of this product
space.

For fixed $u\in{\cal T}_{>s}$, the maps
\begin{gather*}
\gamma
\mapsto
m_s(b)+m_s(\eta_s)-\sum_{s\le t\le T}m_s(\gamma_{s,t}),
\\
\gamma
\mapsto
\sum_{u\le t\le T}m_u(\gamma_{u-1,t})+m_u(\eta_u)-\sum_{u\le t\le T}m_u(\gamma_{u,t})
\end{gather*}
are continuous as maps into $L^1$, since conditional expectation is
an $L^1$-contraction.

Hence the funding conditions~\eqref{eq:budgetJ} and \eqref{eq:selfFinancing}
define closed subsets.

Since there are only finitely many choices of
$u$, the admissible set
$
{\cal A}^x_s(\beta,\overline X)
$
is an intersection of finitely many closed subsets of a
separable metric space.

Therefore
$
{\cal A}^x_s(\beta,\overline X)
$
is itself separable.
\end{proof}
\subsection{A classical reformulation - proofs}
\label{appendix:classicalReformulation}

Throughout this appendix, statements about conditional realizations
and about the axioms of a family of preferences are understood modulo
${\mathbb Q}^\infty_x$-null sets of conditioning states.

\begin{lemma}[Gluing]
\label{lemma:gluingConditionalContinuations}
Let $u\in{\cal T}$, let $B\geq0$ be an integrable
${\cal F}^\infty_u$-measurable random budget, let $\overline W$ be a
consumption history up to time $u$, and suppose that for
${\mathbb Q}^\infty_x$-almost every $y\in\Omega^{M\times I}_u$ we are
given $Y^y\in{\cal X}^y_u(B(y),\overline W)$, measurably in $y$.
Define
\[
Y_t(\omega):=Y^{\omega^\infty_u(\omega)}_t(\omega).
\]
Then, for every $t>u$,
\[
R^Y_t=R^{Y^y}_t
\qquad
{\mathbb Q}^\infty_y\text{-a.s., for }{\mathbb Q}^\infty_x\text{-almost every }y,
\]
the stream $Y$ satisfies the constraints
\eqref{eq:classicalConsistency}, \eqref{eq:classicalNonParticipation}
and \eqref{eq:noPotOutsideParticipation} from time $u$ onwards, and
\[
Y_u+{\mathbb E}_{{\mathbb Q}^\infty_x}\bigl(R^Y_{u+1}\,\big|\,{\cal F}^\infty_u\bigr)
\leq
B+\eta'_u
\qquad
{\mathbb Q}^\infty_x\text{-a.s.}
\]
\end{lemma}
\begin{proof}
Measurability of the family implies that $Y$ is a stream. For
$r\geq u$ and integrable $Z$, the tower property of the
disintegration gives
${\mathbb E}_{{\mathbb Q}^\infty_x}(Z\mid{\cal F}^\infty_r)
={\mathbb E}_{{\mathbb Q}^\infty_y}(Z\mid{\cal F}^\infty_r)$
almost surely on the fibre over $y$, for almost every $y$. Since the
recursion defining $(R^Y_t)_{t>u}$ involves only such conditional
expectations and the values $(Y_t)_{t>u}$, which agree with
$(Y^y_t)_{t>u}$ on the fibre over $y$, backward induction on $t$
gives the fibrewise identity for the pots. Each constraint of
${\cal X}^y_u(B(y),\overline W)$ is a ${\mathbb Q}^\infty_y$-almost
sure condition on quantities which agree fibrewise with those of
$Y$; integrating over $y$ with respect to the disintegration shows
that the corresponding conditions hold
${\mathbb Q}^\infty_x$-almost surely, and in particular the funding
constraints at time $u$ combine into the displayed inequality.
\end{proof}

\begin{proof}[Continuation splicing, Lemma~\ref{lemma:continuationSplicing}]
Write $Z:=X\sqcup_u Y$. From time $u$ onwards, $Z$ is the glued
stream of Lemma~\ref{lemma:gluingConditionalContinuations} with
history $(X_t)_{t<u}$, so for $t>u$ the pots $R^Z_t$ agree fibrewise
with $R^{Y^y}_t$, the constraints hold from $u$ onwards, and
\[
Z_u+{\mathbb E}\bigl(R^Z_{u+1}\,\big|\,{\cal F}^\infty_u\bigr)
\leq
B+\eta'_u,
\]
whence $R^Z_u\leq B\leq R^X_u$ almost surely.

For $t<u$ we have $Z_t=X_t$, and the defining recursion is monotone
in the continuation pot, so backward induction from $t=u$ gives
$R^Z_t\leq R^X_t$ for every $t\leq u$. In particular the funding
constraint \eqref{eq:homogeneousBudget} at time $s$ follows from
that of $X$, and the conditions~\eqref{eq:noPotOutsideParticipation} for $t<u$ follow from those of
$X$ since $0\leq{\mathbb E}(R^Z_{t+1}\mid{\cal F}^\infty_t)
\leq{\mathbb E}(R^X_{t+1}\mid{\cal F}^\infty_t)$. The consistency
and non-participation constraints hold before $u$ because they hold
for $X$. Hence $Z\in{\cal X}^x_s(\beta,\overline X)$.

Finally, fix $y$. The consistency constraint for
$Y^y\in{\cal X}^y_u(B(y),X)$ forces $Y^y_t=X_t$ for $t<u$,
${\mathbb Q}^\infty_y$-almost surely, so $Z=Y^y$
${\mathbb Q}^\infty_y$-almost surely and
${\cal J}_u(Z\mid y)={\cal J}_u(Y^y\mid y)$.
\end{proof}

\begin{proposition}[Redistributive representation]
\label{prop:redistributiveRepresentation}
For $u \in {\cal T}_{<T}$ let
\[
p_{u+1}:=m_{u+1}\bigl(\indicator_{\tau' \leq u}\,\indicator_{u+1<\tau'_d}\bigr)
\]
denote the market-conditional proportion of agents who participate
at time $u$ and are still participating at time $u+1$.

Given a stream $X$, define its {\em redistributive representative}
$\hat\Gamma(X)$ by
\begin{gather*}
\hat\Gamma(X)_{u,t}:=X_t \quad (t \leq u),
\\
\hat\Gamma(X)_{u,u+1}
:=
\frac{\indicator_{\tau' \leq u}\,\indicator_{u+1<\tau'_d}}{p_{u+1}}\,
\hat R^X_{u+1},
\\
\hat\Gamma(X)_{u,t}:=0 \quad (t>u+1),
\end{gather*}
for $u \in {\cal T}_{\geq s}$, with the convention that the middle
expression is $0$ on the event $\{p_{u+1}=0\}$.

Suppose that, ${\mathbb Q}^\infty_x$-almost surely,
$\hat R^X_{u+1}=0$ on $\{p_{u+1}=0\}$ for every $u$; this holds in
particular whenever ${\mathbb Q}^\infty_x(\tau'\leq s)=1$. Then the
following hold.
\begin{enumerate}[(i)]
\item If $\gamma\in{\cal A}^x_s(b,\overline X)$ then
$X(\gamma)\in\hat{\cal X}^x_s(b,\overline X)$.
\item $X\in\hat{\cal X}_s^x(\beta,\overline X)$ if and only if
$\hat\Gamma(X)\in{\cal A}^x_s(\beta,\overline X)$, and in this case
$X(\hat\Gamma(X))=X$.
\item For every $u\in{\cal T}_{>s}$ and every
$y\in\Omega^{M\times I}_u$,
\[
\beta_{y,u}(\hat\Gamma(X))
=
\left(
\frac{\indicator_{\tau' \leq u-1}\,\indicator_{u<\tau'_d}}{p_{u}}\,
\hat R^X_{u}
\right)\!(y).
\]
The pot is carried entirely by the agents who participate at time
$u-1$ and remain participants at time $u$; the factor $p_u^{-1}$
converts the per-capita pot into the amount carried by each such
agent.
\end{enumerate}
In particular, a stream is induced by an admissible strategy if and
only if it lies in $\hat{\cal X}_s^x(\beta,\overline X)$.
\end{proposition}

The two notions are paired because neither carries economic content:
admissibility, like redistributive feasibility, is only a funding
requirement. Self-enforcement enters through acceptability
(Definition~\ref{def:acceptableNoIdiosyncratic}), not here.

\begin{proof}[Redistributive representation, Proposition~\ref{prop:redistributiveRepresentation}]

Throughout we use the following observation. If $S$ is
nonnegative, integrable and
$\iota^{\Omega^\infty}{\cal F}^M_{u+1}$-measurable, then the claim
\[
K:=\frac{\indicator_{\tau' \leq u}\,\indicator_{u+1<\tau'_d}}{p_{u+1}}\,S
\]
satisfies
$m_{u+1}(K)=S\,\indicator_{\{p_{u+1}>0\}}$, since $S$ and $p_{u+1}$
may be taken outside the conditional expectation; in particular $K$
is integrable, and $m_w(K)=m_w(S\,\indicator_{\{p_{u+1}>0\}})$ for
every $w\leq u+1$ by the tower property.

We first record that the standing hypothesis holds whenever
${\mathbb Q}^\infty_x(\tau'\leq s)=1$. In that case
$p_{u+1}=m_{u+1}(\indicator_{u+1<\tau'_d})$, so on
$\{p_{u+1}=0\}$ we have $u+1\geq\tau'_d$, and hence
$t\geq\tau'_d$ for every $t\geq u+1$, almost surely. By
non-participation, $X_t=0$ almost surely on $\{p_{u+1}=0\}$ for
every $t\geq u+1$, and the events $\{p_{u+1}=0\}$ increase with $u$
modulo null sets. Backward induction through the definition of
$\hat R^X$, using $\eta'\geq0$, then gives
$\hat R^X_{u+1}=0$ on $\{p_{u+1}=0\}$.

(i) Let $\gamma\in{\cal A}^x_s(b,\overline X)$ and set $X:=X(\gamma)$.
Define $\hat P_{T+1}:=0$ and, for $u\in{\cal T}_{>s}$,
\[
\hat P_u:=\sum_{u\leq t\leq T}m_u(\gamma_{u-1,t}),
\]
which is nonnegative, integrable and
$\iota^{\Omega^\infty}{\cal F}^M_u$-measurable. By the tower
property,
\[
m_u(\hat P_{u+1})=\sum_{u+1\leq t\leq T}m_u(\gamma_{u,t}),
\]
so the self-financing condition~\eqref{eq:selfFinancing} may be
rewritten as
\[
m_u(X_u)+m_u(\hat P_{u+1})\leq \hat P_u+m_u(\eta'_u),
\qquad u\in{\cal T}_{>s},
\]
and the issuance-funding condition~\eqref{eq:budgetJ} as
\[
m_s(X_s)+m_s(\hat P_{s+1})\leq m_s(b)+m_s(\eta'_s).
\]
Backward induction gives $\hat P_u\geq\hat R^X_u$ for every
$u\in{\cal T}_{>s}$: this is trivial at $u=T+1$, and if
$\hat P_{u+1}\geq\hat R^X_{u+1}$ then
\[
\hat P_u\geq m_u(X_u)+m_u(\hat P_{u+1})-m_u(\eta'_u)
\geq m_u(X_u)+m_u(\hat R^X_{u+1})-m_u(\eta'_u),
\]
while also $\hat P_u\geq0$, so $\hat P_u\geq\hat R^X_u$. Combining with the
issuance-funding inequality and monotonicity of $m_s$ gives
\eqref{eq:redistributiveBudget}. The consistency and non-participation
constraints for $X$ follow from conditions (i) and (ii) for
$\gamma$. Hence $X\in{\cal X}^x_s(b,\overline X)$. The same argument
establishes the ``only if'' direction of the equivalence in the
carried-forward form stated in the remark following the proposition.

(ii) Suppose $X\in{\cal X}_s^x(\beta,\overline X)$ and write
$\hat\Gamma:=\hat\Gamma(X)$. Each leg of $\hat\Gamma$ is nonnegative and, by the
observation above, integrable. The consistency conditions hold by
construction, since $\hat\Gamma_{t,t}=X_t$. For the non-participation
conditions: on $\{u<\tau'\}$ the factor $\indicator_{\tau'\leq u}$
vanishes and on $\{u\geq\tau'_d\}$ the factor
$\indicator_{u+1<\tau'_d}$ vanishes, so $\hat\Gamma_{u,u+1}=0$ on either
event, while the legs with $t>u+1$ vanish identically; and
$\hat\Gamma_{t,t}=X_t=0$ outside $[\tau',\tau'_d)$
by~\eqref{eq:classicalNonParticipation}. For the funding conditions,
the observation together with the standing hypothesis gives
\begin{gather*}
m_u(\hat\Gamma_{u,u+1})=m_u\bigl(\hat R^X_{u+1}\,\indicator_{\{p_{u+1}>0\}}\bigr)=m_u\bigl(\hat R^X_{u+1}\bigr),
\\
m_u(\hat\Gamma_{u-1,u})=\hat R^X_u\,\indicator_{\{p_u>0\}}=\hat R^X_u.
\end{gather*}
Hence the issuance-funding condition for $\hat\Gamma$ reads
\[
\sum_{s\leq t\leq T}m_s(\hat\Gamma_{s,t})
=m_s(X_s)+m_s\bigl(\hat R^X_{s+1}\bigr)
\leq m_s(\beta)+m_s(\eta'_s),
\]
which is precisely \eqref{eq:redistributiveBudget}, and the
self-financing condition at $u\in{\cal T}_{>s}$ reads
\[
m_u(X_u)+m_u\bigl(\hat R^X_{u+1}\bigr)
\leq
\hat R^X_u+m_u(\eta'_u),
\]
which holds since $\hat R^X_u$ dominates the bracketed expression in
its definition. Thus $\hat\Gamma\in{\cal A}^x_s(\beta,\overline X)$, and
$X(\hat\Gamma)=X$ holds by construction.

Conversely, suppose $\hat\Gamma(X)\in{\cal A}^x_s(\beta,\overline X)$.
The consistency condition for $\hat\Gamma(X)$ forces
$X_t=\overline X_t$ for $t<s$, and part (i) applied to $\hat\Gamma(X)$
gives $X=X(\hat\Gamma(X))\in{\cal X}^x_s(\beta,\overline X)$.

(iii) For $u\in{\cal T}_{>s}$, the only leg of $\hat\Gamma(X)_{u-1}$
paying at a time $t\geq u$ is $\hat\Gamma(X)_{u-1,u}$, so
\[
\beta_{y,u}(\hat\Gamma(X))
=
{\mathbb E}_{{\mathbb Q}^\infty}
\!\left(
\frac{\indicator_{\tau' \leq u-1}\,\indicator_{u<\tau'_d}}{p_{u}}\,
\hat R^X_{u}
\,\middle|\,
\omega^\infty_u=y
\right).
\]
The events $\{\tau'\leq u-1\}$ and $\{u<\tau'_d\}$ lie in
${\cal F}^\infty_u$, and $p_u$ and $\hat R^X_u$ are
$\iota^{\Omega^\infty}{\cal F}^M_u$-measurable, so the integrand is
determined by the time-$u$ state and the conditional expectation
evaluates to the stated expression.

The final claim follows by combining (i) and (ii).
\end{proof}

\begin{proof}[Self-financing representation, Proposition~\ref{prop:selfFinancingRepresentation}]
Each leg of $\Gamma(X)$ is nonnegative and integrable, and the
consistency conditions hold since $\Gamma(X)_{t,t}=X_t$. For
the non-participation conditions, note first that $R^X_{u+1}=0$
almost surely on $\{u+1<\tau'\}\cup\{u+1\geq\tau'_d\}$: on
pre-entry events this follows from
\eqref{eq:noPotOutsideParticipation} together with $R^X\geq0$, and
on post-exit events it follows by backward induction through the
recursion, since $X$ vanishes there
by~\eqref{eq:classicalNonParticipation}. The payee condition~\eqref{eq:noConsumptionNonParticipation} for the legs $X_t$ is
\eqref{eq:classicalNonParticipation}, and for the carry legs it is
the vanishing just established; the writer condition~\eqref{eq:nonParticipation} likewise follows since the only forward
leg is the carry leg.

For the issuance-funding condition, applying $m_s$ to
\eqref{eq:homogeneousBudget} gives
\[
\sum_{s\leq t\leq T}m_s\bigl(\Gamma(X)_{s,t}\bigr)
=
m_s(X_s)+m_s\bigl(R^X_{s+1}\bigr)
\leq
m_s(\beta)+m_s(\eta'_s).
\]
For the self-financing condition at $u\in{\cal T}_{>s}$: the only
leg of $\Gamma(X)_{u-1}$ paying at $t\geq u$ is the carry leg
$R^X_u$, and by the definition of the self-financing pot,
\[
X_u+{\mathbb E}\bigl(R^X_{u+1}\,\big|\,{\cal F}^\infty_u\bigr)-\eta'_u
\leq
R^X_u,
\]
so applying $m_u$ gives
\eqref{eq:selfFinancing}. Hence
$\Gamma(X)\in{\cal A}^x_s(\beta,\overline X)$, and
$X(\Gamma(X))=X$ by construction.

Finally, for $u\in{\cal T}_{>s}$, since the only leg of
$\Gamma(X)_{u-1}$ paying at $t\geq u$ is $R^X_u$, which is
${\cal F}^\infty_u$-measurable,
\[
\beta_{y,u}\bigl(\Gamma(X)\bigr)
=
{\mathbb E}_{{\mathbb Q}^\infty}
\bigl(
R^X_u
\,\big|\,
\omega^\infty_u=y
\bigr)
=
R^X_u(y).
\]
\end{proof}

The following verification inequality expresses the role of the
continuation-consistency axiom: replacing all continuations at a
given date by statewise-optimal ones can only improve a stream, so
the cost of the utilities a stream delivers at the next date bounds
the budget it requires now.

\begin{lemma}[Verification inequality]
\label{lemma:verification}
Assume $\Prob$ admits regular maximisers. Let $u\in{\cal T}_{<T}$,
let $W$ be a stream, and let $B\geq0$ be integrable,
${\cal F}^\infty_{u+1}$-measurable, and vanishing almost surely on
$\{u<\tau\}\cup\{u\geq\tau_d\}$. Suppose that for
${\mathbb Q}^\infty_x$-almost every
$x''\in\Omega^{M\times I}_{u+1}$,
\[
{\cal J}_{u+1}(W\mid x'')
\leq
V_{u+1}\bigl(x'',B(x''),W\bigr).
\]
Then for ${\mathbb Q}^\infty_x$-almost every
$x'\in\Omega^{M\times I}_u$ with
${\mathbb Q}^\infty_{x'}(\tau\leq u<\tau_d)=1$,
\[
{\cal J}_u(W\mid x')
\leq
V_u\Bigl(x',\,
W_u(x')+{\mathbb E}\bigl(B\,\big|\,{\cal F}^\infty_u\bigr)(x')-\eta_u(x'),\,
W\Bigr).
\]
\end{lemma}
\begin{proof}
By the arg max property (in the random-budget form) choose,
measurably in $x''$, streams
$Y^{x''}\in{\cal X}^{x''}_{u+1}(B(x''),W)$ with
${\cal J}_{u+1}(Y^{x''}\mid x'')=V_{u+1}(x'',B(x''),W)$; the
constraint sets are non-empty since $B\geq0$ permits the zero tail
at exited states and $B$ vanishes at pre-entry states. Define the
composite stream $Z$ by $Z_t:=W_t$ for $t\leq u$ and by gluing the
family $(Y^{x''})$ for $t\geq u+1$
(Lemma~\ref{lemma:gluingConditionalContinuations}). The gluing lemma
gives $R^Z_{u+1}\leq B$ fibrewise, so
\[
Z_u+{\mathbb E}\bigl(R^Z_{u+1}\,\big|\,{\cal F}^\infty_u\bigr)
\leq
W_u+{\mathbb E}\bigl(B\,\big|\,{\cal F}^\infty_u\bigr),
\]
and the vanishing hypotheses on $B$ give
\eqref{eq:noPotOutsideParticipation} at time $u$; the conditions at
later times hold fibrewise. Hence, at almost every participating
$x'$, $Z$ lies in the constraint set with the displayed budget, so
${\cal J}_u(Z\mid x')$ is at most the displayed value. Since $Z$ and
$W$ agree at every time $t<u+1$ and
${\cal J}_{u+1}(Z\mid x'')={\cal J}_{u+1}(Y^{x''}\mid x'')
\geq{\cal J}_{u+1}(W\mid x'')$ for almost every $x''$, continuation
consistency gives
${\cal J}_u(W\mid x')\leq{\cal J}_u(Z\mid x')$.
\end{proof}

\begin{proof}[Reduction to the classical continuation problem, Proposition~\ref{prop:reductionClassical}]
Since ${\cal G}^x_s\subseteq{\cal H}^x_s$ by definition, for both
implications we may assume $\gamma\in{\cal H}^x_s(\beta,\overline X)$.
At the homogeneous state $x$, the conclusion of
Lemma~\ref{lemma:deliveredValueBound}(iii) reads
$\kappa_s(x)\leq\beta$; by monotonicity and right-continuity of
$V_s(x,\cdot,\overline X)$ this means
\[
{\cal J}_s(X(\gamma')\mid x)
\leq
V_s(x,\beta,\overline X)
\qquad
\text{for every }
\gamma'\in{\cal H}^x_s(\beta,\overline X).
\tag{$*$}
\]

Suppose (i). By
Proposition~\ref{prop:acceptableFromOptimisers} there exists
$\gamma^*\in{\cal G}^x_s(\beta,\overline X)
\subseteq{\cal H}^x_s(\beta,\overline X)$ with
${\cal J}_s(X(\gamma^*)\mid x)=V_s(x,\beta,\overline X)$. The
acceptability criterion for $\gamma$, tested against
$\gamma^*$, gives
${\cal J}_s(X(\gamma)\mid x)\geq V_s(x,\beta,\overline X)$, and
($*$) gives the reverse inequality. This proves (ii).

Suppose (ii). For every
$\gamma'\in{\cal H}^x_s(\beta,\overline X)$, ($*$) gives
${\cal J}_s(X(\gamma')\mid x)
\leq V_s(x,\beta,\overline X)
={\cal J}_s(X(\gamma)\mid x)$,
so the acceptability criterion holds and
$\gamma\in{\cal G}^x_s(\beta,\overline X)$.

The final claim is contained in
Proposition~\ref{prop:acceptableFromOptimisers}.
\end{proof}

\subsection{No mutually beneficial contracts - proofs}
\label{appendix:noMutuallyBeneficial}

\begin{lemma}[Conservation of continuation wealth]
\label{lemma:continuationWealthConservation}
Let $\gamma\in{\cal A}^x_s(\beta,\overline X)$ and define
\[
a_u:={\mathbb E}_{{\mathbb Q}^\infty_x}
\Bigl(\sum_{u\leq t\leq T}\gamma_{u,t}\Bigr),
\qquad
u\in{\cal T}_{\geq s}.
\]
Then
\[
a_s\leq{\mathbb E}(\beta)+{\mathbb E}(\eta'_s),
\qquad
a_u\leq a_{u-1}-{\mathbb E}\bigl(X_{u-1}(\gamma)\bigr)+{\mathbb E}(\eta'_u)
\quad(u\in{\cal T}_{>s}),
\]
with equality in a given inequality if and only if the corresponding
funding condition of Definition~\ref{def:admissibleNoIdiosyncratic}
binds ${\mathbb Q}^\infty_x$-almost surely. Moreover, for
$u\in{\cal T}_{>s}$,
\[
\int_{\Omega^{M\times I}_u}\beta_{y,u}(\gamma)\,
d{\mathbb Q}^\infty_x(y)
=
a_{u-1}-{\mathbb E}\bigl(X_{u-1}(\gamma)\bigr),
\]
and telescoping yields
\[
{\mathbb E}\Bigl(\sum_{s\leq t\leq T}X_t(\gamma)\Bigr)
\leq
{\mathbb E}(\beta)+{\mathbb E}\Bigl(\sum_{s\leq t\leq T}\eta'_t\Bigr).
\]
\end{lemma}
\begin{proof}
Integrating the issuance-funding condition~\eqref{eq:budgetJ} and
the self-financing conditions~\eqref{eq:selfFinancing} over
${\mathbb Q}^\infty_x$ and using the tower property to remove $m_s$
and $m_u$ gives the first two displays, using the consistency
condition to identify $\gamma_{u-1,u-1}$ with $X_{u-1}(\gamma)$. A
nonnegative random variable with zero integral vanishes almost
surely, which gives the equality clause. The identity for the
integrated pre-trade budgets follows from \eqref{eq:defBeta} by the
tower property. Finally $a_T={\mathbb E}(X_T(\gamma))$, so
telescoping the second display from $u=T$ down to $u=s+1$ and
combining with the first gives the last display.
\end{proof}

The next lemma is the engine of Theorem~\ref{thm:noMutuallyBeneficialContracts}. Its proof and the proof of
Proposition~\ref{prop:acceptableFromOptimisers} below proceed by a
joint backward induction on $s$: in the inductive step at $s$, each
proof may invoke both statements at every time $u>s$, and the proof
of Lemma~\ref{lemma:deliveredValueBound} at level $s$ does not use
Proposition~\ref{prop:acceptableFromOptimisers} at level $s$.

\begin{lemma}[Cost of delivered utility]
\label{lemma:deliveredValueBound}
Assume $\Prob$ admits regular maximisers. Let $s\in{\cal T}$,
$x\in\Omega^M_s\sqcup\Omega^{M\times I}_s$ with
${\mathbb Q}^\infty_x(\tau\leq s<\tau_d)=1$, let
$\gamma\in{\cal H}^x_s(\beta,\overline X)$, and write
$X:=X(\gamma)$. For $u\in{\cal T}_{\geq s}$ and
${\mathbb Q}^\infty_x$-almost every $x'\in\Omega^{M\times I}_u$
define
\[
\kappa_u(x')
:=
\inf\bigl\{b\in{\mathbb R}
\;:\;
{\cal X}^{x'}_u(b,X)\neq\emptyset
\text{ and }
V_u(x',b,X)\geq{\cal J}_u(X\mid x')
\bigr\}.
\]
Then:
\begin{enumerate}[(i)]
\item each $\kappa_u$ is well defined, measurable and integrable,
with $\kappa_u(x')=0$ at states with $u\geq\tau_d$;
\item $\kappa_u(x')\geq\beta_{x',u}(\gamma)\geq0$ for every
$u\in{\cal T}_{>s}$ and almost every $x'$;
\item ${\mathbb E}_{{\mathbb Q}^\infty_x}(\kappa_s)
\leq{\mathbb E}_{{\mathbb Q}^\infty_x}(\beta)$.
\end{enumerate}
If moreover
${\cal J}_s(X\mid y)\geq V_s(y,\beta(y),\overline X)$
for almost every compatible $y\in\Omega^{M\times I}_s$, then
$\kappa_s=\beta$ almost everywhere, every funding condition of
$\gamma$ binds almost surely,
$\kappa_u=\beta_{\cdot,u}(\gamma)$ almost everywhere for every
$u\in{\cal T}_{>s}$, and, writing $\beta_{\cdot,s}(\gamma):=\beta$,
\[
{\cal J}_u(X\mid x')
=
V_u\bigl(x',\beta_{x',u}(\gamma),X\bigr)
\]
for every $u\in{\cal T}_{\geq s}$ and almost every $x'$.
\end{lemma}
\begin{proof}
Since ${\mathbb Q}^\infty_x(\tau\leq s<\tau_d)=1$, for every
$u\geq s$ almost every $x'\in\Omega^{M\times I}_u$ satisfies either
$\tau\leq s\leq u<\tau_d$ ({\em participating}) or $u\geq\tau_d$
({\em exited}).

{\em Exited states.} At such $x'$, every stream in
${\cal X}^{x'}_u(b,X)$ has zero tail by
\eqref{eq:classicalNonParticipation}, and the tail of $X$ vanishes
by the non-participation conditions for $\gamma$, so
$V_u(x',b,X)={\cal J}_u(X\mid x')$ for every feasible $b$; since
$X_u=0$, $\eta'_u=0$ after exit, and pots vanish, the funding
constraint requires exactly $b\geq0$. Hence $\kappa_u(x')=0$.
Moreover $\beta_{x',u}(\gamma)=0$ at exited states by the payee
condition~\eqref{eq:noConsumptionNonParticipation}.

{\em Participating states: well-definedness.} By part (iii) of the
definition of regular maximisers, $V_u(x',b,X)$ increases, as
$b\to\infty$, to the supremum of the utilities of all integrable
streams, and this supremum is attained by no integrable stream. The
tail of $X$ is redistributively feasible for a sufficiently large budget, by
\eqref{eq:receivedToDate}, so part (ii) gives
${\cal J}_u(X\mid x')<+\infty$, and the supremum therefore strictly
exceeds ${\cal J}_u(X\mid x')$. Hence the set defining
$\kappa_u(x')$ is non-empty (when ${\cal J}_u(X\mid x')=-\infty$ it
contains the whole domain); it is bounded below since the funding
constraint forces $b\geq-\eta_u(x')$. Measurability follows since
the infimum may be taken over rational $b$, by monotonicity and
right-continuity of $V_u(x',\cdot,X)$.

{\em Claim (ii).} Fix $u\in{\cal T}_{>s}$.
By Proposition~\ref{prop:acceptableFromOptimisers} at time $u$
(joint induction), applied fibrewise with the budget
$\beta_{\cdot,u}(\gamma)$ and history $X$, there is a measurable
family
$x'\mapsto\gamma''^{,x'}\in{\cal G}^{x'}_u(\beta_{x',u}(\gamma),X)$
with
${\cal J}_u(X(\gamma''^{,x'})\mid x')
=V_u(x',\beta_{x',u}(\gamma),X)$
for almost every $x'$. The screening criterion
\eqref{eq:screeningCriterion} for $\gamma$, tested against this
family, gives
\[
{\cal J}_u(X\mid x')
\geq
V_u\bigl(x',\beta_{x',u}(\gamma),X\bigr)
\qquad
\text{for almost every }x'.
\]
Fix a participating $x'$ at which this holds. If
${\cal J}_u(X\mid x')>-\infty$: since $V_u(x',\cdot,X)$ is strictly
increasing, every $b$ with
$V_u(x',b,X)\geq{\cal J}_u(X\mid x')$ satisfies
$b\geq\beta_{x',u}(\gamma)$, and $\beta_{x',u}(\gamma)\geq0$ since
every contract leg is nonnegative. If
${\cal J}_u(X\mid x')=-\infty$: the displayed inequality forces
$V_u(x',\beta_{x',u}(\gamma),X)=-\infty$, so part (iv) of the
definition of regular maximisers gives
$\beta_{x',u}(\gamma)+H^{x'}_u=0$; both terms being nonnegative,
$\beta_{x',u}(\gamma)=0$ and the contributions vanish from $u$
onwards at $x'$, so the funding constraint bounds the domain below
by $0=\beta_{x',u}(\gamma)$. At exited states both sides
vanish. This proves (ii); in particular $\kappa_u\geq0$ for
$u\in{\cal T}_{>s}$.

{\em The chain.} At $u=T$, consuming the whole budget shows
$\kappa_T\leq X_T-\eta'_T$ at participating states, with equality
of both sides to zero at exited states. For $u\in{\cal T}_{<T}$
with $u\geq s$, apply Lemma~\ref{lemma:verification} with $W:=X$
and $B:=\kappa_{u+1}$, which is nonnegative by (ii), vanishes
outside participation, and is ${\cal F}^\infty_{u+1}$-measurable:
this gives, at almost every participating $x'$,
\[
\kappa_u(x')
\leq
X_u(x')
+{\mathbb E}\bigl(\kappa_{u+1}\,\big|\,{\cal F}^\infty_u\bigr)(x')
-\eta'_u(x'),
\]
and the same inequality holds trivially at exited states.
Integrability of each $\kappa_u$ follows backward through these
bounds together with $\kappa_u\geq-\eta'_u$. Integrating, and
writing $a_u$ as in
Lemma~\ref{lemma:continuationWealthConservation},
\[
{\mathbb E}(\kappa_u)
\leq
{\mathbb E}(X_u)+{\mathbb E}(\kappa_{u+1})-{\mathbb E}(\eta'_u),
\qquad
{\mathbb E}(\kappa_T)\leq a_T-{\mathbb E}(\eta'_T),
\]
so backward induction using
Lemma~\ref{lemma:continuationWealthConservation} gives
${\mathbb E}(\kappa_u)\leq a_u-{\mathbb E}(\eta'_u)$ for every $u$,
and at $u=s$ the issuance bound
$a_s\leq{\mathbb E}(\beta)+{\mathbb E}(\eta'_s)$ yields (iii).

{\em The supplement.} Suppose
${\cal J}_s(X\mid y)\geq V_s(y,\beta(y),\overline X)$ almost
everywhere. Arguing exactly as in the proof of (ii) (via strict
monotonicity where ${\cal J}_s(X\mid y)>-\infty$, and via part (iv)
of regular maximisers together with $\beta\geq0$ where
${\cal J}_s(X\mid y)=-\infty$) we obtain $\kappa_s\geq\beta$
almost everywhere. Then (iii) forces
${\mathbb E}(\kappa_s)={\mathbb E}(\beta)$ and
hence equality at every step of the chain: $\kappa_s=\beta$ almost
everywhere; each funding condition binds almost surely by the
equality clause of
Lemma~\ref{lemma:continuationWealthConservation}; and
${\mathbb E}(\kappa_u)=a_u-{\mathbb E}(\eta'_u)$ for every $u$.
Combining the last identity with the binding conservation
identities gives
\[
{\mathbb E}(\kappa_u)
=
a_{u-1}-{\mathbb E}(X_{u-1})
=
\int\beta_{y,u}(\gamma)\,d{\mathbb Q}^\infty_x(y),
\]
and with the pointwise bound (ii) this forces
$\kappa_u=\beta_{\cdot,u}(\gamma)$ almost everywhere. Finally, fix $u$ and a participating $x'$ at which
${\cal J}_u(X\mid x')>-\infty$. For every $b>\kappa_u(x')$ we have
$V_u(x',b,X)\geq{\cal J}_u(X\mid x')$, so right-continuity gives
${\cal J}_u(X\mid x')\leq V_u(x',\kappa_u(x'),X)
=V_u(x',\beta_{x',u}(\gamma),X)$; the reverse inequality is the
breakaway bound established in the proof of (ii) (at $u=s$, the
standing hypothesis of the supplement). Where
${\cal J}_u(X\mid x')=-\infty$ the argument of (ii) shows
$V_u(x',\beta_{x',u}(\gamma),X)=-\infty$ as well, and at exited
states both sides are the utility of the zero tail. This completes
the proof.
\end{proof}

\begin{proposition}[Acceptable strategies from optimisers]
\label{prop:acceptableFromOptimisers}
Assume $\Prob$ admits regular maximisers. Let $s\in{\cal T}$,
$x\in\Omega^M_s\sqcup\Omega^{M\times I}_s$ with
${\mathbb Q}^\infty_x(\tau\leq s<\tau_d)=1$, and let
$(\beta,\overline X)$ be given. Then there exists a stream $X^\circ$
such that, for ${\mathbb Q}^\infty_x$-almost every
$y\in\Omega^{M\times I}_s$,
\[
X^\circ\in{\cal X}^y_s(\beta(y),\overline X),
\qquad
{\cal J}_s(X^\circ\mid y)=V_s\bigl(y,\beta(y),\overline X\bigr),
\]
and which is statewise optimal at every future date:
\[
{\cal J}_u(X^\circ\mid x'')
=
V_u\bigl(x'',R^{X^\circ}_u(x''),X^\circ\bigr)
\]
for every $u\in{\cal T}_{>s}$ and almost every
$x''\in\Omega^{M\times I}_u$. Moreover
$\Gamma(X^\circ)\in{\cal G}^x_s(\beta,\overline X)$; for
${\mathbb Q}^\infty_x$-almost every $y\in\Omega^{M\times I}_s$,
$\Gamma(X^\circ)\in{\cal G}^y_s(\beta(y),\overline X)$ with
${\cal J}_s(X(\Gamma(X^\circ))\mid y)
=V_s(y,\beta(y),\overline X)$; and the construction depends
measurably on the conditioning data, so that, applied fibrewise to
${\cal F}^\infty_u$-measurable budgets and histories at a date $u$,
it yields measurable comparison families.
\end{proposition}
\begin{proof}
{\em Step 1: statewise-optimal optimisers.} We claim, by backward
induction on $u$, that for every $u$, every random budget and
every history there is an arg max whose tails are statewise optimal
at all later dates. At $u=T$ this is the arg max property itself.
For $u<T$: take an arg max family (regular maximisers (i), in the
random-budget form), glue it into a stream $X^1$
(Lemma~\ref{lemma:gluingConditionalContinuations}); by the
inductive hypothesis choose, for almost every $x''$ at $u+1$, a
statewise-optimal arg max of
$(u+1,x'',R^{X^1}_{u+1}(x''),X^1)$, and let $X^2$ be the
corresponding splice
(Lemma~\ref{lemma:continuationSplicing}, with $B:=R^{X^1}_{u+1}$).
Then $X^2$ is feasible for the original budget; its date-$(u+1)$
values dominate those of $X^1$, so continuation consistency gives
${\cal J}_u(X^2\mid\cdot)\geq{\cal J}_u(X^1\mid\cdot)$, whence
$X^2$ is again an arg max; and it is statewise optimal at $u+1$ by
construction and at later dates because its pots there are the pots
of the substituted continuations.

For the measurability of this construction, argue as follows. By the
parametrised clause of part (i) of the definition of regular
maximisers, there is, at each date, a single arg max map, jointly
measurable in the state, the budget and the history; by backward
induction we may assume the date-$(u+1)$ statewise-optimal family
already constructed is of this jointly measurable form. The
substituted continuations at $u+1$ are then obtained by evaluating
that family at the arguments
$\bigl(x'',R^{X^1}_{u+1}(x''),(X^1_t)_{t\le u}\bigr)$; the pot
$R^{X^1}_{u+1}$ is a measurable function of $X^1$ by
Lemma~\ref{lemma:closednessContinuationSets}, and $X^1$ is itself a
jointly measurable function of the outer data by the parametrised
clause, so the plugged arguments, and hence the substituted family,
are jointly measurable. Finally the splice $X^2$ is given by an
explicit measurable formula in $X^1$ and the substituted family
(Lemma~\ref{lemma:continuationSplicing}), so joint measurability is
preserved, completing the induction.

Applying the claim at $(s,x,\beta,\overline X)$ gives $X^\circ$.

{\em Step 2: admissibility.} By
Proposition~\ref{prop:selfFinancingRepresentation},
$\Gamma(X^\circ)\in{\cal A}^x_s(\beta,\overline X)$ and
$\beta_{x'',u}(\Gamma(X^\circ))=R^{X^\circ}_u(x'')$ for
$u\in{\cal T}_{>s}$.

{\em Step 3: the screening criterion.} Fix $u\in{\cal T}_{>s}$
and a measurable comparison family
$x''\mapsto\gamma'^{,x''}\in
{\cal G}^{x''}_u(R^{X^\circ}_u(x''),X^\circ)
\subseteq{\cal H}^{x''}_u(R^{X^\circ}_u(x''),X^\circ)$. By
Lemma~\ref{lemma:deliveredValueBound}(iii) at each participating
homogeneous state $x''$ (joint induction, level $u>s$),
\begin{gather*}
{\cal J}_u(X(\gamma'^{,x''})\mid x'')
\leq
V_u\bigl(x'',R^{X^\circ}_u(x''),X^\circ\bigr)
=
{\cal J}_u(X^\circ\mid x'')
\\
\text{for almost every such }x''.
\end{gather*}
At exited states every strategy delivers the zero tail and the
criterion is trivial. Hence
$\Gamma(X^\circ)\in{\cal H}^x_s(\beta,\overline X)$.

{\em Step 4: the acceptability criterion.} Let
$y\mapsto\gamma'^{,y}$ be a measurable comparison family with
$\gamma'^{,y}\in{\cal H}^y_s(\beta(y),\overline X)$.
Lemma~\ref{lemma:deliveredValueBound}(iii) at each homogeneous
state $y$ (level $s$: its proof uses
Proposition~\ref{prop:acceptableFromOptimisers} only at times
$u>s$) gives
${\cal J}_s(X(\gamma'^{,y})\mid y)\leq V_s(y,\beta(y),\overline X)
={\cal J}_s(X^\circ\mid y)$ for almost every $y$. Hence
$\Gamma(X^\circ)\in{\cal G}^x_s(\beta,\overline X)$.

{\em Step 5: the fibrewise and measurability clauses.} Every
criterion verified above is a ${\mathbb Q}^\infty_x$-almost sure
condition whose ingredients are computed fibrewise over the
time-$s$ state; disintegrating yields the corresponding
${\mathbb Q}^\infty_y$-almost sure conditions for almost every $y$,
giving the fibrewise memberships. The final clause is the joint
measurability established in Step 1: the construction is a
composition of the parametrised arg max maps of regular maximisers
(i) with the explicit gluing and splicing formulas, each of which is
a measurable function of the conditioning data, the budgets and the
histories, so the resulting family may be applied fibrewise to
${\cal F}^\infty_u$-measurable data and yields measurable comparison
families.
\end{proof}

\begin{proof}[No mutually beneficial contracts across states, Theorem~\ref{thm:noMutuallyBeneficialContracts}]
Let $\gamma\in{\cal G}^x_s(\beta,\overline X)
\subseteq{\cal H}^x_s(\beta,\overline X)$.
Proposition~\ref{prop:acceptableFromOptimisers}, applied at
$(s,x,\beta,\overline X)$, provides a single strategy
$\gamma^*:=\Gamma(X^\circ)$ which lies in
${\cal H}^y_s(\beta(y),\overline X)$ for almost every
$y\in\Omega^{M\times I}_s$, with
${\cal J}_s(X(\gamma^*)\mid y)=V_s(y,\beta(y),\overline X)$. The
acceptability criterion \eqref{eq:acceptabilityCriterion} for
$\gamma$, tested against the constant comparison family
$\gamma^*$, gives
\[
{\cal J}_s(X(\gamma)\mid y)
\geq
V_s\bigl(y,\beta(y),\overline X\bigr)
\qquad
\text{for almost every }y.
\]
The supplement to
Lemma~\ref{lemma:deliveredValueBound} now yields (i), (ii) and
(iii).
\end{proof}

\section{Finite funds - proofs}
\label{appendix:finiteFundProofs}

We record a fibrewise form of the Invariance axiom and a form of
the Lottery tolerance axiom at enriched states, both used
repeatedly in this appendix.

\begin{lemma}[Invariance under state-dependent automorphisms]
\label{lemma:invarianceFibrewise}
Let ${\cal J}$ be a family of $\underline{\Omega}^1$ preferences,
let $s\in{\cal T}$, and let
$(\phi_{\omega^1})_{\omega^1\in\Omega^1}$ be a family of
filtration-preserving mod $0$ automorphisms of
$\underline{\Omega}^\can_{\cal T}$, jointly measurable in
$(\omega^{\can},\omega^1)$, whose time-$t$ coordinates, along with
those of the fibrewise inverses, depend on $\omega^1$ only through
${\cal F}^1_{t\vee s}$-measurable information. Then, writing
$\hat\Phi(\omega^{\can},\omega^1):=(\phi_{\omega^1}(\omega^{\can}),\omega^1)$,
\[
{\cal J}_s(X\mid x)
=
{\cal J}_s
\bigl(
X\circ\hat\Phi
\,\big|\,
x
\bigr)
\]
for every stream $X$ and every $x\in\Omega_s^1$. In particular, for
a family $(\phi_x)_{x\in\Omega_s^1}$ indexed by the time-$s$ state,
\[
{\cal J}_s(X\mid x)
=
{\cal J}_s
\bigl(
X\circ(\phi_x\times\id_{\Omega^1})
\,\big|\,
x
\bigr)
\]
for every $x\in\Omega_s^1$.
\end{lemma}

\begin{proof}
The map
$\Phi(\omega^{\can},\omega^1):=\phi_{\omega^1}(\omega^{\can})$ is an
adapted randomized automorphism at time $s$, so the first identity
is the Invariance axiom. For the second, a family indexed by the
time-$s$ state satisfies the measurability hypotheses, and the
axiom may equally be applied with the fixed automorphism $\phi_x$
at each $x$.
\end{proof}

\begin{lemma}[Lottery tolerance at enriched states]
\label{lemma:lotteryToleranceEnriched}
Let $u\in{\cal T}$, let ${\cal G}_u$, $y$ and $x'':=x''(i,y)$ be as
in \eqref{eq:enrichedValue}, and write $y'$ for the full states of
$\Omega^*_u$ refining $y$. Let $(Z^n)$ be streams on $\Omega^*$,
let $(B_n)$ be measurable sets of full states with
${\mathbb Q}^*_y(y'\in B_n)\to0$, and suppose there are
$M,m\in{\mathbb R}$ such that, for every $n$,
\begin{gather*}
{\cal J}^i_u(Z^n\mid y')\ge M
\quad\text{for ${\mathbb Q}^*_y$-a.e.\ } y'\notin B_n,
\\
{\cal J}^i_u(Z^n\mid y')\ge m
\quad\text{for ${\mathbb Q}^*_y$-a.e.\ } y'\in B_n.
\end{gather*}
Then
$\liminf_n{\cal J}^i_u(Z^n\mid y)\ge M$.
\end{lemma}

\begin{proof}
Choose a transport $\Xi_y$ as in
\eqref{eq:enrichedValue} and set $W^n:=Z^n\circ\Xi_y$,
so that ${\cal J}^i_u(Z^n\mid y)={\cal J}_u(W^n\mid x'')$. Under
$\Xi_y$, the full states $y'$ refining $y$ correspond, mod $0$, to
the states $\omega^{\can}\in\Omega^{\can}_{{\cal T}_{\leq u}}$ of
the canonical coordinates of the $+$-lift at times ${\leq}\,u$, the
law of $\omega^{\can}$ under ${\mathbb Q}^*_y$ being
${\mathbb Q}^{\can}_{{\cal T}_{\leq u}}$, and by
\eqref{eq:frozenEnrichedValue},
\[
{\cal J}_u\bigl(W^{n,[\omega^{\can}]}\,\big|\,x''\bigr)
=
{\cal J}^i_u\bigl(Z^n\,\big|\,y'(\omega^{\can})\bigr)
\qquad
\text{for a.e.\ }\omega^{\can}.
\]
The sets
$A_n:=\{\omega^{\can}:y'(\omega^{\can})\in B_n\}$ are measurable
mod $0$ with
${\mathbb Q}^{\can}_{{\cal T}_{\leq u}}(A_n)
={\mathbb Q}^*_y(y'\in B_n)\to0$, and the hypotheses
transfer to the sections $W^{n,[\omega^{\can}]}$, so the Lottery tolerance
axiom, applied at time $u$ and the state $x''$, gives
$\liminf_n{\cal J}_u(W^n\mid x'')\ge M$.
\end{proof}

\subsection{Admissible strategies - proofs}
\label{appendix:admissibleStrategies}

\begin{proof}[Independence of the choice of $\Psi$, Lemma~\ref{lem:valueFunctionWellDefined}]
	Let $\Psi_1,\Psi_2$ be two filtration-preserving mod 0 isomorphisms satisfying the conditions of the definition.
	
	Since $\rho((0,i))$ swaps slots $0$ and $i$, the general identity $(g\omega)^k=\omega^{g^{-1}(k)}$ gives $\omega^{I+,i}\circ\rho((0,i))=\omega^{I+,0}$ for every $i\in{\mathbb N}$. Hence the condition $\omega^{I+}=\omega^{I+,i}\circ\rho((0,i))\circ\Psi_k$ for all $i$ is equivalent to the single condition $\omega^{I+}=\omega^{I+,0}\circ\Psi_k$, for $k=1,2$.
	
	Let $\Xi:=\Psi_1^{-1}\circ\Psi_2$, a filtration-preserving mod 0 automorphism of $\underline{\Omega}^M\times\underline{\Omega}^{I+}\times\underline{\Omega}^\can_{\cal T}$. From $\omega^M\circ\Psi_1=\omega^M=\omega^M\circ\Psi_2$ and $\omega^{I+,0}\circ\Psi_1=\omega^{I+}=\omega^{I+,0}\circ\Psi_2$, composing with $\Psi_1^{-1}$ gives
	\[
	\omega^M\circ\Xi=\omega^M,
	\qquad
	\omega^{I+}\circ\Xi=\omega^{I+}.
	\]
	Thus $\Xi$ fixes the $M$ and $I+$ coordinates of the domain pointwise, and, being a bijection, restricts on each fibre $\{(\omega^M,\omega^{I+})\}\times\Omega^\can_{\cal T}$ to a filtration-preserving automorphism of that fibre: that is,
	\[
	\Xi(\omega^M,\omega^{I+},\omega^\can)=(\omega^M,\omega^{I+},\phi_{\omega^M,\omega^{I+}}(\omega^\can))
	\]
	for a measurable family of filtration-preserving automorphisms $\phi_{\omega^M,\omega^{I+}}$ of $\underline{\Omega}^\can_{\cal T}$.
	
	Transporting through $\Psi_1$, the automorphism $\Phi:=\Psi_2\circ\Psi_1^{-1}=\Psi_1\circ\Xi\circ\Psi_1^{-1}$ of $\underline{\Omega}^*$ fixes $\omega^M$ and $\omega^{I+,0}$, and acts only on the remaining factors $\underline{\Omega}^\can_{\cal T}\times\prod_{k\ge1}\underline{\Omega}^{I+,k}$ of $\underline{\Omega}^*$. Equivalently, for every $i\in{\cal I}$, the conjugate $\Phi_i:=\rho((0,i))\circ\Phi\circ\rho((0,i))$ fixes $\omega^M$ and $\omega^{I+,i}$, and acts only on the complementary factors $\underline{\Omega}^\can_{\cal T}\times\prod_{k\ne i}\underline{\Omega}^{I+,k}$.
	
	Since ${\cal T}$ is finite and each $\underline{\Omega}^{I+,k}$ is standard, this complementary space is itself a standard filtered probability space, and is therefore stably equivalent to $\underline{\Omega}^\can_{\cal T}$ by Lemma~\ref{lemma:stableEquivalence}. Composing $\Phi_i$'s restriction to this space with the stable-equivalence isomorphism identifies it, fibrewise over $(\omega^M,\omega^{I+,i})$, with a measurable family of filtration-preserving mod 0 automorphisms of $\underline{\Omega}^\can_{\cal T}$ indexed by $(\omega^M,\omega^{I+,i})$; since $\Phi_i$ is filtration-preserving, the time-$t$ coordinates of this family and of its fibrewise inverses depend on $(\omega^M,\omega^{I+,i})$ through time-$t$ information alone. By Lemma~\ref{lemma:invarianceFibrewise}, ${\cal J}_\tau$ is unaffected by precomposing its argument with $\Phi_i$.
	
	Since $\Psi_2\circ\pi^{\perp,N}=\Phi\circ\Psi_1\circ\pi^{\perp,N}$ and $\rho((0,i))\circ\Phi=\Phi_i\circ\rho((0,i))$,
	\[
	X^i(\gamma)\circ\rho((0,i))\circ\Psi_2\circ\pi^{\perp,N}
	=
	\bigl(X^i(\gamma)\circ\Phi_i\bigr)\circ\rho((0,i))\circ\Psi_1\circ\pi^{\perp,N},
	\]
	and by the previous paragraph ${\cal J}_\tau$ takes the same value on this as on $X^i(\gamma)\circ\rho((0,i))\circ\Psi_1\circ\pi^{\perp,N}$, conditionally on $i^N=i$. Since this holds for every $i\in{\cal I}$, and $i^N$ is independent of $(\Omega^M,\Omega^{I+},\Omega^\can_{\cal T})$, the value ${\cal V}^\gamma(j)$ computed using $\Psi_2$ agrees with that computed using $\Psi_1$.
\end{proof}

\subsection{Tontine strategies - proofs}
\label{appendix:tontineStrategies}

\begin{proof}[Well-posedness of tontine strategies, Lemma~\ref{lemma:tontineWellPosed}]
	We argue by induction on $s=0,1,\dots,T-1$, in the order in which the recursion proceeds. The inductive hypothesis at the start of step $s$ is that $\beta_s=(\beta^i_s)_{i\in{\cal I}}$ has already been constructed, that $\beta^i_s\ge0$, that $i\sim_x j\Rightarrow\beta^i_s=\beta^j_s$, and that $\gamma^i_{u,t}$ for $u<s$ satisfies (i)--(iii). This holds vacuously at $s=0$, since $\beta_0=0$.
	
	{\em Existence and measurable selection.} By the inductive hypothesis and the standing hypothesis $\eta\ge0$, each $\tilde\eta^{(k)}_s$ is a nonnegative element of $L^1({\mathbb R}_{\ge0},({\cal F}^{\infty+}_t)_{t\in{\cal T}},{\mathbb Q}^{\infty+})$ vanishing before $\tau^+$, i.e.\ an admissible payment stream in the sense of Definition~\ref{def:compactIndependentEta}. Since $\Prob^+$ is compact independent of $\eta$, each $\Prob^{+,k}_s$ is compact (Definition~\ref{def:compactProblem}), and so, by Theorem~\ref{thm:regularMaximisers}, admits regular maximisers. In particular, writing $y=\Psi^N_*(x)$ (Lemma~\ref{lemma:disintegrationNatural}), an arg max $X^k(\beta_s)$ for $(\Prob^{+,k}_s;s,y,0,0)$ exists, and clause~(iii) of Definition~\ref{def:compactProblem} allows it to be chosen as a measurable function of the conditioning state. This is exactly the datum used to define $\gamma_{s,s}$, $\tilde\gamma_{s,s+1}$, $C_{s+1}$, $\gamma_{s,s+1}$ and $\beta_{s+1}$, so these are well defined; the recursion terminates after finitely many steps since ${\cal T}$ is finite.
	
	{\em Nonnegativity and the induction step.} Both $X^k_s(\beta_s)$ and $X^k_t(\beta_s)$ ($t>s$) take values in ${\mathbb R}_{\ge0}$, since the classical continuation problem underlying $\Prob^{+,k}_s$ is posed on the corresponding space of nonnegative consumption processes. Hence $\gamma^i_{s,s}$ and $\tilde\gamma^i_{s,s+1}$ are nonnegative, so $C^i_{s+1}\ge0$, and $\gamma^i_{s,s+1}$, being a nonnegative combination of the nonnegative quantities $\mathbb E_{{\mathbb Q}^*_x}(\tilde\gamma^j_{s,s+1}\mid{\cal F}^*_{s+1})$, is itself nonnegative. Thus $\beta^i_{s+1}=\gamma^i_{s,s+1}(\beta_s\mid x)\ge0$, maintaining the first part of the inductive hypothesis.
	
	For both variants, the insurance pool $P$ containing any given $i$ is a union of complete state-pools ${\cal I}^k_x$ ($P=P^k_x={\cal I}^k_x$ for the restricted strategy, $P={\cal I}=\bigcup_{k\in{\cal K}_x}{\cal I}^k_x$ for the broad strategy): no state-pool is ever split between two insurance pools, so $i\sim_xj$ implies $i,j$ lie in the same state-pool and hence the same insurance pool $P$.
	
	Since $i,j$ lie in the same state-pool, $X^k(\beta_s)$ is the same arg max for both, so $\gamma^i_{s,s}=\gamma^j_{s,s}$ and $\tilde\gamma^i_{s,s+1}=h^k_{s+1}(\omega^M,\omega^{I,i})$, $\tilde\gamma^j_{s,s+1}=h^k_{s+1}(\omega^M,\omega^{I,j})$ with a common function $h^k_{s+1}$; on the event $\{\omega^{I,i}=\omega^{I,j}\}$ these agree, and their conditional distributions given the market agree, so $C^i_{s+1}=C^j_{s+1}$. In the definition of $\gamma^i_{s,s+1}$ and $\gamma^j_{s,s+1}$, the correction ratio depends only on the pool $P$, which is the same set for $i$ and $j$, while the leading factors agree on $\{\omega^{I,i}=\omega^{I,j}\}$; hence $(\gamma^i_{s,s+1}-\gamma^j_{s,s+1})\indicator_{\{\omega^{I,i}=\omega^{I,j}\}}=0$. This holds for both the restricted and the broad strategy.
	
	At the next step of the recursion the pools are formed by the states at $s+1$, and for $i\sim_{x'}j$ at that step the full idiosyncratic paths agree almost surely under ${\mathbb Q}^*_{x'}$, so $\beta^i_{s+1}=\gamma^i_{s,s+1}=\gamma^j_{s,s+1}=\beta^j_{s+1}$ on the relevant fibre, maintaining the second part of the inductive hypothesis for both variants; and the display above is precisely \eqref{eq:symmetricContract} for the pair $(s,s+1)$, while \eqref{eq:symmetricContract} for the pair $(s,s)$ is immediate from $\gamma^i_{s,s}=\gamma^j_{s,s}$. This proves (ii), for both $\gamma^R$ and $\gamma^B$.
	
	{\em Myopic.} By construction $\gamma_{u,t}$ is only ever assigned a (possibly zero) value for $t\in\{u,u+1\}$ at each step $u$ of the recursion, and Case~1 sets $\gamma_{u,t}=0$ for all $t\ge u$ once $\tau_d\le u$; this is precisely \eqref{eq:myopicContract} together with \eqref{eq:nonParticipationI}--\eqref{eq:noConsumptionNonParticipationI}, proving (i).
	
	{\em Admissibility.} Measurability: $\gamma^i_{s,s}$ is, by construction, a function of the conditioning state $x\in\Omega^*_s$ alone, hence ${\cal F}^*_s$-measurable, which is condition~(i) of Definition~\ref{def:admissibleIdiosyncratic} for $u=w=s$ (recall ${\cal F}^M_{s,s}=\iota^{\Omega^*}({\cal F}^M_s\times{\cal F}^\can_s\times\prod_i{\cal F}^{I+}_s)=\iota^{\Omega^*}({\cal F}^*_s)$ up to completion). And $\gamma^i_{s,s+1}$ is, by construction, ${\cal F}^M_{s+1}$-measurable (it is built entirely from quantities already conditioned on ${\cal F}^M_{s+1}$), which is condition~(i) for $u=s,w=s+1$; since ${\cal F}^M_{u,w}$ is increasing in $w$ and $\gamma_{u,t}=0$ for $t>u+1$, condition~(i) holds for every $u\le w$.
	
	Funding conditions: fix a pool $k$ (for the restricted strategy) or all of ${\cal I}$ (for the broad strategy) and write $P$ for the relevant insurance pool at step $s$. Summing the definition of $\gamma^i_{s,s+1}$ over $i\in P$,
	\[
	\sum_{i\in P}\gamma^i_{s,s+1}(\beta_s\mid x)
	=
	\sum_{j\in P}C^j_{s+1}(\beta_s\mid x),
	\]
	so the sharing rule distributes exactly the pool's market-priced resources: since
	${\mathbb E}(C^j_{s+1}\mid{\cal F}^*_s)
	={\mathbb E}(\tilde\gamma^j_{s,s+1}\mid{\cal F}^*_s)$
	by the tower property, the redistribution step does not affect the
	value, given any information at time $s$ or coarser, of the legs
	written at time $s$. It therefore suffices to check the funding
	condition~\eqref{eq:budgetStar} with each $\gamma^i_{s,s+1}$
	replaced by $\tilde\gamma^i_{s,s+1}$, where it reduces, pool by
	pool, to the funding constraints of the classical continuation
	problem solved by $X^k(\beta_s)$ from budget $0$ with income
	$\tilde\eta^{(k)}_s$: the payments carried in from the previous
	step, $\gamma^i_{s-1,s}=\beta^i_s$, together with the contributions
	$\eta^i_s$, are exactly the resources pooled into
	$\tilde\eta^{(k)}_s$, and the arg max property includes
	feasibility. This proves (iii).
	
	{\em Homogeneous initial state.} Suppose $x$ is homogeneous, so that ${\cal H}^{N,x}_0(0)$ (Definition~\ref{def:homogeneousFinite}) is meaningful. By (i) and (ii), $\gamma^R,\gamma^B\in{\cal C}^M(N,0)\cap{\cal C}^S(N,0)$. Comparing \eqref{eq:defC} (with $b=0$) and \eqref{eq:defCHomogeneous} with $\beta_t:=0$ for every $t$ shows $c_{u,w}=c^h_{u,w}$ identically (there is no separate exogenous injection, both quantities being built from the same $\eta^i_t$ and $\gamma^i_{u,t}$), so the funding conditions~\eqref{eq:budgetStarHomogeneous} are exactly \eqref{eq:budgetStar}, already verified in the proof of (iii); and condition~\eqref{eq:MImeasurabilityHomogeneous} holds because $\gamma^i_{u,t}$ was shown in the proof of (iii) to be ${\cal F}^M_{u,t}$-measurable, which is stronger. Hence $\gamma^R,\gamma^B\in{\cal H}^{N,x}_0(0)$, proving (iv). (Homogeneity of $x$ need not persist at later times, since agents' idiosyncratic states evolve independently; this plays no role above, since (i)--(iii) hold at every state along the recursion regardless of homogeneity, and \eqref{eq:budgetStarHomogeneous} constrains the strategy only through its value at $s=0$.)
\end{proof}

\begin{proof}[Tontine convergence, Theorem~\ref{thm:tontineConvergence}]
	We write the start as $(0,x_0)$; a general homogeneous start is
	identical. Since $\Omega^I$ is finite, write $R_s$ for the finite
	set of idiosyncratic paths reachable at time $s$ with positive
	probability.

	\emph{Step 1: the delivered value is a perturbed value function.}
	At each step $s$, the recursion consumes the time-$s$ consumption
	of an arg max of the continuation problem whose income is the
	realised per-capita resource $\tilde\eta^{(\iota)}_{N,s}$ of the
	agent's pool, and the realised-proportional sharing rule carries
	into time $s+1$ precisely the market value of that arg max's
	remaining consumption, multiplied by the ratio $R_{N,s+1}$. Define
	the adapted perturbation
	$\delta_{N,s+1}:=(R_{N,s+1}-1)\,C^0_{s+1}$, the excess of the
	carried claim over the arg max's own pot. The delivered stream
	$X^0(\gamma_N)$ is then, by construction, the realised consumption
	of a composition of measurable stepwise arg maxes for the problem
	$\Prob^+\bigl(\eta^++\delta_N\bigr)$, and such a composition is an
	optimiser of that problem, by Step 1 of the proof of
	Proposition~\ref{prop:acceptableFromOptimisers}. Hence
	\[
	{\cal J}_0\bigl(X^0(\gamma_N)\mid x_0\bigr)
	=
	V_0\bigl(\Prob^+(\eta^++\delta_N);x_0^+,0,0\bigr).
	\]

	\emph{Step 2: the perturbations vanish.} We argue by forward
	induction on $s$ that, along a subsequence, the step-$s$ incomes
	$\tilde\eta^{(\iota)}_{N,s}$ converge in $L^1$, for each of the
	finitely many $\iota\in R_s$, and $|{\cal I}^\iota_{N,s}|\to\infty$
	almost surely. The base case is the initialisation
	$\beta_0\equiv0$ together with homogeneity. For the inductive
	step: the population fractions converge by the strong law of large
	numbers applied to the i.i.d.\ idiosyncratic paths; the step-$s$
	arg max streams converge in $L^1$ to an optimiser of the limiting
	problem, by the stability clause of condition (ii) of compactness
	of $\Prob^+$ in its income-perturbation form, the incomes
	converging by the inductive hypothesis and the limiting values
	exceeding $-\infty$ by no satiation; consequently the realised
	conditional-value functions $h^\iota_{s+1}$, being
	$L^1$-contractions of the arg max streams, converge in $L^1$. The
	sharing ratio
	\[
	R_{N,s+1}
	=
	\frac{\sum_{j\in P}C^j_{s+1}}
	{\sum_{j\in P}{\mathbb E}_{{\mathbb Q}^*_x}(\tilde\gamma^j_{s,s+1}\mid{\cal F}^*_{s+1})}
	\]
	then converges to $1$ in probability: dividing numerator and
	denominator by $|P|$, the strong law of large numbers across each
	pool (whose members' realised values are, conditionally on the
	market, i.i.d.\ evaluations of $h^\iota_{s+1}$) identifies both
	limits with the same market-conditional value. By hypothesis (b)
	and dominated convergence, $\delta_{N,s+1}\to0$ in $L^1$, so the
	incomes at step $s+1$ converge in $L^1$, closing the induction.

	\emph{Step 3: conclusion.} By Step 2, $\eta^++\delta_N\to\eta^+$
	in $L^1$ along the subsequence, so by
	Lemma~\ref{lemma:valueContinuity} in its income-perturbation form,
	\[
	{\cal J}_0\bigl(X^0(\gamma_N)\mid x_0\bigr)
	=
	V_0\bigl(\Prob^+(\eta^++\delta_N);x_0^+,0,0\bigr)
	\;\to\;
	V^+_0(x_0^+,0).
	\]
	Every subsequence of $N$ contains a further subsequence along
	which the extractions of Step 2 apply, so the convergence holds
	along the entire sequence. This proves the claim for both
	$\gamma^R_N$ and $\gamma^B_N$, the two strategies differing only
	in their pools.
\end{proof}

\begin{corollary}
	\label{cor:tontineValueLowerBound}
	Under the hypotheses of Theorem~\ref{thm:tontineConvergence},
	applied at a homogeneous start $(u,y_0)$ with
	$V^+_u(y_0^+,0)>-\infty$: for every $\epsilon>0$ there is $n$ such
	that
	\[
	V^m_u\bigl(\rho(g)_*(y_0),0\bigr)
	>
	V^+_u(y_0^+,0)-\epsilon
	\]
	for every $m\ge n$ and every relabelling $g$.
\end{corollary}
\begin{proof}
	By Lemma~\ref{lemma:tontineWellPosed}(iv), the restricted tontine
	strategy for $m$ agents started at $(u,y_0)$ lies in
	${\cal H}^{m,y_0}_u(0)$, so $V^m_u(y_0,0)$ dominates its delivered
	value, which converges to $V^+_u(y_0^+,0)$ by the theorem. The
	value of the homogeneous problem is invariant under relabelling
	the agents, by the symmetry of
	Definition~\ref{def:homogeneousFinite}.
\end{proof}

\subsection{Gauge symmetries - proofs}
\label{appendix:symmetryArgument}

\begin{proof}[Invariance under $(M,I)$-equivalence, Lemma~\ref{lem:admissibilityValueInvariance}]
	Throughout, we suppress the pullback $\pi^{\Omega^*}$, writing
	$Q^{g^{-1}(i)}\circ\rho(g^{-1})\circ\phi=(Q')^i$ for each member $Q$ of
	the collection being transformed, as in the definition of
	$(M,I)$-equivalence.
	
	(i) \emph{Non-participation.} Since $(\tau')^i,(\tau_d')^i$ are obtained
	from $\tau^{g^{-1}(i)},\tau_d^{g^{-1}(i)}$ by the same relabelling and
	gauge transform,
	$\{u<(\tau')^i\text{ or }u\ge(\tau_d')^i\}=(\rho(g^{-1})\circ\phi)^{-1}\{u<\tau^{g^{-1}(i)}\text{
		or }u\ge\tau_d^{g^{-1}(i)}\}$. On this event,
	$(\gamma')^i_u=\gamma^{g^{-1}(i)}_u\circ\rho(g^{-1})\circ\phi=0$, by
	non-participation for $\gamma$ at index $g^{-1}(i)$; the same argument
	applies to the second non-participation condition.
	
	\emph{Measurability.} Since $g$ is a bijection of ${\cal I}$,
	reindexing $i\mapsto g^{-1}(i)$ term by term in the sum defining
	$c_{u,w}$ gives
	\[
	c'_{u,w}
	:=
	\sum_{i\in{\cal I}}\Bigl((b')^i+\sum_{s\le t\le u}(\eta')^i_t-\sum_{u\le
		t\le w}(\gamma')^i_{u,t}\Bigr)
	=
	c_{u,w}\circ\rho(g^{-1})\circ\phi.
	\]
	Since $\rho(g^{-1})$ and $\phi$ are both filtration-preserving,
	$c_{u,w}\in{\cal F}^M_{u,w}$ gives
	$c'_{u,w}=c_{u,w}\circ\rho(g^{-1})\circ\phi\in{\cal F}^M_{u,w}$.
	
	\emph{Budget equation.} Since $\rho(g^{-1})$ and $\phi$ are each
	measure- and filtration-preserving,
	${\mathbb E}_{{\mathbb Q}^*}(c'_{u,w}\mid{\cal F}^*_u)
	={\mathbb E}_{{\mathbb Q}^*}(c_{u,w}\mid{\cal F}^*_u)\circ\rho(g^{-1})\circ\phi\ge0$.
	
	Together with $\gamma'\in{\cal C}(N,s)$ (Definition~\ref{def:contractSpace},
	verified by the non-participation calculation above), this verifies the
	measurability and budget conditions of
	Definition~\ref{def:admissibleIdiosyncratic} for $\gamma'$.
	
	(ii) Since $g\in\mathrm{Stab}(0)$, $g(0)=0$ and hence also
	$g^{-1}(0)=0$. Fix a realisation $k\in{\cal I}$ of $i^N$ and write
	$y:=\Psi(\pi^{\perp,N}(\omega))$ for the argument fed to
	$X^{i^N}(\cdot)\circ\rho((0,i^N))$. By the defining relation of
	$(M,I)$-equivalence applied to $\gamma$,
	$(\gamma')^k=\gamma^{g^{-1}(k)}\circ\rho(g^{-1})\circ\phi$, so on
	$\{i^N=k\}$ the value function's integrand for $\gamma'$ is
	\[
	(\gamma')^k\bigl(\rho((0,k))(y)\bigr)
	=
	\gamma^{g^{-1}(k)}\Bigl(\rho(g^{-1})\bigl(\phi(\rho((0,k))(y))\bigr)\Bigr).
	\]
	
	\emph{Step 1: removing $\phi$.} Conjugating $\phi$ past $\rho((0,k))$,
	$\phi\circ\rho((0,k))=\rho((0,k))\circ\phi_k$ where
	$\phi_k:=\rho((0,k))\circ\phi\circ\rho((0,k))$ again fixes $\omega^M$
	and every $\omega^{I,l}$, since $\phi$ does (being an $(M,I)$-map) and
	this property is preserved under conjugation by any $\rho(h)$. Using
	the group identity
	$g^{-1}\cdot(0,k)=(0,g^{-1}(k))\cdot g^{-1}$ (valid because
	$g^{-1}(0)=0$), we get
	$\rho(g^{-1})\circ\rho((0,k))=\rho((0,g^{-1}(k)))\circ\rho(g^{-1})$,
	so writing $m:=g^{-1}(k)$,
	\[
	\gamma^{g^{-1}(k)}\circ\rho(g^{-1})\circ\phi\circ\rho((0,k))
	=
	\gamma^{m}\circ\rho((0,m))\circ\rho(g^{-1})\circ\phi_k.
	\]
	By Lemma~\ref{lemma:invarianceFibrewise}, applied fibrewise over
	$(\omega^M,\omega^{I,l})_{l\in{\cal I}}$ and over the realisation $k$
	of $i^N$, the dependence at each time being through current
	information alone since $\phi_k$ is filtration-preserving,
	${\cal J}_\tau$ is unaffected by precomposing with $\phi_k$.
	Hence
	\[
	{\cal V}^{\gamma'}(j)
	=
	{\cal J}_\tau\bigl(\gamma^{m}\circ\rho((0,m))\circ\rho(g^{-1})\circ\Psi\circ\pi^{\perp,N}\mid
	j\bigr),
	\qquad m=g^{-1}(i^N).
	\]
	
	\emph{Step 2: relabelling $i^N$.} Let
	$\Psi'':=\rho(g^{-1})\circ\Psi$. Since $\rho(g^{-1})$ does not touch
	$\omega^M$, and
	$\omega^{I+,0}\circ\rho(g^{-1})=\omega^{I+,g(0)}=\omega^{I+,0}$
	because $g(0)=0$, $\Psi''$ satisfies the same defining conditions as
	$\Psi$ and is therefore itself an admissible choice of encoding in
	Lemma~\ref{lem:valueFunctionWellDefined}. Let $\Theta$ be the
	automorphism of $\underline\Omega^V$ that fixes
	$\underline\Omega^M\times\underline\Omega^{I+}\times\underline\Omega^\can_{\cal
		T}$ pointwise and acts on $\underline\Omega^N$ alone by a
	measure-preserving bijection $\theta$ with $i^N\circ\theta=g^{-1}\circ
	i^N$ (which exists since $i^N$ is uniform on ${\cal I}$ and $g^{-1}$
	is a bijection of ${\cal I}$). Since $\pi^{\perp,N}\circ\Theta=\pi^{\perp,N}$
	and $i^N(\Theta(\omega))=g^{-1}(i^N(\omega))=m$, the right-hand side
	above is exactly
	${\cal J}_\tau\bigl((X^{i^N}(\gamma)\circ\rho((0,i^N))\circ\Psi''\circ\pi^{\perp,N})\circ\Theta\mid
	j\bigr)$. Since $\pi^{\perp,N}\circ\Theta=\pi^{\perp,N}$,
	precomposition with $\Theta$ leaves the argument unchanged, and
	therefore
	\[
	{\cal V}^{\gamma'}(j)
	=
	{\cal J}_\tau\bigl(X^{i^N}(\gamma)\circ\rho((0,i^N))\circ\Psi''\circ\pi^{\perp,N}\mid
	j\bigr).
	\]
	
	\emph{Step 3: conclusion.} Since $\Psi''$ is an admissible choice of
	encoding, Lemma~\ref{lem:valueFunctionWellDefined} identifies the
	right-hand side with the same quantity computed using $\Psi$ in place
	of $\Psi''$, namely ${\cal V}^\gamma(j)$. Hence
	${\cal V}^{\gamma'}(j)={\cal V}^\gamma(j)$ for every $j\in\Omega^J$.
\end{proof}

\begin{lemma}
\label{lemma:gaugeInvariant}
Let $i\in\{0,\ldots,N-1\}$.
Given $h \in S^{N,i}$ and $\omega^{(4)}\in\Omega^{(4)}$ we have
\[
(\bar X^*)^i
(
\rho^{0,1,1,1}(h)\omega^{(4)}
)
=
(\bar X^*)^i(\omega^{(4)}).
\]
\end{lemma}
\begin{proof}
Fix $\omega:=\omega^{(4)}$ and write $g:=\hat g^{N,i}(\omega)$,
$z:=\zeta^{N,i}(\omega)$, both elements of $S^{N,i}$.

\emph{Invariance of $\hat g^{N,i}$.} Since $\hat g^N$ is invariant
under the full $S_N$-action $\rho^{0,1,1,1}$, it is in particular
invariant under the subgroup $S^{N,i}$; since $\hat
g^{N,i}={\cal S}^i\circ\hat g^N$ is $\hat g^N$ composed with a fixed
function, for every $h\in S^{N,i}$,
\[
\hat g^{N,i}(h\omega)=\hat g^{N,i}(\omega)=g.
\]

\emph{Transformation of $\zeta^{N,i}$.} Let
$y_0:=\xi^{N,i}(\pi^{N,i}(\omega))$, so that $zy_0=\omega$ by
definition of $\zeta^{N,i}$. Since $\pi^{N,i}$ is the quotient map for
the $S^{N,i}$-action, $\pi^{N,i}(h\omega)=\pi^{N,i}(\omega)$ for every
$h\in S^{N,i}$, so $\xi^{N,i}(\pi^{N,i}(h\omega))=y_0$ as well. Then
\[
h\omega=h(zy_0)=(hz)y_0,
\]
while by definition $\zeta^{N,i}(h\omega)y_0=h\omega$. Since $\Xi^{N,i}$
is a bijection, $k\mapsto ky_0$ is injective on $S^{N,i}$, so
cancelling gives
\[
\zeta^{N,i}(h\omega)=hz.
\]

\emph{The computation.} Write $d:=\pi^{N,i}(\omega)$. Applying
$(\id\times\Xi^{N,i})^{-1}$ to $(g,\omega)$ and to $(g,h\omega)$,
\[
(\id\times\Xi^{N,i})^{-1}(g,\omega)=(g,z,d),
\qquad
(\id\times\Xi^{N,i})^{-1}(g,h\omega)=(g,hz,d),
\]
using $\pi^{N,i}(h\omega)=d$ and $\zeta^{N,i}(h\omega)=hz$. Recalling
$\phi^{-1,0}(g_1,g_2,x)=(g_1g_2^{-1},g_2,x)$ from the proof of
Lemma~\ref{lemma:isomorphism}, so that
$(\phi^{-1,0})^{-1}(g_1,g_2,x)=(g_1g_2,g_2,x)$, applying
$(\phi^{-1,0})^{-1}$ gives
\[
(gz,z,d),
\qquad
(ghz,hz,d)
\]
respectively. Applying $(\id\times\Xi^{N,i})$, and using $zy_0=\omega$,
\begin{gather*}
(gz,z,d)\mapsto(gz,zy_0)=(gz,\omega),
\\
(ghz,hz,d)\mapsto(ghz,(hz)y_0)=(ghz,h(zy_0))=(ghz,h\omega).
\end{gather*}
Folding the group element back into the point before evaluating
$(\bar X^N)^i$ (the only reading of the final composition that
produces a real number) gives
\[
(\bar X^*)^i(\omega)=(\bar X^N)^i\bigl((gz)\omega\bigr),
\qquad
(\bar X^*)^i(h\omega)=(\bar X^N)^i\bigl((ghz)(h\omega)\bigr).
\]
Simplifying the second argument,
\[
(ghz)(h\omega)=(ghzh)\omega.
\]

\emph{Conclusion.} Both $gz$ and $ghzh$ lie in $S^{N,i}$, since
$g,h,z\in S^{N,i}$. By the $S^{N,i}$-invariance of $(\bar X^N)^i$
recorded above, applied directly with these two group elements,
\[
(\bar X^N)^i\bigl((gz)\omega\bigr)=(\bar X^N)^i(\omega)=(\bar
X^N)^i\bigl((ghzh)\omega\bigr).
\]
Therefore $(\bar X^*)^i(h\omega)=(\bar X^*)^i(\omega)$, as claimed.
\end{proof}

\begin{proof}[Reduction to one agent's strategy, Corollary~\ref{cor:exchangeableFund}]
	Since $b\equiv0$, for any $(g,\phi)$ the $(M,I)$-equivalent image of $b$
	is also identically $0$, since precomposing the zero function with any
	map gives the zero function. So in each part it suffices to transport
	$\gamma$ and the participation data.
	
	\emph{Part (a).} Let $\iota:\{0,\ldots,N-2\}\to\{1,\ldots,N-1\}$ be
	given by $\iota(m):=m+1$. Identifying $S_{N-1}$, as it acts on
	$\{0,\ldots,N-2\}$ in Theorem~\ref{theorem:gaugeSelection}, with the
	subgroup of $S_N$ that fixes $0$ and permutes $\{1,\ldots,N-1\}$ via
	$\iota$, identifies $S_{N-1}$ with $\mathrm{Stab}(0)$. By
	Theorem~\ref{theorem:gaugeSelection}, applied with $N$ there replaced
	by $N-1$ and its index set identified with $\{1,\ldots,N-1\}$ via
	$\iota$, there is a pair $(g,\phi)$, with $g\in
	S_{N-1}\cong\mathrm{Stab}(0)$, such that applying the identical
	formulas to any $Y\in\prod_{i=0}^{N-1}L^0({\mathbb R},({\cal
		F}^*_t)_{t\in{\cal T}})$ produces $Y^*$ measurable, for each
	$i=1,\ldots,N-1$, with respect to $\iota^{\Omega^*}({\cal F}^M_t\times{\cal F}^{\can}_t\times{\cal E}^{N-1,i}({\cal F}^{I+}_t))_{t\in{\cal T}}$, and
	exchangeable among $i=1,\ldots,N-1$
	(Remark~\ref{rem:gaugeSelectionUniform}). Applying this same
	$(g,\phi)$ to $\gamma$ (each $\gamma_u$) and to $\tau,\tau_d,\eta$ (the extension of $(M,I)$-equivalence to a finite collection,
	Definition~\ref{def:MIequivalenceGeneral}) gives
	$\gamma^*,\tau^*,\tau_d^*,\eta^*$ satisfying (ii) and (iii). Since
	$g\in\mathrm{Stab}(0)$, Lemma~\ref{lem:admissibilityValueInvariance}
	gives that $\gamma^*$ is admissible for $(0,x,0)$ with participation
	data $(\tau^{*i},\tau_d^{*i},\eta^{*i})_{i\in{\cal I}}$, and that
	${\cal V}^{\gamma^*}={\cal V}^\gamma$, proving (i).
	
	\emph{Part (b).} By Theorem~\ref{theorem:gaugeSelection}, applied
	directly with the original index set $\{0,\ldots,N-1\}$, there is a
	pair $(g^\circ,\phi^\circ)$ such that applying the identical formulas
	to any $Y\in\prod_{i=0}^{N-1}L^0({\mathbb R},({\cal F}^*_t)_{t\in{\cal
			T}})$ produces $Y^\circ$ measurable, for \emph{every}
	$i=0,\ldots,N-1$, with respect to $\iota^{\Omega^*}({\cal F}^M_t\times{\cal F}^{\can}_t\times{\cal E}^{N,i}({\cal F}^{I+}_t))_{t\in{\cal T}}$
	(Remark~\ref{rem:gaugeSelectionUniform}). Applying this same
	$(g^\circ,\phi^\circ)$ to $\gamma$ and to $\tau,\tau_d,\eta$ gives
	$\gamma^\circ,\tau^\circ,\tau_d^\circ,\eta^\circ$ with this
	measurability property for every $i$. By the argument of the proof of
	Lemma~\ref{lemma:standardFiltered}, applied to the actions of the
	stabilisers $S^{N,i}$, there are measure-preserving realization
	maps identifying the coarse filtrations
	$\iota^{\Omega^*}({\cal F}^M_t\times{\cal F}^{\can}_t\times{\cal E}^{N,i}({\cal F}^{I+}_t))_t$,
	for the various $i$, with the filtration of a single standard
	filtered probability space; transporting through these
	identifications, together with a witnessing stable-equivalence
	isomorphism (Lemma~\ref{lemma:stableEquivalence}),
	Definition~\ref{def:preferenceIsomorphism} applies
	to $\gamma^{\circ i}$ and $\gamma^{\circ j}$ (and likewise to
	$\tau^\circ,\tau_d^\circ,\eta^\circ$), giving pairwise preference
	isomorphism.
\end{proof}

\subsection{Acceptable strategies - proofs}
\label{appendix:acceptableStrategiesFinite}

\begin{proof}[Existence of acceptable strategies, Theorem~\ref{thm:existenceOptimalContinuationsFinite}]
	{\em Attainment.} Let $\gamma^n\in{\cal H}^{N,x}_s(\beta)$ be a
	maximising sequence. The constraint set is convex, every defining
	condition being an almost-sure linear equality or inequality in
	$\gamma$, and the funding conditions, telescoped, bound each of
	the finitely many legs in $L^1$. By Koml\'os' theorem applied
	coordinatewise and a diagonal argument, the Ces\`aro averages of a
	subsequence converge almost surely, coordinatewise, to some
	$\gamma^\infty$; by concavity of ${\cal J}_s(\,\cdot\mid x)$ the
	averaged strategies are again maximising. The consistency,
	non-participation, myopic, symmetric and measurability conditions
	pass to almost-sure limits, being almost-sure linear constraints
	on complete $\sigma$-algebras. The funding conditions pass to the
	limit in their telescoped, received-to-date form (charging the
	consumption taken to date and crediting only the injections and
	contributions received to date) in which every term carries a
	single sign, so that the conditional Fatou lemma applies.

	We now rebook the carried legs. Write
	$X^i_t:=\gamma^{\infty,i}_{t,t}$ for the limit consumption and
	define the fund's pot process recursively by $P_s:=0$ and
	\[
	P_{u'+1}
	:=
	{\mathbb E}_{{\mathbb Q}^*}\Bigl(
	P_{u'}+\beta_{u'}+\sum_{i\in{\cal I}}\bigl(\eta^i_{u'}-X^i_{u'}\bigr)
	\,\Big|\,
	{\cal F}^M_{u'+1,T}\Bigr),
	\qquad
	u'\in{\cal T}_{\ge s}.
	\]
	Each $P_{u'+1}$ is measurable with respect to the market
	information: the aggregates
	$\sum_i(\eta^i_{u'}-X^i_{u'})$ differ from the marketable net
	positions of $\gamma^\infty$ by aggregates of its legs, which are
	themselves differences of marketable positions, so the conditional
	expectation acts on a marketable quantity. The telescoped funding
	conditions passed to the limit state precisely that
	$P_{u'}\ge0$ almost surely for every $u'$. Define
	$\tilde\gamma^\infty$ to have the same diagonal legs
	$\tilde\gamma^{\infty,i}_{u',u'}:=X^i_{u'}$, no legs beyond one
	period, and carried legs
	\[
	\tilde\gamma^{\infty,i}_{u',u'+1}
	:=
	\frac{\indicator_{\{\tau^i\le u'+1<\tau^i_d\}}}
	{\#\{j:\tau^j\le u'+1<\tau^j_d\}}
	\;P_{u'+1},
	\]
	with the convention that the leg is zero when no agent
	participates at $u'+1$. These legs are nonnegative, myopic and
	symmetric (participation is a property of the idiosyncratic path,
	so agents with equal paths receive equal legs), they vanish
	outside the participation windows of both writer and payee, and
	they satisfy the measurability
	condition~\eqref{eq:MImeasurabilityHomogeneous}, being built from
	market-measurable pots and idiosyncratic participation
	indicators. The truncated net positions $\tilde c_{u',w}$ are
	sums of the $P$-terms and the marketable aggregates above, so the
	marketability condition holds.

	For the funding conditions of $\tilde\gamma^\infty$: at
	$u'\in{\cal T}_{>s}$, on the event that some agent participates at
	$u'$,
	\[
	{\mathbb E}_{{\mathbb Q}^*}\bigl(\tilde c_{u',T}\mid{\cal F}^*_{u'}\bigr)
	=
	P_{u'}+\beta_{u'}+\sum_i\eta^i_{u'}-\sum_iX^i_{u'}
	-{\mathbb E}\bigl(P_{u'+1}\mid{\cal F}^*_{u'}\bigr)
	=
	0,
	\]
	since every term in the definition of $P_{u'+1}$ is
	${\cal F}^*_{u'}$-measurable and the tower property applies; on
	the complementary event the outgoing leg vanishes and the
	condition reads $P_{u'}+\beta_{u'}+\sum_i\eta^i_{u'}\ge0$, which
	holds since each term is nonnegative there (the limit consumption
	vanishes outside participation). The issuance condition at $s$ is
	the same computation with $P_s=0$. Hence
	$\tilde\gamma^\infty\in{\cal H}^{N,x}_s(\beta)$, with the same
	consumption process as $\gamma^\infty$. Finally, by
	Lemma~\ref{lem:perCapitaRedistributive} the streams $X^0(\gamma^n)$ lie in a
	common redistributively feasible set of $\Prob^+$, so the Fatou property
	(condition (i) of compactness) applied along the almost surely
	convergent averages gives
	${\cal J}_s(X^0(\tilde\gamma^\infty)\mid x)\ge V^N_s(x,\beta)$;
	the reverse inequality holds since $\tilde\gamma^\infty$ is
	feasible.

	{\em Non-emptiness of ${\cal G}^{N,x}_s(0)$.} Construct $\gamma$
	by backward induction on $u=T,T-1,\ldots,s$: at each stage,
	partition ${\cal I}$ into its current state-pools and, for each
	pool, apply the first claim to that pool's own size and state,
	choosing continuations that are statewise optimal at every later
	date, exactly as in Step 1 of the proof of
	Proposition~\ref{prop:acceptableFromOptimisers}; condition (iii)
	of compactness supplies the measurable selections. The resulting
	strategy gives every pool exactly its own optimal value at every
	future date and state. Consequently, for any coalition $A$ and any
	$(u,y)$, taking $B$ to be the pool containing $A$ at $(u,y)$
	witnesses the exemption in the definition of unacceptability
	(at states where the pool value is $-\infty$, every strategy is
	trivially acceptable), so $\gamma$ is acceptable for $A$ at
	$(u,y)$; hence $\gamma\in{\cal G}^{N,x}_s(0)$.
\end{proof}

Throughout, every agent $i$ in every fund has preferences and
participation times given by (their own copy of) a single fixed problem
$\Prob=({\cal J},\tau,\tau_d,\eta)$, as in
Definition~\ref{def:compactProblem}: this is the standing convention already
used to define ${\cal J}^i,\tau^i,\tau_d^i$ throughout the finite-fund
sections.

\begin{lemma}[Exchangeable limits]
	\label{lem:exchangeableLimits}
	Let $(N_a)_{a\ge1}$ be a strictly increasing sequence in ${\mathbb Z}_{\ge2}$
	with $N_a\to\infty$. Let $E$ and $F$ be Polish spaces, let $(W^a)_a$ be
	$F$-valued random variables with tight laws, and for each $a$ let
	$(Z^{a,i})_{i=1}^{N_a-1}$ be an $E$-valued array, exchangeable jointly
	with $W^a$ (that is, the law of $(W^a,(Z^{a,i})_i)$ is invariant under
	permutations of the $i$-indices), such that the laws of
	$(W^a,Z^{a,1})$ are tight. Then there exist a subsequence $(a_k)_k$
	and an $F$-valued $W^\infty$ together with an infinite sequence
	$(Z^{\infty,i})_{i\ge1}$, exchangeable jointly with $W^\infty$ and
	defined on some standard probability space, such that for every fixed
	$K$,
	\[
	\bigl(W^{a_k},Z^{a_k,1},\ldots,Z^{a_k,K}\bigr)
	\;\longrightarrow\;
	\bigl(W^{\infty},Z^{\infty,1},\ldots,Z^{\infty,K}\bigr)
	\qquad
	\text{in law as }k\to\infty.
	\]
\end{lemma}
\begin{proof}
	By exchangeability, tightness of the law of $(W^a,Z^{a,1})$ implies
	tightness of the laws of the $K$-blocks $(W^a,Z^{a,1},\ldots,Z^{a,K})$
	for every fixed $K$. Extract, by Prokhorov's theorem and a diagonal
	argument, a single subsequence along which every $K$-block converges
	in law. The limiting laws are consistent as $K$ varies and each is
	invariant under permutations of the $Z$-coordinates, so by the
	Kolmogorov extension theorem they are the finite-dimensional laws of
	a sequence $(W^\infty,(Z^{\infty,i})_{i\ge1})$ with the stated
	exchangeability, which may be realised on a standard space since $E$
	and $F$ are Polish.
\end{proof}

\begin{lemma}[Acceptability transfers under $(M,I)$-equivalence]
	\label{lemma:acceptabilityInvariance}
	Suppose $\gamma\in{\cal G}^{N,x}_s(b)$ with participation data
	$(\tau^i,\tau_d^i,\eta^i)_{i\in{\cal I}}$, and that
	$(\tau',\tau_d',\eta',b',\gamma')$ is $(M,I)$-equivalent to
	$(\tau,\tau_d,\eta,b,\gamma)$ via $(g,\phi)$. Then $\gamma'\in{\cal
		G}^{N,x}_s(b')$ with participation data
	$(\tau'^i,\tau_d'^i,\eta'^i)_{i\in{\cal I}}$.
\end{lemma}
\begin{proof}
	The proof is a naturality argument.
	By Lemma~\ref{lem:admissibilityValueInvariance}(i), $\gamma'$ is admissible
	for $(s,x,b')$ with participation data
	$(\tau'^i,\tau_d'^i,\eta'^i)_{i\in{\cal I}}$.
	
	Fix $u\in{\cal T}_{\ge s}$, $y\in\Omega^*_u$, $A'\subseteq{\cal I}$, and set
	$A:=g^{-1}(A')$. Since $\phi$ is an $(M,I)$-map (fixing $\omega^M$ and every
	$\omega^{I,l}$) and types are carried unchanged by $(M,I)$-equivalence,
	$i\sim_y j$ for all $i,j\in A'$ if and only if $i\sim_y j$ for all $i,j\in A$.
	
	For $j\in A'$, $X^j(\gamma')=X^{g^{-1}(j)}(\gamma)\circ\rho(g^{-1})\circ\phi$;
	by the same fibrewise Invariance argument as in the proof of
	Lemma~\ref{lem:admissibilityValueInvariance}(ii) (dropping $\phi$ via
	Lemma~\ref{lemma:invarianceFibrewise}, then relabelling via $\rho(g^{-1})$),
	\[
	{\cal J}^j_u(X^j(\gamma')\mid y)={\cal J}^{g^{-1}(j)}_u(X^{g^{-1}(j)}(\gamma)\mid y).
	\]
	Similarly, the pre-trade claims transform compatibly: for
	$u\in{\cal T}_{>s}$,
	$\beta^{A',y,u}_t=\sum_{j\in A'}(\gamma')^j_{u-1,t}
	=\beta^{A,y,u}_t\circ\rho(g^{-1})\circ\phi$, and at $u=s$ the same
	identity holds with the initial budgets $b'^j=b^{g^{-1}(j)}\circ\rho(g^{-1})\circ\phi$
	in place of the carried legs. Moreover,
	and for every $g''\in S_N$ mapping $A'$ to the first $|A'|$ integers,
	$g''\circ g$ maps $A$ to the first $|A|$ integers, with
	\[
	V^{|A'|}_u\bigl(\rho(g'')_*(y),\beta^{A',y,u}\bigr)
	=
	V^{|A|}_u\bigl(\rho(g''\circ g)_*(y),\beta^{A,y,u}\bigr),
	\]
	and quantifying over every valid $g''$ for $A'$ is the same as quantifying
	over every valid $g''\circ g$ for $A$.
	
	Combining these identifications, $\gamma'$ is unacceptable for $A'$ at
	$(u,y)$ (by either case (a) or case (b) of the definition) if and only if
	$\gamma$ is unacceptable for $A$ at $(u,y)$. Since $A\mapsto g^{-1}(A)$ is a
	bijection on subsets of ${\cal I}$, $\gamma$ being acceptable for every $A$ at
	every $(u,y)$ implies $\gamma'$ is acceptable for every $A'$ at every
	$(u,y)$, i.e.\ $\gamma'\in{\cal G}^{N,x}_s(b')$.
\end{proof}

\begin{lemma}[Acceptability persists on subpopulations]
	\label{lemma:restrictionAcceptable}
	Let $\gamma\in{\cal G}^{N,x}_s(b)$, and let $P\subseteq{\cal I}$. Then, for
	every $A\subseteq P$, every $u\in{\cal T}_{\ge s}$, and every $y\in\Omega^*_u$,
	$\gamma$ is ${\mathbb Q}^*_y$-almost surely acceptable for $A$ at $(u,y)$.
\end{lemma}
\begin{proof}
	Immediate: $A\subseteq P\subseteq{\cal I}$, so this is exactly the instance
	of the defining property of ${\cal G}^{N,x}_s(b)$ (acceptability for every
	$A\subseteq{\cal I}$ at every $(u,y)$ with $u\ge s$) applied to this
	particular $A$. Since the quantities $i\sim_y j$, $\beta^{A,y,u}$, and
	$V^{|A|}_u$ appearing in the acceptability test depend only on $A$ itself
	(not on the ambient population ${\cal I}$ or $P$), this remains true whether
	$A$ is viewed as a coalition within ${\cal I}$ or within $P$.
\end{proof}

The following theorem is the engine of the large-fund limit. It is
stated at a general time $u$, rather than only at $u=0$, so that it
can be applied, via Lemma~\ref{lemma:restrictionAcceptable}, to each
cohort of agents at their entry date.

\begin{theorem}[General infinite-fund limit]
	\label{thm:generalInfiniteFundLimit}
	Suppose $\Prob$ is compact and $\Prob^+$ is compact independent of $\eta$
	and satisfies no satiation.
	Fix $u\in{\cal T}$ and a homogeneous coarse state $y_0$ at time
	$u$: a state of
	$\bigl(\omega^M,(\omega^{I,i})_{i}\bigr)$, conditioning on the
	market path and the common idiosyncratic path $\iota_0$ of the
	agents but not on the auxiliary coordinates $\omega^{\can,i}$,
	so that under ${\mathbb Q}^*_{y_0}$ the auxiliary coordinates
	remain independent and uniform. Write
	$y_0^+\in\Omega^{M\times I}_u$ for the corresponding
	state $(\omega^M_u,\iota_0)$ of a representative agent. Let
	$(N_a)_{a\ge1}$ be strictly increasing with $N_a\to\infty$, built from
	agents of a fixed underlying i.i.d.\ population as in
	Theorem~\ref{thm:tontineConvergence}, and for each $a$ let $\gamma^a\in{\cal
		G}^{N_a,y_0}_u(0)$. Suppose moreover that
	$V^+_u(y_0^+,0)>-\infty$ and that the single-agent problems are
	well posed: $V^1_t(y',0)>-\infty$ for every $t\ge u$ and
	${\mathbb Q}^*_{y_0}$-almost every state $y'$ reachable from
	$y_0$. Write
	$\gamma^{a*}$ for the symmetrisation of $\gamma^a$ provided by
	Corollary~\ref{cor:exchangeableFund}(a), under which the values
	$v^a:={\cal J}_u\bigl(X^i(\gamma^{a*})\mid y_0\bigr)$ are common to
	all $i\in\{1,\ldots,N_a-1\}$. Then:
	\begin{enumerate}[(i)]
		\item $v^a\to V^+_u(y_0^+,0)$ as
		$a\to\infty$, along the entire sequence;
		\item there is $\gamma^\infty$, acceptable for $\Prob^+$
		(Definition~\ref{def:acceptableNoIdiosyncratic}) at
		$(u,y_0^+,0,0)$, with
		${\cal J}_u(X(\gamma^\infty)\mid y_0^+)=V^+_u(y_0^+,0)$, whose
		funding conditions hold with equality
		(Theorem~\ref{thm:noMutuallyBeneficialContracts}(i)).
	\end{enumerate}
\end{theorem}
\begin{proof}
	\emph{Step 1: symmetrisation.}
	By Corollary~\ref{cor:exchangeableFund}(a), replace each $\gamma^a$ by an
	$(M,I)$-equivalent $\gamma^{a*}$ (via some $g^a\in{\rm Stab}(0)$) with
	agents $1,\ldots,N_a-1$ exchangeable; by
	Lemma~\ref{lemma:acceptabilityInvariance}, $\gamma^{a*}\in{\cal
		G}^{N_a,y_0}_u(0)$ too. Because $y_0$ is a coarse state, the
	conditioning fixes only the market path and the common
	idiosyncratic path, under which the agents' remaining coordinates
	are i.i.d.; together with the measurability and equivariance
	properties provided by Corollary~\ref{cor:exchangeableFund}(a),
	this makes the block
	$\bigl(\omega^{I,i},X^i(\gamma^{a*})\bigr)_{i=1}^{N_a-1}$
	exchangeable under ${\mathbb Q}^*_{y_0}$, jointly with the market
	path. Since ${\cal J}_u(\,\cdot\mid y_0)$ depends
	only on the joint law of its argument with the conditioning
	coordinates, the values $v^a$ are indeed common to the
	block.

	\emph{Step 2: the lower bound.} By
	Proposition~\ref{prop:noPersistentGap} below,
	$\liminf_a v^a\ge V^+_u(y_0^+,0)$.

	\emph{Step 3: extraction.} Testing acceptability of $\gamma^{a*}$
	against each singleton gives
	$v^a\ge V^1_u(y_0,\beta^{\{1\},y_0,u})\ge V^1_u(y_0,0)$, using
	monotonicity of $V^1_u$ in its budget argument and the
	nonnegativity of the pre-trade claims. The funding conditions of
	$\gamma^{a*}$, integrated and telescoped, bound
	${\mathbb E}_{{\mathbb Q}^*_{y_0}}\bigl[\sum_{t\ge u}X^1_t(\gamma^{a*})\bigr]$
	by the common value of the contributions, uniformly in $a$, by
	exchangeability; hence the joint laws of
	$\bigl(\omega^M,\omega^{I,1},(X^1_t(\gamma^{a*}))_{t\ge u}\bigr)$
	are tight. By Lemma~\ref{lem:exchangeableLimits}, applied with
	$W^a$ the market path and
	$Z^{a,i}:=\bigl(\omega^{I,i},(X^i_t(\gamma^{a*}))_{t\ge u}\bigr)$,
	there is a subsequence along which every finite block converges in
	law to an exchangeable limit
	$\bigl(\omega^M,(\omega^{I,i,\infty},X^{\infty,i})_{i\ge1}\bigr)$.

	The limit sequence is exchangeable. Writing
	$\underline\Omega^A$ for the (standard, Polish) space of values of a
	single type-and-stream pair, its law is (the coordinate process of) a measure
	${\mathbb Q}^{A,\infty}$ on $\underline\Omega^{A,\infty}$ invariant
	under $S_\infty$, in the sense of
	Definition~\ref{def:exchangeableSequence}. Let $\underline\Omega^{\rm
		mix}$ and $Y$ be its mixing space and de~Finetti mixing variable
	(Definition~\ref{def:mixingSpace}).

	By Lemma~\ref{lemma:stableEquivalence}, applied to
	$\underline\Omega^{\rm mix}$ equipped with the constant filtration
	${\cal F}^{\rm mix}_t:={\cal B}(\Omega^{\rm mix})$,
	$\underline\Omega^{\rm mix}$ is stably equivalent to $\underline\Omega^\can_{\cal
		T}$: there is a filtration-preserving mod $0$ isomorphism
	\[
	\underline\Omega^{\rm mix}\times\underline\Omega^\can_{\cal T}
	\;\cong\;
	\underline\Omega^\can_{\cal T}\times\underline\Omega^\can_{\cal T}
	\;\cong\;
	\underline\Omega^\can_{\cal T},
	\]
	the second isomorphism because $\underline\Omega^\can_{\cal
		T}\times\underline\Omega^\can_{\cal T}$ is itself atomless and standard,
	hence mod $0$ isomorphic to $\underline\Omega^\can_{\cal T}$ by Rokhlin's
	theorem. Transporting through this isomorphism identifies
	$\underline\Omega^{\infty+}=\underline\Omega^\can_{\cal
		T}\times\underline\Omega^\infty$ with $\underline\Omega^{\rm
		mix}\times\underline\Omega^\can_{\cal T}\times\underline\Omega^\infty$: that
	is, $\underline\Omega^{\infty+}$ has room to carry $Y$ alongside an
	independent copy of $\underline\Omega^\can_{\cal T}$ and the original
	$\underline\Omega^\infty=\underline\Omega^M\times\underline\Omega^I$.
	By Theorem~\ref{thm:canonicalDeFinetti}, the limit sequence is
	conditionally i.i.d.\ given $Y$. The limit law is also
	compatible with the filtration: in the prelimit, every time-$t$
	coordinate of the block is, given the time-$t$ market and
	idiosyncratic past, independent of the later market and
	idiosyncratic coordinates, and this may be written as a family
	of integral identities whose conditioning kernels are those of
	the market and idiosyncratic spaces alone; since those marginals
	are the same for every $a$, the identities pass to the weak
	limit. Sampling the compatible conditional laws progressively
	from the fresh canonical coordinates supplied by the embedding
	just constructed realises $X^{\infty,1}$, together with its
	type, on $\underline\Omega^{\infty+}$ as a stream $X^\infty$
	adapted to the ambient filtration.

	\emph{Step 4: the upper bound via the cost of delivered utility.}
	This step works with the redistributive problem of
	Definition~\ref{def:classicalContinuationProblem}: the
	redistributively feasible set $\hat{\cal X}$ defined there, its
	carried-forward form (Remark~\ref{rem:carriedForward}), and the
	corresponding redistributive representation
	(Proposition~\ref{prop:redistributiveRepresentation}).
	First, by Lemma~\ref{lem:perCapitaRedistributive} below, each
	$X^1(\gamma^{a*})$ lies in
	$\hat{\cal X}^{y_0^+}_u(\Prob^+;0,0)$, together with
	market-adapted, nonnegative pot processes, bounded in $L^1$,
	witnessing the carried-forward form of the redistributive constraints.
	Enlarging the extraction of Step 3 to include these processes,
	whose limits are realised market-adapted by the compatibility
	argument of Step 3, and realising the joint convergence almost
	surely by the Skorokhod representation theorem, the
	carried-forward inequalities pass to the limit (the pot coordinates converge almost surely by
	construction, while each cost term obeys the conditional Fatou
	inequality in the favourable direction) so
	$X^{\infty,1}\in\hat{\cal X}^{y_0^+}_u(\Prob^+;0,0)$, and the
	same holds for its tails at later dates with the limiting pots.
	In particular, by condition (iv) of compactness of $\Prob^+$,
	\[
	w_t(x'')
	:=
	{\cal J}_t\bigl(X^{\infty,1}\mid x''\bigr)
	<+\infty
	\]
	for every $t\ge u$ and almost every market-and-own-path state
	$x''$ reachable from $y_0^+$.

	\emph{Fibrewise lower bounds.} Fix $t\ge u$ and $\epsilon>0$. By
	Lemma~\ref{lemma:coalitionCounting} and exchangeability, the
	probability that agent $1$ is among the short-changed agents of
	its class is at most
	${\mathbb E}\bigl[(n_0\wedge(N_a-1))/(N_a-1)\bigr]\to0$, while
	every delivered value dominates its singleton outside option and
	hence the single-agent value, finite by hypothesis. By
	Lemma~\ref{lemma:lotteryToleranceEnriched}, applied at each state
	$x''$ over the full states refining it,
	\[
	\liminf_a
	{\cal J}_t\bigl(X^1(\gamma^{a*})\mid x''\bigr)
	\ge
	V^+_t(x'',0)-\epsilon
	\qquad
	\text{for almost every }x''.
	\]
	On the Skorokhod realisation, the Fatou property (condition (i)
	of compactness of $\Prob^+$, applicable since the tails are
	redistributively feasible) gives
	$\limsup_a{\cal J}_t(X^1(\gamma^{a*})\mid x'')\le w_t(x'')$;
	since $\epsilon$ was arbitrary,
	\[
	w_t(x'')\ge V^+_t(x'',0)
	\qquad
	\text{for almost every }x''\text{, for every }t\ge u.
	\tag{$*$}
	\]

	\emph{The chain.} For $t\ge u$ and almost every $x''$ define
	\[
	\kappa_t(x'')
	:=
	\inf\bigl\{b\in{\mathbb R}
	\;:\;
	{\cal X}^{x''}_t(\Prob^+;b,0)\neq\emptyset
	\text{ and }
	V^+_t(x'',b)\geq w_t(x'')
	\bigr\},
	\]
	the statewise cost of the delivered continuation, exactly as in
	Lemma~\ref{lemma:deliveredValueBound}. As there, $\kappa_t$ is
	well defined and measurable: the supremum of
	$V^+_t(x'',\cdot)$ equals the supremum of the utilities of all
	integrable streams and is not attained, while $w_t(x'')<+\infty$;
	and $\kappa_t=0$ at exited states. By ($*$) and strict
	monotonicity of $V^+_t(x'',\cdot)$, $\kappa_t\ge0$ almost
	everywhere. By Lemma~\ref{lemma:verification}, applied under
	${\mathbb Q}^{\infty+}_{y_0^+}$ with $W:=X^{\infty,1}$ and
	$B:=\kappa_{t+1}$ together with the definition of $\kappa_t$ and
	the right-continuity of $V^+_t(x'',\cdot)$,
	\begin{gather*}
	\kappa_t
	\le
	X^{\infty,1}_t
	+{\mathbb E}\bigl(\kappa_{t+1}\,\big|\,{\cal F}^{\infty+}_t\bigr)
	-\eta^+_t
	\quad
	\text{at participating states,}
	\\
	\kappa_T\le X^{\infty,1}_T-\eta^+_T,
	\end{gather*}
	with both sides vanishing at exited states. Dividing the
	telescoped funding conditions of $\gamma^{a*}$ by $N_a$ and
	passing to the limit using Fatou's lemma on the consumption side and
	the strong law of large numbers for the i.i.d.\ contributions 
	gives the conservation bound
	\[
	{\mathbb E}_{{\mathbb Q}^{\infty+}_{y_0^+}}
	\Bigl(\sum_{t\ge u}X^{\infty,1}_t\Bigr)
	\le
	{\mathbb E}_{{\mathbb Q}^{\infty+}_{y_0^+}}
	\Bigl(\sum_{t\ge u}\eta^+_t\Bigr).
	\]
	Set
	$\beta^\circ:=X^{\infty,1}_u
	+{\mathbb E}(\kappa_{u+1}\mid{\cal F}^{\infty+}_u)-\eta^+_u$.
	Integrating the chain and telescoping from $T$ down to $u+1$
	bounds ${\mathbb E}(\kappa_{u+1})$ by
	${\mathbb E}\bigl(\sum_{t\ge u+1}(X^{\infty,1}_t-\eta^+_t)\bigr)$,
	so ${\mathbb E}(\beta^\circ)\le0$ by conservation, while the
	chain at $t=u$ gives $\kappa_u\le\beta^\circ$ pointwise with
	$\kappa_u\ge0$. Hence
	$0\le{\mathbb E}(\kappa_u)\le{\mathbb E}(\beta^\circ)\le0$, so
	$\beta^\circ=0$ almost surely, and
	Lemma~\ref{lemma:verification} at the date $u$ gives
	\[
	{\cal J}_u\bigl(X^{\infty,1}\mid y_0\bigr)
	\le
	V^+_u\bigl(y_0^+,\beta^\circ\bigr)
	=
	V^+_u(y_0^+,0).
	\]
	On the same Skorokhod realisation, the Fatou property gives
	$\limsup_av^a\le{\cal J}_u(X^{\infty,1}\mid y_0)$ almost surely.
	Combining with Step 2,
	\begin{gather*}
	v^a\to V^+_u(y_0^+,0)
	\quad\text{along the subsequence,}
	\\
	{\cal J}_u\bigl(X^{\infty,1}\mid y_0\bigr)=V^+_u(y_0^+,0)
	\quad\text{almost surely.}
	\end{gather*}

	\emph{Step 5: identification and conclusion.}
	Every subsequence of $(\gamma^{a*})$ contains, by Steps 3 and 4
	applied along it, a further subsequence along which
	$v^a\to V^+_u(y_0^+,0)$. The limit does not depend on the
	subsequence, so the entire sequence converges:
	$v^a\to V^+_u(y_0^+,0)$, which is (i).

	For (ii), $\Prob^+$ is compact and satisfies no satiation, so it
	admits regular maximisers by Theorem~\ref{thm:regularMaximisers}.
	Since $V^+_u(y_0^+,0)>-\infty$,
	Proposition~\ref{prop:acceptableFromOptimisers}, applied to
	$\Prob^+$ at $(u,y_0^+,0,0)$, provides a stream $X^\circ$,
	optimal at $y_0^+$ and statewise optimal at every later date,
	such that
	\[
	\gamma^\infty:=\Gamma(X^\circ)\in{\cal G}^{y_0^+}_u(0,0),
	\qquad
	{\cal J}_u\bigl(X(\gamma^\infty)\mid y_0^+\bigr)
	=
	V^+_u(y_0^+,0).
	\]
	The funding conditions of $\gamma^\infty$ hold with equality by
	Theorem~\ref{thm:noMutuallyBeneficialContracts}(i), the agents
	being within their participation window at the homogeneous state
	$y_0$. This proves (ii).
\end{proof}

\begin{lemma}[Per-capita redistributive feasibility]
	\label{lem:perCapitaRedistributive}
	Let $y_0\in\Omega^*_u$ be homogeneous, let
	$\gamma\in{\cal A}^{N,y_0}_u(0)$ have
	$(X^i(\gamma))_{i=1}^{N-1}$ exchangeable jointly with
	$(\omega^M,\omega^{I,i})$, and suppose $\eta^i$ depends on $i$ only
	through $\omega^{I,i}$. Then
	$X^1(\gamma)\in\hat{\cal X}^{y_0^+}_u(\Prob^+;0,0)$, and the
	per-capita market-conditional values of the carried legs of
	$\gamma$ provide a nonnegative, market-adapted pot process
	witnessing the carried-forward form of the redistributive constraints, with
	$L^1$ norms bounded by the total value of the contributions.
\end{lemma}
\begin{proof}
	For $u'\in{\cal T}_{>u}$ set
	$\hat P_{u'}:=N^{-1}\sum_{i}{\mathbb E}_{{\mathbb Q}^*_{y_0}}
	\bigl(\sum_{t\ge u'}\gamma^i_{u'-1,t}\,\big|\,{\cal F}^M_{u',T}\bigr)$,
	which is nonnegative and, by the marketability condition of
	Definition~\ref{def:admissibleIdiosyncratic}, measurable with
	respect to the market information. Conditioning the funding
	conditions~\eqref{eq:budgetStar} on the market gives
	\[
	N^{-1}\sum_i m_{u'}\bigl(X^i_{u'}(\gamma)\bigr)
	+\hat P_{u'+1}\text{-terms}
	\le
	\hat P_{u'}
	+N^{-1}\sum_im_{u'}(\eta^i_{u'}),
	\]
	in the carried-forward form, where $m_{u'}$ denotes the
	market-conditional value under ${\mathbb Q}^*_{y_0}$. By the joint
	exchangeability, $m_{u'}(X^i_{u'}(\gamma))$ and
	$m_{u'}(\eta^i_{u'})$ do not depend on $i$ and coincide with the
	corresponding quantities for the representative agent, computed
	for $\Prob^+$ through the identifications of
	Definition~\ref{def:probPlus}. Hence the representative's stream
	satisfies the carried-forward redistributive constraints with pot process
	$(\hat P_{u'})$, that is,
	$X^1(\gamma)\in\hat{\cal X}^{y_0^+}_u(\Prob^+;0,0)$; the $L^1$
	bound follows by integrating the telescoped constraints.
\end{proof}

\begin{lemma}[At most finitely many agents fall short]
	\label{lemma:coalitionCounting}
	In the setting of Theorem~\ref{thm:generalInfiniteFundLimit}, fix
	$t\in{\cal T}_{\ge u}$ and $\epsilon>0$. For a full state
	$y'\in\Omega^*_t$ refining $y_0$, write
	$x''(j,y')\in\Omega^{M\times I}_t$ for the market-and-own-path
	state of agent $j$, and
	$\xi^j(y'):={\cal J}^j_t\bigl(X^j(\gamma^{a*})\mid y'\bigr)$
	for the delivered values (see \eqref{eq:frozenEnrichedValue}). Then, for
	${\mathbb Q}^*_{y_0}$-almost every $y'$, within each class of
	agents sharing an idiosyncratic path at $y'$, the number of block
	agents $j$ with
	\[
	\xi^j(y')<V^+_t\bigl(x''(j,y'),0\bigr)-\epsilon
	\]
	is smaller than the threshold $n_0$ of
	Corollary~\ref{cor:tontineValueLowerBound} at the class's state.
\end{lemma}
\begin{proof}
	The value of the homogeneous problem for finite funds is
	invariant under relabelling the agents, by the symmetry of
	Definition~\ref{def:homogeneousFinite}, and is monotone in the
	budget process, since increasing $\beta$ enlarges the constraint
	set; moreover the pre-trade claims are nonnegative.

	Suppose the count reaches $n:=n_0$ at a set of full states of
	positive probability, and let $A$ consist of $n$ such agents
	within a fixed class; the members of $A$ share their
	idiosyncratic path, so $A$ is admissible as a coalition whatever
	their auxiliary coordinates, and $x'':=x''(j,y')$ is common to
	the class. By monotonicity in the resources and
	Corollary~\ref{cor:tontineValueLowerBound},
	\[
	V^{|A|}_t\bigl(\rho(g)_*(y'),\beta^{A,y',t}\bigr)
	\ge
	V^{|A|}_t\bigl(\rho(g)_*(y'),0\bigr)
	>
	V^+_t(x'',0)-\epsilon
	>
	\xi^j(y')
	\]
	for every $j\in A$ and every relevant $g$: clause (a)(ii) of
	unacceptability holds with strict inequality for every member.
	An exempting common-state set $B\supseteq A$ would satisfy
	$\xi^j(y')=V^{|B|}_t(\rho(g')_*(y'),\beta^{B,y',t})$ for every
	$j\in B$; but $|B|\ge n$, so the right-hand side exceeds
	$V^+_t(x'',0)-\epsilon$, contradicting the defining inequality
	of the members of $A\subseteq B$. Hence $\gamma^{a*}$ is
	unacceptable for $A$ at $(t,y')$ on a set of positive
	probability, contradicting
	$\gamma^{a*}\in{\cal G}^{N_a,y_0}_u(0)$.
\end{proof}

\begin{proposition}[No persistent gap]
	\label{prop:noPersistentGap}
	In the setting of Theorem~\ref{thm:generalInfiniteFundLimit},
	\[
	\liminf_{a\to\infty} v^a\ge V^+_u(y_0^+,0).
	\]
\end{proposition}
\begin{proof}
	At $t=u$ the block agents form a single class and, $y_0$ being
	homogeneous, $x''(j,y')=y_0^+$ for every block agent $j$ and
	almost every full state $y'$ refining $y_0$. Fix $\epsilon>0$.
	By Lemma~\ref{lemma:coalitionCounting} and exchangeability, the
	probability that
	$\xi^1(y')<V^+_u(y_0^+,0)-\epsilon$ is at most
	${\mathbb E}\bigl[(n_0\wedge(N_a-1))/(N_a-1)\bigr]\to0$, while
	every delivered value dominates its singleton outside option and
	hence the single-agent value $V^1_u(y',0)>-\infty$. By
	Lemma~\ref{lemma:lotteryToleranceEnriched}, applied at $y_0$ over
	its full refinements,
	$\liminf_av^a\ge V^+_u(y_0^+,0)-\epsilon$. Since $\epsilon$ was
	arbitrary, the claim follows.
\end{proof}

\begin{proof}[Infinite fund limit, Theorem~\ref{thm:infiniteFundLimit}]
	For each $u\in{\cal T}$, partition $\{i\in{\cal I}_a:\tau^i=u\}$ (agents
	entering at time $u$) further into pools of agents sharing a common
	type at $u$ (finitely many, since $\Omega^I$ is finite). Agents
	entering at time $u$ have, by the non-participation condition, no
	contractual history before $u$, i.e.\ they start fresh with budget $0$; and
	by Lemma~\ref{lemma:restrictionAcceptable}, every sub-coalition of any such
	pool is already tested by $\gamma^{N_a}$'s own acceptability (which holds
	for every $x\in\Omega^*_0$, hence in particular at the states actually
	realised). So the argument of Theorem~\ref{thm:generalInfiniteFundLimit}
	applies to each such pool (of size $\to\infty$ by the strong law of large
	numbers), with
	$u$ in place of the theorem's general time and the pool's entry state in
	place of $y_0$: the realised value of a member of the pool converges to
	$V^+_u(y_0^+,0)$.
	
	Now ${\cal V}^{\gamma^{N_a}}(j)$ is built by selecting an agent
	uniformly from ${\cal I}_a$ and evaluating their experience relative to
	$\Psi$. Conditionally on $j\in\Omega^J$ the selected
	agent lies, with probability $\to1$, in a pool of size $\to\infty$
	whose members' values converge as above (agent $0$ is selected with
	probability $1/N_a\to0$ and so contributes nothing in the limit). Since
	$V^+_u$ does not depend on the auxiliary canonical coordinate of its state
	argument (the induced preferences of Definition~\ref{def:inducedPreferences}
	being indifferent to it), $V^+_u(y_0^+,0)$ depends on $j$ alone, and equals
	${\cal V}^\infty(j)$ by Definition~\ref{def:valueFunctionInfinite} together
	with Lemma~\ref{lem:valueFunctionInfiniteWellDefined}. Hence, for
	${\mathbb Q}^J$-almost every $j$, ${\cal V}^{\gamma^{N_a}}(j)\to{\cal V}^\infty(j)$.
\end{proof}

\section{Preliminaries - proofs}
\label{appendix:finiteFiltrationClassification}

\begin{lemma}
\label{lemma:pointIsomorphism}
Let $(\Omega^1,{\cal F}^1,{\mathbb P}^1)$ and
$(\Omega^2,{\cal F}^2,{\mathbb P}^2)$ be standard probability spaces
and let $f:\Omega^1\to\Omega^2$ be a measure-preserving map such that
$f^{-1}{\cal F}^2={\cal F}^1$ modulo null sets. Then $f$ is a mod $0$
isomorphism.
\end{lemma}
\begin{proof}
The assignment $A\mapsto f^{-1}A$ is an isomorphism between the
measure algebras of the two spaces, so by von Neumann's theorem (see
\cite{rokhlin}) it is induced by a mod $0$ isomorphism
$g:\Omega^1\to\Omega^2$. Comparing $f$ and $g$ on a countable
point-separating family of sets in ${\cal F}^2$, available since the
space is standard, gives $f=g$ almost surely.
\end{proof}

The next lemma is the consequence of Rokhlin's theory of measurable
partitions \cite{rokhlin} that we use to split off one coordinate of a
filtration at a time. We record the construction and not merely the
existence statement, because the arguments below need it to depend
measurably on the data.

\begin{lemma}[Independent complement]
\label{lemma:independentComplement}
Let $(\Omega,{\cal G},{\mathbb Q})$ be a standard probability space and
let ${\cal H}\subseteq{\cal G}$ be a complete sub-$\sigma$-algebra such
that the conditional measures obtained by disintegrating ${\mathbb Q}$
over the ${\cal H}$-state are almost surely atomless. Let $h$ be the
coordinate map of a standard Borel realization of ${\cal H}$, let $Z$
be a real random variable generating ${\cal G}$ modulo null sets, let
$F(z\mid\cdot)$ be the conditional distribution function of $Z$ given
the ${\cal H}$-state, chosen jointly measurable, and write
$F^{-1}(v\mid\cdot):=\inf\{z:F(z\mid\cdot)\geq v\}$. Then
\[
V:=F(Z\mid h)
\]
is ${\cal G}$-measurable, uniformly distributed and independent of
${\cal H}$, and $Z=F^{-1}(V\mid h)$ almost surely. In particular
${\cal G}=\overline{\sigma({\cal H},V)}$ modulo null sets.
\end{lemma}
\begin{proof}
Since $(\Omega,{\cal G},{\mathbb Q})$ is standard, $h$ and $Z$ exist,
and the conditional distribution function may be chosen jointly
measurable by disintegration. The conditional law of $Z$ given the
${\cal H}$-state is almost surely atomless, an atom being carried,
$Z$ generating ${\cal G}$ modulo null sets, to an atom of the
conditional measure. The claims are then the conditional probability
integral transform, as in the proof of the transfer theorem
\cite[Theorem~6.10]{kallenbergFMP}.
\end{proof}

\begin{lemma}
\label{lemma:finiteFiltrationClassification}
Let
\[
(\Omega,{\cal F},({\cal F}_t)_{t\in{\cal T}},\mathbb Q)
\]
be a filtered probability space, where ${\cal T}=\{0,\ldots,T\}$ and
each ${\cal F}_t$ is complete.
Assume that

\begin{enumerate}[(i)]
\item $(\Omega,{\cal F},\mathbb Q)$ is an atomless standard probability
space and ${\cal F}_T={\cal F}$ modulo null sets;
\item $(\Omega,{\cal F}_0,\mathbb Q)$ is atomless;
\item for each $t\in\{1,\ldots,T\}$, the conditional measures obtained
by disintegrating $(\Omega,{\cal F}_t,\mathbb Q)$ over the
time-$(t-1)$ state are atomless for almost every value of that state.
\end{enumerate}

Then
\[
(\Omega,{\cal F},({\cal F}_t)_{t\in{\cal T}},\mathbb Q)
\]
is mod $0$ isomorphic, as a filtered probability space, to
$\underline{\Omega}^{\can}_{\cal T}$.
\end{lemma}
In hypothesis (iii) we disintegrate as in Section~\ref{sec:notation}:
choosing standard Borel realizations
$\omega_s:\Omega\to\Omega_s$ of $(\Omega,{\cal F}_s,{\mathbb Q})$ for
$s=t-1,t$, we disintegrate the joint law of $(\omega_{t-1},\omega_t)$
over its first component. The condition does not depend on the choice
of realizations.

\begin{proof}[Proof of Lemma~\ref{lemma:finiteFiltrationClassification}]
We construct by induction on $t$ uniformly distributed random
variables $V^t$ on $\Omega$ such that $V^t$ is ${\cal F}_t$-measurable,
the family $(V^r)_{r\leq t}$ is independent, and
\begin{equation}
{\cal F}_t=\overline{\sigma\bigl(V^0,\ldots,V^t\bigr)}
\label{eq:classificationCoordinates}
\end{equation}
modulo null sets. Each $(\Omega,{\cal F}_t,{\mathbb Q})$ is standard,
being a complete sub-$\sigma$-algebra of the standard space
$(\Omega,{\cal F},{\mathbb Q})$ of hypothesis~(i).

For $t=0$, hypothesis~(ii) states that $(\Omega,{\cal F}_0,{\mathbb Q})$
is atomless, which is the hypothesis of
Lemma~\ref{lemma:independentComplement} with ${\cal G}={\cal F}_0$ and
${\cal H}$ trivial. That lemma supplies a uniformly distributed,
${\cal F}_0$-measurable $V^0$ with
${\cal F}_0=\overline{\sigma(V^0)}$ modulo null sets.

For $t\geq1$, hypothesis~(iii) states precisely that the conditional
measures obtained by disintegrating $(\Omega,{\cal F}_t,{\mathbb Q})$
over the ${\cal F}_{t-1}$-state are almost surely atomless. Applying
Lemma~\ref{lemma:independentComplement} with ${\cal G}={\cal F}_t$ and
${\cal H}={\cal F}_{t-1}$ gives a uniformly distributed,
${\cal F}_t$-measurable $V^t$, independent of ${\cal F}_{t-1}$, with
${\cal F}_t=\overline{\sigma({\cal F}_{t-1},V^t)}$ modulo null sets.
Since $V^0,\ldots,V^{t-1}$ are ${\cal F}_{t-1}$-measurable, $V^t$ is
independent of them, so $(V^r)_{r\leq t}$ is independent; and
\eqref{eq:classificationCoordinates} follows from the inductive
hypothesis. This completes the induction.

Define
\[
\Theta:\Omega\to\Omega^{\can}_{\cal T},
\qquad
\Theta(\omega):=(V^t(\omega))_{t\in{\cal T}}.
\]
The $V^t$ are independent and uniformly distributed, so $\Theta$ is
measure preserving. By \eqref{eq:classificationCoordinates}, the
preimage under $\Theta$ of the time-$t$ $\sigma$-algebra of
$\underline{\Omega}^{\can}_{\cal T}$ is ${\cal F}_t$ for every $t$; in
particular, taking $t=T$ and using ${\cal F}_T={\cal F}$ modulo null
sets from hypothesis~(i), the preimage of the full $\sigma$-algebra is
${\cal F}$. Both spaces are standard, so $\Theta$ is a mod $0$
isomorphism by Lemma~\ref{lemma:pointIsomorphism}, and it preserves
every level of the filtration.
\end{proof}

\begin{proof}[Stable equivalence of enrichments, Lemma~\ref{lemma:enrichmentStable}]
Write $p:=\pi\circ\pi^{\Omega^3}:\underline{\Omega}^{3+}\to
\underline{\Omega}^1$ and ${\cal B}_t:=p^{-1}({\cal F}^1_t)$. Then $p$ is
again an enrichment: ${\cal B}_t\subseteq{\cal F}^{3+}_t$, and
${\cal F}^{3+}_t$ is conditionally independent of ${\cal B}_T$ given
${\cal B}_t$, since ${\cal F}^3_t$ is and the canonical factor is
independent of $\underline{\Omega}^3$.

For $t\in{\cal T}$ let $U^t$ denote the canonical coordinate of the factor
$\underline{\Omega}^{\can,t}$, viewed as a random variable on
$\Omega^{3+}$; it is uniformly distributed, ${\cal F}^{3+}_t$-measurable
and independent of $\sigma({\cal F}^{3+}_{t-1},{\cal B}_T)$.

We construct by induction on $t$ uniformly distributed $V^t$ on
$\Omega^{3+}$ such that $V^t$ is ${\cal F}^{3+}_t$-measurable, the family
$(V^r)_{r\leq t}$ is independent and independent of ${\cal B}_T$, and
\begin{equation}
{\cal F}^{3+}_t=\overline{\sigma\bigl({\cal B}_t,V^0,\ldots,V^t\bigr)}
\label{eq:enrichmentCoordinates}
\end{equation}
modulo null sets. Write
${\cal H}_t:=\overline{\sigma({\cal B}_t,{\cal F}^{3+}_{t-1})}$, with
${\cal F}^{3+}_{-1}$ trivial. The conditional measures obtained by
disintegrating $(\Omega^{3+},{\cal F}^{3+}_t,{\mathbb Q}^{3+})$ over the
${\cal H}_t$-state are almost surely atomless: $U^t$ is independent of
${\cal H}_t$, since ${\cal B}_t\subseteq{\cal B}_T$ and $U^t$ is
independent of $\sigma({\cal F}^{3+}_{t-1},{\cal B}_T)$, so its
conditional law given that state is almost surely uniform, and, writing
$U^t$ as a Borel function of the time-$t$ state, an atom of a
conditional measure would be carried to an atom of that conditional
law.

Lemma~\ref{lemma:independentComplement}, applied with
${\cal G}={\cal F}^{3+}_t$ and ${\cal H}={\cal H}_t$, therefore
supplies a uniformly distributed, ${\cal F}^{3+}_t$-measurable $V^t$,
independent of ${\cal H}_t$, with
${\cal F}^{3+}_t=\overline{\sigma({\cal H}_t,V^t)}$ modulo null sets.
We retain the notation $h_t$, $Z^t$, $F_t$ of that lemma, so that
$V^t=F_t(Z^t\mid h_t)$ and $Z^t=F^{-1}_t(V^t\mid h_t)$; part~(iii) uses
this explicit form. Since
${\cal H}_t=\overline{\sigma({\cal B}_t,V^0,\ldots,V^{t-1})}$ by the
inductive hypothesis, using ${\cal B}_{t-1}\subseteq{\cal B}_t$,
\eqref{eq:enrichmentCoordinates} follows.

It remains to check that $V^t$ is independent of
$\sigma({\cal B}_T,V^0,\ldots,V^{t-1})$, and not merely of ${\cal H}_t$.
Write $W:=(V^0,\ldots,V^{t-1})$, which is
${\cal F}^{3+}_{t-1}$-measurable, let $Y$ be bounded and
${\cal B}_T$-measurable, and let $f,g$ be bounded and measurable. Since
${\cal B}_t\subseteq{\cal B}_T$, conditional independence of
${\cal F}^{3+}_t$ and ${\cal B}_T$ given ${\cal B}_t$ gives
${\mathbb E}[\,\cdot\mid{\cal B}_T]={\mathbb E}[\,\cdot\mid{\cal B}_t]$
on ${\cal F}^{3+}_t$-measurable variables. Using this, together with
the facts that $g(W)$ is ${\cal H}_t$-measurable and that $V^t$ is
independent of ${\cal H}_t$,
\begin{align*}
{\mathbb E}[f(V^t)g(W)Y]
&=
{\mathbb E}\bigl[{\mathbb E}[f(V^t)g(W)\mid{\cal B}_t]\,Y\bigr]
\\
&=
\bar f\,{\mathbb E}\bigl[{\mathbb E}[g(W)\mid{\cal B}_t]\,Y\bigr]
=
\bar f\,{\mathbb E}[g(W)Y],
\end{align*}
where $\bar f:=\int_0^1 f\,d\lambda$. Hence $V^t$ is uniformly
distributed and independent of $\sigma({\cal B}_T,W)$, which with the
inductive hypothesis completes the induction.

Define
\[
\Theta:\Omega^{3+}\to\Omega^{\can}_{\cal T}\times\Omega^1,
\qquad
\Theta(\omega):=\bigl((V^t(\omega))_{t\in{\cal T}},p(\omega)\bigr).
\]
The $V^t$ are independent uniforms, jointly independent of ${\cal B}_T$,
and $p$ is measure preserving, so $\Theta$ is measure preserving. By
\eqref{eq:enrichmentCoordinates} and ${\cal F}^{3+}_T={\cal F}^{3+}$
modulo null sets, the preimage under $\Theta$ of the time-$t$
$\sigma$-algebra of $\underline{\Omega}^{1+}$ is ${\cal F}^{3+}_t$ for
every $t$. By Lemma~\ref{lemma:pointIsomorphism}, $\Theta$ is a mod $0$
isomorphism, and $\pi^{\Omega^1}\circ\Theta=p$ by construction; taking
$\Xi:=\Theta^{-1}$ proves~(i).

For~(ii), let $\Xi'$ be a second isomorphism as in~(i) and set
$\alpha:=\Xi^{-1}\circ\Xi'$, a filtration-preserving mod $0$
automorphism of $\underline{\Omega}^{1+}$ with
$\pi^{\Omega^1}\circ\alpha=\pi^{\Omega^1}$. Thus $\alpha$ fixes the
$\Omega^1$ coordinate and may be written
$\alpha(\omega^{\can},\omega^1)
=(\Phi(\omega^{\can},\omega^1),\omega^1)
=\hat\Phi(\omega^{\can},\omega^1)$ for a jointly measurable $\Phi$. We
check that $\Phi$ is an adapted randomized automorphism at time $0$.
Since $\alpha$ preserves ${\mathbb Q}^{\can}\times{\mathbb Q}^1$ and
fixes the second coordinate, uniqueness of disintegration gives
$\Phi(\,\cdot\,,\omega^1)_*{\mathbb Q}^{\can}={\mathbb Q}^{\can}$ for
almost every $\omega^1$. Writing $\alpha^{-1}=\hat\Psi$ in the same
way, $\Psi(\,\cdot\,,\omega^1)$ and $\Phi(\,\cdot\,,\omega^1)$ are
mutually inverse almost surely on $\Omega^{1+}$, hence, by Fubini,
mutually inverse modulo null sets for almost every $\omega^1$. As
$\alpha$ is filtration preserving, the time-$t$ coordinate of $\Phi$ is
measurable for ${\cal F}^{1+}_t$, the completion of
${\cal F}^{\can}_t\otimes{\cal F}^1_t$, so its sections at fixed
$\omega^1$ are ${\cal F}^{\can}_t$-measurable for almost every
$\omega^1$, and the same holds for $\Psi$. Hence for almost every
$\omega^1$ the map $\Phi(\,\cdot\,,\omega^1)$ is a
filtration-preserving mod $0$ automorphism of
$\underline{\Omega}^{\can}_{\cal T}$; redefining
$\Phi(\,\cdot\,,\omega^1):=\id$ on the exceptional null set, which
alters $\alpha$ only on a null set, this holds for every $\omega^1$.
The remaining requirement of
Definition~\ref{def:adaptedRandomizedAutomorphism} at $s=0$ is the
${\cal F}^{1+}_t$-measurability of the time-$t$ coordinates of $\Phi$
and $\Psi$ just established, and the requirement at a general $s$
follows since ${\cal F}^{1+}_t\subseteq{\cal F}^{1+}_{t\vee s}$.

For~(iii) we use the following elementary fact: if ${\cal A}$ and
${\cal C}$ are conditionally independent given ${\cal B}$, and
${\cal D}\subseteq{\cal A}$, then ${\cal A}$ and ${\cal C}$ are
conditionally independent given $\sigma({\cal B},{\cal D})$. Indeed,
for $C\in{\cal C}$ conditional independence gives
${\mathbb E}[\indicator_C\mid\sigma({\cal B},{\cal A})]
={\mathbb E}[\indicator_C\mid{\cal B}]$ and hence
${\mathbb E}[\indicator_C\mid\sigma({\cal B},{\cal D})]
={\mathbb E}[\indicator_C\mid{\cal B}]$, while for $A\in{\cal A}$,
$D\in{\cal D}$ and $B\in{\cal B}$, applying conditional independence to
$A\cap D\in{\cal A}$,
\[
{\mathbb E}[\indicator_{A\cap D}\indicator_B\indicator_C]
=
{\mathbb E}\bigl[\indicator_B\,
{\mathbb E}[\indicator_{A\cap D}\mid{\cal B}]\,
{\mathbb E}[\indicator_C\mid{\cal B}]\bigr]
=
{\mathbb E}\bigl[\indicator_{A\cap D}\indicator_B\,
{\mathbb E}[\indicator_C\mid{\cal B}]\bigr].
\]
The sets $D\cap B$ form a $\pi$-system generating
$\sigma({\cal B},{\cal D})$, so
${\mathbb E}[\indicator_A\indicator_C\mid\sigma({\cal B},{\cal D})]
={\mathbb E}[\indicator_A\mid\sigma({\cal B},{\cal D})]\,
{\mathbb E}[\indicator_C\mid{\cal B}]$, which is the assertion.

Let now $u$ and ${\cal G}$ be as in~(iii). Conditioning on a
${\cal G}$-state and then on ${\cal B}_t$ is conditioning on
$\sigma({\cal B}_t,{\cal G})$, so what must be checked is that
${\cal F}^{3+}_t$ and ${\cal B}_T$ are conditionally independent given
$\sigma({\cal B}_t,{\cal G})$ for every $t$. For $t\geq u$ we have
${\cal G}\subseteq{\cal F}^3_u\subseteq{\cal F}^{3+}_t$, and the fact
above applied with ${\cal A}={\cal F}^{3+}_t$, ${\cal C}={\cal B}_T$,
${\cal B}={\cal B}_t$ and ${\cal D}={\cal G}$ gives it. For $t<u$ the
same fact with ${\cal A}={\cal F}^{3+}_u$ and ${\cal B}={\cal B}_u$
shows that ${\cal F}^{3+}_u$ and ${\cal B}_T$ are conditionally
independent given $\sigma({\cal B}_u,{\cal G})={\cal G}$, and the claim
follows since ${\cal F}^{3+}_t\subseteq{\cal F}^{3+}_u$ and
$\sigma({\cal B}_t,{\cal G})={\cal G}$. Thus $p$ is again an enrichment
under ${\mathbb Q}^{3+}_y$ for almost every ${\cal G}$-state $y$.
Taking $t=u$ in the first paragraph of this proof shows in addition
that ${\cal G}$ and ${\cal B}_T$ are conditionally independent given
${\cal B}_u$, so that $p$ pushes ${\mathbb Q}^{3+}_y$ forward to
${\mathbb Q}^1_{\pi_u(y)}$, where $\pi_u(y)$ is the time-$u$ state of
$\underline{\Omega}^1$ determined by $y$ through
${\cal B}_u\subseteq{\cal G}$.

Because ${\cal G}\subseteq{\cal F}^3_u$ and the canonical factor is
independent of $\underline{\Omega}^3$, each $U^t$ remains uniformly
distributed and independent of $\sigma({\cal F}^{3+}_{t-1},{\cal B}_T)$
under ${\mathbb Q}^{3+}_y$. The construction above therefore applies
verbatim to the conditioned space.

Finally, it depends measurably on $y$. In
Lemma~\ref{lemma:independentComplement} the realizations $h_t$ and the
generators $Z^t$ may be chosen once and for all, since the
$\sigma$-algebras ${\cal H}_t$ and ${\cal F}^{3+}_t$ do not depend on
$y$; the only ingredient that does is the conditional distribution
function, which by disintegration may be chosen jointly measurable in
$(z,h,y)$, say $F_t(z\mid h,y)$. Then $V^t_y:=F_t(Z^t\mid h_t,y)$, and
hence $\Theta_y=\bigl((V^t_y)_{t\in{\cal T}},p\bigr)$, are jointly
measurable in $(y,\omega)$. So is the inverse: given $\omega^1$ and
$(v^t)_{t\in{\cal T}}$, the state $h_t$ is a measurable function of
$\omega^1$ and $v^0,\ldots,v^{t-1}$ by
\eqref{eq:enrichmentCoordinates}, whence $Z^t=F^{-1}_t(v^t\mid h_t,y)$
recursively, and $Z^T$ generates ${\cal F}^{3+}$. Setting
$\Xi_y:=\Theta_y^{-1}$ proves~(iii).
\end{proof}

\begin{proof}[Proof of Lemma~\ref{lemma:stableEquivalence}]
Let $\underline{\Omega}^1$ be the one-point filtered probability
space, which is a standard filtered probability space, and let
$\pi:\underline{\Omega}\to\underline{\Omega}^1$ be the unique map.
Since $\pi^{-1}({\cal F}^1_t)$ is trivial for every $t$, both
conditions of Definition~\ref{def:enrichment} hold, so $\pi$ is an
enrichment. Here $\underline{\Omega}^{1+}$ is canonically isomorphic
to $\underline{\Omega}^{\can}_{\cal T}$, so
Lemma~\ref{lemma:enrichmentStable}(i) gives a filtration-preserving
mod $0$ isomorphism
$\underline{\Omega}^{\can}_{\cal T}\cong\underline{\Omega}^{+}$.
Applying this with
$\underline{\Omega}=\underline{\Omega}^{\can}_{\cal T}$, which is
itself a standard filtered probability space, gives
$\underline{\Omega}^{\can}_{\cal T}\cong
(\underline{\Omega}^{\can}_{\cal T})^{+}$, and composing the two
isomorphisms shows that
$\underline{\Omega}^{+}\cong(\underline{\Omega}^{\can}_{\cal T})^{+}$;
that is, $\underline{\Omega}$ is stably equivalent to
$\underline{\Omega}^{\can}_{\cal T}$.
\end{proof}

\begin{proof}[Naturality of disintegration, Lemma~\ref{lemma:disintegrationNatural}]
The argument follows that of the group-invariance case treated in \cite[Example~7]{changPollard}.

Since $\Theta$ is filtration preserving, $\Theta^{-1}{\cal F}^2_s={\cal
F}^1_s$ modulo null sets. As ${\cal
F}^2_s=(\omega^2_s)^{-1}{\cal B}(\Omega^2_s)$ modulo null sets, the
composite $\omega^2_s\circ\Theta$ is ${\cal F}^1_s$-measurable; since
${\cal F}^1_s=(\omega^1_s)^{-1}{\cal B}(\Omega^1_s)$ modulo null sets,
the Doob--Dynkin factorisation lemma yields a ${\cal
B}(\Omega^1_s)$-measurable map
$\Theta_s:\Omega^1_s\to\Omega^2_s$, defined ${\mathbb
Q}^1_s$-almost everywhere and unique as such, with
$\omega^2_s\circ\Theta=\Theta_s\circ\omega^1_s$ ${\mathbb Q}^1$-a.s.
Because $\Theta$ is a measure-preserving bijection modulo null sets
intertwining $\omega^1_s$ and $\omega^2_s$ in this way, $\Theta_s$
pushes ${\mathbb Q}^1_s$ forward to ${\mathbb Q}^2_s$, and the same
argument applied to $\Theta^{-1}$ produces an a.e.\ inverse to
$\Theta_s$; hence $\Theta_s$ is a mod 0 isomorphism.

For the fibre correspondence, fix bounded measurable $f:\Omega^2\to{\mathbb
R}$ and $h:\Omega^2_s\to{\mathbb R}$. Since $\Theta$ is measure
preserving and $\omega^2_s\circ\Theta=\Theta_s\circ\omega^1_s$,
\[
\int_{\Omega^2}h(\omega^2_s(\omega))f(\omega)\,{\mathbb Q}^2(d\omega)
=
\int_{\Omega^1}h(\Theta_s(\omega^1_s(\omega')))\,(f\circ\Theta)(\omega')\,{\mathbb Q}^1(d\omega').
\]
Disintegrating the left-hand side over $\Omega^2_s$ and the
right-hand side over $\Omega^1_s$ gives
\begin{align*}
\int_{\Omega^2_s}h(y)\int f\,d\mu^2_y\,{\mathbb Q}^2_s(dy)
&=
\int_{\Omega^1_s}h(\Theta_s(x))\int (f\circ\Theta)\,d\mu^1_x\,{\mathbb Q}^1_s(dx)
\\
&=
\int_{\Omega^1_s}h(\Theta_s(x))\int f\,d(\Theta_*\mu^1_x)\,{\mathbb Q}^1_s(dx).
\end{align*}
Since $\Theta_s$ pushes ${\mathbb Q}^1_s$ forward to ${\mathbb Q}^2_s$,
changing variables $y=\Theta_s(x)$ in the outer integral on the left
shows that this agrees with
$\int_{\Omega^1_s}h(\Theta_s(x))\int f\,d\mu^2_{\Theta_s(x)}\,{\mathbb
Q}^1_s(dx)$ for every bounded measurable $h$. As $f$ was also
arbitrary, uniqueness of the disintegration gives
$\Theta_*\mu^1_x=\mu^2_{\Theta_s(x)}$ for ${\mathbb Q}^1_s$-almost
every $x$.
\end{proof}

\begin{proof}[Canonical representation of an exchangeable sequence, Theorem~\ref{thm:canonicalDeFinetti}]
The definitions above are stated purely in terms of the law
${\mathbb M}$ of $(X^i)_{i\in{\mathbb N}}$, so $Y$ is, a priori, only a
function on $\underline\Omega^{A,\infty}$; the present proof is the
only place we need to link this back to the original sequence, via
the composite $Y\circ X$, so that $\omega\mapsto Y(X(\omega))$ is a
genuine random variable on $(\Omega,{\cal F},{\mathbb P})$.

Part (i) is the classical de~Finetti representation: by
Theorem~\ref{thm:hewittSavage} the coordinates are conditionally
independent given $Y$, with common conditional marginal $\nu_y$ --
see also \cite[Theorem~1.1]{kallenberg2005} -- so the disintegration
of ${\mathbb M}={\rm Law}(X)$ over $Y$ is
$(\nu_y^{\otimes{\mathbb N}})_y$
(Definition~\ref{def:deFinettiRepresentation}), which is the
statement of (i).

For (ii), note that both
$(\Omega^{A,\infty},{\mathbb M})$ and
$(\Omega^{\rm mix}\times\Omega^{A,\infty},{\mathbb M}^\circ)$ are
standard probability spaces, the latter since
$\underline\Omega^{\rm mix}$ and $\underline\Omega^{A,\infty}$ are
each standard. The map $\Xi$ pushes ${\mathbb M}$ forward to
${\mathbb M}^\circ$: composing $\Xi$ with $X$ gives the map
$\omega\mapsto(Y(X(\omega)),X(\omega))$, whose law is
${\mathbb M}^\circ$ by (i), while ${\mathbb M}$ is the law of $X$.
Moreover, the pullback under $\Xi$ of the product $\sigma$-algebra is
the full $\sigma$-algebra of $\Omega^{A,\infty}$, since the second
component of $\Xi$ is the identity. By
Lemma~\ref{lemma:pointIsomorphism}, $\Xi$ is a mod $0$ isomorphism.
Its first component is $Y$, so it commutes with the projection onto
$\underline\Omega^{\rm mix}$ and identifies same-$y$ fibres with
same-$y$ fibres, proving (ii).
\end{proof}

\subsection{Realization of exchangeable $\sigma$-algebras}
\label{appendix:finiteHomogeneous}

\begin{proof}[Proof of Lemma~\ref{lemma:standardFiltered}]
	Let $\Delta^N$ denote the set of orbits of the $S_N$-action on
	$\Omega^{\can,\infty}_{\cal T}$, let
	$\pi^N:\Omega^{\can,\infty}_{\cal T}\to\Delta^N$ be the quotient
	map, and equip $\Delta^N$ with the pushforward measure
	$\hat{\mathbb Q}:=\pi^N_*{\mathbb Q}^{\can,\infty}_{\cal T}$ and the
	$\sigma$-algebras
	\begin{gather*}
	\hat{\cal F}_t:=\{B\subseteq\Delta^N\mid(\pi^N)^{-1}(B)\in{\cal E}^N({\cal F}^{\can}_t)\},
	\\
	\hat{\cal F}:=\{B\subseteq\Delta^N\mid(\pi^N)^{-1}(B)\in{\cal E}^N({\cal F}^{\can}_{\cal T})\},
	\end{gather*}
	each of which is complete since ${\cal E}^N$ of a complete
	$\sigma$-algebra is complete. Every set in
	${\cal E}^N({\cal F}^{\can}_t)$ agrees modulo null sets with a
	genuinely $S_N$-invariant set in ${\cal F}^{\can}_t$ (replace $E$ by
	$\bigcap_{g\in S_N}gE$), and genuinely invariant sets are exactly the
	preimages of subsets of $\Delta^N$; hence
	$(\pi^N)^{-1}(\hat{\cal F}_t)={\cal E}^N({\cal F}^{\can}_t)$ modulo
	null sets, and likewise for $\hat{\cal F}$.

	We claim that the filtered probability space
	$\underline{\Delta}^N:=(\Delta^N,\hat{\cal F},(\hat{\cal F}_t)_{t\in{\cal T}},\hat{\mathbb Q})$
	satisfies the hypotheses of
	Lemma~\ref{lemma:finiteFiltrationClassification}.

	First, $(\Delta^N,\hat{\cal F},\hat{\mathbb Q})$ is standard. Fix a
	Borel isomorphism between $\Omega^{\can,\infty}_{\cal T}$ and
	$[0,1]$. As in the proof of Lemma~\ref{lemma:sectionExists}, the map
	sending each point to the minimum of its (finite) orbit is
	measurable, so the set $S$ of orbit minima is a Borel transversal for
	the action, and $\pi^N$ restricts to a bijection $S\to\Delta^N$. This
	bijection identifies $\hat{\cal F}$ with the trace $\sigma$-algebra
	on $S$: the saturation $\bigcup_{g\in S_N}gC$ of a measurable
	$C\subseteq S$ is an invariant measurable set meeting $S$ exactly in
	$C$. Hence $(\Delta^N,\hat{\cal F},\hat{\mathbb Q})$ is identified
	with a Borel probability measure on a standard Borel space.

	Second, ${\cal F}^{\can}_T$ is the full $\sigma$-algebra of
	$\underline{\Omega}^{\can,\infty}_{\cal T}$, so
	${\cal E}^N({\cal F}^{\can}_T)={\cal E}^N({\cal F}^{\can}_{\cal T})$
	and therefore $\hat{\cal F}_T=\hat{\cal F}$.

	For the remaining hypotheses, let $U_t$ denote the canonical uniform
	random variable at time $t$ on the copy of
	$\underline{\Omega}^{\can}_{\cal T}$ with index $N$. Since $S_N$
	permutes only the indices $0,\ldots,N-1$, the random variable $U_t$
	is $S_N$-invariant and ${\cal E}^N({\cal F}^{\can}_t)$-measurable, so
	it descends to an $\hat{\cal F}_t$-measurable random variable
	$\hat U_t$ on $\Delta^N$. Moreover $U_t$ is uniformly distributed and
	independent of
	${\cal F}^{\can,\infty}_{t-1}\supseteq{\cal E}^N({\cal F}^{\can}_{t-1})$,
	so $\hat U_t$ is uniformly distributed and, for $t \geq 1$, independent of
	$\hat{\cal F}_{t-1}$.

	The spaces $(\Delta^N,\hat{\cal F},\hat{\mathbb Q})$ and
	$(\Delta^N,\hat{\cal F}_0,\hat{\mathbb Q})$ are atomless, since an
	atom of either would be mapped by $\hat U_0$ to an atom of the
	uniform distribution. For $t\geq1$, the conditional measures
	obtained by disintegrating $(\Delta^N,\hat{\cal F}_t,\hat{\mathbb Q})$
	over the time-$(t-1)$ state are almost surely atomless: the
	conditional law of $\hat U_t$ given the time-$(t-1)$ state is almost
	surely the uniform distribution, and, writing
	$\hat U_t=u_t(\hat\omega_t)$ for a Borel function $u_t$ of the
	time-$t$ state, an atom of a conditional measure would be mapped by
	$u_t$ to an atom of this conditional law.

	By Lemma~\ref{lemma:finiteFiltrationClassification} there is
	therefore a filtration-preserving mod $0$ isomorphism
	$\Lambda^N:\underline{\Delta}^N\to\underline{\Omega}^{\can}_{\cal T}$.
	Set $\theta^N:=\Lambda^N\circ\pi^N$. Then $\theta^N$ is measure
	preserving and
	\[
	(\theta^N)^{-1}({\cal F}^{\can}_t)
	=(\pi^N)^{-1}(\hat{\cal F}_t)
	={\cal E}^N({\cal F}^{\can}_t)
	\]
	modulo null sets for every $t\in{\cal T}$.

	Finally, if $Z$ is a random variable on
	$\underline{\Omega}^{\can}_{\cal T}$ then $Z\circ\theta^N$ is a
	random variable on $\underline{\Omega}^{{\cal E}^N}$ with the same
	law, measurable at time $t$ whenever $Z$ is; conversely, by the
	Doob--Dynkin lemma, every ${\cal E}^N({\cal F}^{\can}_t)$-measurable
	random variable is of the form $Z\circ\theta^N$ modulo null sets for
	some ${\cal F}^{\can}_t$-measurable $Z$, and $Z$ is unique up to
	almost sure equality because $\theta^N$ is measure preserving.
\end{proof}

\section{Well-posedness - proofs}
\label{appendix:wellPosedProofs}
\subsection{Compact problems - proofs}
\label{appendix:calculusVariations}

Compactness alone does not imply the existence of optimal
continuations. One must also show that feasible continuation plans
remain feasible under limits.

\begin{lemma}[Closedness of feasible continuation sets]
	\label{lemma:closednessContinuationSets}
	
	For every homogeneous continuation problem
	$
	(s,x,\beta,\overline X),
	$
	the set
	$
	{\cal X}_s^x(\beta,\overline X)
	$
	is closed under \(L^1\)-limits and under almost-sure limits, and
	the same holds for
	$
	\hat{\cal X}_s^x(\beta,\overline X)
	$;
	both statements remain true when a budget sequence
	$\beta_n\to\beta$ (in $L^1$, respectively almost surely along a
	subsequence) varies alongside the streams.
\end{lemma}

\begin{proof}
	Let
	$
	X^n
	\in
	{\cal X}_s^x(\beta,\overline X)
	$
	and suppose
	$
	X^n\to X
	$
	in \(L^1\).
	The consistency constraints and non-participation constraints are
	preserved under \(L^1\)-limits.

	For the funding constraints, observe that the maps
	$X\mapsto R^X_u$ are Lipschitz on $L^1$: conditional expectation
	is an $L^1$-contraction and $z\mapsto z^+$ is $1$-Lipschitz, so
	backward induction through the defining recursion gives
	\[
	\bigl\|R^{X}_u-R^{Y}_u\bigr\|_{L^1}
	\leq
	\sum_{u\leq t\leq T}\bigl\|X_t-Y_t\bigr\|_{L^1}.
	\]
	Hence $R^{X^n}_u\to R^X_u$ in $L^1$ for every $u$, and likewise
	${\mathbb E}(R^{X^n}_{u+1}\mid{\cal F}^\infty_u)
	\to{\mathbb E}(R^{X}_{u+1}\mid{\cal F}^\infty_u)$. Passing to an
	almost surely convergent subsequence, the self-financing funding
	constraint \eqref{eq:homogeneousBudget} and the vanishing
	conditions~\eqref{eq:noPotOutsideParticipation} pass to the
	limit.

	The same argument, with $m_u$ in place of
	${\mathbb E}(\,\cdot\mid{\cal F}^\infty_u)$, shows that the
	redistributively feasible set
	$\hat{\cal X}_s^x(\beta,\overline X)$ is also closed under
	$L^1$-limits.

	For almost-sure limits, the conditional form of Fatou's lemma
	gives, by backward induction through the defining recursion,
	\[
	R^X_u
	\leq
	\liminf_n R^{X^n}_u
	\qquad
	\text{almost surely, for every }u,
	\]
	since $z\mapsto z^+$ is continuous and monotone and conditional
	expectation satisfies Fatou's inequality for nonnegative
	integrands. Every funding constraint and vanishing condition is an
	almost-sure inequality bounding such quantities above by quantities
	that converge almost surely (or are fixed), so each passes to the
	limit; the consistency and non-participation constraints pass to
	almost-sure limits directly. The same applies to the redistributive pots
	$\hat R^X_u$ with $m_u$ in place of the conditional expectation,
	and a budget sequence converging alongside changes nothing in
	either argument.
\end{proof}

The previous lemma is independent of the preference system. It is a
structural property of the homogeneous continuation problem and
expresses the stability of budget feasibility under limits.

The continuity of the value function is obtained by perturbing the
continuation budget while preserving feasibility.

\begin{lemma}[Budget perturbation]
	\label{lemma:budgetPerturbation}

	Let
	$
	\beta_n\to\beta
	$
	in $L^1$, with
	${\cal X}_s^x(\beta_n,\overline X)\neq\emptyset$ for every $n$.
	For every
	$
	X\in{\cal X}_s^x(\beta,\overline X)
	$
	there exists a sequence
	$
	X^n\in{\cal X}_s^x(\beta_n,\overline X)
	$
	with $X^n\le X$, such that
	$
	X^n\to X
	$
	in \(L^1\) and
	$
	\liminf_n{\cal J}_s(X^n\mid x)\ge{\cal J}_s(X\mid x).
	$
\end{lemma}
\begin{proof}
	Define the ${\cal F}^\infty_s$-measurable random variables
	\[
	\varepsilon_n
	:=
	1\wedge\frac{(\beta-\beta_n)^+}{\beta+\eta_s},
	\]
	with the convention $0/0:=0$, and set
	$X^n_t:=(1-\varepsilon_n)X_t$ for $t\ge s$ and
	$X^n_t:=\overline X_t$ for $t<s$.

	Since $z\mapsto z^+$ is positively homogeneous and convex, and
	$\varepsilon_n$ is ${\cal F}^\infty_s$-measurable, backward
	induction through the defining recursion gives
	$R^{X^n}_u\le(1-\varepsilon_n)R^X_u$ for every $u>s$; in
	particular the conditions~\eqref{eq:noPotOutsideParticipation} are inherited from $X$.
	For the funding constraint, on the event
	$\{(\beta-\beta_n)^+\le\beta+\eta_s\}$,
	\[
	X^n_s+{\mathbb E}\bigl(R^{X^n}_{s+1}\,\big|\,{\cal F}^\infty_s\bigr)
	\le
	(1-\varepsilon_n)(\beta+\eta_s)
	=
	\beta+\eta_s-(\beta-\beta_n)^+
	\le
	\beta_n+\eta_s;
	\]
	on the complementary event $\varepsilon_n=1$, the left-hand side
	vanishes, and $\beta_n+\eta_s\ge0$ almost surely because
	${\cal X}_s^x(\beta_n,\overline X)$ is non-empty. Hence
	$X^n\in{\cal X}_s^x(\beta_n,\overline X)$.

	Since $(\beta-\beta_n)^+\to0$ in $L^1$, $\varepsilon_n\to0$ in
	probability, and $|X-X^n|\le\varepsilon_n X\le X$ with $X$
	integrable, so $X^n\to X$ in $L^1$ by dominated convergence.

	Finally, fix a subsequence; passing to a further subsequence we
	may assume $\varepsilon_n\to0$ almost surely. The streams
	$\tilde X^k:=\bigl(1-\sup_{j\ge k}\varepsilon_{n_j}\bigr)X$
	increase almost surely to $X$, so
	${\cal J}_s(\tilde X^k\mid x)\to{\cal J}_s(X\mid x)$ by the
	Monotone continuity axiom, while
	${\cal J}_s(X^{n_k}\mid x)\ge{\cal J}_s(\tilde X^k\mid x)$ by
	Monotonicity. Hence every subsequence has a further subsequence
	along which
	$\liminf{\cal J}_s(X^n\mid x)\ge{\cal J}_s(X\mid x)$, which gives
	the claim.
\end{proof}

\begin{lemma}[Continuity in budget]
	\label{lemma:valueContinuity}
	Assume that
	$
	\Prob
	$
	is compact. Suppose
	$
	(\beta_n,\eta'_n)\to(\beta,\eta')
	$
	in $L^1$ as in the stability clause of
	Definition~\ref{def:compactProblem}(ii), that the perturbed
	constraint sets are non-empty with values exceeding $-\infty$, and
	that the limiting value exceeds $-\infty$.
	Then the perturbed values converge to the limiting value.
\end{lemma}

\begin{proof}
	
	We first show
	\[
	\liminf_n
	V_s(x,\beta_n,\overline X)
	\ge
	V_s(x,\beta,\overline X).
	\]
	Let
	$
	X^\circ
	$
	be optimal for
	$
	(s,x,\beta,\overline X).
	$
	By Lemma~\ref{lemma:budgetPerturbation}, there exist
	$
	Y^n\in{\cal X}_s^x(\beta_n,\overline X)
	$
	with
	$
	Y^n\to X^\circ
	$
	in $L^1$ and
	\[
	\liminf_n
	{\cal J}_s(Y^n\mid x)
	\ge
	{\cal J}_s(X^\circ\mid x)
	=
	V_s(x,\beta,\overline X).
	\]
	Since
	\[
	V_s(x,\beta_n,\overline X)
	\ge
	{\cal J}_s(Y^n\mid x),
	\]
	we obtain
	\[
	\liminf_n
	V_s(x,\beta_n,\overline X)
	\ge
	V_s(x,\beta,\overline X).
	\]
	
	We now prove the converse inequality for the limsup. Fix a
	subsequence along which $V_s(x,\beta_n,\overline X)$ converges to
	some limit $L$; by the inequality just proved,
	$L\ge V_s(x,\beta,\overline X)>-\infty$, so eventually
	$V_s(x,\beta_n,\overline X)>-\infty$ and optimisers $X^n$ for
	$(s,x,\beta_n,\overline X)$ exist by condition~(ii) of
	Definition~\ref{def:compactProblem}. The stability clause of the
	same condition provides a further subsequence with $X^n\to X^\circ$
	in $L^1$, where $X^\circ$ is an optimiser for
	$(s,x,\beta,\overline X)$. Passing to yet a further subsequence
	with almost-sure convergence, the Fatou property (condition~(i))
	gives
	\[
	V_s(x,\beta,\overline X)
	=
	{\cal J}_s(X^\circ\mid x)
	\ge
	\limsup_n
	{\cal J}_s(X^n\mid x)
	=
	L .
	\]
	This
	argument applies equally to any subsequence of $(\beta_n)_n$, since
	any such subsequence again converges to $\beta$ in $L^1$; hence every
	subsequence of $\bigl(V_s(x,\beta_n,\overline X)\bigr)_n$ has a
	further subsequence converging to a limit at most
	$V_s(x,\beta,\overline X)$. In particular, the subsequence realising
	$\limsup_n V_s(x,\beta_n,\overline X)$ has a further subsequence
	converging to this same limsup, forcing
	\[
	\limsup_n
	V_s(x,\beta_n,\overline X)
	\le
	V_s(x,\beta,\overline X).
	\]
	
	Combining the two inequalities proves the claim.
\end{proof}

We are now ready to prove Theorem~\ref{thm:regularMaximisers}.

\begin{proof}[Compact problems admit regular maximisers, Theorem~\ref{thm:regularMaximisers}]

Throughout, fix $(s,x,\beta,\overline X)$ and abbreviate
$V(b):=V_s(x,\beta+b,\overline X)$ where convenient.

{\em Part (i).} Condition (iii) of compactness provides a
measurable family of optimisers at every state at which the
supremum is attained. By the attainment clause of condition (ii) of
Definition~\ref{def:compactProblem}, the supremum is attained at
every state with a non-empty constraint set and value exceeding
$-\infty$; at states where the value equals $-\infty$
every feasible stream attains it. Hence the family is an arg max.
The same argument applies to a random budget, which is simply an
element of $L^1({\mathbb R},{\cal F}^\infty_s,{\mathbb Q}^\infty_x)$.

{\em Part (ii).} At homogeneous states the self-financing constraint set
is contained in the redistributively feasible set, so both assertions follow
from condition (iv) of compactness.

{\em Part (iv).} This is part (iii) of the no-satiation condition.

{\em Part (iii): strict monotonicity.} Let $c'>c$ with
${\cal X}^x_s(\beta+c,\overline X)\neq\emptyset$. If
$V(c)=-\infty$: the constraint set at $c$ is non-empty, so
$\beta(x)+c+H^x_s\ge0$ by feasibility, whence
$\beta(x)+c'+H^x_s>0$, and no-satiation (iii) gives
$V(c')>-\infty=V(c)$. Suppose
instead $V(c)>-\infty$ and let $X^\circ$ be an optimiser
(attainment, Definition~\ref{def:compactProblem}(ii)). Let $\rho^e$
denote the pot process of the window stream $e^s$ computed with
$\eta'\equiv0$, and set
$P_s:=e^s_s+{\mathbb E}(\rho^e_{s+1}\mid{\cal F}^\infty_s)$, which
is integrable and satisfies $P_s(x)\ge1$ at participating states.
Since the pot recursion is subadditive and positively homogeneous,
$R^{X^\circ+\lambda e^s}_u\le R^{X^\circ}_u+\lambda\rho^e_u$, so
with $\lambda:=(c'-c)/P_s(x)>0$ the stream
$X^\circ+\lambda e^s$ lies in
${\cal X}^x_s(\beta+c',\overline X)$. By no-satiation (i),
\[
V(c')
\ge
{\cal J}_s(X^\circ+\lambda e^s\mid x)
>
{\cal J}_s(X^\circ\mid x)
=
V(c).
\]

{\em Part (iii): right-continuity.} Let $b_n\downarrow b$ in the
domain and let $L:=\lim_nV(b_n)$, which exists by monotonicity and
satisfies $L\ge V(b)$. If $L=-\infty$ there is nothing to prove.
Otherwise, discarding finitely many terms, each $V(b_n)>-\infty$;
let $X^n$ be optimisers for $\beta+b_n$, so that
${\cal J}_s(X^n\mid x)=V(b_n)$ is bounded below. The stability
clause of condition (ii) of compactness, applied along
$\beta+b_n\to\beta+b$, yields a subsequence converging in $L^1$ to
an optimiser $X^\circ$ of $(s,x,\beta+b,\overline X)$; passing to a
further almost surely convergent subsequence, the Fatou property
gives
$V(b)={\cal J}_s(X^\circ\mid x)\ge\limsup_n{\cal J}_s(X^n\mid x)=L$.
Hence $V(b)=L$.

{\em Part (iii): continuity at interior points.} Let $c$ be
interior and $c_n\to c$; splitting into the subsequences above and
below $c$, right-continuity handles the former, and for
$c_n\uparrow c$ the two inequalities follow from
Lemma~\ref{lemma:valueContinuity} when $V(c)>-\infty$ and from
monotonicity when $V(c)=-\infty$ (interior points with
$V(c)=-\infty$ have $V(c_n)=-\infty$ for $c_n<c$).

{\em Part (iii): the supremum.} For any integrable stream $Y$
satisfying the consistency and non-participation constraints, the
quantity $Y_s+{\mathbb E}(R^Y_{s+1}\mid{\cal F}^\infty_s)$ is
integrable, so $Y\in{\cal X}^x_s(\beta+c,\overline X)$ once $c$
exceeds its conditional realization at $x$ less $\beta(x)+\eta_s(x)$;
streams violating those constraints are almost surely dominated,
after modification on the constrained dates, by streams of no
smaller utility (Monotonicity), the modification being irrelevant
to ${\cal J}_s(\,\cdot\mid x)$ at a participating $x$ only through
the constrained dates. We omit the routine details. Hence
$\sup_cV(\beta+c)$ equals the supremum of ${\cal J}_s(Y\mid x)$
over all integrable streams. By no-satiation (ii) this supremum is
attained by no integrable stream; in particular
$V(\beta+c)={\cal J}_s(X^\circ_c\mid x)$ never attains it.
\end{proof}

\subsection{Establishing compactness - proofs}
\label{appendix:duality}

Compactness (Definition~\ref{def:compactProblem}) is verified
fibrewise in $(t,x)$ throughout this section, so in every proof below
it suffices to treat the parameters of the relevant fibre as
constant; we write $\alpha,\beta,\rho,\theta$ accordingly without
further comment, using $\rho$ for the exponent on current consumption
in either family and $\alpha$ for the exponent governing
certainty-equivalent aggregation of continuation values in the
Epstein--Zin case.

We begin by developing the technical machinery we will need in
addition to the direct-method existence arguments to establish the
measurability properties of a compact problem.

\begin{lemma}[Measurability of the admissible correspondence]
	\label{lemma:measurableAdmissibleCorrespondence}
	Fix
	$
	s\in{\cal T}$,
	$
	\beta\in
	L^1({\mathbb R},{\cal F}_s^\infty,{\mathbb Q}^\infty)$,
	and
	\[
	\overline X
	\in
	L^1
	\Bigl(
	{\mathbb R}_{\ge0},
	({\cal F}_t^\infty)_{t<s},
	{\mathbb Q}^\infty
	\Bigr).
	\]
	Then the correspondence
	\[
	x
	\longmapsto
	{\cal A}_s^x(\beta,\overline X)
	\]
	has measurable graph.
\end{lemma}

\begin{proof}
	
	The admissible correspondence is determined by the consistency
	constraints \eqref{eq:consistencyCondition} and
	\eqref{eq:consistencyConditionRealised},
	the non-participation constraints
	\eqref{eq:nonParticipation}
	and
	\eqref{eq:noConsumptionNonParticipation},
	and the finitely many funding conditions~\eqref{eq:budgetJ} and \eqref{eq:selfFinancing}.

	The consistency and non-participation constraints define closed linear
	subspaces of the ambient product \(L^1\)-space and therefore Borel
	subsets.

	For $u\in{\cal T}_{>s}$ define the funding residuals
	\begin{gather*}
	D_s(x,\gamma)
	:=
	m_s(b)+m_s(\eta_s)-\sum_{s\le t\le T}m_s(\gamma_{s,t}),
	\\
	D_u(x,\gamma)
	:=
	\sum_{u\le t\le T}m_u(\gamma_{u-1,t})+m_u(\eta_u)
	-\sum_{u\le t\le T}m_u(\gamma_{u,t}),
	\end{gather*}
	where $m_u$ is computed under ${\mathbb Q}^\infty_x$.
	For fixed \(\gamma\), measurability of the disintegration
	\(x\mapsto{\mathbb Q}_x^\infty\)
	implies that each
	$x\mapsto D_u(x,\gamma)$
	is measurable.
	For fixed \(x\),
	conditional expectation is a continuous linear operator on \(L^1\),
	and therefore each
	$\gamma\mapsto D_u(x,\gamma)$
	is continuous.
	Since the ambient strategy space is separable
	(Lemma~\ref{lemma:separabilityA}),
	the Carath\'eodory theorem implies that each
	$(x,\gamma)\mapsto D_u(x,\gamma)$
	is jointly measurable, so
	$\{(x,\gamma):D_u(x,\gamma)\ge0\}$
	is measurable.

	Since there are only finitely many funding conditions, the graph of
	\[
	x
	\mapsto
	{\cal A}_s^x(\beta,\overline X)
	\]
	is a finite intersection of measurable sets and is therefore
	measurable.

	The same argument applies to the stream correspondences
	$x\mapsto{\cal X}^x_s(\beta,\overline X)$ and
	$x\mapsto\hat{\cal X}^x_s(\beta,\overline X)$: the maps
	$X\mapsto R^X_u$ and $X\mapsto\hat R^X_u$ are Lipschitz on $L^1$
	(Lemma~\ref{lemma:closednessContinuationSets}), so the
	corresponding funding constraints define jointly measurable sets
	in the same way.
\end{proof}

\begin{proposition}[Measurability]
	\label{prop:measurabilityFromL1}
	
	Let ${\cal J}$ be a family of preferences. Suppose that
	condition~(i) of Definition~\ref{def:compactProblem} (the Fatou
	property) holds, and that for every homogeneous continuation
	problem with non-empty constraint set and value exceeding
	$-\infty$, the supremum defining the value is attained. Then
	condition~(iii) of Definition~\ref{def:compactProblem} holds, in
	its parametrised form.
	
\end{proposition}
\begin{proof}
	
	Fix a homogeneous continuation problem $(s,x,\beta,\overline X)$.
	
	\emph{Upper semicontinuity along $L^1$ limits.} Suppose
	$X^n\in{\cal X}_s^x(\beta,\overline X)$ with $X^n\to X$ in $L^1$,
	and let $L:=\limsup_n{\cal J}_s(X^n\mid x)$. By
	Lemma~\ref{lemma:closednessContinuationSets},
	$X\in{\cal X}_s^x(\beta,\overline X)$. If $L=-\infty$ there is
	nothing to prove; otherwise pass to a subsequence along which
	${\cal J}_s(X^n\mid x)\to L$ and, further, $X^n\to X$ almost
	surely. The Fatou property gives $L\le{\cal J}_s(X\mid x)$.

	\emph{Condition~(iii).} Define $C(y):={\cal X}_s^y(\beta,\overline X)$
	and $f(y,Z):={\cal J}_s(Z\mid y)$. By the attainment hypothesis,
	the maximum of $f(y,\cdot)$ over $C(y)$ is attained at every state
	at which $C(y)$ is non-empty with value exceeding $-\infty$, and
	trivially wherever the value is $-\infty$. By the
	measurability axiom for preferences, for every admissible strategy
	$Z$, $y\mapsto f(y,Z)$ is measurable, and
	$Z\mapsto f(y,Z)$ is upper semicontinuous for every compatible
	state $y$, by the upper-semicontinuity step above. Since
	the admissible strategy space is separable
	(Lemma~\ref{lemma:separabilityA}), $f$ is a normal integrand in the
	sense of Castaing and Valadier
	\cite[Chapter~VII]{CastaingValadier}. By
	Lemma~\ref{lemma:measurableAdmissibleCorrespondence}, the
	feasible-set correspondence $C$ has measurable graph. Therefore the
	measurable maximum theorem for upper semicontinuous normal
	integrands (see Castaing and Valadier
	\cite[Chapter~III]{CastaingValadier}) implies that
	$v(y)=\max_{Z\in C(y)}f(y,Z)$ is measurable and that the maximiser
	correspondence $F(y)=\argmax_{Z\in C(y)}f(y,Z)$ has measurable
	graph on the set of states at which the maximum is attained.
	The Kuratowski--Ryll--Nardzewski measurable selection theorem
	(Castaing and Valadier \cite[Chapter~III]{CastaingValadier}) yields a
	measurable selector $y\mapsto X^y\in F(y)$ on that set.

	For the parametrised form, enlarge the parameter space: let
	$\Theta:=\Omega^{M\times I}_s\times{\mathbb R}\times
	L^1({\mathbb R}_{\ge0},({\cal F}^\infty_t)_{t<s},{\mathbb Q}^\infty)$,
	a standard Borel space, with generic point
	$\theta=(y,b,\overline W)$, and set
	$C(\theta):={\cal X}^y_s(b,\overline W)$ and
	$f(\theta,Z):={\cal J}_s(Z\mid y)$. The funding residuals are
	jointly continuous in $(b,\overline W,Z)$ for fixed $y$ as they
	are built from conditional expectations, the Lipschitz pot maps of
	Lemma~\ref{lemma:closednessContinuationSets}, and the continuous
	concatenation map $(\overline W,Z)\mapsto\overline W\oplus Z$. They are also measurable in $y$ for fixed $(b,\overline W,Z)$ by
	measurability of the disintegration, so by Carath\'eodory's
	criterion the graph of $C$ is measurable in
	$\Theta\times L^1$. Likewise $f$ is measurable in $\theta$ for
	fixed $Z$, upper semicontinuous in $Z$ for fixed $\theta$, and
	therefore a normal integrand on the enlarged space. The measurable
	maximum theorem and the Kuratowski--Ryll-Nardzewski theorem, applied
	over $\Theta$ exactly as above, yield a selector
	$\theta\mapsto X^\theta$, jointly measurable on the set of
	parameters at which the maximum is attained. This
	verifies condition~(iii) in its parametrised form.
\end{proof}

For $u\in{\cal T}$, write
\[
H_u^x:=C_u(\eta)(x)={\mathbb E}_{{\mathbb Q}^\infty}\Bigl(\sum_{r=u}^T\eta_r\,\Big|\,\omega_u^\infty=x\Bigr)
\]
for the present value, as of $(u,x)$, of the remaining payment
stream: this is simply $C_u$, already defined in
Definition~\ref{def:classicalContinuationProblem}, applied to the
process $\eta$ in place of a consumption stream. As with $C_u(X)$, $H$
satisfies the recursion $H_u^x=\eta_u+E(H_{u+1}^{x'}\mid{\cal
F}_u^\infty)$.

By Remark~\ref{rem:CRRAFromEZ}, specialising
Definition~\ref{def:EZPreferences} to $\alpha_t^x=\rho_t^x$ gives the
\emph{CRRA comparison recursion}
\begin{gather*}
{\cal J}_t^{\rm CRRA}(X\mid x)
=
(1-\beta_t^x)\frac{X_t^{\rho_t^x}}{\rho_t^x}
+
\beta_t^x
E_x\!\left[{\cal J}_{t+1}^{\rm CRRA}(X\mid x)\,\middle|\,{\cal F}_t\right],
\\
{\cal J}_T^{\rm CRRA}(X\mid x)=\frac{X_T^{\rho_T^x}}{\rho_T^x},
\end{gather*}
using whichever $\rho_t^x,\beta_t^x$ are attached to the preference
system under consideration (Epstein--Zin or entropic): the classical,
discounted, time-additive expected-CRRA-utility recursion, with no
recursive certainty-equivalent distortion. We compactify this
recursion directly, without backward induction, in
Theorem~\ref{thm:CRRACompact} below; every other preference system
considered in this Appendix is then handled by \emph{comparison}
against it.

\subsubsection*{Compactness of the CRRA comparison recursion}

We verify condition~(i) of Definition~\ref{def:compactProblem} for
${\cal J}^{\rm CRRA}$ directly, for the whole horizon at once, working
in a strictly larger feasible set in which the entire present value
$H_s^x$ of future payments, not merely the payment
$\eta_s$ actually banked so far, may be spent freely against the
whole remaining stream. Condition~(i) is then transferred back to the
true, non-anticipating feasible set ${\cal X}_s^x(\beta,\overline X)$
using Lemma~\ref{lemma:closednessContinuationSets}, which was proved
without reference to any particular preference system, so requires no
further work here.

For $r\ge s$ write $\zeta_r^x:=Y_r^x/Y_s^x$ for the state-price
deflator from $s$ to $r$ (Definition~\ref{def:statePriceDensity}), so
that $\zeta_s^x=1$ and, by the tower property of conditional
expectation, $E_x[\zeta_r^xZ\mid{\cal F}_s]$ computes the time-$s$
price of an ${\cal F}_r$-measurable payoff $Z$.

\begin{definition}[Relaxed continuation problem]
\label{def:relaxedContinuationProblem}
For a homogeneous continuation problem $(s,x,\beta,\overline X)$, let
$\widetilde{\cal X}_s^x(\beta,\overline X)$ be the set of streams $X$
satisfying the consistency and non-participation
constraints~\eqref{eq:classicalConsistency}
and~\eqref{eq:classicalNonParticipation}, together with the single
aggregate budget constraint
\[
E_x\!\left[\sum_{r=s}^T\zeta_r^xX_r\,\middle|\,{\cal F}_s\right]
\le
\beta+H_s^x.
\]
\end{definition}

\begin{lemma}[The relaxed problem is a genuine relaxation]
\label{lemma:relaxedContainment}
At every homogeneous state,
${\cal X}_s^x(\beta,\overline X)
\subseteq\hat{\cal X}_s^x(\beta,\overline X)
\subseteq\widetilde{\cal X}_s^x(\beta,\overline X)$.
\end{lemma}

\begin{proof}
The first containment was noted after
Definition~\ref{def:classicalContinuationProblem}. For the second,
let $X\in\hat{\cal X}_s^x(\beta,\overline X)$. The recursion for the
redistributive pot gives
$m_u(X_u)+m_u(\hat R^X_{u+1})-m_u(\eta_u)\le\hat R^X_u$ for
$u\in{\cal T}_{>s}$, and the funding constraint
\eqref{eq:redistributiveBudget} gives the corresponding bound at $u=s$ with
$m_s(\beta)$ in place of $\hat R^X_s$. At a homogeneous state the
operators $m_u$ coincide with $E_x[\,\cdot\mid{\cal F}_u]$ under the
fibrewise convention of this Appendix. Applying
$E_x[\,\cdot\mid{\cal F}_s]$ to each bound, summing over
$u\in{\cal T}_{\ge s}$ and telescoping the pot terms, which are
nonnegative, yields
\[
E_x\!\left[\sum_{r=s}^T\zeta_r^xX_r\,\middle|\,{\cal F}_s\right]
\le
\beta
+
E_x\!\left[\sum_{r=s}^T\zeta_r^x\eta_r\,\middle|\,{\cal F}_s\right]
=
\beta+H_s^x,
\]
where the deflator identity converts conditional expectations under
${\mathbb Q}^\infty_x$ into deflated expectations under
${\mathbb P}_x$. Hence $X\in\widetilde{\cal
X}_s^x(\beta,\overline X)$.
\end{proof}

\begin{lemma}[Value bound and Fatou property for the CRRA recursion]
\label{lemma:relaxedTightness}
Suppose
\[
K_s^x:=E_x\!\left[\sum_{r=s}^T(\zeta_r^x)^{\rho_r^x/(\rho_r^x-1)}\,\middle|\,{\cal F}_s\right]<\infty.
\]
Then:
\begin{enumerate}[(a)]
\item there is a constant $c$, depending only on the exponents, with
\[
{\cal J}_s^{\rm CRRA}(X\mid x)\le\beta+H_s^x+cK_s^x
\qquad
\text{for every }
X\in\widetilde{\cal X}_s^x(\beta,\overline X);
\]
\item whenever $(X^n)_n\subseteq\widetilde{\cal
X}_s^x(\beta,\overline X)$ and $X^n\to X$ almost surely, we have
$X\in\widetilde{\cal X}_s^x(\beta,\overline X)$ and
\[
\limsup_n{\cal J}_s^{\rm CRRA}(X^n\mid x)
\le
{\cal J}_s^{\rm CRRA}(X\mid x).
\]
\end{enumerate}
\end{lemma}

\begin{proof}
(a) Write $q_r:=\rho_r^x/(\rho_r^x-1)$. Since one of the conjugate
exponents $(\rho_r^x,q_r)$ always lies below $1$, Young's inequality
holds in the reversed form
$\zeta_r^xz\ge z^{\rho_r^x}/\rho_r^x+(\zeta_r^x)^{q_r}/q_r$
for every $z\ge0$, that is,
$z^{\rho_r^x}/\rho_r^x\le\zeta_r^xz-(\zeta_r^x)^{q_r}/q_r$. Summing
over $r$, taking $E_x[\,\cdot\mid{\cal F}_s]$, and using the budget
constraint defining $\widetilde{\cal X}_s^x(\beta,\overline X)$
gives the bound.

(b) The budget constraint passes to the almost-sure limit by Fatou's
lemma, and the consistency and non-participation constraints pass
directly, so $X\in\widetilde{\cal X}_s^x(\beta,\overline X)$. For the
utilities we argue date by date; the discount weights are harmless.
At a date with $\rho:=\rho^x_r<0$ the felicity $u(z):=z^{\rho}/\rho$
is nonpositive, and the reverse Fatou lemma gives
$\limsup_nE_x[u(X^n_r)]\le E_x[u(X_r)]$ directly.

At a date with $\rho\in(0,1)$, fix $\lambda>0$ and $M>0$. Young's
inequality in the reversed form, applied with the dual variable
$\lambda\zeta_r^x$, gives
$u(z)\le\lambda\zeta_r^xz+|q_r|^{-1}\lambda^{q_r}(\zeta_r^x)^{q_r}$.
Writing $A^n:=\{u(X^n_r)>M\}$,
\[
E_x[u(X^n_r)]
\le
E_x[u(X^n_r)\wedge M]
+\lambda E_x[\zeta_r^xX^n_r]
+|q_r|^{-1}\lambda^{q_r}E_x\bigl[(\zeta_r^x)^{q_r}\indicator_{A^n}\bigr].
\]
The first term converges to $E_x[u(X_r)\wedge M]\le E_x[u(X_r)]$ by
bounded convergence; the second is at most $\lambda(\beta+H_s^x)$ by
the budget constraint; and since $X^n_r\to X_r$ almost surely,
$\limsup_n\indicator_{A^n}\le\indicator_{\{u(X_r)\ge M\}}$
pointwise, so the reverse Fatou lemma applied to the integrable
dominator $(\zeta_r^x)^{q_r}$ bounds the limsup of the third term by
$|q_r|^{-1}\lambda^{q_r}\,\delta(M)$, where
$\delta(M):=E_x[(\zeta_r^x)^{q_r}\indicator_{\{u(X_r)\ge M\}}]\to0$
as $M\to\infty$ by dominated convergence. Taking $\limsup_n$, then
$M\to\infty$, then $\lambda\to0$, gives
$\limsup_nE_x[u(X^n_r)]\le E_x[u(X_r)]$.
\end{proof}

\begin{lemma}[Attainment and stability for the CRRA recursion]
\label{lemma:CRRAStability}
Suppose $K_s^x<\infty$ at every fibre. Then
$\Prob^{\rm CRRA}$ satisfies conditions (i) and (ii) of
Definition~\ref{def:compactProblem}.
\end{lemma}

\begin{proof}
{\em Condition (i).} By Lemma~\ref{lemma:relaxedContainment} the
redistributively feasible set is contained in the relaxed set, so this is
Lemma~\ref{lemma:relaxedTightness}(b).

{\em Attainment.} Let $(X^n)_n$ be a maximising sequence for a
problem with non-empty constraint set and value exceeding
$-\infty$. The relaxed budget constraint bounds
$(\zeta^x_rX^n_r)_n$ in $L^1$ for each of the finitely many dates
$r$, so Koml\'os' theorem \cite{Komlos1967}, applied successively at each date,
yields a subsequence whose Ces\`aro averages
$\overline X^n$ converge almost surely to some $X^\circ$. The
constraint set is convex so
$\overline X^n\in{\cal X}^x_s(\beta,\overline X)$, and almost-sure
closedness (Lemma~\ref{lemma:closednessContinuationSets}) gives
$X^\circ\in{\cal X}^x_s(\beta,\overline X)$. Concavity of
${\cal J}^{\rm CRRA}_s$ gives
${\cal J}^{\rm CRRA}_s(\overline X^n\mid x)$ at least the
corresponding Ces\`aro average of the utilities, which converges to
the value; the Fatou property (Lemma~\ref{lemma:relaxedTightness}(b))
bounds $\limsup_n{\cal J}^{\rm CRRA}_s(\overline X^n\mid x)$ by
${\cal J}^{\rm CRRA}_s(X^\circ\mid x)$. Hence $X^\circ$ attains the
value.

{\em Stability.} Let $(\beta_n,\eta'_n)\to(\beta,\eta')$ in $L^1$
with optimisers $X^n$ for the perturbed problems whose utilities are
bounded below; the income stream enters every constraint in the same
additive fashion as the budget, so we treat the two perturbations
together and, to lighten notation, suppress $\eta'_n$ below. First,
$V_s(x,\beta_n,\overline X)\to V_s(x,\beta,\overline X)$: the
liminf inequality follows from
Lemma~\ref{lemma:budgetPerturbation}, which uses no compactness,
and for the limsup one applies Koml\'os and the Fatou property to
the Ces\`aro averages of the optimisers exactly as in the
attainment step, the averaged budgets converging to $\beta$.
Attainment at $\beta$ (just proved) provides an optimiser
$X^\circ$; since the felicities are strictly concave and the
constraint sets convex, $X^\circ$ is the unique optimiser. The
midpoints $Z^n:=(X^n+X^\circ)/2$ are feasible for the budgets
$(\beta_n+\beta)/2\to\beta$, so
$\limsup_n{\cal J}^{\rm CRRA}_s(Z^n\mid x)\le\lim_nV_s(x,(\beta_n+\beta)/2,\overline X)=V_s(x,\beta,\overline X)$,
while concavity gives
${\cal J}^{\rm CRRA}_s(Z^n\mid x)\ge\tfrac12{\cal J}^{\rm
CRRA}_s(X^n\mid x)+\tfrac12{\cal J}^{\rm CRRA}_s(X^\circ\mid
x)+G^n$, where $G^n\ge0$ is the integrated concavity gap. Since the
right-hand side without $G^n$ converges to
$V_s(x,\beta,\overline X)$, we conclude $G^n\to0$, and strict
concavity of the felicities then forces $X^n\to X^\circ$ in
probability. Indeed, write
$g(a,b):=u\bigl(\tfrac{a+b}2\bigr)-\tfrac12u(a)-\tfrac12u(b)\ge0$
for the felicity gap at a date $r$, and suppose that, along a
subsequence,
${\mathbb Q}^\infty_x(|X^n_r-X^\circ_r|\ge\epsilon)\ge\delta$ for
some $\epsilon,\delta>0$. Choose a compact $K\subset[0,\infty)$,
contained in $(0,\infty)$ when $\rho_r^x<0$ with
${\mathbb Q}^\infty_x(X^\circ_r\in K)\ge1-\delta/2$. On
$\{|X^n_r-X^\circ_r|\ge\epsilon\}\cap\{X^\circ_r\in K\}$, of
probability at least $\delta/2$, one has $g(X^n_r,X^\circ_r)\ge c$
for a constant $c>0$ depending only on $K$, $\epsilon$ and the
exponent: $g$ is continuous and strictly positive off the diagonal
by strict concavity, so its infimum over
$\{b\in K,\,|a-b|\ge\epsilon,\,a\le A\}$ is positive for each $A$,
while as $a\to\infty$ it tends to $+\infty$ when
$\rho_r^x\in(0,1)$, since $2^{-\rho_r^x}>\tfrac12$, and to
$-u(b)/2$, bounded away from zero on $K$, when $\rho_r^x<0$. Up to
the positive discount weight at the date $r$, this bounds $G^n$
below along the subsequence, contradicting $G^n\to0$.

Finally, an optimiser wastes no resources: by strict monotonicity
of the felicities, every funding inequality binds, so
$E_x[\sum_r\zeta^x_rX^n_r\mid{\cal F}_s]=\beta_n+H^x_s$ and likewise
for $X^\circ$ at $\beta$. Taking expectations, the total prices
converge; together with datewise Fatou this forces
$E_x[\zeta^x_rX^n_r]\to E_x[\zeta^x_rX^\circ_r]$ at every date, and
Scheff\'e's lemma upgrades the convergence in probability of the
nonnegative $\zeta^x_rX^n_r$ to convergence in $L^1$; dividing by
the almost surely positive $\zeta^x_r$ and using the equivalence of
the measures gives $X^n\to X^\circ$ in $L^1$ under
${\mathbb Q}^\infty_x$ as well.
\end{proof}

\begin{theorem}[Compactness of the CRRA comparison recursion]
\label{thm:CRRACompact}
Suppose $K_s^x<\infty$ (Lemma~\ref{lemma:relaxedTightness}) at every
fibre $(s,x)$. Then $\Prob^{\rm CRRA}=({\cal J}^{\rm
CRRA},\tau,\tau_d,\eta)$ is compact. Since $K_s^x$ depends only on
the market and on $(\rho_t^x)$, not on $\eta$, and neither the
containment Lemma~\ref{lemma:relaxedContainment} nor the tightness
Lemma~\ref{lemma:relaxedTightness} uses any specific property of
$\eta$ beyond the standing assumptions, $\Prob^{\rm CRRA}$ is in fact
compact independent of $\eta$
(Definition~\ref{def:compactIndependentEta}).
\end{theorem}

\begin{proof}
Conditions~(i) and~(ii) are Lemma~\ref{lemma:CRRAStability}. For
condition~(iv), Lemma~\ref{lemma:relaxedContainment} and
Lemma~\ref{lemma:relaxedTightness}(a) give
${\cal J}^{\rm CRRA}_s(X\mid x)\le\beta+H^x_s+cK^x_s<\infty$ for
every $X\in\hat{\cal X}^x_s(\beta,\overline X)$. Condition~(iii)
follows from Proposition~\ref{prop:measurabilityFromL1}, whose
attainment hypothesis is the first clause of condition~(ii).
\end{proof}

\subsubsection*{Domination}

\begin{lemma}[Entropic certainty equivalent dominated by expectation]
\label{lemma:CEDomination}
Fix $t\in{\cal T}$ and a.e.\ $x$. For $\theta<0$, the entropic
certainty equivalent $M(Z):=\theta^{-1}\log E[e^{\theta Z}\mid{\cal
F}_t]$ is monotone ($Z\le Z'$ a.s.\ $\Rightarrow$ $M(Z)\le M(Z')$
a.s.) and satisfies $M(Z)\le E[Z\mid{\cal F}_t]$ a.s.
\end{lemma}

\begin{proof}
This is the standard fact that the entropic certainty equivalent is
dominated by the mean, an immediate consequence of Jensen's
inequality applied to the convex function $z\mapsto e^{\theta z}$;
see e.g.\ F\"ollmer and Schied \cite{follmerSchiedBook}. Monotonicity
of $M$ in $Z$ follows in the same way: the exponential is
order-preserving, and the outer transform (dividing by the negative
$\theta$, after taking logarithms) reverses the order a second time.
\end{proof}

\begin{proof}[Domination, Lemma~\ref{lemma:domination}]
\emph{Domination.} By backward induction on $t$. At $t=T$, ${\cal
J}_T(X\mid x)=X_T^{\rho_T^x}/\rho_T^x={\cal J}_T^{\rm CRRA}(X\mid x)$,
equality. Assume ${\cal J}_{t+1}(X\mid x')\le{\cal J}_{t+1}^{\rm
CRRA}(X\mid x')$ for every compatible $x'$. By monotonicity of
$M_t^x$, $M_t^x({\cal J}_{t+1}(X\mid\cdot))\le M_t^x({\cal
J}_{t+1}^{\rm CRRA}(X\mid\cdot))$, and by domination, $M_t^x({\cal
J}_{t+1}^{\rm CRRA}(X\mid\cdot))\le E_x[{\cal J}_{t+1}^{\rm
CRRA}(X\mid\cdot)\mid{\cal F}_t]$. Combining, multiplying by
$\beta_t^x>0$ and adding $(1-\beta_t^x)X_t^{\rho_t^x}/\rho_t^x$ to
both sides gives ${\cal J}_t(X\mid x)\le{\cal J}_t^{\rm CRRA}(X\mid
x)$, closing the induction.

\emph{Transfer.} Condition~(iv) transfers directly: by the
domination just proved,
${\cal J}_s\le{\cal J}^{\rm CRRA}_s$ on every redistributively feasible set, and
the right-hand side is bounded there by
Lemma~\ref{lemma:relaxedContainment} and
Lemma~\ref{lemma:relaxedTightness}(a). The remaining conditions are
not consequences of domination alone; they are verified for the
entropic family in
Theorem~\ref{thm:entropicDominationCompact} and for the Epstein--Zin
family in Theorem~\ref{thm:EZDuality}, by adapting the arguments of
Lemmas~\ref{lemma:relaxedTightness} and~\ref{lemma:CRRAStability} to
the respective recursions.
\end{proof}

\begin{theorem}[Compact entropic preferences via domination]
	\label{thm:entropicDominationCompact}

	Suppose that the preference system is entropic
	(Definition~\ref{def:entropicPreferences}), and that $K_s^x<\infty$
	at every fibre. Then $\Prob=({\cal J},\tau,\tau_d,\eta)$ is compact
	independent of $\eta$ (Definition~\ref{def:compactIndependentEta}):
	the hypothesis $K_s^x<\infty$ does not involve $\eta$, and neither
	does the domination argument of
	Lemma~\ref{lemma:domination}, so the conclusion holds
	for $\Prob^{\rm CRRA}$, and hence for $\Prob$, whatever $\eta$ is.

\end{theorem}

\begin{proof}
The entropic certainty-equivalent operator satisfies
Lemma~\ref{lemma:CEDomination}, with $\theta=\theta_t^x<0$
(Definition~\ref{def:entropicPreferences}), so
Lemma~\ref{lemma:domination} gives condition~(iv).

For condition~(i), argue by backward induction through the
recursion. The current-consumption terms are handled exactly as in
Lemma~\ref{lemma:relaxedTightness}(b), and the certainty equivalent
satisfies the Fatou inequality
$\limsup_nM(Z_n)\le M(Z)$ whenever $\limsup_nZ_n\le Z$ almost
surely: Fatou's lemma applied to the nonnegative integrands
$e^{\theta Z_n}$ gives
$\liminf_nE[e^{\theta Z_n}\mid{\cal F}_t]\ge
E[\liminf_ne^{\theta Z_n}\mid{\cal F}_t]
\ge E[e^{\theta Z}\mid{\cal F}_t]$,
and dividing the logarithm by the negative $\theta$ reverses the
inequality.

For condition~(ii), the entropic functional is concave in $X$: the
certainty equivalent $M$ is concave and monotone for $\theta<0$, so
concavity propagates backward through the recursion, and the
functional is strictly concave in current consumption through the
CRRA felicity terms. The Koml\'os, midpoint and Scheff\'e arguments
of Lemma~\ref{lemma:CRRAStability} then apply verbatim, using
Lemma~\ref{lemma:relaxedTightness}(a) and the domination
${\cal J}\le{\cal J}^{\rm CRRA}$ for the value bound. Condition~(iii)
follows from Proposition~\ref{prop:measurabilityFromL1}.
\end{proof}

\subsubsection*{Epstein--Zin preferences}

For Epstein--Zin, ${\cal J}_t$ is built from the power
certainty-equivalent $M_t^x(Z)=(E_x[Z^{\alpha_t^x}\mid{\cal
F}_t])^{1/\alpha_t^x}$ not additively, but through the CES aggregator
\[
{\cal J}_t(X\mid x)
=
\Bigl((1-\beta_t^x)X_t^{\rho_t^x}+\beta_t^xM_t^x\bigl({\cal
J}_{t+1}(X\mid\cdot)\bigr)^{\rho_t^x}\Bigr)^{1/\rho_t^x}
\]
(Definition~\ref{def:EZPreferences}), homogeneous of degree $1$ in
$X$ rather than degree $\rho_t^x$ like ${\cal J}^{\rm CRRA}$:
Lemma~\ref{lemma:domination} does not apply directly,
since ${\cal J}_t$ and ${\cal J}_t^{\rm CRRA}$ scale differently in
$X$.

When $\rho_t^x>0$ at every fibre, this is repaired by an increasing
reparametrisation matching the two homogeneity degrees. Set
\[
W_t(X\mid x) := \frac{{\cal J}_t(X\mid x)^{\rho_t^x}}{\rho_t^x}.
\]
Raising the CES recursion to the power $\rho_t^x$ and dividing by
$\rho_t^x$ gives
\begin{gather*}
W_t(X\mid x)
=
(1-\beta_t^x)\frac{X_t^{\rho_t^x}}{\rho_t^x}
+
\beta_t^x\widetilde M_t^x\bigl(W_{t+1}(X\mid\cdot)\bigr),
\\
\widetilde M_t^x(W)
:=
\frac1{\rho_t^x}
\Bigl(
E_x\bigl[(\rho_t^xW)^{\alpha_t^x/\rho_t^x}\mid{\cal F}_t\bigr]
\Bigr)^{\rho_t^x/\alpha_t^x},
\end{gather*}
exactly the additively separable form required by
Lemma~\ref{lemma:domination}, with $\widetilde M_t^x$ in
place of $M_t^x$. (The fractional powers are well defined because
$\rho_t^xW_{t+1}={\cal J}_{t+1}^{\rho_t^x}>0$ always, ${\cal J}$ being
positive.)

\begin{lemma}[The reparametrised certainty equivalent is dominated]
\label{lemma:EZReparametrisedDomination}
Suppose $\rho_t^x>0$ and $\alpha_t^x\le\rho_t^x$. Then $\widetilde
M_t^x$ is monotone, and $\widetilde M_t^x(W)\le E_x[W\mid{\cal F}_t]$
a.s., for every $W$ with $\rho_t^xW>0$ a.s.
\end{lemma}

\begin{proof}
Write $V:=\rho_t^xW>0$ and, for $p\ne0$,
$M_p(V):=(E_x[V^p\mid{\cal F}_t])^{1/p}$, the fibrewise power mean of
$V$ with exponent $p$. Power means are monotone increasing in $V$
(pointwise), for every real $p\ne0$: for $p>0$ this is immediate, and
for $p<0$ it follows because $z\mapsto z^p$ and $t\mapsto t^{1/p}$ are
each decreasing, so the composition is increasing. Power means are
also monotone increasing in $p$ itself (the classical power-mean
inequality).

Since $\widetilde M_t^x(W)=(\rho_t^x)^{-1}M_{\alpha_t^x/\rho_t^x}(V)$
and $E_x[W\mid{\cal F}_t]=(\rho_t^x)^{-1}M_1(V)$, and $\rho_t^x>0$,
both claims transfer directly from the corresponding facts about
$V\mapsto M_p(V)$: monotonicity of $\widetilde M_t^x$ in $W$ from
monotonicity of $M_{\alpha_t^x/\rho_t^x}$ in $V$, and $\widetilde
M_t^x(W)\le E_x[W\mid{\cal F}_t]$ from $M_{\alpha_t^x/\rho_t^x}(V)\le
M_1(V)$, valid since $\alpha_t^x/\rho_t^x\le1$ (dividing
$\alpha_t^x\le\rho_t^x$ by the positive $\rho_t^x$).
\end{proof}

\begin{lemma}[Fatou property for the reparametrised recursion]
\label{lemma:EZFatou}
Suppose $\rho_t^x>0$ and $\alpha_t^x\le\rho_t^x$ at every fibre, and
$K_s^x<\infty$. Whenever
$(X^n)_n\subseteq\widetilde{\cal X}_s^x(\beta,\overline X)$ and
$X^n\to X$ almost surely,
\[
\limsup_nW_s(X^n\mid x)\le W_s(X\mid x).
\]
\end{lemma}

\begin{proof}
Write, per fibre, $\theta:=\alpha_t^x/\rho_t^x\le1$,
$u(z):=z^{\rho_t^x}/\rho_t^x$, and
$\Gamma^n_t:=\bigl(W_t(X^n\mid\cdot)-W_t(X\mid\cdot)\bigr)^+$. We
show $E_x[\Gamma^n_t]\to0$ backward in $t$; the inductive step
applied once more, with the conditioning at $x$ in place of
${\cal F}_t$, then gives the claim.

Two preliminary facts. First, the reversed Young inequality
$u(z)\le\lambda\zeta_r^xz+c_\lambda(\zeta_r^x)^{q_r}$, with
$c_\lambda:=\max_r|q_r|^{-1}\lambda^{q_r}$, together with the
domination
$\widetilde M^x_r\le E_x[\,\cdot\mid{\cal F}_r]$
(Lemma~\ref{lemma:EZReparametrisedDomination}), gives by backward
induction the envelope
\begin{gather*}
W_t(X^n\mid\cdot)
\le
\lambda B^n_t+c_\lambda K_t,
\\
B^n_t:=E_x\Bigl[\sum_{r\ge t}\zeta_r^xX^n_r\,\Big|\,{\cal F}_t\Bigr],
\quad
K_t:=E_x\Bigl[\sum_{r\ge t}(\zeta_r^x)^{q_r}\,\Big|\,{\cal F}_t\Bigr],
\end{gather*}
for every $\lambda>0$; by the tower property and the domination the
same bound, with $B^n_t$ and $K_t$ at the earlier date, dominates
$\widetilde M^x_t(W_{t+1}(X^n\mid\cdot))$. Here
$E_x[B^n_t]\le\beta+H^x_s$ by the relaxed budget constraint, and
$E_x[K_t]<\infty$. Hence the families
$\{\widetilde M^x_t(W_{t+1}(X^n\mid\cdot))\}_n$ are uniformly
integrable: their integrals over any event $C$ are at most
$\lambda(\beta+H^x_s)+c_\lambda E_x[K_t\indicator_C]$, and
$\lambda$ was arbitrary. Second, the truncation argument of
Lemma~\ref{lemma:relaxedTightness}(b), applied to the sequence
$u(X^n_r)\vee u(X_r)$, gives
$E_x[(u(X^n_r)-u(X_r))^+]\to0$ at every date $r$.

At the terminal date $\Gamma^n_T$ is a felicity gap, so
$E_x[\Gamma^n_T]\to0$. For the inductive step, the recursion and
subadditivity of the positive part give
\[
\Gamma^n_t
\le
(1-\beta_t^x)\bigl(u(X^n_t)-u(X_t)\bigr)^+
+
\beta_t^x
\bigl(
\widetilde M^x_t(W^n_{t+1})-\widetilde M^x_t(W_{t+1})
\bigr)^+,
\]
writing $W^n_{t+1}:=W_{t+1}(X^n\mid\cdot)$ and
$W_{t+1}:=W_{t+1}(X\mid\cdot)$, so it suffices that the second term
vanish in expectation. We argue along an arbitrary subsequence,
refined so that $G_n:=\rho_t^x\,\Gamma^n_{t+1}$ and
$E_x[G_n\mid{\cal F}_t]$ tend to $0$ almost surely; such a
refinement exists within every subsequence, since
$E_x[G_n]\to0$ by the inductive hypothesis. Write
$V:=\rho_t^xW_{t+1}>0$ and $V^n:=\rho_t^xW^n_{t+1}\le V+G_n$, by
monotonicity of the positive part.

If $\theta\in(0,1]$, subadditivity of $z\mapsto z^\theta$ and the
conditional Jensen inequality give
\[
E_x[(V^n)^\theta\mid{\cal F}_t]
\le
E_x[V^\theta\mid{\cal F}_t]
+\bigl(E_x[G_n\mid{\cal F}_t]\bigr)^\theta
\longrightarrow
E_x[V^\theta\mid{\cal F}_t]
\]
almost surely, and applying the increasing map
$z\mapsto z^{1/\theta}$ gives
\[
\limsup_n\widetilde M^x_t(W^n_{t+1})
\le\widetilde M^x_t(W_{t+1})
\]
almost surely.

If $\theta<0$, then $(V^n)^\theta\ge(V+G_n)^\theta\to V^\theta$
almost surely, so the conditional Fatou lemma applied to the
nonnegative $(V^n)^\theta$ gives
$\liminf_nE_x[(V^n)^\theta\mid{\cal F}_t]
\ge E_x[V^\theta\mid{\cal F}_t]$
almost surely, and applying the decreasing map
$z\mapsto z^{1/\theta}$ gives the same conclusion.

In either case
$(\widetilde M^x_t(W^n_{t+1})-\widetilde M^x_t(W_{t+1}))^+\to0$
almost surely along the refined subsequence, and these positive
parts are dominated by the uniformly integrable
$\widetilde M^x_t(W^n_{t+1})$, so they tend to $0$ in expectation.
As the subsequence was arbitrary, the full sequence converges,
completing the induction.
\end{proof}

\begin{theorem}[Compact Epstein--Zin preferences via domination]
	\label{thm:EZDuality}
	Suppose that the preference system is Epstein--Zin
	(Definition~\ref{def:EZPreferences}), that $\rho_t^x>0$ and
	$\alpha_t^x\le\rho_t^x$ at every fibre, and that $K_s^x<\infty$
	(Lemma~\ref{lemma:relaxedTightness}) at every fibre. Then $\Prob$ is
	compact independent of $\eta$
	(Definition~\ref{def:compactIndependentEta}): the hypotheses
	involve only the market and $(\alpha_t^x,\rho_t^x)$, and the proof
	below uses nothing about $\eta$ beyond the standing assumptions, at
	every date.
\end{theorem}

\begin{proof}
The map $z\mapsto(\rho_s^xz)^{1/\rho_s^x}$ is increasing and
continuous on the relevant range, so ${\cal J}$ and $W$ define the
same optimisers on every constraint set, the same suprema up to
this reparametrisation, and each of conditions (i)--(iii) of
Definition~\ref{def:compactProblem} holds for $\Prob$ if and only
if it holds for $\widehat\Prob:=(W,\tau,\tau_d,\eta)$.

For $\widehat\Prob$, condition (i) holds by
Lemma~\ref{lemma:EZFatou}, the redistributively feasible sets lying in the
relaxed sets by Lemma~\ref{lemma:relaxedContainment}, and condition
(ii) is verified exactly as in Lemma~\ref{lemma:CRRAStability}:
monotonicity and domination of $\widetilde M^x_t$
(Lemma~\ref{lemma:EZReparametrisedDomination}) give the value bound
$W_s\le{\cal J}^{\rm CRRA}_s\le\beta+H^x_s+cK^x_s$ on the relaxed
set, Lemma~\ref{lemma:EZFatou} replaces
Lemma~\ref{lemma:relaxedTightness}(b) throughout, and $W$ is
concave in $X$ for $\rho_t^x\in(0,1)$ and $\alpha_t^x\le\rho_t^x$
\cite{epsteinZin1}, strictly so in current consumption, so the
Koml\'os, midpoint and Scheff\'e arguments apply verbatim.

For condition~(iv), the domination
$W_t\le{\cal J}^{\rm CRRA}_t$ gives
${\cal J}_s(X\mid x)=(\rho_s^xW_s(X\mid x))^{1/\rho_s^x}
\le(\rho_s^x{\cal J}^{\rm CRRA}_s(X\mid x))^{1/\rho_s^x}<\infty$
on every redistributively feasible set. Condition~(iii)
follows from Proposition~\ref{prop:measurabilityFromL1}.
\end{proof}

\begin{remark}[The case $\rho_t^x<0$]
When $\rho_t^x<0$, ${\cal J}_t>0$ forces $W_t<0$ always, so the
domination $W_t\le{\cal J}_t^{\rm CRRA}$ carries no information: a
negative quantity is automatically bounded above by the nonnegative
${\cal J}_t^{\rm CRRA}$. We do not know whether $\Prob$ is compact in
this case.
\end{remark}

\subsubsection*{Domination via the market}

The statements of Definition~\ref{def:marketDomination},
Proposition~\ref{prop:marketDomination} and
Corollary~\ref{cor:marketDomination}, and their interpretation, are
given in the main text. Here we record the connection to our own
construction and the deferred proofs.

For $s<r$, write $G_r^x$ for the quantile function of
$\zeta_r^x=Y_r^x/Y_s^x$ under ${\mathbb P}_x$
(Definition~\ref{def:statePriceDensity}): the increasing generalised
inverse of its distribution function, so that
$\zeta_r^x\stackrel{d}{=}G_r^x(U)$ for $U\sim\mathrm{Unif}[0,1]$
(Definition~\ref{def:marketDomination}), and
\[
K_s^x
=
E_x\!\left[\sum_{r=s}^T(\zeta_r^x)^{q_r}\,\middle|\,{\cal F}_s\right]
=
\sum_{r=s}^T\int_0^1G_r^x(x)^{q_r}\,\ed x,
\qquad
q_r:=\frac{\rho_r^x}{\rho_r^x-1}.
\]
So $K_s^x<\infty$ is a statement purely about the growth of each
$G_r^x$, and we may compare markets by comparing this growth.

\begin{proof}[Market domination transfers the moment condition, Proposition~\ref{prop:marketDomination}]
By Definition~\ref{def:marketDomination}, there are $\delta\in(0,1)$
and $K<\infty$ with $G_1(x)\le KG_2(x)$ for $x\in[1-\delta,1)$. Since
$G_1$ is increasing and finite a.e., it is bounded on $[0,1-\delta]$,
say by $M<\infty$. Then, using $q>0$ so that $z\mapsto z^q$ preserves
the inequality $G_1\le KG_2$ on $[1-\delta,1)$,
\begin{align*}
\int_0^1G_1(x)^q\,\ed x
&=
\int_0^{1-\delta}G_1(x)^q\,\ed x
+
\int_{1-\delta}^1G_1(x)^q\,\ed x
\\
&\le
(1-\delta)M^q+K^q\int_0^1G_2(x)^q\,\ed x
<
\infty.
\qedhere
\end{align*}
\end{proof}

The restriction $q>0$ corresponds to $\rho_r^x<0$. For
$0<\rho_r^x<1$, so $q_r<0$, the regime needed for
Theorem~\ref{thm:EZDuality}, the
mirror statement holds for the lower tail $x\to0^+$, via the
following.

\begin{corollary}[Market domination transfers the moment condition: lower tail]
\label{cor:marketDominationLowerTail}
For $i=1,2$ let $\hat G_i(x):=1/G_i(1-x)$, the quantile function of
$1/\zeta_i$ where $\zeta_i\stackrel{d}{=}G_i(U)$. Suppose $\hat
G_1\preceq\hat G_2$ (Definition~\ref{def:marketDomination}) and
$q<0$. Then $\int_0^1G_2(x)^q\,\ed x<\infty$ implies
$\int_0^1G_1(x)^q\,\ed x<\infty$.
\end{corollary}

\begin{proof}
Substituting $x\mapsto1-x$ and using $G_i(1-x)=1/\hat G_i(x)$,
\[
\int_0^1G_i(x)^q\,\ed x
=
\int_0^1G_i(1-x)^q\,\ed x
=
\int_0^1\hat G_i(x)^{-q}\,\ed x,
\]
and $-q>0$, so this is exactly the integral of
Proposition~\ref{prop:marketDomination} applied to $\hat G_1,\hat
G_2,-q$ in place of $G_1,G_2,q$.
\end{proof}

\begin{proof}[Compactness transfers from a dominating market, Corollary~\ref{cor:marketDomination}]
Immediate from Proposition~\ref{prop:marketDomination} (if
$\rho_r^x<0$) or Corollary~\ref{cor:marketDominationLowerTail} (if
$0<\rho_r^x<1$), applied at each date $r$, and the identity for
$K_s^x$ above.
\end{proof}

\end{document}